\documentclass[a4paper]{article}
\pdfoutput=1 

\usepackage{biblatex}
\usepackage{appendix}
\usepackage{amsmath}
\usepackage{authblk}

\usepackage{appendix}
\usepackage[group-separator={,}]{siunitx}

\usepackage[a4paper,top=1.96cm,bottom=1.96cm,left=3cm,right=3cm,marginparwidth=1.75cm]{geometry}

\usepackage{amsfonts}
\usepackage{amssymb}
\usepackage{amsthm} % For theorems
\usepackage{dsfont}
\usepackage{centernot}
\usepackage{graphicx}
\usepackage{mathtools}
\usepackage[colorinlistoftodos]{todonotes}
\usepackage{mathrsfs}
\usepackage{pifont}
\newcommand\followUpInFuture[1][ ]{{}}

\newcommand\introsection[1]{\subsection{#1}}

\newcommand\Independence{Irreducibility}
\newcommand\independence{irreducibility}

\newcommand\independent{irreducible}

\newcommand\epsEuc{\varepsilon}
\newcommand\epsLor{\epsilon}
\newcommand\epsShort{\epsLor}

\newtheorem{theorem}{Theorem}[section]

\theoremstyle{definition}
\newtheorem{corollary}[theorem]{Corollary} %theorem not section?
\newtheorem{lemma}[theorem]{Lemma} %theorem not section?

\theoremstyle{definition}
\newtheorem{definition}[theorem]{Definition}

\theoremstyle{definition}

\theoremstyle{definition}

\theoremstyle{remark}
\newtheorem*{remark}{Remark}

\newenvironment{subproof}[1][\proofname]{%
  \begin{proof}[#1]%
}{%
  \end{proof}%
}

\DeclareMathOperator{\Tr}{Tr}
 
\DeclareMathOperator{\sgn}{sign}

\newcommand\hide[1]{{}}
\newcommand\comCGL[1]{{{\color{red} [CGL: #1]}}}

 \newcommand\fourvec[4]{{\left(\begin{array}{c} #1\\ #2\\ #3\\ #4\end{array}\right)}}
 \newcommand\fourvecT[4]{{\left(#1,\left( #2, #3, #4\right)\right)}}
 \newcommand\fourvecMassForm[4]{{
  \left[#1;(#2, #3, #4)\right]
 }}

\DeclareMathOperator{\Aut}{Aut}
 
 \newcommand\equivdef{ \underset {\text{def}} \equiv}
 
  \DeclareRobustCommand{\gyzpi}{%
  \left(
  {\overset{\pi}{\underset {yz}{\curvearrowleft}} }
  \right)
  }
  
  \DeclareRobustCommand{\gaxrot}[1]{%
  \left(
  {\overset{{#1}}{\underset {xy}{\curvearrowleft}} }
  \right)
  }

\DeclareMathOperator{\TransversePlaneOp}{TP}

\newcommand\nicety[3]{\left[\begin{array}{cc|ccc}
     m_p & m_q & m_{#1} & m_{#2} & m_{#3} \\
     0 & 0 & \mathbf{#1} e^{i\theta} & \mathbf{#2} e^{i\theta} & \mathbf{#3} e^{i\theta} \\
     p & -p & {#1}_z & {#2}_z & {#3}_z
\end{array}\right]}

\newcommand\oddity[3]{\left[\begin{array}{cc|ccc}
     m_q & m_p & m_{#1} & m_{#2} & m_{#3} \\
     0 & 0 & \mathbf{#1}^* e^{i\theta} & \mathbf{#2}^* e^{i\theta} & \mathbf{#3}^* e^{i\theta} \\
     p & -p & -{#1}_z & -{#2}_z & -{#3}_z
\end{array}\right]}

\newcommand\polarnicety[6]{\left[\begin{array}{cc|ccc}
     m_p & m_q & m_{#1} & m_{#2} & m_{#3} \\
     0 & 0 & {|\mathbf{#1}|} e^{i(\theta+{#4})} & {|\mathbf{#2}|} e^{i(\theta+{#5})} & {|\mathbf{#3}|} e^{i(\theta+{#6})} \\
     p & -p & {#1}_z & {#2}_z & {#3}_z
\end{array}\right]}

\newcommand\polaroddity[6]{\left[\begin{array}{cc|ccc}
     m_q & m_p & m_{#1} & m_{#2} & m_{#3} \\
     0 & 0 & {|\mathbf{#1}|} e^{i(\theta-{#4})} & {|\mathbf{#2}|} e^{i(\theta-{#5})} & {|\mathbf{#3}|} e^{i(\theta-{#6})} \\
     p & -p & -{#1}_z & -{#2}_z & -{#3}_z
\end{array}\right]}

\newcommand\GRAMTWO[4]{
G\!\left(\!\!\!\!
\begin{array}{c@{,\ \!}c}
     {#1} & {#2} \\
     {#3} & {#4}
\end{array}\!\!\!\!
\right)
}
\newcommand\GRAMTHR[6]{
G\!\left(\!\!\!\!
\begin{array}{c@{,\ }c@{,\ }c}
     {#1} & {#2} & {#3} \\
     {#4} & {#5} & {#6}
\end{array}\!\!\!\!
\right)
}

\newcommand\PerpABFE[4]{
{\GRAMTWO {#1} {#3}
      {#2} {#3}} #4 0
}
\newcommand\PerpABF[3]{
\PerpABFE {#1} {#2} {#3} =
}

\newcommand\SYMGRAMTWO[2]{
\Delta_2\left(
{#1}, {#2}
\right)
}

\newcommand\SYMGRAMTHR[3]{
\Delta_3\left(
{#1}, {#2}, {#3}
\right)
}

\newcommand{\tmark}{\text{\ding{51}}}%
\newcommand{\xmark}{\text{\ding{55}}}%

\newcommand\numberthis{\addtocounter{equation}{1}\tag{\theequation}}

\newcommand\False{\text{FALSE}}
\newcommand\True{\text{TRUE}}
\newcommand\isTrue[1]{\overbrace{{#1}}^{\True}}
\newcommand\isFalse[1]{\overbrace{{#1}}^{\False}}

\newcommand\PCom{\mathscr{P}^{{abc}}_{pq}(\mathbb{C})}

\newcommand\POne{\mathscr{P}^{abc}_{pq}(1)}
\newcommand\PTwo{\mathscr{P}^{abc}_{pq}(2)}
\newcommand\PThr{\mathscr{P}_{pq}(3)}

\newcommand\FZerFIXEDWIDTH{F^{\phantom{abc}}_{pq}(0)}
\newcommand\FOneFIXEDWIDTH{F^{\phantom{abc}}_{pq}(1)}
\newcommand\FTwoFIXEDWIDTH{F^{\phantom{abc}}_{pq}(2)}
\newcommand\FThrFIXEDWIDTH{F^{{abc}}_{pq}(3)}

\newcommand\FZer{F_{pq}(0)}
\newcommand\FOne{F_{pq}(1)}
\newcommand\FTwo{F_{pq}(2)}
\newcommand\FThr{F^{abc}_{pq}(3)}

\newcommand\FZerOneTwo{
F^{\cdot\cdot\cdot}_{pq}(\!\!\!\!\!
\text{ \tiny{
 $\begin{array}{c}0\\1\\2\end{array}$
 }
 }\!\!\!\!)}
 
 \newcommand\PComU[1]{{\left(\POne\right)^3 + i {#1} \PTwo }}
 \newcommand\FComU[1]{{\FOne + i {#1} \FTwo }}

 \newcommand\PComUS[1]{{
 \mathscr{P}^{abc}_{pq}(\mathbb{C}(#1))
 }}
 \newcommand\FComUS[1]{{
 F_{pq}(\mathbb{C}(#1))
 }}

\newcommand\cFingerConstraintABCPreFinal[3]{
(p^2=q^2) 
\land
({#1}^2={#2}^2)
\land
\left(\PerpABF {#3} {p-q} {p+q}\right)
\land
\left(\PerpABF {{#1}+{#2}} {p-q} {p+q}\right)
\land
\left(\SYMGRAMTHR {{#1}-{#2}} p q =0\right)
}

\newcommand\smallcFingerConstraintABCAOEN[7]{
(p^2 #6 q^2) 
#4
({#1}^2  #6 {#2}^2)
#4
\left(\PerpABFE {{#1}+{#2}} {p-q} {p+q} #6 \right)
#4
\left(\SYMGRAMTHR {{#1}-{#2}} p q  #6  0\right)
}

\newcommand\cFingerConstraintABCAOEN[7]{
(p^2 #6 q^2) 
#4
({#1}^2  #6 {#2}^2)
#4
\left(\PerpABFE {a+b+c} {p-q} {p+q} #6\right)
#4
\left(\PerpABFE {{#1}+{#2}} {p-q} {p+q} #6 \right)
#4
\left(\SYMGRAMTHR {{#1}-{#2}} p q  #6  0\right)
}

\newcommand\cFingerConstraintABC[3]{
\cFingerConstraintABCAOEN #1 #2 #3 \land \lor = \ne
}

\newcommand\newaddonEN[4]{
{#1} #4 {#2}
}
\newcommand\theaddonEN[4]{
\newaddonEN{#1}{#2} #3 #4
}
\newcommand\theaddon[2]{
\newaddonEN{#1}{#2} = \ne
}

\newcommand\easyPartOfAxialCatapultConstraintABCVersionA[3]{
({#1}^2={#2}^2)
\land
\left( \PerpABF{{#1}-{#2}}{p-q}{p+q} \right)
\land
\left(
\SYMGRAMTWO {#1} {p+q}
=
\SYMGRAMTWO {#2} {p+q}
\right)
\land
\left(
\SYMGRAMTHR {#1} p q
{\color{blue}>0}
\right)
}

\newcommand\easyPartOfAxialCatapultConstraintABCVersionB[3]{
({#1}^2={#2}^2)
\land
\left( \PerpABF{{#1}-{#2}}{p-q}{p+q} \right)
\land
\left(
\GRAMTWO {{#1}-{#2}} {p+q}
         {{#1}+{#2}} {p+q} =0
\right)
\land
\left(
\SYMGRAMTHR {#1} p q
{\color{blue}>0}
\right)
}

\newcommand\easyPartOfAxialCatapultConstraintABCWithoutOffAxisAOEN[7]{
({#1}^2 #6 {#2}^2)
#4
\left( \PerpABFE{{#1}-{#2}}{p-q}{p+q} #6 \right)
#4
\left(
\GRAMTWO {{#1}-{#2}} {p+q}
         {{#1}+{#2}} {p+q} #6 0
\right)
}

\newcommand\axialCatapultConstraintABCPreFinal[3]{
\easyPartOfAxialCatapultConstraintABCVersionB{#1}{#2}{#3}
\qquad
\\
\land
\left(
\PerpABF{{#1}-{#2}}{a+b+c}{p+q}
\right)
\land
\left(
\theaddon{#1}{#2}
\right)
}

\newcommand\smallaxialCatapultConstraintABCAOEN[7]{
\easyPartOfAxialCatapultConstraintABCWithoutOffAxisAOEN{#1}{#2}{#3} #4 #5 #6 #7
#4
\left(
\theaddonEN{#1}{#2} #6 #7 
\right)
}

\newcommand\axialCatapultConstraintABCAOEN[7]{
\easyPartOfAxialCatapultConstraintABCWithoutOffAxisAOEN{#1}{#2}{#3} #4 #5 #6 #7
#4
\left(
\PerpABFE{{#1}-{#2}}{a+b+c}{p+q} #6
\right)
#4
\left(
\theaddonEN{#1}{#2} #6 #7 
\right)
}

\newcommand\axialCatapultConstraintABC[3]{
\axialCatapultConstraintABCAOEN {#1} {#2} {#3} \land \lor = \ne
}

\newcommand\pqabColEvent{\ensuremath{\Omega_{(pq)\rightarrow (ab)X}}}

\newcommand\pqabcColEvent{\ensuremath{\Omega_{(pq)\rightarrow (abc)X}}}

\newcommand\pqabAllEvent{\ensuremath{\Omega_{(pq)(ab)X}}}

\newcommand\pqabcAllEvent{\ensuremath{\Omega_{(pq)(abc)X}}}

\newcommand\chiral{\aleph}
\newcommand\nonchiral{\heartsuit}

\newcommand\tinyseparatorlor[1]{
\raisebox{0.8mm}{\rule{0.15\textwidth}{0.10mm}}
#1
\raisebox{0.8mm}{\rule{0.15\textwidth}{0.10mm}}
}
\newcommand\smallseparatorlor[1]{
\raisebox{0.8mm}{\rule{0.35\textwidth}{0.10mm}}
#1
\raisebox{0.8mm}{\rule{0.35\textwidth}{0.10mm}}
}
\newcommand\separatorlor[1]{
\raisebox{0.8mm}{\rule{0.45\textwidth}{0.20mm}}
#1
\raisebox{0.8mm}{\rule{0.45\textwidth}{0.20mm}}
}
\def\toperator{
\raisebox{0.8mm}{\rule{0.99\textwidth}{0.5mm}}
}

\newcommand\toparotor[1]{
\raisebox{0.8mm}{\rule{#1}{0.5mm}}
}
\newcommand\separotor[2]{
\raisebox{0.8mm}{\rule{{#2}}{0.20mm}}
#1
\raisebox{0.8mm}{\rule{{#2}}{0.20mm}}
}

\newcommand\bignonchiralcondAOEN[4]
{
\toperator \\
( 
   ([a,b,p,q] #3 0)
   #1
   ([b,c,p,q] #3 0)
   #1
   ([c,a,p,q] #3 0) 
)
\\ \separatorlor #2 \nonumber \\
\axialCatapultConstraintABCAOEN a b    c   #1 #2 #3 #4
\\ \smallseparatorlor #2 \nonumber \\
\axialCatapultConstraintABCAOEN b c    a   #1 #2 #3 #4
\\ \smallseparatorlor #2 \nonumber \\
\axialCatapultConstraintABCAOEN c a    b   #1 #2 #3 #4
\\ \separatorlor #2 \nonumber \\
  (p^2 #3 q^2)
  #1
  \left(\PerpABFE {a} {p-q} {p+q} #3 \right)
  #1
  \left(\PerpABFE {b} {p-q} {p+q} #3 \right)
  #1
  \left(\PerpABFE {c} {p-q} {p+q} #3 \right)
\\ \separatorlor #2 \nonumber \\
  \cFingerConstraintABCAOEN a b     c  #1 #2 #3 #4
\\ \smallseparatorlor #2 \nonumber \\
  \cFingerConstraintABCAOEN b c     a  #1 #2 #3 #4
\\ \smallseparatorlor #2 \nonumber \\
  \cFingerConstraintABCAOEN c a     b  #1 #2 #3 #4
\\
\toperator
}

\newcommand\smallnonchiralcondAOEN[4]
{
\toperator \\
( 
   ([a,b,p,q] #3 0)
)
\\ \separatorlor #2 \nonumber \\
\smallaxialCatapultConstraintABCAOEN a b    c   #1 #2 #3 #4
\\ \separatorlor #2 \nonumber \\
  (p^2 #3 q^2)
  #1
  \left(\PerpABFE {a} {p-q} {p+q} #3 \right)
  #1
  \left(\PerpABFE {b} {p-q} {p+q} #3 \right)
\\ \separatorlor #2 \nonumber \\
  \smallcFingerConstraintABCAOEN a b     c  #1 #2 #3 #4
\\
\toperator
}

\newcommand\bignonchiralcond{
\bignonchiralcondAOEN \land \lor = \ne
}

\newcommand\ANTIbignonchiralcond{
\bignonchiralcondAOEN \lor \land \ne =
}

\newcommand\ANTIsmallnonchiralcond{
\smallnonchiralcondAOEN \lor \land \ne =
}

\newcommand\tritribox[9]{
\left\{
\begin{array}{ccc}
     {#1} & {#2} & {#3} \\
     \cline{1-3} 
     {#4} & {#5} & {#6} \\
     \cline{1-3} 
     {#7} & {#8} & {#9}
     \end{array}
     \right\}
}

\newcommand\fourthreebox[3]{
\left\{
\begin{array}{cccc}
     #1  \\
     \cline{1-4} 
     #2 \\
     \cline{1-4} 
     #3
     \end{array}
     \right\}
}

\newcommand\twothreebox[3]{
\left\{
\begin{array}{cc}
     #1  \\
     \cline{1-2} 
     #2 \\
     \cline{1-2} 
     #3
     \end{array}
     \right\}
}

\newcommand\nibbledchocolatewafer{
\left[
\begin{array}{c}
\toparotor {0.5\textwidth}
\\
        (\POne=0)
        \land
        (\PTwo=0)
        \land
        (\PThr\ne0)
\\
\separotor \land {0.22\textwidth}
\\
        (\FTwo=0)
        \land
        (\FThr\ne0)
\\
\toparotor {0.5\textwidth}
\end{array}
\right]
}

\newcommand\timesavingsandwichfilling{\left\{
\begin{array}{c}
(g^{a-b}_{a+b}\ne0) \lor (g^{a-b}_{a+b+c}\ne0)
\lor
(f^{ab}\ne0) \lor (g^{a-b}_{p-q}\ne0) \\
\separotor \land {0.22\textwidth}\\
(g^{b-c}_{b+c}\ne0) \lor (g^{b-c}_{a+b+c}\ne0)
\lor
(f^{bc}\ne0) \lor (g^{b-c}_{p-q}\ne0) \\
\separotor \land {0.22\textwidth}\\
(g^{c-a}_{c+a}\ne0) \lor (g^{c-a}_{a+b+c}\ne0)
\lor
(f^{ca}\ne0) \lor (g^{c-a}_{p-q}\ne0) 
\end{array}
\right\}}

\newcommand\mirrorcondright[3]{
\phantom{#3}  #2  {#1}   #2  {#3}
}
\newcommand\mirrorcondleft[3]{
{#3}  #2  {#1}   #2  \phantom{#3}
}

\title{Lorentz and permutation invariants of particles III:\\
constraining non-standard sources of 
parity violation}
\author[1]{Christopher G.~Lester%\thanks{lester@hep.phy.cam.ac.uk}
}
\author[1]{Ward Haddadin%\thanks{C.C@university.edu}
}
\author[1]{Ben Gripaios%\thanks{B.B@university.edu}
}
\affil[1]{Cavendish Laboratory, University of Cambridge}
\usepackage{hyperref}

\begin{document}
\maketitle

%\comCGL{XXXXX coefficients in the epsilons}
 
\begin{abstract}
Comparisons of the positive and negative halves of the distributions of parity-odd event variables in particle-physics experimental data can provide sensitivity to sources of non-standard parity violation.  Such techniques benefit from lacking first-order dependence on simulations or theoretical models, but have hitherto lacked systematic means of enumerating all discoverable signals.  To address  that issue this paper seeks to construct sets of parity-odd event variables which may be proved to be able to reveal the existence of \textit{any} Lorentz-invariant source of  non-standard parity violation which could be visible in data consisting of groups of real non space-like four-momenta exhibiting certain permutation symmetries. 
\end{abstract}

\clearpage
\tableofcontents

%%%%%%%%%%%%%%%%%%%%
%%%%%%%%%%%%%%%%%%%%
%%%%%%%%%%%%%%%%%%%%

\clearpage

\clearpage

\section{Preamble concerning the scope and aims of this paper}

\begin{itemize}
\item
This paper's main purpose is to derive two sets of `event variables' satisfying particular requirements which, the authors believe, could make then useful to collider experiments seeking to look for non-standard sources of parity violation.
\item
The particular properties we demand for these sets of event variables are set out in Section~\ref{sec:ourobjectives}.
\item
Although the paper's introduction provides arguments explaining  \textbf{why} the demanded particular properties might expected to be useful, the primary purpose of the paper is simply to derive a set of event variables \textbf{satisfying} the stated properties.  No evidence of the utility of these event variables (or the lack thereof) is presented here.\footnote{Such evidence will be for later papers.}  \textit{The purpose of this paper is only to provide falsifiable arguments in support of our claim that the two  sets of event variables we derive do indeed have the properties which were demanded of them.}
\item
\textbf{One set of event variables} is designed to analyse collisions resulting in \textbf{two} main final-state particles which are indistinguishable except by their four-momenta (e.g.~these could be dijet events). This set is found to comprises three event variables ($X_1$, $X_2$ and $X_3$) which are defined in the steps leading to equations \eqref{eq:unhattedx1} to \eqref{eq:unhattedx3}.
\item
\textbf{The second set of event variables} is designed to analyse collisions resulting in \textbf{three} main final-state particles which are indistinguishable except by their four-momenta.  This set is found to comprises nineteen event variables ($V_1$ to $V_{19}$) contained within the four sub-sets shown in \eqref{eq:s19ALL}.
\item
The `take home' messages for a physicist reading this paper are as follows:
\begin{enumerate}
\item
\textit{That if he or she has a collider with unpolarised beams, and has a detector which is blind to polarisations or flavours but not four-momenta of detected particles, and has access to all the events having three or fewer particles in the final state, then every single possible source of non-standard parity-violation, no matter how bizarre, provided that it discoverable `at all' within that dataset, would generate an observable asymmetry in at least one of the event variables derived in this paper, given enough luminosity, and given an appropriate (parity-even) event selection.\footnote{The comment about the need for `an appropriate (parity-even) event selection' refers forward to remarks made in Corollary~\ref{cor:theneedforparityeveneventselectionsaswellaspoddvars}. }}
\item
More succinctly: he or she could search model independently for all potential sources of non-standard parity violation by looking for asymmetries in only the event variables described in this paper.  These variables `cover all bases'.
\item
Computer code to evaluate the event variables may be found as ancillary support files attached to \cite{Lester:2020jrg}.
\item
Not all event variables will be relevant to all colliders.   For example: some event variables collapse to zero if the two beams of the collider have particles of equal mass\footnote{For example: the two proton beams of the Large Hadron Collider would satisfy this constraint while the proton-electron beams of HERA would not.}, while others collapse to zero if final state particles are treated as massless or degenerate in mass.  In most analyses, therefore, the number of variables needing to be considered would be much smaller than the above counts of `three' and `nineteen' might suggest.
\item
If the initial or final state particles could be distinguished by more than just their four-moment (e.g.~if flavour or helicity information were available) then the variables derived herein would no longer be sufficient to catch all forms of parity violation which would then be (in principle) observable. Nonetheless, the procedures used to generate the event variables in this paper could be adapted to the new symmetries and new variables generated which would have the desired coverage.
\end{enumerate}
\item
Some results are have been generalised to situations of relevance beyond colliders\footnote{For example, \eqref{eq:s28ALL} provides a set of twenty-eight variables which could be used in dark matter detection experiments with five particles in the final state.}, but these are not of primary interest to the paper.
\end{itemize}

\clearpage
\section{Introduction}

\introsection{Non-standard parity violation, and why searches for it are important}

It is rightly beyond doubt that the laws of physics violate parity.  The elegant experiments of the 1950s \cite{Wu:1957my,Garwin:1957hc,PhysRev.109.1015} unambiguously showed that the weak interaction of the Standard Model can tell the difference between our universe and its mirror image.  Nonetheless, only a single attempt \cite{Lester:2019bso} has yet been made to look in  Large Hadron Collider (LHC) data to see whether they provide  evidence of non-standard parity violating mechanisms which operate only at high energies probed by that machine.
As one of us has already noted in the introduction of
\cite{Lester:2019bso}, the scarcity of tests of non-standard sources of parity violation at these new energies:
\begin{quote}
\ldots is largely a pragmatic response to the obstacles presented by the LHC:  its beams are not polarised,  its detectors are not sensitive to polarizations, and it is mathematically impossible to construct a parity violating spin-averaged matrix element within any $CP$-conserving Locally Lorentz-Invariant quantum field theory (LLIQFT), effective or otherwise.\footnote{To reveal the existence of parity-violation a model must possess at least one matrix element having both a parity-even and a parity-odd part.
After trace identities have removed spinor sums, the only parity-odd expressions which can remain in a Lorentz-invariant $|M|^2$ are contractions of the totally antisymmetric alternating tensor with groups of four linearly independent four-momenta: $\epsilon_{\mu\nu\sigma\tau}a^\mu b^\nu c^\sigma d^\tau$.  \label{fn:epsilon} While such terms are parity-odd, they are also time-odd. 
Assuming $CPT$-symmetry, such a matrix element therefore also violates $CP$.  $CP$-conserving local Lorentz-invariant quantum field theories therefore cannot generate  parity violating differential cross sections. 
}
Given that the only route to probing parity-violation in the Standard Model at the LHC would be from within the $C$-even part of its (very small) $CP$-violating sector,\footnote{One could, in principle, demonstrate parity-violation unambiguously by using a `genuine $CP$-odd' observable (such as one of those described in \cite{Han:2009ra})
 on a $C$-even final state.} it is not surprising that no such analyses have yet been performed.\footnote{This limitation does not prevented the LHC from making measurements of parity-violating parameters within models {\it in which a particular mechanism of parity-violation is present by assumption}.
For example, the differences between the axial and vector couplings of the $Z$-boson in the Standard Model violate parity and were  measured in \cite{Aad:2015uau,Khachatryan:2016yte}.  However, neither of these papers incontrovertibly demonstrates that nature violates parity

The reason is simple:  the angles from which forward-backward asymmetries are calculated are even under parity,  unlike primary observables from the experiments of the 1950s.   
 The very same forward backward asymmetries could therefore also be explained, at least in principle, by some alternative parity {\it conserving} theory.
}
\end{quote}
It would be wrong, however, to conclude that genuine tests of parity violation are therefore of limited value.
On the contrary:
a purely data-driven search for non-standard sources of parity violation was proposed and tested \cite{Lester:2019bso} on CMS Open Data \cite{CMS:QCDData} and demonstrated that genuine tests of parity-invariance are very straightforward to make at the LHC and can (in principle) provide very large signatures for models which are either (a) strongly $CP$-violating LLIQFTs, or (b) not even LLIQFTs at all!
While it is true that there is an overwhelming theoretical preference for LLIQFTs (the Standard Model itself is one, as are most popular extensions including those featuring supersymmetry, leptoquarks, technicolor, axions, additional gauge interactions, {\it etc}.) there is no law of nature which {\it  demands} that  new physics be describable only by such theories.
Moreover, given the lack of evidence for new physics found at LHC thus far, the need for the community to search in all possible hiding places is surely greater than ever.  In particular, it is hard to imagine any reason why every possible attempt should not be made to test and re-test the fundamental symmetries of nature every time a door opens onto a new energy range.

We argue, therefore, that it is important to have tools which allow us to systematically discover and/or exclude as-yet unconstrained source of parity violation both in non-LLIQFTs and in strongly $CP$-violating LLIQFTs.  Among those, we choose to restrict ourselves to considering only  Lorentz-invariant theories, however there is no requirement that others need make the same choice.

\introsection{The lack of frameworks for structuring such searches}
Having established the need for such searches, the next concern is whether a framework  already exist which  support their execution.   This question was faced by  \cite{Lester:2019bso} and answered in the negative.  A significant obstacle to the generality of the results of that particular search was the observation that the number of potential parity-odd variables with which it could have looked for signals was unlimited in size and lacked any coherent structure.  Consequently the event-variables which it used were poorly motivated.  A secondary consequence of this arbitrariness was that, although that search saw no evidence of non-standard parity violation, it could not make quantitative statements about the extent of areas remaining most untested after its publication, nor give any guide as to where best future searches should probe.

It is to find solutions to such problems that the present paper directs itself.

\introsection{Reducing a problem from  infinite size to finite size}

We seek to bring some order to the field of searches for non-standard sources of Lorentz-invariant parity violation so that such searches may become more tractable, and so that they may lose some of the arbitrariness that would otherwise  be present in their choices of search strategy or event variables. 
We will attempt to do so by proposing a set of soon-to-be-defined requirements on groups of parity-odd event variables 
(they shall be called %\textbf{parity oddness},  
\textbf{sufficiency},  
\textbf{necessity/\independence}, \textbf{minimality},
\textbf{reality},
\textbf{continuity} and
\textbf{invariance under certain symmetries})
using which we will prove that the number of event-variables needing to be considered (given a class of events, such as `events with three jets') may be reduced from a potentially infinite size to one of of finite size, without any loss of generality.

So far as the authors are aware there are no existing works aiming to achieve the same objective.  The closest in spirit is perhaps \cite{Han:2009ra}, which reminds its readers that `genuine' tests of $CP$-violation may be made at colliders like the LHC. Nonetheless, its primary target is $CP$ not $P$, and even to the extent to which its results apply to pure parity violation, it has other objectives in its sights. Other works (such as  \cite{Butter:2017cot} and \cite{Bogatskiy:2020tje})  share our interest in answering questions involving the coupling of Lorentz invariance and permutation invariances. However they choose (actively in the case of \cite{Butter:2017cot}\footnote{Indeed, the top-tagger of \cite{Butter:2017cot} actively inverts the parity ~50\% of its jet clusters (aiming to put any tertiary source of jet substructure in the first quadrant of the its jets' cones) in order to gain  performance from the resulting standardization of the data it inputs to its neural net. That such a regularisation step may be performed by that work is because no non-trivial parity structure is expected in lone quark or gluon jets.}, and implicitly in the case of \cite{Bogatskiy:2020tje}) to make themselves completely blind to any potential sources of parity violation. [It should be emphasised that the last remark is not intended as a criticism of either of the aforementioned works! On the contrary: given the lack of relevance of parity-violation to the goals that each work targets, each work has a good reason for avoiding sensitivity to parity as an extraneous complicating factor.]

\introsection{Our objectives}
\label{sec:ourobjectives}We set out to find sets $S=\{V_i\mid i=1,2,\ldots\}$ of event variables $V_i$ (functions on particle physics events) with the following properties:
\begin{itemize}
\item \textbf{[reality]:}
that each event variable $V_i$ shall map an event of a (suitably defined) class $\Omega$ to a real number,
\item \textbf{[parity oddness]:}
that each event variable $V_i$ shall be parity odd (that is to say, would change sign under $\vec  x\rightarrow -\vec x$),
\item \textbf{[continuity]:}
that the mapping shall be continuous (that is to say: small changes to the energy, or the mass or to any  component of the momentum of a particle in the event shall lead to only small changes in the value of the event variable  $V_i$),
\item \textbf{[invariance under specific symmetries]:}
that all event variables $V_i$ shall be invariant under a common symmetry group of choice (for example in each of the illustrative derivations in this paper we require Lorentz-invariance together with invariance under arbitrary permutation of the momenta of objects belonging to a common class such as `all jets' or or `all photons'),
\item \textbf{[sufficiency]:}
that the set $S=\{V_i\mid i=1,2,\ldots\}$ of all the variables shall have the collective property that at least one $V_i$ shall evaluate to a non-zero value for every event in $\Omega$ which is chiral (an event will be defined to be chiral if it cannot be mapped onto  itself by a parity reversal followed by the action of an arbitrary element of the chosen symmetry group), 
\item  \textbf{[necessity/\independence]:} % WHAT WAS CALLED INDEPENDENCE IS NOW IRREDUCIBILITY .... HENCE COMMAND
that every variable in the set $S$ shall be \textbf{necessary} (or, equivalently, that the set itself shall be \textbf{\independent}), by which we mean that there shall be no variable $V\in S$ such that the set $S\setminus\{V\}$ has the same \textbf{sufficiency} property possessed by $S$, %\footnote{Note that we shall use \textbf{\independence}  and \textbf{necessity} as related terms: the former is property of a \textit{set} of variables, while the latter is the property of the \textit{variables} in such a set. In short: a set of variables, each of which is \textit{necessary}, shall be termed an \textit{irreducible} set.}
and
\item \textbf{[minimality]:}
that where two sets $S_1$ and $S_2$ share all the above properties but differ in cardinality, we shall favour the one with the smaller number of elements.
\end{itemize}
The \textbf{reality} requirement is present purely to ensure that we can count event variables on a well defined basis.  Any event variable could be decomposed into a parity-even and parity-odd part. The former parts cannot ever help us identify instances of parity-violation so the presence of the requirement of \textbf{parity oddness} in the variables which we seek should therefore be self explanatory. We hope that the motivation for the \textbf{sufficiency} and \textbf{necessity} criteria is also readily apparent: without the former property the set of variables would lack sensitivity to any parity-violating theory able to map all of its chiral events to zeros in the elements $V_i$, while without the latter property sets of variables could be large and wasteful.  The \textbf{minimality} constraint likewise seeks to push us toward simple solutions.

The benefits of requiring  \textbf{invariance under specific symmetries} are numerous. If the intended use of the event variables is as inputs to a machine learning algorithm, then a clear benefit is that such an algorithm will not need to use precious training data to learn to about symmetries which are already obvious from context.\footnote{For example: knowing that there is no information content in the ordering of $n$ jet momenta in some set $s$ has the effect of multiplying the value of a training set by a factor of $n$-factorial.  Knowing that there is no information in a global rotation allows one training data sample to stand in for an infinite number of copies of itself sitting at all possible rotations.} In short: such algorithms should generalise better. This could be especially important if the training set is size-constrained or costly to produce. If the event variable is destined for some other process, then it remains the case that use of invariant variables removes the opportunity for an analysis (whether intentionally or accidentally) to make selections or predictions which are unphysical in a set of ways we can choose.

%This is particularly relevant for permutation symmetries -- is merely a consequence of the fact that scientific conclusions should be independent of irrelevant labels.  If an event contains three jets, it should not matter whether the physicist records them in the order one-two-three or three-one-two.  Even if those three jets are different and have different properties (energies, directions, \textit{etc.}) they are not themselves labelled by nature as different objects. No information is therefore lost by demanding that event variables should not changes after re-ordering of momenta of objects reconstructed as instances of the same thing (e.g.~all jets or all electrons or all positrons), and only irrelevant unphysical information is  added if arbitrary labels are included.

\introsection{The continuity requirement in more detail}

\label{sec:continuityindetail}

The reason for including the requirement of \textbf{continuity}, however, deserves in depth discussion.  In truth we impose this constraint simply because we are prejudiced toward believing that it is `a good thing' for any two events to share near-identical descriptions when projected onto event-variables-space \textbf{if} those events,  for all practical purposes, are almost indistinguishable from each other experimentally.  Why? Because any subsequent analysis can only see data through the lens provided by the event variables, and if near-indistinguishable events were to map to separated parts of event-variable-space, then trivial perturbations in data-space map into large perturbations in the analysis which must then be un-learned or compensated for by the analyser, be he or she a human or a neural net or some other machine learning algorithm.  This seems to add at best an unnecessary complication, and at worst a  dangerous one.\footnote{The correction can fail to be done well, or may not be done at all, or may use precious resources of either the analyser or the neural net, etc.}   

The requirement of \textbf{continuity} is, however, a very strong one.  Most of the difficulty in the rest of the paper comes from our having demanded it.

\textbf{Continuity} is a requirement which is not present in the  most common descriptions of complex particle physics events, all of which record contain data structures which store jet momenta (or the momenta of other groups of similar particles) in order of decreasing transverse momentum.  While this (or any similar) sorting process renders any subsequent event-variables invariant with respect to permutations of the orders of the jets (prior to sorting) -- this being a property which we also desire -- it is a not-often-appreciated fact that this same sorting process has the negative consequence that it makes the map from events to  event variables discontinuous.  For example: the dot product between the  momenta of the first and second most energetic jets in an event (or the related variable $\cos{(\Delta\theta(j_1,j_2))}$) can change  discontinuously when the properties of those three jets are continuously varied. All that is necessary for this to happen is for the the second and third most energetic jets to approach each other's energy while each has a  different angles to the most energetic jet.  This means that  machine learning algorithms given such inputs are being asked not only to do important tasks (separating $b$-jets from light quark jets, say) but are also being asked to waste resources in learning to do similar things in widely separated part of the input space which superficially appear to  represent very different events but which, in actual fact, may actually represent events which could be  almost indistinguishable from each other!   \textbf{Continuity}, therefore, might be expected to be beneficial in allowing machine learning algorithms to focus their finite resources on the relevant rather than the irrelevant tasks in hand.

Were the requirement for \textbf{continuity} to be removed, many of the task of the paper would becomes almost trivial.  Construction of the desired sets of event-variables would begin by sorting the momenta within each class of identical particles so as to produce intermediate momenta with the requisite permutation invariance.\footnote{This step is not entirely trivial as structures are needed to cope with every possible tie-breaker situation that can arise during the sort process.  Tie-breaker situations cannot be brushed under the carpet either as they can be the norm rather than exceptions. For example, the lack of a frame-independent way of distinguishing one beam from another in a collider colliding identical particles creates many problems.}  Then all pseudo-scalar Lorentz-contractions of those intermediate momenta would be formed and would themselves constitute a \textbf{sufficient} though  not usually \textbf{\independent} set of event-variables. Further work could be done to try to pare that set down to one which was \textbf{\independent} -- if the resulting reduction were felt necessary.  Alternatively one could simply accept the additional computational cost of using a \textbf{sufficient} but not \textbf{\independent} set of event-variables.

It remains to be determined whether our prejudice motivating the requirement of \textbf{continuity} is worth the additional complication and cost it imposes on the construction of event variables.  Our hunch is that the requirement \textit{is} an important one, but that it only becomes so when the number $n$ of momenta subject to a permutation symmetry has grown sufficiently large that the cost of coping with a number of discontinuities likely to grow as $n!$ (when working with na\"ively constructed event variables) outweighs the costs of deriving a much smaller number of \textbf{continuous} but individually more powerful  event variables.\footnote{The previously mentioned \cite{Butter:2017cot} perhaps takes a different view.  It comments that ``While one could use advanced pre-processing beyond some kind of ordering of the input 4-momenta, our earlier study [\textit{by which it means }\cite{Kasieczka:2017nvn}]  suggests that this is not necessary.''  This statement is made more interesting by the fact that \cite{Kasieczka:2017nvn} appears to make no statements about the benefits or dis-benefits of  four-momentum ordering strategies.  It appears to make no statements about four-momentum ordering strategies of any kind.  Instead \cite{Kasieczka:2017nvn} processes images of jets projected onto eta-phi space.  Such objects have lost all notion of momentum ordering by an entirely different means.  If background work in support of \cite{Kasieczka:2017nvn} had indeed considered issues connected with ordering of four-momenta, but the conclusions from those studies were simply not mentioned the final paper, it is unclear whether that work considered small sets of four-momenta (in which case we would likely agree) or large sets of four-momenta (in which case there could be scope for disagreement or debate between us).}
  It is therefore possible that the most complex permutation symmetry we have used in our illustrations  ($S_2\times S_3$) is simply too simple to see any significant benefit deriving from the \textbf{continuity} requirement we have imposed.\footnote{In the $(pq)\rightarrow(abc)+X$ scenario which we consider toward the end of the paper, a na\"ive transverse-momentum-based sorting approach for the $(abc)$ momenta in the $(pq)$ rest frame would result in only one discontinuity among the sorted $(abc)$ dot products, and one `separated duplication' caused by the inability to distinguish in a frame-independent way between the concepts of positive and negative rapidity when $p$ and $q$ are presumed to represent identical particles.}
It would be interesting for future work to try to create variables for $S_2\times S_n$ symmetries so as to probe performance for all values of $n$ between 2 and $O(10)$ -- or the maximum number of jets recorded in appreciable numbers within LHC runs.\followUpInFuture[Note this proposal for future owrk on $S_2\times S_n$ to test the importance of continuity.]

\introsection{Our use of illustrative examples}

Unfortunately we are not able to present a turn-handle algorithm for generating the appropriate  set of event variables, having all the aforementioned properties, given a `general' or `arbitrary' class of events  $\Omega$. Instead, this paper illustrates the process of how one can go about constructing a set of event variables with the right properties \textit{given}   a concrete class of event, $\Omega_1$, which one is interested in analysing.  This process is then repeated for a number of other event classes $\Omega_2$, $\Omega_3$, \ldots. The event classes have (we hope) been chosen to be sufficiently interesting that they illustrate the structure and nature of the calculations which would need to be performed for more complicated event classes beyond those considered.

The simplest event classes  we consider are those one could use when looking for non-standard sources of Lorentz-invariant parity violation in two-jet or three-jet events at colliders.  The most general class of events we consider is one in which each event is assumed to contain five important momenta among which an $S_2$ interchange symmetry is present between two of them, and an $S_3$  interchange symmetry is present between the others. In this class of events there are no restrictions on whether those momenta are in the initial or final state, or are even observed at a collider.

The only high-level physics-based requirements which have been woven into to the very fabric of our framework (and which might be absent from a purely abstract mathematical attempt at solving the same problem) are: (i) that momenta which comprise event-data are assumed to be time-like or null (not space-like) and
(ii)
they are assumed to have entirely real (not complex) components.

We caution the reader against becoming  anxious upon seeing the long list of special cases presented in the mid sections of the paper. Although some of these sections  (e.g.~Section~\ref{sec:movingontononcollevtsinpqabc}) have alarming titles, the final results which they lead to are more general than they may at first appear. Although it was necessary to break many proofs down into small components, each acting on very narrow classes of events, these are only ingredients forming part of a divide-and-conquer strategy aiming at investigating a simpler whole.

\introsection{The bigger picture}

As the bulk of the paper comprises theorems and proofs tailored very much to the details of the specific 
event classes used as illustrations, it may be helpful to contrast those direct attacks on the problem with a more general description of the fundamental ideas underlying our results.  This might both help readers to see the `wood for the trees' and pave the way for future generalizations (e.g.~with additional symmetries or more complicated events).

To that end: suppose we have some physical system, whose dynamics is known to be invariant under a symmetry group $H$
%was $G$ but changed to $H$ by Ward
and which is known to have some notion of parity, under which the dynamics may or may not be invariant.\footnote{$H$ will later be the Lorentz group without parity, together with some permutation symmetry.}  Rather than thinking of parity as some definite transformation of the physical system (say `$\vec x \rightarrow -\vec x$'), it is helpful to think of it as an extension of the group $H$ by $\mathbb Z/2$. That is to say, we have another group $G$, of which $H$ is a normal subgroup, such that $G/H$ is isomorphic to $\mathbb Z/2$, and that the non-trivial element of this $\mathbb Z/2$ subgroup is what we call parity.\footnote{$G$ will later be the Lorentz group with parity, together with some permutation symmetry.}  The phenomenologically interesting question then becomes: we know that $H$ is a symmetry of our system, but is $G$?   

To give a somewhat trivial example, suppose our system is invariant under translations in space. Then, rather than arbitrarily choosing coordinates $x$ on our space and defining parity as the transformation $\vec x \rightarrow -\vec x$ reflecting points in our arbitrarily chosen origin, we should think about the reflections through all possible origins, each of which is related to every other by a translation. 

Now suppose we collect some data, which we call events, each of which is a point in a space $X$ with the property that $G$ (and $H$) have a well-defined action on $X$ (for example, we might choose some space of particles' energy and momenta, measured in a particular reference frame).\footnote{$X$ will later be the space of all momenta. } Because we know $H$ to be a symmetry of the system, it makes sense to collate events which lie on a common orbit of $H$, since there is nothing to be learnt by keeping them separate. Now, because $H$ is a normal subgroup of $G$, there is a well-defined action of $G/H \cong \mathbb Z/2$, i.e.~parity, on the space $X/H$ of $H$-orbits of $X$. Moreover, we can partition the orbits themselves into non-chiral orbits, which are fixed under the action of $\mathbb Z/2$ and the chiral orbits, which are not fixed (and which therefore contain precisely two distinct $H$-orbits of $X$). The chiral orbits are the ones that are of interest to us, because they allow us the possibility of testing directly whether parity is indeed a symmetry of our dynamical system. One can do this, for example, by examining whether, in a run of observations of events, we obtain a statistically-significant discrepancy in the number of events observed in each member of any such pair of orbits. 

This simple observation is hindered, in general, by the fact that both the space of possible events and the resulting space of possible $H$-orbits will be infinitely large, requiring us to somehow combine many orbits in order to have a chance of reaching a statistically significant discrepancy in a finite number of observations. We would like to do so in a way which discards as little information as possible. One way to do so is to look for observables (i.e.~nice functions on the underlying space $X$, e.g.~polynomials) which are not only $H$-invariant (so induce well-defined functions on $X/H$ which are oblivious to useless information), but which are also parity odd (meaning that they return values of opposite sign on each of an orbit pair), giving one the hope of establishing whether parity is a symmetry or not. Since any function can be decomposed into a piece which is parity odd and parity even, there is again no loss of information here; rather, we are simply throwing away the useless, parity even part. 

Ideally, one would like to have a set of functions of this nature that is large enough to be able to separate all of the $H$-orbits from one another, so that again no physical information is lost in passing from events to observables. This is, in essence, the goal of this paper. In order to achieve it, we make full use of the freedom to restrict the possible events to those which are physically observable at colliders and elsewhere, enabling us to reduce the size of the set of functions as much as possible.

\introsection{Alternative data-only constraints on non-standard sources of parity violation} 

If it were desired to retain the ability to make statements about non-standard parity-violation which are primarily based on data alone (not using comparisons of data and Monte Carlo  at first order), and if it were acceptable to lose the requirements connected with \textit{sufficiency}, an interesting approach would be to randomly flip the parity of 50\% of LHC events in software, and then attempt to train a machine learning algorithm to distinguish the flipped from the non-flipped events.  Such methods would be many in number, limited only by the creativity of the algorithms used. They would frequently come with neither the qualified coverage guarantees  provided by the sets of variables proposed in the present paper, nor with assurances that algorithms is blind to (perhaps permutation) symmetries known to be uninteresting, nor with assurances that the algorithms respect symmetries known or presumed by choice to be fundamental (perhaps Lorentz-invariance). Such methods would, however, have set up costs for new event classes which are orders of magnitude below those associated with the construction of sets of variables in the manner suggested in the current paper.  While the present paper is the third in a series beginning with \cite{Gripaios:2020ori} and \cite{Gripaios:2020hya}, the fourth\footnote{Likely to be `Mastandrea and Lester \textit{et.al.}'} may well be one which compares the current proposal to such alternative methods.

\introsection{The structure of the rest of the paper}

%\noindent The rest of the note is structured as follows:
\begin{itemize}
\item
Section~\ref{sec:definitionsrelatingtocollevents} prepares the ground for many subsequent calculations, most importantly by providing  axiomatic definitions of what we wish to consider as  events at colliders, what we mean by classes of events, a nomenclature for describing the interchange symmetries on them which may be of interest, and what it means for events to be chiral or non-chiral under parity,
\textit{etc.}
\item
Section~\ref{sec:condsforvariouschiraleventclasses} then determines explicit conditions under which events in certain classes actually are chiral. The classes of events considered  include (in no particular order) the following:
\begin{enumerate}
\item
$(pq)\rightarrow (ab)+X$,
\item
$(pq)\rightarrow (abc)+X$,
\item
$X\rightarrow (pq)+(ab)+Y$, and
\item
$X\rightarrow (pq)+(abc)+Y$,
\end{enumerate}
\item
For each of the classes of events just mentioned, Section~\ref{sec:wherewearenow} then obtains sets of parity-odd event variables about which strong statements can be made concerning their \textbf{sufficiency} and \textbf{\independence} properties.  This section is, arguably, the most important of the paper, since its outputs are the main results contained herein.  In particular:
\item
For events of the form $(pq)\rightarrow (ab)+X$, $(ab)\rightarrow (pq)+X$ or $X\rightarrow (ab)+(pq)+Y$,
Section~\ref{sec:exploratorysectionforpqabeventvars} identifies a set $S$ containing three Lorentz-invariant event variables ($S=\{X_1, X_2, X_3\}$) which are insensitive to $(pq)$- and $(ab)$-permutations and which
are proved to satisfy the desired \textbf{reality}, 
\textbf{parity oddness}, 
\textbf{continuity},
\textbf{invariance}, 
\textbf{necessity}, and
\textbf{sufficiency}
properties defined in the introduction. Each of the variables in $S$ is defined in a manifestly frame independent way.
\item
Section~\ref{sec:exploratorysectionforpqabeventvars} also identifies a related set $\hat S$ containing three similar  Lorentz-invariant variables ($\hat S=\{\hat X_1, \hat X_2, \hat X_3\}$). These may only be used to ascribe parities to events of the form $(pq)\rightarrow (ab)+X$.  Despite this limitation, the variables of $\hat S$ are arguably simpler and easier to interpret than those of $S$, and so may be of use to some.
\item
No claim is made that either set $S$  or $\hat S$ is globally \textbf{minimal} for the class of events with which it is concerned.  That is to say:  it has not been proved that there does not exist some other set $S_?$ containing only one or two event variables but which retains all the  stated properties.
Nonetheless, we \textit{conjecture} that there is indeed  no smaller set $S_?$ and that $S$ and $\hat S$ are therefore examples of \textbf{minimal} sets.
\followUpInFuture[Ben: Can you see a way of proving this conjecture?]
%\item
%The results of Section~\ref{sec:exploratorysectionforpqabeventvars}, though perhaps of some academic interest, are in themselves of insufficient merit to warrant publication on their own.  This section exists primarily to provide a pedagogical `warm up' for the more interesting job of doing the same for events of the form  $(pq)\rightarrow (abc)+X$ which now begins in Section~\ref{sec:exploratorysectionforpqabceventvars}.
\item
Section~\ref{sec:exploratorysectionforpqabceventvars} directs itself toward the generation of a set,  $S_{19}$, containing nineteen Lorentz-invariant event variables ($S_{19}=\{V_1, \ldots, V_{19}\}$) which are insensitive to $(pq)$- and $(abc)$-permutations and which, for events of the form $(pq)\rightarrow (abc)+X$, are proved to satisfy the \textbf{reality}, 
\textbf{parity-oddness}, 
\textbf{continuity}, 
\textbf{invariance}, 
\textbf{necessity} and
\textbf{sufficiency}
properties already described. The event variables in $S_{19}$ are defined in a manifestly frame independent way. 
\item
No claim or conjecture is made that $S_{19}$  is globally \textbf{minimal}.  That is to say it has not been proved that there does not exist some other set $S_\#$ containing  eighteen or fewer event variables which retains all the stated properties including sufficiency.  Nonetheless, Section~\ref{sec:discOfMinimialityOfS19} discusses the circumstantial evidence which suggests that $S_{19}$ might not be a \textbf{minimal} set.
\item
Sections~\ref{sec:bringbothnoncolltypestogeterinpqabc} and \ref{sec:pqabcbothcollandnoncollvars} work toward finding a superset of $S_{19}$ named  $S_{28}$ containing twenty-eight Lorentz-invariant event variables  ($S_{28}=\{V_1, \ldots, V_{28}\}$) which are insensitive to $(pq)$- and $(abc)$-permutations.  For events having an interchangeable pair of particles $(pq)$ and an interchangeable triplet of particles $(abc)$ (neither of which need be in the initial state) the event variables in $S_{28}$ are proved to satisfy the \textbf{reality}, 
\textbf{parity-oddness}, 
\textbf{continuity}, 
\textbf{invariance} and
\textbf{sufficiency}
properties already described. The event variables in $S_{28}$ are defined in a manifestly frame independent way.  The set $S_{28}$ is \textit{not} demonstrated to be \textbf{\independent}.
\item
Section~\ref{sec:supportmaterials} describes downloadable support materials which may assist readers wishing to validate calculations using the event variables $V_1$ to $V_{28}$.
\item
Section~\ref{sec:discussion} and beyond then discuss the potential uses of the event variables previously defined, their relationship to related ideas, and scope for future development.
\item
Core notation may be found in Appendix~\ref{appendix:notation} and a small number of non-standard mathematical identities may be found in Appendix~\ref{app:mathematicalidentitissdf}.
\end{itemize}

\section{Preliminary definitions and calculations.}

This Section~\ref{sec:definitionsrelatingtocollevents} contains important definitions and notation for key physical concepts which are needed to support later proofs.  Inevitably, the definitions of these higher-level concepts depend on deeper notation which some (but perhaps not all) readers may regard as `standard'. An example might include the four-vector notation `$p=(E,\vec p)$' seen in Definition~\ref{def:firstdeftousealorentsvec}.  Such  `underlying' or `semi-standard' notation is important to define, but it could be a distraction to define it \textit{here}. We therefore provide it instead in a dedicated Appendix~\ref{appendix:notation}.

\label{sec:definitionsrelatingtocollevents} 

\hide{
\begin{definition}
$\mathbb{O}$ is defined to be the set containing the `zero' four-vector:  $\mathbb{O}=\{(0,\vec 0)\}$. 
\end{definition}
\begin{definition}
$\mathbb{N}$ is defined to be the set containing all forward-time-directed null four-vectors:  $\mathbb{N}=\left\{\left(|\vec p|,\vec p\right) \mid \vec p\in\mathbb{R}^3\right\}$. 
\end{definition}
\begin{definition}
$\mathbb{M}$ is defined to be the set containing all forward-time-directed massive four-vectors:  $\mathbb{M}=\left\{\left(\sqrt{m^2+|\vec p|^2},\vec p\right) \mid \vec p\in\mathbb{R}^3,m>0\right\}$. 
\end{definition}
}

\begin{definition}
$\mathbb{A}$ is defined to be the set containing all four-vectors having real energies and real momentum components:  $\mathbb{A}=\left\{p  \mid p=\left(E,\vec p\right), \vec p\in\mathbb{R}^3,E\in \mathbb{R}\right\}.$
\label{def:firstdeftousealorentsvec}
\end{definition}
\begin{definition}
$\mathbb{V}$ is defined to be the proper subset of $\mathbb{A}$ containing all forward-time-directed non-spacelike four-vectors:  $\mathbb{V}=\left\{p  \mid p=\left(\sqrt{m^2+|\vec p|^2},\vec p\right), \vec p\in\mathbb{R}^3,m\ge0\right\}\subset\mathbb{A}$. 
\end{definition}
\begin{definition}
$\mathbb{M}$ is defined to be the proper subset of $\mathbb{V}$ containing all forward-time-directed timelike non-massless four-vectors:  $\mathbb{M}=\left\{p  \mid p=\left(\sqrt{m^2+|\vec p|^2},\vec p\right), \vec p\in\mathbb{R}^3,m>0\right\}\subset\mathbb{V}$. 
\end{definition}
\begin{definition}
$\mathbb{N}$ is defined to be the proper subset of $\mathbb{V}$ containing all forward-time-directed null (massless)  four-vectors:  $
\mathbb{N}=
\left\{
  p  
  \mid 
  p=\left(
        |\vec p|,\vec p
     \right), 
     \vec p\in\mathbb{R}^3
\right\}\subset\mathbb{V}$. 
\end{definition}
\begin{remark}
%$\mathbb{M}\subset\mathbb{V}\subset\mathbb{A}$,
%$\mathbb{N}\subset\mathbb{V}\subset\mathbb{A}$, 
$\mathbb{M}\cup\mathbb{N}=\mathbb{V}$ and $\mathbb{M}\cap\mathbb{N}=\emptyset$. The zero-vector is contained within $\mathbb{N}$.
\end{remark}
\begin{remark}
A generic element of $\mathbf{x}\in\mathbb{A}^n$ (or
$\mathbb{V}^n$ or $\mathbb{M}^n$ or $\mathbb{N}^n$)
will frequently use \textbf{bold-face} to emphasise that it has $n$ components.  Where these components themselves need to be indexed, the most common notation will take the form  $\mathbf{x}=(p_1,\ldots,p_n)$ for $p_i\in \mathbb{A}$ (or
$\mathbb{V}^n$ or $\mathbb{M}^n$ or $\mathbb{N}^n$).
\end{remark}
\begin{remark}
The set $\mathbb{A}$ can easily be extended to a vector space over the field $\mathbb R$ by the addition of the usual operators allowing vector addition and scalar multiplication.  \textbf{The same is not true for $\mathbb V$ or $\mathbb{M}$ or $\mathbb{N}$} since these do not contain inverse elements for the vector addition operator.  Nonetheless, Lemmas~\ref{lem:vsumsareinv} and \ref{lem:msummareinm} are related and still hold.
\end{remark}

\begin{lemma}
If $p$ and $q$ are both in $\mathbb{V}$ then $p.q\ge 0$.\label{lem:pqpos}
\begin{proof}
$
    p.q
    =
    %\left(
    \sqrt{m_p^2+|\vec p|^2}
    \sqrt{m_q^2+|\vec q|^2}
    -\vec p \cdot \vec q
    %\right)
   % \\
\ge 
    %\left(
  |\vec p|
  |\vec q|
    -|\vec p|
  |\vec q|\cos{\theta_{pq}}
    %\right) 
  %  \\
    \ge
    %\left(
  |\vec p|
  |\vec q|
    -|\vec p|
  |\vec q|
    %\right)
  %  \\
    =0.
$
\end{proof}
\end{lemma}

\begin{lemma}
\label{lem:vsumsareinv}
If $p$ and $q$ are both in $\mathbb{V}$, and $\lambda$ and $\mu$ are non-negative real numbers, then $\lambda p + \mu q \in \mathbb{V}$. \begin{proof}
It is evident from the condition on $\lambda$ and $\mu$ that the time component of $\lambda p + \mu q$ will be non-negative (as required). The only non-trivial check required, therefore, is that $\lambda p + \mu q$ has a non-negative squared-mass. This may be shown by considering $(\lambda p + \mu q)^2 = \lambda^2 m_p^2 + \mu^2 m_q^2 + 2 \lambda \mu (p\cdot q) \ge2 \lambda \mu (p\cdot q) \ge 0$. (The last inequality uses Lemma~\ref{lem:pqpos}.)
\end{proof}
\end{lemma}
\begin{lemma}
\label{lem:msummareinm}
If $p$ and $q$ are both in $\mathbb{M}$, and $\lambda$ and $\mu$ are positive real numbers, then $\lambda p + \mu q \in \mathbb{M}$. \begin{proof}
The proof for this lemma is a trival extension of the proof of Lemma~\ref{lem:vsumsareinv}. The only difference is that this time use may be made of the fact that $\lambda^2 m_p^2$ and $\mu^2 m_q^2$ are both strictly greater than zero.
\end{proof}
\end{lemma}
%\begin{remark}$\mathbb{V}=\mathbb{N}\cup\mathbb{M}.$\end{remark}

\begin{lemma}
\label{lem:eppos}
If $p\in\mathbb{V}$ and $p\ne0$ then $E_p>0$ in every frame.\footnote{We exclude from the concept of `frame' anything reached only as a limit (i.e.~with an infinite Lorentz `$\gamma$').} 
\begin{proof}
If $p$ is in $\mathbb{V}$ and $p\ne 0$ the either $\vec p\ne0$ or $m\ne 0$.  Since $E_p=\sqrt{m^2+|\vec  p|^2}$ then $E_p>0$.
\end{proof}
\end{lemma}

\begin{lemma}
\label{lem:pqbiggerthanmpmq}
If $p$ and $q$ are both in $\mathbb{V}$ then $p.q\ge m_p m_q$.
\begin{proof}
If $m_p=0$ or $m_q=0$ then the result is already proved by Lemma~\ref{lem:pqpos}.  It is only necessary to prove the result, therefore, in the case that $m_p\ne 0 \ne m_q$.  In this case:
\begin{align*}
    p.q
    &=
    %\left(
    \sqrt{m_p^2+|\vec p|^2}
    \sqrt{m_q^2+|\vec q|^2}
    -\vec p \cdot \vec q
    %\right)
   \\
&= 
    %\left(
   \sqrt{m_p^2+|\vec p|^2}
    \sqrt{m_q^2+|\vec q|^2}
    -|\vec p|
  |\vec q|\cos{\theta_{pq}}
    %\right) 
  %  
   \\
&\ge
    %\left(
   \sqrt{m_p^2+|\vec p|^2}
    \sqrt{m_q^2+|\vec q|^2}
    -|\vec p|
  |\vec q|
    %\right) 
  %  
   \\
&= m_p m_q \left(
    %\left(
   \sqrt{1+x^2}
    \sqrt{1+y^2}
    -x
 y
  \right)
    %\right) 
  %  
\end{align*}
where $x=|\vec p|/m_p$ and $y=|\vec q|/m_q$.  The lemma shall therfore be proved if it can be demonstrated that the function $f(x,y)=\sqrt{1+x^2}\sqrt{1+y^2}-xy$ is greater than or equal to one for all non-negative $x$ and $y$. Trivially $f(x,y)\ge 1$ if  either $x=0$ or $y=0$.  We will only need to show that this is also true when both $x$ and $y$ are strictly positive: $x>0$ and $y>0$.
At extrema of $f$:
\begin{align*}
0=\frac{\partial f}{\partial x} &=\frac{x \sqrt{1+y^2}}{\sqrt{1+x^2}} - y 
\qquad\text{and} \\
0=\frac{\partial f}{\partial y} &=\frac{y \sqrt{1+x^2}}{\sqrt{1+y^2}} - x.
\end{align*}
These constraints are redundant with each other and may be jointly written as \begin{align}y \sqrt{1+x^2} = x \sqrt{1+y^2}
\label{eq:yghbresdfresdfv}
\end{align}
whose solutions are contained within those of $
y^2(1+x^2)=x^2(1+y^2)$ which is equivalent to $y^2 = x^2$ or even $x=y$ given, as we have already noted, that
it suffices to consider only those  solutions having $x\ge 0$ and $y\ge 0$.  We may limit ourselves therefore to checking for extrema among the cases $(x,y)=(\lambda,\lambda)$  for   $\lambda\ge0$ -- and moreover note that all such potential solutions do indeed solve $\frac{\partial f}{\partial x}=\frac{\partial f}{\partial y}=0$  (meaning that they are not artefects introduced when \eqref{eq:yghbresdfresdfv} was squared). We therefore try such solutions in $f(x,y)$:
\begin{align*}
f(\lambda, \lambda ) 
&=
\sqrt{1+\lambda^2 }\sqrt{1+\lambda^2 }-\lambda^2 
=1
\end{align*}
independently of $\lambda$.  We therefore see that $f(x,y)$ has a degenerate extremum along the line $(x,y)=(\lambda,\lambda)$ and has no other extrema, except perhaps at infinity. The lemma will therefore be proved if we can demonstrate that the degenerate extremum just found is a minimum.  To do this it suffices to find one point above the line and one point below it\footnote{Finding two points excludes the possibility of a degenerate inflection.} having values of $f(x,y)$ which are greater than 1.  The proof is concluded, therefore, by noting that $f(1,0)=f(0,1)=\sqrt 2>1$. 
\end{proof}
\end{lemma}

\begin{definition}
A partition of $n$ things into $m$ classes having $n_i>0$ things in each class $i$ will be denoted $P(n,m,\mathbf{n})$ where $\mathbf{n}=(n_1, \ldots, n_m)$ and $n=\sum_{i=1}^m n_i.$
\end{definition}
\begin{remark}
When analysing an event of the form $$pp\rightarrow jjj+\mu^++\mu^-+X$$ then the partition of interest would be $P(7,4,(2,3,1,1))$. In this, the $n=7$ shows us we understand the momenta of seven objects (two incoming protons and five outgoing identified particles), $m=4$ tells us that there are four identifiable classes of particles in the event (namely: (i) incoming protons; (ii) jets, (iii) muons; and (iv) anti-muons) while the $\mathbf{n}=(2,3,1,1)$ shows us how many of those particles there are in each successive category.
\end{remark}
\begin{definition}
We define the group $S(\mathbb{V},P(n,m,\mathbf{n}))\in \Aut(\mathbb{V}^n)$ to be 
$$S(\mathbb{V},P(n,m,\mathbf{n}))=SO^+(1,3)\times S_{n_1}\times S_{n_2} \times \cdots \times S_{n_m}$$ where
$SO^+(1,3)$ is the part of the Lorentz Group which is connected to the identity\footnote{$SO^+(1,3)$ is sometimes called the group of proper (i.e.~not parity-altering) orthochronous (i.e.~not time-reversing) Lorentz transformations.},  $S_{n_1}$ is the symmetric group of order $n_1$ which permutes the first $n_1$ momenta within any  $\mathbf{x}\in\mathbb{V}^n$, and $S_{n_2}$ is the symmetric group of order $n_2$ which permutes the next $n_2$ momenta within any  $\mathbf{x}\in\mathbb{V}^n$, and so on. The Lorentz group acts on every element of $\mathbf{x}\in\mathbb{V}^n$ coherently.
\end{definition}
\begin{definition}
For every integer $n\ge 1$, an \textbf{event} $e$ in an \textbf{event class} $\mathscr{Ec}(\mathbb{V}^n,G)$ is an element $\mathbf{x}\in\mathbb{V}^n$ together with a symmetry group $G \in\Aut(\mathbb{V}^n)$:
\begin{align}
\mathscr{Ec}(\mathbb{V}^n,G)&=\mathbb{V}^n\times \Aut(\mathbb{V}^n) \\&= \left\{
e
\mid
e=(\mathbf{x},G),
\mathbf{x}\in \mathbb{V}^n, 
G\in \Aut(\mathbb{V}^n)
\right\}.
\end{align}
\end{definition}

\begin{remark}
We are primarily interested in \textbf{event classes} $\mathscr{Ec}(\mathbb{V}^n,G)$ for which the symmetry group $G$ is of the form $G=S(\mathbb{V},P(n,m,\mathbf{n}))$.  In a small abuse of notation, we therefore extend the notation of an \textbf{event class} $\mathscr{Ec}(\mathbb{V}^n,G)$ by making the following short-hand definition:
\end{remark}
\begin{definition}
$$
\mathscr{Ec}(P(n,m,\mathbf{n}))
\equivdef \mathscr{Ec}(\mathbb{V}^n,
S(\mathbb{V},P(n,m,\mathbf{n}))
).$$
\end{definition}
\begin{definition}
We define two events $e$ and $f$ to be \textbf{equivalent} \label{def:eventequivalence} (written `$e\sim f$') if and only if (i) they are in the same event class, and (ii) they can be transformed into each other by an element of the group each contains.  More formally, if $e=(\mathbf{x},G_e)\in\mathscr{E}(\mathbb{V}^{n(e)},G_e)$
and
$f=(\mathbf{y},G_f)\in\mathscr{E}(\mathbb{V}^{n(f)},G_f)$
then
$$(e\sim f) \iff (\text{$n(e)=n(f)$ and $G_e=G_f$ and $\exists$ $g\in G_e$ s.t.~$\mathbf{x} = g\mathbf{y}$}).$$
\end{definition}
\begin{remark}
The relation $\sim$ just given is an equivalence relation.
\end{remark}
\begin{definition}
\label{def:paroponthings}
The parity operator $\mathscr{P}$ is defined to act on  $p\in\mathbb{V}$, $\mathbf{x}\in\mathbb{V}^n$  and $e\in\mathscr{E}(V^n,G)$ as follows:
\begin{align}
\mathscr{P}\cdot p&\equiv \mathscr{P} \cdot \left(E,\vec p\right) \equivdef \left(E,-\vec p\right),
\\
\mathscr{P}\cdot \mathbf{x} &\equiv \mathscr{P} \cdot \left(p_1,\ldots,p_n\right) \equivdef \left(\mathscr{P}\cdot p_1,\ldots,\mathscr{P}\cdot p_n\right),\qquad\text{and}
\\
\mathscr{P}\cdot e &\equiv \mathscr{P} \cdot \left(\mathbf{x},G\right) \equivdef
\left(\mathscr{P} \cdot\mathbf{x},G\right).
\end{align}
\end{definition}
\begin{definition}[\textbf{event chirality}]
An \textbf{event} $e\in\mathscr{E}(\mathbb{V}^n,G)$ is defined to be \textbf{non-chiral} \label{def:unhanded}if and only if $e\sim \mathscr{P} \cdot e,$
 where $\sim$ is the equivalence relation given in Definition~\ref{def:eventequivalence}.  Similarly, an \textbf{event} for which this is not true may be termed \textbf{chiral}.
\end{definition}
\begin{definition}
If $\mathscr{E}$ is a set of events, we define $\chiral(\mathscr{E})$ to be the set of \textbf{chiral} events in $\mathscr{E}$, and $\nonchiral(\mathscr{E})$ to be the set of \textbf{non-chiral} events in $\mathscr{E}$:
\begin{align}
    \chiral(\mathscr{E}) 
    &=
    \left\{ e \mid e\in\mathscr{E}, e \not\sim \mathscr{P}\cdot e\right\},\text{\qquad and}
    \\
    \nonchiral(\mathscr{E}) 
    &=
    \left\{ e \mid e\in\mathscr{E}, e \sim \mathscr{P}\cdot e\right\}.
\end{align}
\begin{remark}
$\nonchiral(\mathscr{E}) \cup \chiral(\mathscr{E})=\mathscr{E}$ and $\nonchiral(\mathscr{E}) \cap \chiral(\mathscr{E})=\emptyset$.
\end{remark}
\end{definition}
\begin{definition}
Since every event $e=(\mathbf{x},G)$ contains a group $G\in\Aut(\mathbb{V}^n)$, we can define an action of $G$ on the events \label{def:actionGonE} in $\mathscr{E}(\mathbb{V}^n,G)$ in a natural way. Concretely, for any $G\in\Aut(\mathbb{V}^n)$,  we define the action of $g\in G$ on $e\in\mathscr{E}(\mathbb{V}^n,G)$ as follows:
$$
g\cdot e \equiv g\cdot (\mathbf{x},G) \equivdef (g\cdot\mathbf{x},G).
$$
\end{definition}
\begin{corollary}
\label{cor:simplehandedness}
It follows from Definitions~\ref{def:unhanded} and \ref{def:actionGonE} that an \textbf{event} $e\in\mathscr{E}(\mathbb{V}^n,G)$ will be \textbf{non-chiral} if and only if there exists a $g\in G$ such that $g\cdot e = \mathscr{P} \cdot e.$
\end{corollary}

\begin{lemma}
\label{lem:D2neverpositive}
If $p\in\mathbb{V}$ and $q\in\mathbb{V}$ then $p^2 q^2  \le (p.q)^2$ (or equivalently $\Delta_2(p,q) \le0$).
\begin{proof}
The result may be obtained by squaring Lemma~\ref{lem:pqbiggerthanmpmq}.
\end{proof}
\end{lemma}

\begin{lemma}
\label{lem:collisioncond}
If $p\in\mathbb{V}$ and $q\in\mathbb{V}$ then $p^2 q^2  < (p.q)^2$ (or equivalently $\Delta_2(p,q) < 0$) is a necessary and sufficient condition for $p$ and $q$: (a) to have a centre-of-mass frame; and (b) to have equal and opposite non-zero three-momenta in that frame.
\begin{proof}
To prove necessity, we first assume that  $p$ and $q$ have a centre of mass frame in which $\vec p$ and $\vec q$ are not only equal and opposite (which would be required by definition of such a frame) but that they also have $|\vec p|>0$ in that frame.  Concretely, we assume  $p^\mu=\left(\sqrt{m_p^2+|\vec p|^2},\vec p\right)$ and  $q^\mu=\left(\sqrt{m_q^2+|\vec p|^2},-\vec p\right)$ for some $\vec p \ne \vec 0$.  Hence, in such a case: \begin{align}
    (p.q)^2-p^2q^2 
    &=
    (\sqrt{m_p^2+|\vec p|^2} \sqrt{m_q^2+|\vec p|^2}+|\vec p|^2)^2 - m_p^2 m_q^2
    \\
    &=
    (m_p^2+|\vec p|^2)(m_q^2+|\vec p|^2)
+
    2 |\vec p|^2\sqrt{m_p^2+|\vec p|^2} \sqrt{m_q^2+|\vec p|^2}   
    +|\vec p|^4 - m_p^2 m_q^2
    \\
    &=
     %m_p^2 m_q^2+
     (m_p^2+m_q^2) |\vec p|^2+
     2|\vec p|^4
+
    2 |\vec p|^2\sqrt{m_p^2+|\vec p|^2} \sqrt{m_q^2+|\vec p|^2}   
   % +|\vec p|^4 
    %- m_p^2 m_q^2
    \\
    &=
     |\vec p|^2\left(
     (m_p^2+m_q^2)+
     2|\vec p|^2
+
    2 \sqrt{m_p^2+|\vec p|^2} \sqrt{m_q^2+|\vec p|^2}   
     \right)\label{eq:uhndsghsf}
    \\
    &\ge 4 |\vec p|^4,
\end{align}
which is greater than zero (by assumption), so concluding the proof of necessity.

To prove sufficiency we will need to show that $(p^2 q^2< (p.q)^2)$ implies  (($p$ and $q$ have a centre of mass frame) and ($p$ and $q$ are not stationary in that frame)).  This may be done by proving the contrapositive statement:  ``If (($p$ and $q$ do not have a centre of mass frame) or ($p$ and $q$ have a centre of mass frame but are stationary in it)) then $(p^2 q^2\ge (p.q)^2)$''.  If $p$ and $q$ have a centre of mass frame but are stationary in it, then equation (\ref{eq:uhndsghsf}) shows that $p^2q^2=(p.q)^2$ agreeing with half of the contrapositive statement. It therefore only remains to show what happens when $p$ and $q$ fail to have a rest frame.  We need only consider the case in which $p+q$ is non-zero, since if $p+q=0$ then $p^2 q^2\ge (p.q)^2$ is trivially satisfied.  We therefore assume in what follows that $p+q\ne 0$.

Failure to have a rest frame means that the velocity of $(p+q)$ in the any frame must be equal to the velocity of light, or else it would be possible to catch up with $(p+q)$ from some frame.  Let us therefore evaluate the velocity of $(p+q)$ in an arbitrary frame defined (up to rotation) by  unit-mass four-momentum   $\Lambda$ which is taken to be at rest in that frame.  In that frame, we know from Lemma~\ref{lem:eppos} that $E_{p+q}>0$ since Lemma~\ref{lem:vsumsareinv} places $p+q$ in $\mathbb{V}$, and $p+q\ne0$ by assumption. Therefore we may safely say that:
\begin{align}
  \left.  \beta_{p+q} \right|_\Lambda
    &= 
    \left. \frac {
    |{\vec p}|_{p+q} 
    }{
     E_{p+q} 
    }\right|_\Lambda
    %\\
    = 
    \left. \frac {\sqrt{
    E_{p+q}^2-m_{p+q}^2
    }
    }{
   E_{p+q} 
    }\right|_\Lambda
    %\\
     = 
    \frac {
    \sqrt{
    (\Lambda.(p+q))^2-(p+q)^2
    }
    }{
     \Lambda. (p+q)    }
\end{align}
and so a failure of $(p+q)$ to have a rest frame must imply that 
\begin{align}
  ( \left. \beta_{p+q}^2 \right|_\Lambda
    =1)
    &\implies
   %\left( 
   \frac {
    {
    (\Lambda.(p+q))^2-(p+q)^2
    }
    }{
     (\Lambda. (p+q))^2    }
     =1
   %\right)
   \\
    &\implies
  % \left( 
  {
    {
    (\Lambda.(p+q))^2-(p+q)^2
    }
    }={
     (\Lambda. (p+q))^2    }
   %\right)
    \\
    &\implies
  % \left( 
    (p+q)^2=0
  % \right)  
   \\
    &\implies
   p^2+q^2+2 (p.q)=0
   \\
    &\implies
   %\left(
   p^2=q^2=(p.q)=0
   %\right)
\end{align}
using Lemma~\ref{lem:pqpos} three times in the last step.  Accordingly, we see that it is indeed the case that $p^2q^2\ge (p.q)^2$ since both left and right hand sides are identically zero.
\end{proof}
\end{lemma}

\begin{lemma}
\label{lem:whatisnotacolleventONE}
If $p\in\mathbb{V}$ and $q\in\mathbb{V}$ then $p^2 q^2  = (p.q)^2$ (or equivalently $\Delta_2(p,q)=0$) is a necessary and sufficient condition for there to exist non-negative real numbers $\lambda\ge0$ and $\mu\ge0$, not both zero, such that  $\lambda p=\mu q$.
\begin{proof}
To prove necessity, assume first that $\lambda\ge$, $\mu\ge0$, $\lambda p = \mu q$ and that $\lambda$ and $\mu$ are not both zero. Without loss of generality, assume that $\mu$ is non-zero.  This assumption permits $q$ to be written as $q=\frac{\lambda}{\mu}p$ and so both $p^2 q^2$ and $(p.q)^2$ are equal to $(p^2)^2 \frac{\lambda^2}{\mu^2}$.

To prove sufficiency, we begin with the assumption that $(p.q)^2=p^2q^2$ and the knowledge that $p\in\mathbb{V}$ and $q\in\mathbb{V}$.  The latter two permit us to write $p$ and $q$ as by $p^\mu=\left(\sqrt{m_p^2+|\vec p|^2},\vec p\right)$ and  $q^\mu=\left(\sqrt{m_q^2+|\vec q|^2},\vec q\right)$.  To complete the proof we will consider two cases separately: (i) both $p$ and $q$ are massless, and (ii) at least one of $p$ and $q$ (without loss of generality $p$) has mass $m>0$.

In case (i) our parameterisation changes to $p^\mu=(|\vec p|,\vec p)$ and  $q^\mu=(|\vec q|,\vec q)$ and so:
\begin{align}
0 &=
    (p.q)^2-p^2q^2 
    \\
    &=
    (|\vec p||\vec q|-\vec p.\vec q)^2 - 0^20^2
    \\
    &=
   ( |\vec p| |\vec q|(1-\cos\theta) )^2
\end{align}
which tells us that  $
\vec p=0$, $
\vec q=0$, or $((\cos\theta=0)\land(\vec p\ne0\ne\vec q))$.
In the sub-case that $\vec p=0$ then $(\lambda,\mu)=(1,0)$ satisfies the lemma.
In the sub-case that $\vec q=0$ then $(\lambda,\mu)=(0,1)$ satisfies the lemma.
In the remaining sub-case, $(\lambda,\mu)=(|\vec q|,|\vec p|)$ satisfies the lemma.

In case (ii) at least one of $p$ and $q$ has a non-zero mass. Without loss of generality, assume that it is $p$ which has a non-zero mass $m_p\ne 0$.  If the relationship $\lambda p=\mu q$ can be proved in one frame it will be true in all frames, so we choose to prove it in a convenient frame, namely the rest-frame of $p$. In that frame our parameterization becomes:
$p^\mu=(m_p,\vec 0)$ and  $q^\mu=\left(\sqrt{m_q^2+|\vec q|^2},\vec q\right)$, and so
\begin{align}
0
&=
(p.q)^2-p^2q^2 
\\
&=
\left(m_p\sqrt{m_q^2+|\vec q|^2}-0\right)^2- m_p^2 m_q^2
\\
&=
    m_p^2 ( m_q^2+|\vec q|^2 )-m_p^2 m_q^2
\\
&=
    m_p^2 |\vec q|^2
\end{align}
and so either $m_p=0$ or $\vec q=0$.  In the sub-case in which $m_p$=0, then $p=(0,\vec 0)$ and so $(\lambda,\mu)=(1,0)$ satisfies the lemma.  In the sub-case in which $\vec q=0$ we have $q=(m_q,\vec 0)$ and so $(\lambda,\nu)=(m_q,m_p)$ satisfies the lemma (recall that $m_p$ is non-zero).
\end{proof}
\end{lemma}

\begin{definition}[\textbf{collision events}]
\label{def:collevent}
Motivated by Lemma~\ref{lem:collisioncond} we define the set of all \textbf{collision events} with a symmetry group $G$ to be the set of all \textbf{events} of the form $e=((p,q,\ldots),G)\in\mathscr{E}(\mathbb{V}^{n},G)$ for which $p^2 q^2<(p.q)^2$ (equivalently $\SYMGRAMTWO p q <0$):
$$
\mathscr{E}^c(\mathbb{V}^n,G)
=
\left\{
e
\mid
e=((p,q,\ldots),G)\in\mathscr{E}(\mathbb{V}^{n},G), p^2 q^2<(p.q)^2
\right\}
\subset
\mathscr{E}(\mathbb{V}^n,G).
$$
\begin{remark}Note to have enough space to hold $p$ and $q$  every \textbf{collision event} will need to have $n\ge2$.
\end{remark}
\end{definition}

\begin{lemma}
\label{lem:noncolfinegrainedprops}
If $p\in\mathbb V$ and $q\in\mathbb V$ and there exist $\lambda\ge0$ and $\mu\ge0$, not both zero, such that $\lambda p=\mu q$ then one of the following seven non-overlapping conditions pertains:
\begin{enumerate}
\item
$p=q=0$
\item
$p=0$ while $q\in\mathbb{M}$ (and so $q\ne0$),
\item
$q=0$ while $p\in\mathbb{M}$ (and so $p\ne0$),
\item
$p=0$ while $q\ne 0$ is an otherwise unconstrained element of $\mathbb{N}$,
\item
$q=0$ while $p\ne0$ is an otherwise unconstrained element of $\mathbb{N}$,
\item
$p=\mu q\ne0$ and $q=\lambda p\ne 0$ for some $\lambda>0$, $\mu>0$, and both $p$ and $q$ are in $\mathbb{M}$,
\item
$p=\mu q\ne0$ and $q=\lambda p\ne 0$ for some $\lambda>0$, $\mu>0$, and both $p$ and $q$ are in $\mathbb{N}$.
\end{enumerate}
\begin{proof}
Left as an exercise for the reader.
\end{proof}
\end{lemma}

\begin{lemma}
\label{lem:yhnhvfgdy}
If $p$ and $q$ are both in $\mathbb{V}$ then the following five statements (one of which makes used of a symbol $\Delta_2$ for a symmetric Gram Determinant, defined in Appendix ~\ref{app:gramnotation}) are equivalent:
\begin{itemize}
\item
$(p.q)=m_p m_q$,
\item
$(p.q)^2=p^2 q^2$,
\item
$\Delta_2(p,q)=0$,
\item
there exists a $\lambda\ge0$ and a $\mu\ge0$, not both zero, such that $\lambda p = \mu q$,
\hide{\item
there exists a $0\le\theta_{pq}\le\frac{\pi}{2}$ such that $ p \cos \theta =  q \sin \theta$,
\item
there exists a $-\frac{\pi} 4\le\theta_{pq}\le\frac{\pi}{4}$ such that $ p \cos(\theta_{pq}+\frac \pi 4) =  q \sin(\theta_{pq}+\frac \pi 4)$,}
\item
the event falls into any of the categories listed in Lemma~\ref{lem:noncolfinegrainedprops}.
\end{itemize}
\begin{proof}
The equivalence of the first three statements is a trivial consequence of the definitions of  invariant mass and Gram Determinants.  Equivalence of the fourth to the second follows from Lemma~\ref{lem:whatisnotacolleventONE}. Equivalence of the fifth to the fourth follows from Lemma~\ref{lem:noncolfinegrainedprops}.
\end{proof}
\end{lemma}

\begin{definition}[\textbf{non-collision events}]\label{def:noncollevents}
Lemmas~\ref{lem:pqbiggerthanmpmq}, \ref{lem:D2neverpositive} and \ref{lem:whatisnotacolleventONE}
show us that the only \textbf{events} with symmetry group $G$  which the definition of \textbf{collision event} \textit{excludes} are: (i)
\textbf{events} with symmetry group $G$ having fewer than two particles (i.e.~events in $\mathscr{E}(\mathbb{V}^{0},G)\cup\mathscr{E}(\mathbb{V}^{1},G)$), and (ii)
\textbf{events}  $e=((p,q,\ldots),G)\in\mathscr{E}(\mathbb{V}^{n},G)$ for which any of the  five (equivalent) statements in Lemma~\ref{lem:yhnhvfgdy} applies.
We therefore define an \textbf{event} in either category (i) or category (ii) above to be a \textbf{non-collision event}.
\end{definition}

\begin{lemma}[\textbf{every collision event has an axis in a frame in which $(p+q)$ is stationary}]
\label{lem:collisionshaveaxis}
If follows from Lemma~\ref{lem:collisioncond} and Definition~\ref{def:collevent} that the directions of mutual approach or departure of $p$ and $q$ in a \textbf{collision event} can always be used to  define a unique axis in the rest frame of ($p+q$).
\end{lemma}

\begin{remark}[\textbf{The relevance of of non-collision events}]
Collider physicists short on time might wish to skip the parts of this paper dealing with \textbf{non-collision events}. These events are included more for mathematical completeness  than for their usefulness in actual experiments.  The main reason that \textbf{non-collision events} are irrelevant to most real physics detectors is that non-collision events have $p$ and $q$ in a configuration in which each `overlaps' with the other (either as massive particles sharing the same rest frame, or as light-like momenta going in the same direction as each other, albeit with possibly different energies).  Not only do these properties  prevent $p$ and $q$ from representing colliding particles, they often also prevent $p$ and $q$ from being observed as momenta of final-state particles.  An exception could be if the detector is one which can distinguish collinear photons of different energies via the use of (say) scintillators with narrow-band sensitivities. Another exception could be the case in which overlapping massive-particle momenta can be inferred from subsidiary decay products (e.g.~two $Z$-bosons sharing a rest frame could be simultaneously observed if one $Z$-boson decayed to a dimuon pair while the other decaying to a dielectron pair).\footnote{Examples of concrete \textbf{non-collision events} are given in \textit{Remarks} at the end of Corollaries~\ref{cor:whennoncollpqabnoncollevisnonchiral} and \ref{cor:whennoncollpqmasslessabnoncollevisnonchiral}.}
\end{remark}

\begin{corollary}
\label{cor:somethiungusingaG}
If  $e = ((p,q,\cdots),G)$ is a \textbf{collision event} then $p+q\in \mathbb{M}.$
\begin{proof}
That $e = ((p,q,\cdots),G)$ is a \textbf{collision event} tells us that $p,q\in \mathbb{V}$ and $p^2q^2<(p.q)^2$. From Lemma~\ref{lem:vsumsareinv} it is already clear that $p+q\in\mathbb{V}$, so it only remains to show that $(p+q)^2>0.$ Since $p,q\in\mathbb{V}$, then $p^2\ge 0$, $q^2\ge 0$ and Lemme~\ref{lem:pqpos} tells us that $(p.q)\ge 0.$ 
Therefore $p^2q^2<(p.q)^2$ may be re-written $\sqrt{p^2} \sqrt{q^2} < (p.q)$ and so $(p+q)^2=p^2+q^2+2(p.q)>p^2+q^2+2\sqrt{p^2}\sqrt{q^2}=
\sqrt{p^2}+\sqrt{q^2})^2\ge 0$, and thus $(p+q)^2>0.$
\end{proof}
\end{corollary}
\begin{definition}[\textbf{transverse plane}]
\label{def:transverseplane}
If $e$ is a \textbf{collision event}, we define its \textbf{transverse plane}, $\TransversePlaneOp(e)$, to be the set containing all momenta in $\mathbb{V}$ whose spatial parts are perpendicular to $\vec p - \vec q$ in the $(p+q)$ rest frame:
$$
\TransversePlaneOp(e) = \left\{ \ r \ \middle|\  r\in\mathbb{V},\ 
\PerpABF r {p-q} {p+q}
 \right\}
$$
in which $G(\cdots)$ is a `Gram Determinant' as defined in Appendix ~\ref{app:gramnotation}, not a symmetry group $G$ such as that used in Corollary~\ref{cor:somethiungusingaG}.
\begin{remark}
This definition makes sense as Lemma~\ref{lem:adotbinsomeframe} tells us that a Gram determinant of the form shown measures the dot product between $\vec r$ and $\vec p-\vec q$ in the $(p+q)$-rest frame. 
\end{remark}
\end{definition}

\begin{lemma}
\label{lem:TransPlaneReflectionLemma}
The momenta $a$ and $b$ in a  \textbf{collision event}  $e = ((p,q,a,b,\cdots),G)$ are reflections of each other in $\TransversePlaneOp(e)$ if and only if
$$
\left(\PerpABF{a+b}{p-q}{p+q} \right)
\land
(\SYMGRAMTHR {a-b} p q =0).
$$
\begin{proof}
The first criterion says that in the $(p+q)$-rest frame $\vec a+\vec b$ is perpendicular to the beam axis (see Lemma~\ref{lem:adotbinsomeframe}).  The second criterion says that $(\vec a-\vec b)$ has no transverse component in the $(p+q)$-rest frame (this may be seen by coupling Lemma~\ref{lem:transversemom} with invariance properties of Gram determinants).
If the former is true, it asserts that $a_z=-b_z$ in the $(p+q)$-rest frame. If the latter is true it asserts that $\vec a_T=\vec b_T$ in the same frame.  If both are true it asserts that if $\vec a=(x,y,z)$ then $\vec b=(x,y,-z)$ which indeed shows that $\vec a$ and $\vec b$ are reflections of each other in the transverse plane.  Conversely, assuming that $\vec a=(x,y,z)$ then $\vec b=(x,y,-z)$
the given condition may be shown to be true by reversing the argument just given.
\end{proof}
\end{lemma}

%\section{Lorentz invariant expressions for energies, momenta and products thereof in particular frames}

\begin{lemma}
\label{lem:energyinframe}
If $a\in \mathbb{A}$ and $\Lambda\in\mathbb{M}$ then the energy of $a$ in the rest frame of $\Lambda$ is given by:
$$
\left. E_a \vphantom{\int} \right|_{\Lambda}
=
\frac {a.\Lambda} {\sqrt{\Lambda^2}}.
$$
\end{lemma}
\begin{lemma}
\label{lem:adotbinsomeframe}
If $a,b\in \mathbb{A}$ and $\Lambda\in\mathbb{M}$ then use of Lemma~\ref{lem:energyinframe} shows that $\vec a.\vec b$ in the rest frame of $\Lambda$ is given by:
$$
\left.
(\vec a . \vec b)\ 
\vphantom{\int}
\right|_{\Lambda}
=
\left.
E_a E_b
\vphantom{\int}
\right|_{\Lambda}
-
\left.
(
E_a E_b
-\vec a . \vec b
)
\vphantom{\int}
\right|_{\Lambda}
=
\frac {a.\Lambda} {\sqrt{\Lambda^2}}
\frac {b.\Lambda} {\sqrt{\Lambda^2}}
- a.b
=
\frac {
(a.\Lambda)(
b.\Lambda)
-
(a.b) \Lambda^2}
{\Lambda^2}
=
-\frac 1 {\Lambda^2}
\GRAMTWO a \Lambda
         b \Lambda
$$
in which we have used a Gram determinant and notation from Definition~\ref{def:gram}.
\end{lemma}
\begin{lemma}
\label{lem:threemominframe}
If $a\in \mathbb{A}$ and $\Lambda\in\mathbb{M}$, then using Lemma~\ref{lem:adotbinsomeframe} the square of the magnitude of the three momentum of $a$  in the rest frame of $\Lambda$ is:
$$
\left. |\vec a|^2 \vphantom{\int} \right|_{\Lambda}
=
\left. (\vec a.\vec a) \vphantom{\int} \right|_{\Lambda}
=-\frac 1 {\Lambda^2}
\GRAMTWO a \Lambda
         a \Lambda
= \frac{-\SYMGRAMTWO a \Lambda } {\Lambda^2}
$$
in which we have used a symmetric Gram determinant and notation from Definition~\ref{def:symgram}.
\end{lemma}

\begin{lemma}
\label{lem:transversemom}
If $a,p\in \mathbb{A}$, $\Lambda\in\mathbb{M}$, and $-\SYMGRAMTWO {p} {\Lambda} >0$, then the square of the transverse momentum of $\vec a$ with respect to the axis which $\vec p$ defines is: $$\left.|{\vec a}_T|^2 \vphantom{\int} \right|_{\Lambda}=\frac{
\SYMGRAMTHR  a p \Lambda 
}{
-\SYMGRAMTWO p \Lambda   %{\Lambda^2}
}.$$
\begin{proof}
An ugly argument follows:
\begin{align*}
\left. 
|{\vec a}_T|^2 \vphantom{\int} \right|_{\Lambda}
&=
\left. 
|{\vec a}|^2
\vphantom{\int}\right|_{\Lambda}
-
\left. 
\frac { (\vec a.\vec  p)^2} { |\vec p|^2} 
\vphantom{\int}\right|_{\Lambda}
=
 \frac{-\SYMGRAMTWO a \Lambda } {\Lambda^2}
-
\frac{
\frac 1 {\Lambda^4}
\GRAMTWO a \Lambda
         p \Lambda ^2
}{
\frac{-\SYMGRAMTWO p \Lambda } {\Lambda^2}
}
\\
&=
 \frac{-\SYMGRAMTWO a \Lambda } {\Lambda^2}
-
\frac{
\GRAMTWO a \Lambda
         p \Lambda ^2
}{
-\SYMGRAMTWO p \Lambda   {\Lambda^2}
}
=
 \frac{
 \SYMGRAMTWO a \Lambda \SYMGRAMTWO p \Lambda 
 -
 \GRAMTWO a \Lambda
         p \Lambda ^2}
{
-\SYMGRAMTWO p \Lambda   {\Lambda^2}
}
\\
&=
\frac{
\left|\begin{array}{cc}
     a^2 & a.\Lambda  \\
     a.\Lambda & \Lambda^2 
\end{array}\right|  
\left|\begin{array}{cc}
     p^2 & p.\Lambda  \\
     p.\Lambda & \Lambda^2 
\end{array}\right| 
 -
\left|\begin{array}{cc}
     a.p & a.\Lambda  \\
     p.\Lambda & \Lambda^2 
\end{array}\right|^2 
}{
-\SYMGRAMTWO p \Lambda   {\Lambda^2}
}
\\
&=
\frac{
(a^2 \Lambda^2-(a.\Lambda)^2)
(p^2 \Lambda^2-(p.\Lambda)^2)
-
((a.p) \Lambda^2 - (a.\Lambda)(p.\Lambda))^2
}{
-\SYMGRAMTWO p \Lambda   {\Lambda^2}
}
\\
&=
\frac{
(
a^2 \Lambda^2 
p^2 \Lambda^2
-
a^2 \Lambda^2
(p.\Lambda)^2
-
(a.\Lambda)^2
p^2 \Lambda^2
+
(a.\Lambda)^2
(p.\Lambda)^2
)
-
(
(a.p)^2 \Lambda^4 
+
(a.\Lambda)^2(p.\Lambda)^2
-2
(a.p) \Lambda^2  (a.\Lambda)(p.\Lambda)
)
}{
-\SYMGRAMTWO p \Lambda   {\Lambda^2}
}
\\
&=
\frac{
a^2 \Lambda^2 
p^2 \Lambda^2
-
a^2 \Lambda^2
(p.\Lambda)^2
-
(a.\Lambda)^2
p^2 \Lambda^2
+
(a.\Lambda)^2
(p.\Lambda)^2
-
(a.p)^2 \Lambda^4 
-
(a.\Lambda)^2(p.\Lambda)^2
+2
(a.p) \Lambda^2  (a.\Lambda)(p.\Lambda)
}{
-\SYMGRAMTWO p \Lambda   {\Lambda^2}
}
\\
&=
\frac{
a^2 \Lambda^2 
p^2 %\Lambda^2
-
a^2 %\Lambda^2
(p.\Lambda)^2
-
(a.\Lambda)^2
p^2 %\Lambda^2
%+
%(a.\Lambda)^2
%(p.\Lambda)^2
-
(a.p)^2 \Lambda^2 %\Lambda^2 
%-
%(a.\Lambda)^2(p.\Lambda)^2
+2
(a.p) %\Lambda^2  
(a.\Lambda)(p.\Lambda)
}{
-\SYMGRAMTWO p \Lambda   %{\Lambda^2}
}
=
\frac{
\SYMGRAMTHR  a p \Lambda 
}{
-\SYMGRAMTWO p \Lambda   %{\Lambda^2}
}.
\end{align*}
A more elegant argument would start by proving that 
$$
\left. 
(\vec a \times \vec b) .(\vec c \times \vec d)
\vphantom{\int}\right|_{\Lambda}
=
\frac{
\GRAMTHR a b \Lambda
         c d \Lambda
}{
\Lambda^2
}
$$
and would then use that result to evaluate the RHS of:
$$
\left. 
| {\vec a}_T |^2 
\vphantom{\int}\right|_{\Lambda}
=
\left.
\frac{\left| \vec a \times \vec p\right|^2}{|\vec p|^2} 
\vphantom{\int}\right|_{\Lambda}
.
$$
\end{proof}
\end{lemma}

\begin{lemma}
\label{lem:abcvectripprodinLIform}
In the rest frame of $\Lambda$ the vector triple product $\vec a \times \vec b \cdot \vec c$ of the spatial momenta within $a^\mu$, $b^\mu$ and $c^\mu$ is given by:
\begin{align}
\left.\left(\vec a \times \vec b \cdot \vec c\right)\right|_{\Lambda}
&=\epsLor_{0123}\frac{
\epsLor_{\mu\nu\sigma\tau} \Lambda^\mu a^\nu b^\sigma c^\tau}{\sqrt{\Lambda^2}}\qquad\text{(see notes on `epsilon notation' in Section~\ref{sec:epsilonNotation})}\label{eq:aputativetripleprod}
\\
&= \epsLor_{0123}\frac{[\Lambda,a,b,c]}{\sqrt{\Lambda^2}}\qquad\qquad\qquad\qquad\text{(see same section for notation).}\label{eq:somethingdependingoneps0123}
\end{align}
\begin{proof}
The expression of \eqref{eq:aputativetripleprod} is manifestly Lorentz invariant.  It is therefore true if it is true in some frame. In the rest frame of $\Lambda$ it takes the form $\Lambda^\mu =\left(\sqrt{\Lambda^2},(0,0,0)\right)$.
In this form the required result follows naturally, with the $\epsLor_{0123}$ needed to ensure the sign of the answer is independent of the choice of sign convention within $\epsLor$.\end{proof}
\end{lemma}

\begin{lemma}
\label{lem:transmomforcolevent}
In the $(p+q)$ rest frame of a \textbf{collision event}  $e = ((p,q,a,\cdots),G)$ the square of the transverse momentum of $a$ with respect to the beam axis is given by:
\begin{align}
\left. |{\vec a}_T|^2 \vphantom{\int}\right|_{\Lambda}
&=\label{eq:transmomforcollevent}
\frac{
\SYMGRAMTHR  a p q 
}{
-\SYMGRAMTWO p q   %{\Lambda^2}
}.
\end{align}
\begin{proof}
Lemma~\ref{lem:transversemom} gives us the first step: 
\begin{align*}
\left. |{\vec a}_T|^2 \vphantom{\int}\right|_{\Lambda}
&=
\frac{
\SYMGRAMTHR  a p {p+q} 
}{
-\SYMGRAMTWO p {p+q}   %{\Lambda^2}
}
\end{align*}
This simplifies to the result given using the property that matrix determinants (in this case Gram determinants) are insensitive to the addition of any multiple of some row (or column) to another row (or column).  Note also that \textbf{collision events} have $\SYMGRAMTWO p q <0$ by Definition~\ref{def:collevent}, so there is no risk of division by zero in Equation~(\ref{eq:transmomforcollevent}).
\end{proof}
\end{lemma}

\begin{lemma}\label{lem:modamodbisdotprodrel}
For any two vectors $\vec a,\vec b\in\mathbb{R}^n$:
$
(|\vec a|=|\vec b|)
\iff
((\vec a-\vec b).(\vec a+\vec b)=0).
$
\begin{proof}
$%\begin{align*}
%(a_z=b_z)\land
(|\vec a|=|\vec b|)
\iff
%(a_z=b_z)\land
(|\vec a|^2=|\vec b|^2)
%\\
\iff 
(|\vec a|^2 - |\vec b|^2=0)
%\\
\iff 
(\vec a.\vec a  -  \vec b.\vec b=0)
%\\
\iff 
(\vec a.\vec a + \vec a.\vec b - \vec b.\vec a  -  \vec b.\vec b=0)
%\\
\iff 
((\vec a-\vec b).(\vec a+\vec b)=0).
%\end{align*}
$
\end{proof}
\end{lemma}
\hide{
\begin{corollary}
\comCGL{[This corollary is correct but not helpful. It is just a special case of the identity given in Lemma~\ref{lem:momdifferencesidentity}.]}
For a \textbf{collision event}:
$%\begin{align}
( 
\SYMGRAMTWO a {p+q} 
=
\SYMGRAMTWO b {p+q}
)
\iff
\left(
\PerpABF{a-b}{a+b}{p+q}
\right).
%\end{align}
$
\begin{proof}
The proof follows immediately from Lemma~\ref{lem:modamodbisdotprodrel} in conjunction with  Lemmas~\ref{lem:adotbinsomeframe}
and 
\ref{lem:threemominframe} and
Definition~\ref{def:collevent}.
\end{proof}
\end{corollary}
}

%\section{Parity of collision events of the form $(pq)\rightarrow (abc) +X$ and $(pq)\rightarrow (ab) +X$}

\begin{definition}
$\qquad\pqabcColEvent \equivdef \mathscr{E}^c(P(5,2,(2,3)))\qquad\text{with}\qquad \mathbf{x}=(p,q,a,b,c).$
\begin{remark}
This definition is useful for referring quickly and concisely to \textbf{collision events} of the form $(pq)\rightarrow (abc)+X$ which live in $\mathscr{E}^c(P(5,2,(2,3)))$ with  $\mathbf{x}=(p,q,a,b,c)$ and have the symmetry group $G=SO^+(1,3)\times S_2(pq) \times S_3(abc).$
\end{remark}
\end{definition}
\begin{definition}
$\qquad\pqabcAllEvent \equivdef \mathscr{E}(P(5,2,(2,3)))\qquad\text{with}\qquad \mathbf{x}=(p,q,a,b,c).$
\begin{remark}
This definition is useful for referring quickly and concisely  to \textbf{events} having entities  $p$, $q$, $a$, $b$, $c$ and possibly also other things $X$, which live in $\mathscr{E}(P(5,2,(2,3)))$ with  $\mathbf{x}=(p,q,a,b,c)$ and have the symmetry group $G=SO^+(1,3)\times S_2(pq) \times S_3(abc).$
\end{remark}
\end{definition}

\begin{definition}
$\qquad\pqabColEvent \equivdef \mathscr{E}^c(P(4,2,(2,2)))\qquad\text{with}\qquad \mathbf{x}=(p,q,a,b).$
\begin{remark}
This definition is useful for referring quickly and concisely to \textbf{collision events} of the form $(pq)\rightarrow (ab)+X$ which live in $\mathscr{E}^c(P(4,2,(2,2)))$ with  $\mathbf{x}=(p,q,a,b)$ and have the symmetry group $G=SO^+(1,3)\times S_2(pq) \times S_2(ab).$
\end{remark}
\end{definition}
\begin{definition}
$\qquad\pqabAllEvent \equivdef \mathscr{E}(P(4,2,(2,2)))\qquad\text{with}\qquad \mathbf{x}=(p,q,a,b).$
\begin{remark}
This definition is useful for referring quickly and concisely  to \textbf{events} having entities  $p$, $q$, $a$, $b$ and possibly also other things $X$, which live in $\mathscr{E}(P(4,2,(2,2)))$ with  $\mathbf{x}=(p,q,a,b)$ and have the symmetry group $G=SO^+(1,3)\times S_2(pq) \times S_2(ab).$
\end{remark}
\end{definition}

\section{Conditions under which certain types of event are chiral}
\label{sec:condsforvariouschiraleventclasses}

\subsection{\textbf{Non-collision events} $e\in\pqabAllEvent$}

\begin{lemma}
\label{lem:everynoncollevinpqabxallisnonchiral}
Every \textbf{non-collision event} in $\pqabAllEvent$ is \textbf{non-chiral}.
\begin{proof}
Definition \ref{def:noncollevents} reminds us that for every non-collision event in $\pqabAllEvent$ there will exist a $\lambda\ge0$ and a $\mu\ge0$ (not both zero) such that $\lambda p=\mu q$.  Equivalently, $p$ and $q$ must be linearly dependent.  This means that of the four vectors in the set $\{p^\mu, q^\mu, a^\mu,b^\mu\}$ at most three are linearly independent. This is too few independent momenta to create a non-zero pseudoscalar though dotting with an $\epsLor_{\mu\nu\sigma\tau}$.
\end{proof}
\end{lemma}

\subsection{\textbf{Non-collision events} $e\in\pqabcAllEvent$ for which at least one of $p$ and $q$ is massive}
\label{sec:movingontononcollevtsinpqabc}

In this section we will find a condition which specifies whether or not a \textbf{non-collision event} $e\in\pqabcAllEvent$ is \textbf{non-chiral} given that which at least one of $p$ and $q$ is massive.  The condition is given in Corollary~\ref{cor:whennoncollpqabnoncollevisnonchiral}.
\begin{lemma}
For each \textbf{non-collision event} $e\in\pqabcAllEvent$ for which at least one of $p$ and $q$ is massive, there is a frame in which both $p$ and $q$ have no spatial momentum.  It is the rest frame of $(p+q)$.
\begin{proof}
Since at least one of $p$ and $q$ has a mass assume, without loss of generality, that $m_p>0$, i.e.~$p\in\mathbb{M}$.  The particle $p$ therefore has a rest frame.  We will show that in this frame the particle $q$ also has no spatial momentum.  Definition~\ref{def:noncollevents} tells us that  our non-collision event must be in one of the seven configurations listed in Lemma~\ref{lem:noncolfinegrainedprops}.  Our additional requirement that (without loss of generality) $p\in\mathbb{M}$ further constrains our event $e$ to be in configuration 3 or configuration 6 of Lemma~\ref{lem:noncolfinegrainedprops}. In configuration 3 we see that $q$ has no spatial momentum since it is the zero four-vector: $q=0$.  In configuration $6$ we are told that there exists a $\lambda>0$ such that $q=\lambda p$ which shows us that in a frame in which $p$ has no spatial momentum then $q$ shall also have none. 
\end{proof}
\end{lemma}

\begin{corollary}
In the $(p+q)$ rest frame every \textbf{non-collision event} $e\in\pqabcAllEvent$ for which at least one of $p$ and $q$ is massive may be parameterized by the five non-negative masses $m_p$, $m_q$, $m_a$, $m_b$ and $m_c$ together with the three spatial momenta $\vec a$, $\vec b$ and $\vec c$ which in that frame $\vec p=\vec q=\vec 0$. At least one of $m_p$ and $m_q$ is non-zero.  We may represent such an event notationally as follows:
\begin{align}
e = 
\left\{
\begin{array}{cc|ccc}
p & q & a & b & c \\
\hline
m_p & m_q & m_a & m_b & m_c \\
\vec 0 & \vec 0 & \vec a & \vec b & \vec c 
\end{array}
\right\}.
\label{form:atleastonemassivenoncoll}
\end{align}
\end{corollary}

\begin{remark}
The task in-hand is to identify which events of the form  \eqref{form:atleastonemassivenoncoll} are \textbf{chiral} and which are not.  Recall that Definition~\ref{def:unhanded} says that an \textbf{event} is \textbf{non-chiral} if and only if its parity inverted form,
\begin{align}
\mathscr P\cdot e
=
\left\{
\begin{array}{cc|ccc}
p & q & a & b & c \\
\hline
m_p & m_q & m_a & m_b & m_c \\
\vec 0 & \vec 0 & -\vec a & -\vec b & -\vec c 
\end{array}
\right\},
\end{align}
can be mapped onto $e$ by the action of an element of the symmetry group, which in our case consists of Lorentz boosts, rotations, permutations of $(pq$) and permutations of $(abc)$.  To get a match with $p$ and $q$ it is clear that by working in the $(p+q)$ frame we can now exclude Lorentz boosts from further consideration.  Likewise, $p$ and $q$ do not need permuting, and even if they did, such permutations would only affect $m_p$ and $m_q$.  We therefore only need consider permutations of $(abc)$ and global rotations $R$.
In principle there are six permutations of $(abc)$ to consider: $1$, $(ab)$, $(bc)$, $(ca)$, $(abc)$ and $(cba)$.  We shall only consider $h_1=R$, $h_2=(ab)\cdot R$ and $h_3=(abc)\cdot R$ since the effects of the remaining ones may be inferred from these three by symmetry.
\begin{alignat}{5}
h_1  \cdot \mathscr P\cdot e
&=&
R \cdot \mathscr P\cdot e
&&=
\left\{
\begin{array}{cc|ccc}
p & q & a & b & c \\
\hline
m_p & m_q & m_a & m_b & m_c \\
\vec 0 & \vec 0 & -R \vec a & - R\vec b & -R\vec c 
\end{array}
\right\},
\\
h_2  \cdot \mathscr P\cdot e
&=&
(ab) \cdot R \cdot \mathscr P\cdot e
&&=
\left\{
\begin{array}{cc|ccc}
p & q & a & b & c \\
\hline
m_p & m_q & m_b & m_a & m_c \\
\vec 0 & \vec 0 & -R \vec b & - R\vec a & -R\vec c 
\end{array}
\right\},
\\
h_3 \cdot  \mathscr P\cdot e
&=&\ 
(abc) \cdot R \cdot \mathscr P\cdot e
&&=
\left\{
\begin{array}{cc|ccc}
p & q & a & b & c \\
\hline
m_p & m_q & m_b & m_c & m_a \\
\vec 0 & \vec 0 & -R \vec b & - R\vec c & -R\vec a 
\end{array}
\right\}.  
\end{alignat}
\end{remark}

\subsubsection{Case $h_1$:  `$R\cdot \mathscr P \cdot e = e$'}

For events to be \textbf{non-chiral} under this case we shall need $\vec a=-R \vec a$, $\vec b=-R \vec b$ and $\vec c=-R \vec c$.   
%This requires that  $R$ be a rotation of 180 degrees about some axis, and requires that all of $\vec a$, $\vec b$ and $\vec c$ should lie in the plane which is orthogonal to the axis of this rotation. [It may be helpful to refer to Lemma~\ref{lem:whenonplaneperttorot}.]
If one or more of $\vec a$, $\vec b$ or $\vec c$ is zero, then the two remaining vectors will always be in a common plane, and so $R$ could be taken to be a 180-degree rotation about an axis normal to that common plane.  The only non-trivial constraint therefore comes when all of $\vec a$, $\vec b$ and $\vec c$ are non-zero. In such a case, those three vectors must be required to be in a common plane, and this may be enforced by requiring that  $[a,b,c,p+q]=0$ using the square-bracket contraction  notation defined in \eqref{eq:epscontractionnotation}. Since this statement is also trivially true in the case that any of $\vec a$, $\vec b$ or $\vec c$ is zero, we can say that Case $h_1$ is satisfied if and only if $$[a,b,c,p+q]=0.$$ 
\subsubsection{Case $h_2$: `$(ab)\cdot R\cdot \mathscr P \cdot e = e$'}
For events to be \textbf{non-chiral} under this case we shall need $m_a=m_b$, $\vec a=-R \vec b$, $\vec b=-R \vec a$ and $\vec c=-R \vec c$.   Note that these relations require that:  $(\vec a=R^2 \vec a)$;  $(\vec b=R^2\vec b)$ and; $((\vec a\ne \vec 0)\iff(\vec b\ne\vec 0)\iff(R^2=1))$.   Case $h_2$ will be satisfied, therefore, if $m_a=m_b$ and least one of the following is true:
\begin{enumerate}
\item
$\vec a=\vec b=\vec c=\vec 0$; or
\item
$\vec a=\vec b=0$ and $\vec c\ne\vec 0$ (in which case $R$ must be a rotation of 180 degrees about some axis perpendicular to $\vec c$); or
\item
$\vec a\ne\vec0\ne\vec b$ and $R=1$  (in which case it must be the case that $\vec a=-\vec b$ and $\vec c=\vec0$); or
\item
$\vec a\ne\vec0\ne\vec b$ and $R$ is a rotation of 180 degrees about some axis (in which case it must be the case that $\vec c$ is in the plane orthogonal to the axis of rotation).
\end{enumerate}
Sub-cases 1, 2 and 3 all have $[a,b,c,p+q]=0$ and so are also sub-cases of $h_1$ (though $R$ may have a different interpretation).  The only `new' way of being \textbf{non-chiral} provided by Case $h_2$ is therefore sub-case 4. It says that $m_a=m_b$ and that $\vec a$ and $\vec b$ are reflections of each other in a plane which also holds $\vec c$. We may exclude from consideration both the case in which $\vec a=\vec b$ and the case $\vec c=\vec 0$, since both of these place $\vec a$, $\vec b$ and $\vec c$ into a common plane, thus making them sub-cases of Case $h_1$.  The  geometry of the new \textbf{non-chiral} case is shown in Figure~\ref{fig:diag2}.
\begin{figure}
\begin{centering}
\includegraphics[width=0.5\textwidth]{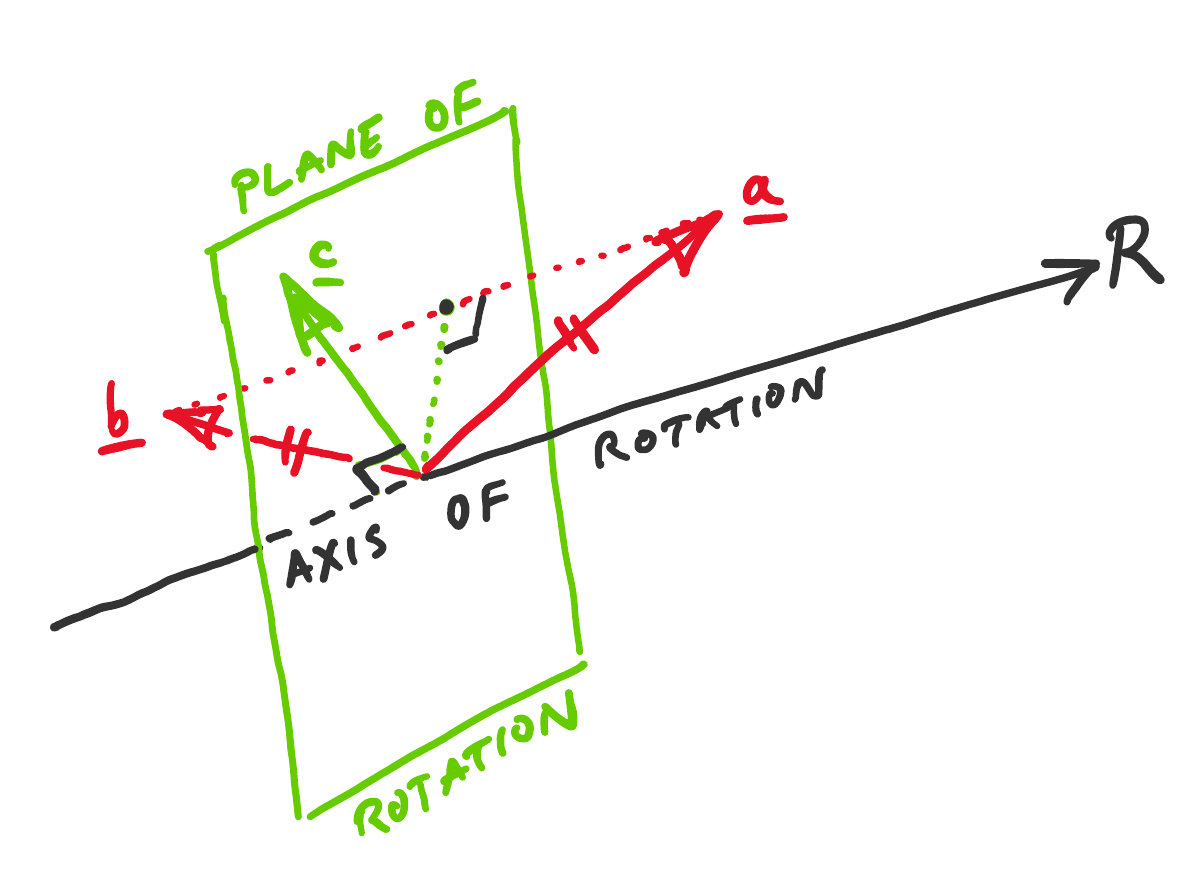}
\caption{
Diagram of a \textbf{non-chiral} configuration of momenta described in the text.
\label{fig:diag2}
}
\end{centering}
\end{figure} The condition shown can be summarised by the requirement:
\begin{align}
\left(m_a=m_b \right)
\land
\left(|\vec a|=|\vec b| \right)
\land
\left(\vec a\cdot\vec c = \vec b\cdot \vec c\right)
\end{align}
or in Lorentz-invariant form as
\begin{align}
\left(a^2 =b^2 \right)
\land
\left(
\GRAMTWO a {p+q} a {p+q} =\GRAMTWO b {p+q} b {p+q}
\right)
\land
\left(\GRAMTWO a {p+q} c {p+q} =\GRAMTWO b {p+q} c {p+q}\right).\label{eq:prerfgjhgfd}
\end{align}
Using Lemma~\ref{lem:tetanotherlemmaaboutgramdetsfughfjdn} (having set $s$ therein to $p+q$) we note that \eqref{eq:prerfgjhgfd} may be written more symmetrically as
\begin{align}
\left(a^2-b^2=0\right)\land
\left(\GRAMTWO a {p+q} a {p+q}-\GRAMTWO b {p+q} b {p+q}=0\right)\land
\left(\GRAMTWO {a-b} {p+q} {a+b+c} {p+q}=0\right)
\label{eq:possvariationone}
\end{align}
which, in the compressed notation later to be introduced in \eqref{eq:defofgramshorthand},
may be written as
\begin{align}
\left(a^2-b^2=0\right)\land
\left(g^a_a-g^b_b=0\right)\land
\left(g^{a-b}_{a+b+c} =0\right)
\label{eq:possvariationtwo}
\end{align}
or, if using also  \eqref{eq:gaagbbdifferencelemma}, may be written as:
\begin{align}
\left(a^2-b^2=0\right)\land
\left(g^{a-b}_{a+b}=0\right)\land
\left(g^{a-b}_{a+b+c} =0\right)
\label{eq:possvariationthr}
\end{align}
according to taste.  Note the similarity of the above constraints to those in the relevant row(s) of  \eqref{eq:secondtimeCisdefined}.
\subsubsection{Case $h_3$: `$(abc)\cdot R\cdot \mathscr P \cdot e = e$'}
For events to be \textbf{non-chiral} under this case we shall need $m_a=m_b=m_c$, $\vec a=-R \vec b$, $\vec b=-R \vec c$ and $\vec c=-R \vec a$.   Note that these relations require that: $(\vec a,\vec b,\vec c) =-R^3\cdot(\vec a,\vec b,\vec c)$.  By Lemma~\ref{lem:whenonplaneperttorot}, therefore, $\vec a$, $\vec b$ and $\vec c$ must all lie on a common plane which is perpendicular to the axis of rotation of $R$.  This means that any configuration satisfying Case $h_3$ is already to be found as a special case of Case $h_1$.  Case $h_3$ therefore adds no new ways of being \textbf{non-chiral}.
\hide{; $|\vec a|=|\vec b|=|\vec c|$; and $((\vec a=\vec b=\vec c=\vec 0)\lor(R^3=-1))$ is true.   Case $h_3$ will be satisfied, therefore, if $m_a=m_b=m_c$ and least one of the following is true:
\begin{enumerate}
\item
$\vec a=\vec b=\vec c=\vec 0$; or
\item
$|\vec a|=|\vec b|=|\vec c|\ne0$ and $R$ is a rotation of 180 degrees about some axis; or
\item
$|\vec a|=|\vec b|=|\vec c|\ne0$ and $R$ is a rotation of 60 degrees about some axis.
\end{enumerate}
Sub-case 1 is already a member of Case $h_1$ and so is not a `new' condition for \textbf{non-chiral} events and may be ignored.
Sub-case 2 implies that $R^2=1$ and hence we have $\vec a=-R\vec b =-R(-R\vec c)= R^2\vec c = \vec c$ and thus by symmetry $\vec a=\vec b=\vec c$.  This is already a sub-case of Case $h_1$ and so is not a `new' condition for \textbf{non-chiral} events and may be ignored. Sub-case 3 implies that }

\subsubsection{Summary of all cases}
\begin{corollary}
\label{cor:whennoncollpqabnoncollevisnonchiral}
A \textbf{non-collision event} $e\in\pqabcAllEvent$ for which at least one of $p$ and $q$ is massive is \textbf{non-chiral} if and only if
\begin{gather*}
[a,b,c,p+q]=0 \\
\separotor \lor {3.6cm}
\\
\left\{
\begin{array}{c}
\left(a^2-b^2=0\right)\land
\left(g^a_a-g^b_b=0\right)\land
\left(g^{a-b}_{a+b+c} =0\right)
\\
\separotor \lor {3.3cm}
\\
\left(b^2-c^2=0\right)\land
\left(g^b_b-g^c_c=0\right)\land
\left(g^{b-c}_{a+b+c} =0\right)
\\
\separotor \lor {3.3cm}
\\
\left(c^2-a^2=0\right)\land
\left(g^c_c-g^a_a=0\right)\land
\left(g^{c-a}_{a+b+c} =0\right)
\end{array}
\right\}.
\end{gather*}
Note that the conditions inside the curly brakets have been written in the form of \eqref{eq:possvariationtwo} however they could just have easily been written using the alternative variations \eqref{eq:possvariationone} or  \eqref{eq:possvariationthr}.  Note also use of the  square-bracket contraction  notation defined in \eqref{eq:epscontractionnotation}.
\begin{proof}
This result is just an `or' of the conditions obtained after consideration of cases $h_1$ to $h_3$ above, together with the other cases that would be obtained from them by symmetry.
\end{proof}
\end{corollary}

\begin{remark}
The following event
\begin{align}
e=\left[
\begin{array}{l}
p^\mu = \left(10, \left(0, 0, 0 \right)\right),\\
q^\mu = \left(10, \left(0, 0, 0 \right)\right),\\
a^\mu = \left(3, \left(1, 0, 0 \right)\right),\\
b^\mu = \left(2, \left(0, 1, 0 \right)\right),\\
c^\mu = \left(1, \left(0, 0, 1 \right)\right]
\end{array}
\right]
\label{firstnoncollecenteindfgdfg}
\end{align}
(which is written using the Lorentz vector
the display conventions given in Appendix~\ref{sec:lorentzvectornotation}) 
is an example of a \textbf{non-collision event} $e\in\pqabcAllEvent$ for which at least one of $p$ and $q$ is massive and which is \textbf{chiral}.  The \textbf{non-collision} status may be verified by checking that $p$ and $q$ meet the requirements of Definition~\ref{def:noncollevents}. The \textbf{chiral} status may be checked by noting that since $a^2=3^2-1^2=8$, $b^2=2^2-1^2=3$ and $c^2=1^2-1^2=0$, every line of the condition given in Corollary~\ref{cor:whennoncollpqabnoncollevisnonchiral} is violated because:
\begin{align*}
%(a+b+c)^2  \ge a^2+b^2+c^2 = 8+3+0 = 11 &> 0,\\
%p &\ne 0,\\
[a,b,c,p+q]\propto \left|\ 
\begin{matrix}
3 & 2 & 1 & 20 \\
1 & 0 & 0 & 0 \\
0 & 1 & 0 & 0 \\
0 & 0 & 1 & 0
\end{matrix}
\ 
\right| = -20 &\ne 0,\\
a^2-b^2=8-3 = +5&\ne0,\\
b^2-c^2=3-0 = +3&\ne0,\\
c^2-a^2=0-8 = -8&\ne0.
\end{align*}
\end{remark}

\subsection{\textbf{Non-collision events}
$e\in\pqabcAllEvent$ for which both $p$ and $q$ are massless}
\label{sec:ncevpqmasslessabc}

In this section we will find a condition which specifies whether or not a \textbf{non-collision event} $e\in\pqabcAllEvent$ is \textbf{non-chiral} given that both of $p$ and $q$ are massless.  The condition is given in Corollary~\ref{cor:whennoncollpqmasslessabnoncollevisnonchiral}.

If $(a+b+c)^2=0$ then all of $a$, $b$ and $c$ must be massless and must all point in the same direction (i.e.~there will exist $\alpha\ge0$ and $\beta\ge0$, not both zero, such that $\alpha a=\beta b$, \textit{etc}). In such a case there are at most three linearly independent Lorentz-vectors in the event, namely $p$ and $q$ together with at most one of $a$, $b$ and $c$. These are too few  to construct a non-zero pseudoscalar from contraction with an epsilon alternating tensor.  Such an event therefore \textbf{non-chiral}.

Similarly, if $p=0=q$ then again there are too few linearly independent Lorentz-vectors to construct a non-zero pseudoscalar.  Such events are therefore also \textbf{non-chiral}.

If $(a+b+c)^2>0$ and at least one of $p$ and $q$ is not the zero-vector, then more work needs to be done to determine whether or not the event is \textbf{chiral}. However, in this case the non-zero mass of $(a+b+c)$ means there is a natural frame in which we can work, namely that in which $(a+b+c)$ is at rest.  In the remainder of Section~\ref{sec:ncevpqmasslessabc} (and its subsections) we will work in that frame and will assume (without loss of generality) that $p$ is not the zero-vector.

Although we have already chosen a rest-frame in which to work, we have not yet used the freedom to orient that frame.
Definition~\ref{def:noncollevents} tells us that, since we are considering only \textbf{non-collision events}, there must exist a $\lambda\ge0$ and a $\mu\ge 0$, not both zero, such that $\lambda p =  \mu q$.  This means that in the frame in which we have chosen to work, at least one of $\vec p$ and $\vec q$ is non-zero, and furthermore if both are non-zero then both point in the same direction.  We therefore choose to orient our frame such that $\vec p$ and $\vec q$ are parallel to the $z$-axis and so that neither of them has a negative $z$-component.\footnote{At least one of them will have a positive $z$-component.}

\begin{corollary}
Every \textbf{non-collision event} $e\in\pqabcAllEvent$ having $(a+b+c)^2>0$ and having at least one of $p$ and $q$ non-zero may, in the $(a+b+c)$-rest frame, be parameterized by the three non-negative masses $m_a$, $m_b$ and $m_c$, two non-negative scalars $p_z$ and $q_z$ (at least one of which is non-zero), three transverse momenta $\vec a_T$, $\vec b_T$ and $\vec c_T$ (each having two components) jointly satisfying $a_T +\vec b_T+\vec c_T=\vec 0_T$ and three other scalars $a_z$, $b_z$ and $c_z$ jointly satisfying $a_z+b_z+c_z=0$. We may represent such an event notationally as follows:
\begin{align}
\hat e
=
\left[\begin{array}{cc|ccc}
     0 &0 & m_a & m_b & m_c \\
     0 & 0 & \mathbf{a} & \mathbf{b} & \mathbf{c} \\
     p_z & q_z & a_z & b_z & c_z
\end{array}\right]\qquad \text{with}\qquad\left(\begin{array}{l}
p_z, q_z,m_a,m_b,m_c\ge0, \\
p_z+q_z>0, \\
\mathbf{a},\mathbf{b},\mathbf{c}\in\mathbb{C},\\
\mathbf{a}+\mathbf{b}+\mathbf{c}=0,\\
a_z,b_z,c_z\in\mathbb{R}\ \text{and}\ 
a_z+b_z+c_z=0
\end{array}\right).\label{form:bothmasslesscoll}
\end{align}
\end{corollary}

\begin{remark}
The task in-hand is to identify which events of the form  \eqref{form:bothmasslesscoll} are \textbf{chiral} and which are not.  Recall that Definition~\ref{def:unhanded} says that an \textbf{event} is \textbf{non-chiral} if and only if its parity inverted form,
\begin{align}
\mathscr P\cdot e
=
\left[\begin{array}{cc|ccc}
     0 &0 & m_a & m_b & m_c \\
     0 & 0 & -\mathbf{a} & -\mathbf{b} & -\mathbf{c} \\
     -p_z & -q_z & -a_z & -b_z & -c_z
\end{array}\right]
,
\end{align}
can be mapped onto $e$ by the action of an element of the symmetry group, which in our case consists of Lorentz boosts, rotations, permutations of $(pq$) and permutations of $(abc)$. 
Since every element $g$ in the symmetry group has an inverse, we may instead choose to apply the above definition not to $\mathscr P\cdot e$ but to $g\cdot\mathscr P\cdot e$.  We will do so, choosing $g$ to be a rotation of 180 degrees about the $y$-axis, this being something that will ensure that the momenta of $p$ and $q$ are thereby already mapped into agreement:
\begin{align}
R_y(\pi) \cdot \mathscr P\cdot e
=
\left[\begin{array}{cc|ccc}
     0 &0 & m_a & m_b & m_c \\
     0 & 0 & \mathbf{a}^* & \mathbf{b}^* & \mathbf{c}^* \\
     p_z & q_z & a_z & b_z & c_z
\end{array}\right].
\end{align}
What elements of the symmetry group can be used to map $R_y(\pi) \cdot \mathscr P\cdot e$ back onto $e$?  Since we already have a match for $p$ and $q$ we must not break this match. This rules out the use of rotations except for those about the $z$-axis.  Boosts are also ruled out for those same reason.  Interchanges of $p$ and $q$ could, in principle, be applied if $p_z=q_z$. However such swaps achieve nothing.  All we need consider, therefore, are rotations $R_z(\theta)$ about the $z$-axis by an arbitrary angle $\theta$, together with any permutation of $a$, $b$ and $c$.  Though there are six such permutations, it will only be necessary to examine three of them (1, (ab) and (abc)) since the behaviour of the others may be determined by symmetry from these.   We therefore consider three cases as follows:
\begin{align}
 R_z(\theta) \cdot R_y(\pi) \cdot \mathscr P\cdot e
&=
\left[\begin{array}{cc|ccc}
     0 &0 & m_a & m_b & m_c \\
     0 & 0 & \mathbf{a}^* e^{i\theta} & \mathbf{b}^* e^{i\theta}& \mathbf{c}^* e^{i\theta}\\
     p_z & q_z & a_z & b_z & c_z
\end{array}\right],
\\
(ab)\cdot R_z(\theta) \cdot R_y(\pi) \cdot \mathscr P\cdot e
&=
\left[\begin{array}{cc|ccc}
     0 &0 & m_b & m_a & m_c \\
     0 & 0 & \mathbf{b}^* e^{i\theta} & \mathbf{a}^* e^{i\theta}& \mathbf{c}^* e^{i\theta}\\
     p_z & q_z & b_z & a_z & c_z
\end{array}\right],
\\
(abc)\cdot R_z(\theta) \cdot R_y(\pi) \cdot \mathscr P\cdot e
&=
\left[\begin{array}{cc|ccc}
     0 &0 & m_b & m_c & m_a \\
     0 & 0 & \mathbf{b}^* e^{i\theta} & \mathbf{c}^* e^{i\theta}& \mathbf{a}^* e^{i\theta}\\
     p_z & q_z & b_z & c_z & a_z
\end{array}\right].
\end{align}
The above cases will be considered with the complex number $\mathbf{a}$, $\mathbf{b}$ and $\mathbf{c}$ in polar form, \textit{vis} $\mathbf{a}=|\mathbf{a}|e^{i\alpha}$,
$\mathbf{b}=|\mathbf{b}|e^{i\beta}$ and
$\mathbf{c}=|\mathbf{c}|e^{i\gamma}$.
\end{remark}

\subsubsection{Case $h_1$: `$R_z(\theta) \cdot R_y(\pi) \cdot \mathscr P\cdot e = e$'}  

This case requires $$
((\mathbf{a}=0)\lor(\theta-\alpha+2 n_a \pi = \alpha))
\land
((\mathbf{b}=0)\lor(\theta-\beta+2 n_b \pi = \beta))
\land
((\mathbf{c}=0)\lor(\theta-\gamma+2 n_c \pi = \gamma))
$$
which is equivalent to
$$
((\mathbf{a}=0)\lor(\alpha=\frac 1 2 \theta+ n_a \pi ))
\land
((\mathbf{b}=0)\lor(\beta=\frac 1 2 \theta+ n_b \pi ))
\land
((\mathbf{c}=0)\lor(\gamma=\frac 1 2 \theta+ n_c \pi ))
$$
which is the same as requiring that all non-zero elements of the set $\{\mathbf{a},\mathbf{b},\mathbf{c}\}$ must have the same complex argument  to within a $\pi$, This is itself the same as saying that $\vec a$, $\vec b$, $\vec c$, $\vec p$ and $\vec q$ are all in a common plane.   Note: it is no surprise that $\vec a$, $\vec b$, $\vec c$ are in a common plane.   This is inevitable given that we are working in the $(a+b+c)$ rest frame.  The `extra' requirement here is only that $p$ and $q$ can be added while remaining in one plane.  Since $\vec p$ and $\vec q$ are necessarily collinear (and not anti-parallel) it is sufficient to test that $(\vec p+\vec q)$ lie in the same plane as each of the other two vectors. I.e.~in the $(a+b+c)$ rest frame we wish to enforce the requirement that $\vec a\times \vec b\cdot(\vec p+\vec q)=\vec b\times \vec c\cdot(\vec p+\vec q)=\vec c\times \vec a\cdot(\vec p+\vec q)=0$.  Using the square-bracket contraction notation defined in \eqref{eq:epscontractionnotation} this constraint may be written in Lorentz invariant form as 
$[a,b,p+q,a+b+c]=[b,c,p+q,a+b+c]=[c,a,p+q,a+b+c]=0$. However, given standard properties of determinants these are all equivalent to the single constraint $$[a,b,c,p+q]=0.$$

\subsubsection{Case $h_2$: `$(ab)\cdot R_z(\theta) \cdot R_y(\pi) \cdot \mathscr P\cdot e = e$'}  

This case requires $m_a=m_b$, $a_z=b_z$ and $|\mathbf{a}|=|\mathbf{b}|$ together with $$
((\mathbf{a}=0)\lor(
(\theta-\beta+2 n_a \pi = \alpha)
\land
(\theta-\alpha+2 n_b \pi = \beta)
))
\land
((\mathbf{c}=0)\lor(\theta-\gamma+2 n_c \pi = \gamma))
$$
which is equivalent to $m_a=m_b$, $a_z=b_z$ and $|\mathbf{a}|=|\mathbf{b}|$ together with \begin{align}
\label{eq:sduhjhbnghytfdfghjk}
((\mathbf{a}=0)\lor(
\frac 1 2 (\alpha+\beta) = \frac 1 2 \theta+ n \pi
))
\land
((\mathbf{c}=0)\lor(\gamma = \frac 1 2 \theta+ n_c \pi )).
\end{align}
Given that we have a constraint $|\mathbf{a}|=|\mathbf{b}|$ we may break this down into four sub-cases:
\begin{enumerate}
\item
$|\mathbf{a}|=|\mathbf{b}|=|\mathbf{c}|=0$,
\item
$|\mathbf{a}|=|\mathbf{b}|=0$ and $|\mathbf{c}|\ne0$, and
\item
$|\mathbf{a}|=|\mathbf{b}|\ne0$. \end{enumerate}
Sub-cases 1 and 2 above are manifestly already contained within Case $h_1$ and so do not represent `new' sources of \textbf{non-chirality} for these events and so may be ignored.
Only sub-case 3 tells us something new.
It tells us that a \textbf{non-collision event} of the type we are considering will be \textbf{non-chiral} if, in the $(a+b+c)$ rest frame, $\vec a$ and $\vec b$  have transverse\footnote{Here `transverse' means with respect to the $(\vec p+\vec q)$-axis in our frame.)} momenta with equal magnitudes and have the same component in the $(\vec p+\vec q)$-direction provided that, whenever $\vec c$ has a non-zero transverse momentum, $\vec c$ lies in the plane  in which $\vec a$ mirrors $\vec b$ and which also includes $(\vec p+\vec q)$.

Note that the mirror plane just mentioned is poorly defined if $\vec a$ and $\vec b$ are axial (that is to say they are along $\vec p+\vec q$ direction) however this case may be ignored since it is a sub-case of Case $h_1$.  The case where $\vec a=\vec b$ and off axis may also be ignored as another a sub-case of Case $h_1$, even though the mirror plane is well defined in this case.

The forms of \textbf{non-chiral event} contained in the above and not already contained within the $[a,b,c,p+q]=0$ condition of Case $h_1$ may be written
\begin{align}
(m_a=m_b)
\land
(a_z=b_z)
\land
(|\vec a|=|\vec b|)
\land
((\vec a-\vec b)\cdot \vec c = 0)
\end{align}
which, having defined $\Sigma=a+b+c$, can re-written in Lorentz-invariant form as
\begin{align}
(a^2=b^2)
\land
\left(
\GRAMTWO {a-b} \Sigma {p+q} \Sigma
=0
\right)
\land
\left(
\SYMGRAMTWO a \Sigma
=
\SYMGRAMTWO b \Sigma
\right)
\land
\left(
\GRAMTWO a \Sigma c \Sigma
=
\GRAMTWO b \Sigma c \Sigma
\right).\label{eq:prerfddsggjhgfd}
\end{align}
Using Lemma~\ref{lem:tetanotherlemmaaboutgramdetsfughfjdn} (having set $s$ therein to $\Sigma$) we note that \eqref{eq:prerfddsggjhgfd} may be written more symmetrically as
\begin{align}
(a^2=b^2)
\land
\left(
\GRAMTWO {a-b} \Sigma {p+q} \Sigma
=0
\right)
\land
\left(
\SYMGRAMTWO a \Sigma
=
\SYMGRAMTWO b \Sigma
\right)
\land
\left(
\GRAMTWO a \Sigma \Sigma \Sigma
=
\GRAMTWO b \Sigma \Sigma \Sigma
\right)
\end{align}
however the properties of Gram determinants ensure that $\left(
\GRAMTWO a \Sigma \Sigma \Sigma
=
\GRAMTWO b \Sigma \Sigma \Sigma
\right)$ is always true (since the left and right hand sides of this equation are both identically zero) so we may write our condition in the simpler and final form:
\begin{align}
(a^2=b^2)
\land
\left(
\GRAMTWO {a-b} \Sigma {p+q} \Sigma
=0
\right)
\land
\left(
\SYMGRAMTWO a \Sigma
=
\SYMGRAMTWO b \Sigma
\right)
.
\end{align}
\subsubsection{Case $h_3$: `$(abc)\cdot R_z(\theta) \cdot R_y(\pi) \cdot \mathscr P\cdot e = e$'}  

In this case, although many constraints are required to be satisfied we will only need to consider three. They are those which say $a_z=b_z=c_z$. These, together with the already made requirement that $a_z+b_z+c_z=0$ reduce to $a_z=b_z=c_z=0$.  These place all of $\vec a$, $\vec b$ and $\vec c$ on the plane transverse to $\vec p+\vec q$ meaning that, in the $(a+b+c)$ rest frame $\vec a \times \vec b \cdot (p+q)$, $\vec b \times \vec c \cdot (p+q)$ and $\vec c \times \vec a \cdot (p+q)$  must all be zero.  This case is therefore a sub-case of Case $h_1$ and adds nothing new.
\subsubsection{Summary of all cases}
\begin{corollary}
\label{cor:whennoncollpqmasslessabnoncollevisnonchiral}
A \textbf{non-collision event} $e\in\pqabcAllEvent$ for which both $p$ and $q$ are massless is \textbf{non-chiral} if and only if 
\begin{gather*}
(a+b+c)^2=0 \\
\separotor \lor {3.6cm} \\
(p=0) \land (q=0) \\
\separotor \lor {3.6cm} \\
[a,b,c,p+q]=0 \\
\separotor \lor {3.6cm}
\\
\left\{
\begin{array}{c}
(a^2=b^2)
\land
\left(
\GRAMTWO {a-b} \Sigma {p+q} \Sigma
=0
\right)
\land
\left(
\SYMGRAMTWO a \Sigma
=
\SYMGRAMTWO b \Sigma
\right)
\\
\separotor \lor {4cm}
\\
(b^2=c^2)
\land
\left(
\GRAMTWO {b-c} \Sigma {p+q} \Sigma
=0
\right)
\land
\left(
\SYMGRAMTWO b \Sigma
=
\SYMGRAMTWO c \Sigma
\right)
\\
\separotor \lor {4cm}
\\
(c^2=a^2)
\land
\left(
\GRAMTWO {c-a} \Sigma {p+q} \Sigma
=0
\right)
\land
\left(
\SYMGRAMTWO c \Sigma
=
\SYMGRAMTWO a \Sigma
\right)
\end{array}
\right\}
\end{gather*}
wherein $\Sigma$ stands for $a+b+c$ and where the square-bracket contraction $[a,b,c,p+q]$ uses notation defined in \eqref{eq:epscontractionnotation}.
\begin{proof}
This result is just an `or' of the conditions obtained after consideration of cases $h_1$ to $h_3$ above, together with the other cases that would be obtained from them by symmetry, and together with the two earlier results just preceding them concerning (i) the case $(a+b+c)^2=0$ and (ii) the case $p=q=0$.
\end{proof}
\end{corollary}

\begin{remark}
The following event
\begin{align}
e=\left[
\begin{array}{l}
p^\mu = \left(2, \left(0, 0, 2 \right)\right),\\
q^\mu = \left(1, \left(0, 0, 1 \right)\right),\\
a^\mu = \left(3, \left(1, 0, 0 \right)\right),\\
b^\mu = \left(2, \left(0, 1, 0 \right)\right),\\
c^\mu = \left(2, \left(-1, -1, 0 \right)\right]
\end{array}
\right]
\label{secondnoncollecenteindfgdfg}
\end{align}
(which is written using the Lorentz vector
the display conventions given in Appendix~\ref{sec:lorentzvectornotation}) 
is an example of a \textbf{non-collision event} $e\in\pqabcAllEvent$ for which both $p$ and $q$ are massless and which is \textbf{chiral}.  The \textbf{non-collision} status may be verified by checking that $p$ and $q$ meet the requirements of Definition~\ref{def:noncollevents}. The \textbf{chiral} status may be checked by noting that since $a^2=3^2-1^2=8$, $b^2=2^2-1^2=3$ and $c^2=2^2-1^2-1^2=2$, every line of the condition given in Corollary~\ref{cor:whennoncollpqmasslessabnoncollevisnonchiral} is violated because:
\begin{align*}
(a+b+c)^2  \ge a^2+b^2+c^2 = 8+3+2 = 13 &> 0,\\
p &\ne 0,\\
[a,b,c,p+q]\propto \left|\ 
\begin{matrix}
3 & 2 &  2 & 3 \\
1 & 0 & -1 & 0 \\
0 & 1 & -1 & 0 \\
0 & 0 &  0 & 3
\end{matrix}
\ 
\right| = 21 &\ne 0,\\
a^2-b^2=8-3 = +5&\ne0,\\
b^2-c^2=3-2 = +1&\ne0,\\
c^2-a^2=2-8 = -6&\ne0.
\end{align*}
\end{remark}

\subsection{\textbf{Collision events} in either $\pqabcColEvent$ or $\pqabColEvent$}

\label{sec:averyimportantpairofresults}

%When are events of this type non-chiral?
\begin{theorem}
\label{thm:nonchiralpqabcEvents}
\label{thm:chiralpqabcEvents}
A \hypertarget{link:chiralpqabcevents}{ \textbf{collision event} $e\in\pqabcColEvent$ is \textbf{chiral}} if and only if $C$ is true, where
\begin{gather*}
C\equiv\numberthis \label{eq:veryfirstCdef}\\
\ANTIbignonchiralcond
\end{gather*}
which makes use of square-bracket contraction notation defined in \eqref{eq:epscontractionnotation} and which uses notation for Gram determinants $G$ and $\Delta$  defined in \eqref{eq:gram} and \eqref{eq:symmgram}.
\end{theorem}
%\begin{remark}
%See also \hyperlink{link:abridgedchiralpqabcevents}{the abridged form}.
%\end{remark}
\begin{proof} 
The proof of Theorem~\ref{thm:nonchiralpqabcEvents} is long and not particularly illuminating.
Exposition of it is therefore delayed until Sections~\ref{sec:longproofstart} to \ref{sec:longproofend} so as not to disrupt the arguments about to be made.
%It has therefore been deposited in  Section~\ref{sec:longproof} of the Appendix rather than here so as not to disrupt the overall structure of the paper. %Sub-Section~\ref{ss:endofGIANTPROOF001}.
\end{proof}

\begin{corollary}
It follows from Theorem~\ref{thm:nonchiralpqabcEvents} almost by inspection\footnote{The result may be obtained either by substituting the zero four-vector in place of $c$ in the statement of Theorem~\ref{thm:nonchiralpqabcEvents} and then simplifying the resulting expression, or by considering the steps of the proof of Theorem~\ref{thm:nonchiralpqabcEvents} which would have been required had there been no $c$ to begin with.} that 
%a \textbf{collision event} $e\in\pqabEvent$ is \textbf{non-chiral} if and only if  
%\begin{gather*}
%\smallnonchiralcond
%\end{gather*}
%or, equivalently, 
a \textbf{collision event} $e\in\pqabColEvent$ is \textbf{chiral} if and only if $B$ is true, where
\begin{gather*}
B\equiv\numberthis \label{eq:veryfirstBdef}\\
\ANTIsmallnonchiralcond
\end{gather*}
which makes use of square-bracket contraction notation defined in \eqref{eq:epscontractionnotation} and which uses notation for Gram determinants $G$ and $\Delta$  defined in \eqref{eq:gram} and \eqref{eq:symmgram}.
\end{corollary}

\begin{definition}
\label{def:helpfulnotationforsymms}
Because the constraints $C$ and $B$ just written are rather big and bulky, we will benefit from writing them in a compressed notation. It will be particularly important to be able to see, at a glance, which quantities within them change sign under interchange of $p$ and $q$ (these will be denoted with a subscript ${\cdot}_{pq}$), which quantities are completely antisymmetric in $a$, $b$ and/or $c$ (these will be denoted with superscripts ${\cdot}^{abc}$ or ${\cdot}^{ab}$ etc.), and which change sign under parity (these will be denoted with a $\mathscr{P}$).   Having made those symmetries clear, it will then be easier to plot a path toward combining the ingredients into functions which are invariant under the required interchange symmetries. We therefore make the following definitions:
%See \hyperlink{link:chiralpqabcevents}{the chiral event cond}.
\hypertarget{link:abridgedchiralpqabcevents}{}
\begin{align}
\mathscr{P}^{ab}_{pq}(0)
&\equiv
\epsShort^{ab}_{pq}
\equiv
[a,b,p,q],\qquad\text{(square-bracket contraction defined in \eqref{eq:epscontractionnotation})}
\\
%\delta(x)
%&=
%\SYMGRAMTHR x p q,
%  \\
  \POne\label{eq:p1def}
  &\equiv
  \epsShort^{ab}_{pq} 
  +
  \epsShort^{bc}_{pq} 
  +
  \epsShort^{ca}_{pq},
\\
\PTwo\label{eq:p2def}
&\equiv
  \epsShort^{ab}_{pq} 
  \epsShort^{bc}_{pq} 
  \epsShort^{ca}_{pq},
  \\
  \PThr
  &\equiv
  (\epsShort^{ab}_{pq} -
  \epsShort^{bc}_{pq} )
  (\epsShort^{bc}_{pq} -
  \epsShort^{ca}_{pq})
  (\epsShort^{ca}_{pq}-
  \epsShort^{ab}_{pq} ),\label{eq:p3def}
    \\
    \PCom %\mathscr{P}^{abc}_{pq}(\mathbb{C})
    &\equiv
    \POne
    + i
    \PTwo
    \\
  g^{x}_{y} 
  &\equiv \GRAMTWO{x}{p+q}
          {y}{p+q},\label{eq:defofgramshorthand}
     \qquad\text{(using a Gram determinant defined in \eqref{eq:gram})}
%  \\
%  \mathscr{P}^{abc}(8)
%  &\equiv [a,b,c,p+q]
%  \\
%  \mathscr{P}^{abc}_{pq}(9)
%  &\equiv [a,b,c,p-q]
  \\  
  G^{\phantom{ab}}_{pq}(0)\equiv
  \FZerFIXEDWIDTH
  %F^{\phantom{abc}}_{pq}(0)
  &\equiv
p^2-q^2,
\\
\FOneFIXEDWIDTH
  %F^{\phantom{abc}}_{pq}(1)
  \label{eq:F1def}
  &\equiv
  g^{a+b+c}_{p-q},
\\
 % F^{\phantom{abc}}_{pq}(2)
  \FTwoFIXEDWIDTH
  \label{eq:F2def}
  &\equiv
  g^a_{p-q} 
  g^b_{p-q} 
  g^c_{p-q},
\\ 
\FThrFIXEDWIDTH
  %F^{abc}_{pq}(3)  
  &\equiv
  g^{a-b}_{p-q} 
  g^{b-c}_{p-q} 
  g^{c-a}_{p-q},\label{eq:F3abcpqdef}
  \\
 \FZerOneTwo
 %F^{\cdot\cdot\cdot}_{pq}(\!\!\!\!\!
%\text{ \tiny{
 %$\begin{array}{c}0\\1\\2\end{array}$
 %}
 %}\!\!\!\!)
 &\equiv\left(
 \FZer,
 \FOne,
 \FTwo
 \right),
  \\
    G^{\phantom{ab}}_{pq}(1) &\equiv
  g^{a+b}_{p-q},
  \\
    G^{ab}_{pq}(2) &\equiv
  g^{a-b}_{p-q},\qquad\qquad\text{and}
    \\
%  g(x)
%  &= \GRAMTWO{x}{p+q}
%          {a+b+c}{p+q},
%\\  
  f^{ab}
  &\equiv
a^2-b^2.
\end{align}
\end{definition}
\begin{lemma}
\label{lem:pqabcchiralcond}
A \textbf{collision event} $e\in\pqabcColEvent$ is \textbf{chiral} if and only if $C$ is true where:
\begin{gather*}
C = \numberthis \label{eq:secondtimeCisdefined} \\
\toperator \\
 %%%%%%%%%%%%%%%%%% h1 %%%%%%%%%%%%%%%%%%%
(
\POne
 % \mathscr{P}^{abc}_{pq}(1)
  \ne0
)
\lor
(
\PTwo
%\mathscr{P}^{abc}_{pq}(2)
  \ne0
)
\lor
(
\PThr
%\mathscr{P}^{\phantom{abc}}_{pq}(3)
  \ne0
)
\\ 
\separatorlor \land
\\ %%%%%%%%%%%%%%%%%% h2,3,4: %%%%%%%%%%%%%%%%%%%
\left\{
\begin{array}{c}
(f^{ab}\ne0) \lor (g^{a-b}_{p-q}\ne0) \lor (g^{a-b}_{a+b}\ne0) \lor (g^{a-b}_{a+b+c}\ne0) \lor (a=b) \\ 
\smallseparatorlor \land \\
(f^{bc}\ne0) \lor (g^{b-c}_{p-q}\ne0) \lor (g^{b-c}_{b+c}\ne0) \lor (g^{b-c}_{a+b+c}\ne0) \lor (b=c) \\ 
\smallseparatorlor \land \\
(f^{ca}\ne0) \lor (g^{c-a}_{p-q}\ne0) \lor (g^{c-a}_{c+a}\ne0) \lor (g^{c-a}_{a+b+c}\ne0) \lor (c=a) 
\end{array}
\right\}
\\
\separatorlor \land
\\ %%%%%%%%%%%%%%%%%% h7: %%%%%%%%%%%%%%%%%%%
(
\FZer
 %F^{\phantom{abc}}_{pq}(0)
  \ne0
)
\lor
(
  \FOne
 %F^{\phantom{abc}}_{pq}(1)
  \ne0
)
\lor
(
\FTwo
 %F^{\phantom{abc}}_{pq}(2)
  \ne0
)
\lor
(
\FThr
 %F^{abc}_{pq}(3)
  \ne0
)
\\ 
\separatorlor \land
\\ %%%%%%%%%%%%%%%%%% h8,9,10: %%%%%%%%%%%%%%%%%%%
(
 \FZer
 %F^{\phantom{abc}}_{pq}(0)
  \ne0
)
\lor
(
  \FOne
 %F^{\phantom{abc}}_{pq}(1)
  \ne0
)
\lor
\left\{
\begin{array}{c}
    (f^{ab}\ne0) \lor (g^{a+b}_{p-q}\ne0) \lor (\Delta_3(a-b,p,q)\ne0) \\
    \tinyseparatorlor \land \\
    (f^{bc}\ne0) \lor (g^{b+c}_{p-q}\ne0) \lor (\Delta_3(b-c,p,q)\ne0) \\
    \tinyseparatorlor \land \\
    (f^{ca}\ne0) \lor (g^{c+a}_{p-q}\ne0) \lor (\Delta_3(c-a,p,q)\ne0) 
\end{array}
\right\}
\\
\toperator
\end{gather*}
\begin{proof}
This is an almost trivial re-writing of the condition $C$ (of \eqref{eq:veryfirstCdef})   using  the notation of  Definition~\ref{def:helpfulnotationforsymms} together with a few applications of Lemma~\ref{lem:tripledecomp}.
\end{proof}
\end{lemma}

\begin{lemma}
\label{lem:pqabchiralcond}
A \textbf{collision event} $e\in\pqabColEvent$ is \textbf{chiral} if and only if $B$ is true where:
\begin{gather*}
B = \numberthis\label{eq:secondtimeBisdefined}\\
\toperator \\
 %%%%%%%%%%%%%%%%%% h1 %%%%%%%%%%%%%%%%%%%
  \mathscr{P}^{ab}_{pq}(0)
  %\epsShort^{ab}_{pq} 
  \ne 0
\\ 
\separatorlor \land
\\ %%%%%%%%%%%%%%%%%% h2,3,4: %%%%%%%%%%%%%%%%%%%
%\left\{
%\begin{array}{c}
(f^{ab}\ne0) \lor
%(g^{a-b}_{p-q}\ne0) 
(G^{ab}_{pq}(2) \ne0)
\lor (g^{a-b}_{a+b}\ne0) 
%\lor (g^{a-b}_{a+b+c}\ne0)
%\lor (a=b) %\\ 
%\smallseparatorlor \land \\
%(f^{bc}\ne0) \lor (g^{b-c}_{p-q}\ne0) \lor (g^{b-c}_{b+c}\ne0) \lor (g^{b-c}_{a+b+c}\ne0) \lor (b=c) \\ 
%\smallseparatorlor \land \\
%(f^{ca}\ne0) \lor (g^{c-a}_{p-q}\ne0) \lor (g^{c-a}_{c+a}\ne0) \lor (g^{c-a}_{a+b+c}\ne0) \lor (c=a) 
%\end{array}
%\right\}
\\
\separatorlor \land
\\ %%%%%%%%%%%%%%%%%% h7: %%%%%%%%%%%%%%%%%%%
(
 G^{\phantom{ab}}_{pq}(0)
  \ne0
)
\lor
(
  G^{\phantom{ab}}_{pq}(1)
  \ne0
)
\lor
(
G^{ab}_{pq}(2)
  \ne0
)
\\ 
\separatorlor \land
\\ %%%%%%%%%%%%%%%%%% h8,9,10: %%%%%%%%%%%%%%%%%%%
(
 G^{\phantom{ab}}_{pq}(0)
  \ne0
)
\lor
(
  G^{\phantom{ab}}_{pq}(1)
  \ne0
)
\lor
\left\{
%\begin{array}{c}
    (f^{ab}\ne0) %\lor (g^{a+b}_{p-q}\ne0) 
    \lor (\Delta_3(a-b,p,q)\ne0) %\\
    %\tinyseparatorlor \land \\
    %(f^{bc}\ne0) \lor (g^{b+c}_{p-q}\ne0) \lor %(\Delta_3(b-c,p,q)\ne0) \\
    %\tinyseparatorlor \land \\
    %(f^{ca}\ne0) \lor (g^{c+a}_{p-q}\ne0) \lor (\Delta_3(c-a,p,q)\ne0) 
%\end{array}
\right\}
.
\\
\toperator
\end{gather*}
\begin{proof}
This is largely just a matter of substituting the new notation into the expression already given in \eqref{eq:veryfirstBdef}).  Only two things are worth remarking:
\begin{itemize}
\item
the `$\ldots \lor(a=b)$' part of \eqref{eq:veryfirstBdef} has been omitted as it is incompatible with the very first `$\epsShort^{ab}_{pq}\ne 0$' requirement, and 
\item na\"ive notational substitution for $\left(\GRAMTWO a {p+q} {p-q} {p+q}\ne 0 \right)\lor \left(\GRAMTWO b {p+q} {p-q} {p+q}\ne0\right)$ would result in $\left(g^a_{p-q}\ne 0\right)\lor\left(g^b_{p-q}\ne 0\right)$ however it has been rendered instead as
$(
  G^{\phantom{ab}}_{pq}(1)
  \ne0
)
\lor
(
G^{ab}_{pq}(2) 
  \ne0
)$
meaning $\left(g^{a+b}_{p-q}\ne 0\right)\lor\left(g^{a-b}_{p-q}\ne 0\right)$ since we crave statements which are either $(ab)$-even or $(ab)$-odd.
\end{itemize}
\end{proof}
\end{lemma}

\subsubsection{Start of proof of Theorem~\ref{thm:nonchiralpqabcEvents}}

%On the chirality of  $(pq)\rightarrow (abc)+X$ events at colliders

\label{sec:longproofstart}

%\begin{proof}%%MATCHES GIANTPROOF001
\begin{remark}
\followUpInFuture[It feels like there must be a much better way of proving Theorem~\ref{thm:chiralpqabcEvents} than the method I used.]
Given Lemma~\ref{lem:collisionshaveaxis} we can work with any event in $\pqabcColEvent$ in the $(p+q)$ rest frame, having aligned $\vec{p}$ with the positive $z$-axis and $\vec{q}$ with the negative $z$-axis.  In such a frame, the following quantities in $\hat e$ uniquely describe any \textbf{event} in $\pqabcColEvent$:
\begin{align}
\hat e
=
\left[\begin{array}{cc|ccc}
     m_p & m_q & m_a & m_b & m_c \\
     0 & 0 & \mathbf{a} & \mathbf{b} & \mathbf{c} \\
     p & -p & a_z & b_z & c_z
\end{array}\right]\qquad \text{with}\qquad\left(\begin{array}{l}
m_p, m_q,m_a,m_b,m_c\ge0, \\
p>0, \\
\mathbf{a},\mathbf{b},\mathbf{c}\in\mathbb{C},\ \text{and}\\
a_z,b_z,c_z\in\mathbb{R}
\end{array}\right),
\end{align}
provided that we take, respectively, the real and imaginary parts of $\mathbf{a}$ to represent the $x$ and $y$ components of $a$ (and similarly for $\mathbf{b}$ and $\mathbf{c}$).
\end{remark}
\begin{definition}
We name a matrix $\hat e$ which has the properties shown above to be the \textbf{event representation} of an event $e\in\pqabcColEvent$. 
\end{definition}
\begin{remark}
One may naturally extend the action of the parity operator, $\mathscr{P}$, and of elements of the symmetry group $G$, to an action on    \textbf{event representations} $\hat e$.  We will make such an extensions implicitly rather than explicitly, as happens in the next lemma:
\end{remark}
\begin{lemma}
\begin{align}
\label{eq:perep}
\mathscr{P} \cdot \hat e
=
\left[\begin{array}{cc|ccc}
     m_p & m_q & m_a & m_b & m_c \\
     0 & 0 & -\mathbf{a} & -\mathbf{b} & -\mathbf{c} \\
     -p & p & -a_z & -b_z & -c_z
\end{array}\right].
\end{align}
\begin{proof}
The result follows directly from Definition~\ref{def:paroponthings}.
\end{proof}
\end{lemma}

\noindent Recall that we are attempting to find the conditions under which \textbf{collision events} $e\in\pqabcColEvent$ are \textbf{non-chiral}. Corollary~\ref{cor:simplehandedness} reminds us that in order to identify such events it suffices to find all those $\hat e$ for which there exists a $g\in G$ such that $g\cdot \hat e = \mathscr{P} \cdot \hat e$.  However, because
\begin{align}
\label{eq:gpe}
\gyzpi\cdot\mathscr{P} \cdot \hat e
=
\left[\begin{array}{cc|ccc}
     m_p & m_q & m_a & m_b & m_c \\
     0 & 0 & -\mathbf{a}^* & -\mathbf{b}^* & -\mathbf{c}^* \\
     p & -p & a_z & b_z & c_z
\end{array}\right]
%\qquad \text{with}\qquad\left(\begin{array}{l}
%m_p, m_q,m_a,m_b,m_c\ge0, \\
%p>0, \\
%a,b,c\in\mathbb{C},\ \text{and}\\
%a_z,b_z,c_z\in\mathbb{R}
%\end{array}\right)
\end{align}
has fewer minus signs than $\mathscr{P}\cdot\hat e$ and is identical to $\hat e$ in all but three `cells', it is more convenient for us to find the \textbf{non-chiral collision events} by searching for all $\hat e$ for which there exists a $g\in G$ such that  $g\cdot \hat e = \gyzpi \cdot \mathscr{P} \cdot \hat e$.  The inclusion of the $\gyzpi$ has no effect on the set of events so obtained since $G$ is a group. Furthermore, every element of $g\in G$ can be written as a product of an element of $SO^+(1,3)$ with one of the twelve elements of the following group $H$:
$$
H=\left\{1,\gyzpi (pq)\right\}\otimes\left\{
     1, (ab), (ac), (bc), (abc), (cba)
     \right\}
$$
in which, once more, the $\gyzpi$ operator has been added only for convenience, to reduce the number of minus signs appearing in later work.  We can therefore reduce the task of finding \textbf{non-chiral collider events} to having to find only those events $\hat e$ for which there exists an element $s\in SO^+(1,3)$ such that $s\cdot h\cdot \hat e = \gyzpi \cdot \mathscr{P} \cdot \hat e$ for some $h$ in $H$.

It is easy to check that every for every element $h\in H$: 
\begin{align}
s\cdot h\cdot \hat e 
&=
s\cdot h \cdot
\left[\begin{array}{cc|ccc}
     m_p & m_q & m_a & m_b & m_c \\
     0 & 0 & \mathbf{a} & \mathbf{b} & \mathbf{c} \\
     p & -p & a_z & b_z & c_z
\end{array}\right]
\\
&=
s\cdot 
\left[\begin{array}{cc|ccc}
     ? & ? & ? & ? & ? \\
     0 & 0 & \mathbf{?} & \mathbf{?} & \mathbf{?} \\
     p & -p & ? & ? & ?
\end{array}\right]\qquad\text{with $p>0$.}\label{eq:targetforshe}
\end{align}
Since we require that the RHS of Equation~(\ref{eq:targetforshe}) must match the $\gyzpi\cdot \mathscr{P}\cdot \hat e$ shown in Equation~(\ref{eq:gpe}) we can see that the only elements $s\in SO^+(1,3)$ which need be considered are those  which leave the spatial components of $\vec p$ and $\vec q$ invariant.\footnote{That we can make such a statement derives ultimately from Lemma~\ref{lem:collisionshaveaxis}.}   Any boost would necessarily modify some component of $\vec p$ or $\vec q$ we know that the only permissible options for $s$ are rotations which, since they must leave the $z$-axis invariant, are rotations about $z$.  We have therefore reduced the task of proving Theorem~\ref{thm:nonchiralpqabcEvents} to finding the events $\hat e$ for which there exists an angle $\theta$ and an element $h\in H$ such that
\begin{align}
\gaxrot{\theta}\cdot h\cdot \hat e = \gyzpi \cdot \mathscr{P} \cdot \hat e \label{eq:oneoftwelve}
\end{align}
whose RHS we already know. 
\begin{definition}
Let us define $\nonchiral(h,\theta,\pqabcColEvent)$ to be the set of \textbf{non-chiral collision events} $e\in\nonchiral(\pqabcColEvent)\subset\pqabcColEvent$ which satisfy Equation~(\ref{eq:oneoftwelve}):
$$
\nonchiral(h,\theta,\pqabcColEvent)
\equivdef
\left\{
\ \ e \ \ \middle|\ \ 
e\in\pqabcColEvent,
\gaxrot{\theta}\cdot h\cdot \hat e = \gyzpi \cdot \mathscr{P} \cdot \hat e
\right\}.
$$
Similarly, let us define $\nonchiral(h,\pqabcColEvent)$ to be the set of \textbf{non-chiral collision events} $e\in\nonchiral(\pqabcColEvent)\subset\pqabcColEvent$ which satisfy Equation~(\ref{eq:oneoftwelve}) \textit{for some $\theta$}:
\begin{align*}
\nonchiral(h,\pqabcColEvent)
&\equivdef
\left\{
\ \ e \ \ \middle|\ \ 
e\in\pqabcColEvent,
\gaxrot{\theta}\cdot h\cdot \hat e = \gyzpi \cdot \mathscr{P} \cdot \hat e, \theta\in\mathbb{R}
\right\}
\\
&=
\bigcup_{\theta\in\mathbb{R}} \nonchiral(h,\theta,\pqabcColEvent).
\end{align*}
\end{definition}
\begin{corollary}
In terms of the quantities just defined, our goal, $\nonchiral(\pqabcColEvent)$, may be written as:
$$
\nonchiral(\pqabcColEvent) = \bigcup_{h\in H} \nonchiral(h,\pqabcColEvent).
$$
\end{corollary}

The action of $\gaxrot{\theta}$ on any event representation $\hat e$ is nice, as it only multiplies each element of the complex row of the representation by $e^{i\theta}$.  Since there are only twelve elements of $H$ (we will define them to be
$\{$
$h_1{=}1$,
$h_2{=}(ab)$,
$h_3{=}(ac)$,
$h_4{=}(bc)$,
$h_5{=}(abc)$,
$h_6{=}(cba)$,
$h_7{=}\gyzpi(pq)$,
$h_8{=}\gyzpi(pq)(ab)$,
$h_9{=}\gyzpi(pq)(ac)$,
$h_{10}{=}\gyzpi(pq)(bc)$,
$h_{11}{=}\gyzpi(pq)(abc)$,
$h_{12}{=}\gyzpi(pq)(cba)$
$\}$)
we can now enumerate all twelve LHSs of Equation~(\ref{eq:oneoftwelve}):
\begin{align}
  \begin{array}{rr}
%%%%%%%%%%%%%%%%%%%%%%%%%%
{
\gaxrot{\theta} h_1  \hat e = 
\nicety a b c},
&
{
\gaxrot{\theta} h_7  \hat e =
\oddity a b c},
\\ %%%%%%%%%%%%%%%%%%%%%%%%%%
{
\gaxrot{\theta} h_2  \hat e =
\nicety b a c},
&
{
\gaxrot{\theta} h_8  \hat e =
\oddity b a c},
\\  %%%%%%%%%%%%%%%%%%%%%%%%%%
{\color{blue}
\gaxrot{\theta} h_3  \hat e =
\nicety c b a},
&
{\color{blue}
\gaxrot{\theta} h_9  \hat e =
\oddity c b a},
\\  %%%%%%%%%%%%%%%%%%%%%%%%%%
{\color{blue}
\gaxrot{\theta} h_4  \hat e =
\nicety a c b},
&
{\color{blue}
\gaxrot{\theta} h_{10} \hat e =
\oddity a c b},
\\  %%%%%%%%%%%%%%%%%%%%%%%%%%
{
\gaxrot{\theta} h_5  \hat e =
\nicety b c a},
&
{
\gaxrot{\theta} h_{11}   \hat e =
\oddity b c a},
\\  %%%%%%%%%%%%%%%%%%%%%%%%%%
{\color{blue}
\gaxrot{\theta} h_6  \hat e =
\nicety c a b},
&
{\color{blue}
\gaxrot{\theta} h_{12}  \hat e =
\oddity c a b},
          %%%%%%%%%%%%%%%%%%%%%%%%%%
    \end{array}
    \label{eq:6black6blue}
\end{align}
at least one of which will equal
$$
%%% This is a copy of the {eq:gpe} equation
\gyzpi\cdot\mathscr{P} \cdot \hat e
=
\left[\begin{array}{cc|ccc}
     m_p & m_q & m_a & m_b & m_c \\
     0 & 0 & -\mathbf{a}^* & -\mathbf{b}^* & -\mathbf{c}^* \\
     p & -p & a_z & b_z & c_z
\end{array}\right].
$$
for some $\theta$ if $\hat e$ is \textbf{non-chiral}. In (\ref{eq:6black6blue}) six of of the constraints have been coloured blue.  This is because these constraints are strongly related to another black one in the list. (We will see how later.) The `hard work' only needs to be invested in the black cases. Once those are done, results relating to the remaining blue ones will follow trivially.

We now take the black constraints from (\ref{eq:6black6blue}) and re-write their complex numbers in polar form by making the following replacements:
\begin{align}
    \mathbf{a} &\rightarrow |\mathbf{a}| e^{i\alpha}, \\
    \mathbf{b} &\rightarrow |\mathbf{b}| e^{i\beta}, \qquad\text{and}\\
    \mathbf{c} &\rightarrow |\mathbf{c}| e^{i\gamma}.
\end{align}
This results in:
\begin{align}
  \begin{array}{rr}
%%%%%%%%%%%%%%%%%%%%%%%%%%
\gaxrot{\theta} h_1  \hat e = 
\polarnicety a b c \alpha \beta \gamma,
\\
\gaxrot{\theta} h_2  \hat e =
\polarnicety b a c \beta \alpha \gamma,
\\
\gaxrot{\theta} h_5  \hat e =
\polarnicety b c a \beta \gamma \alpha,
\\
\gaxrot{\theta} h_7  \hat e =
\polaroddity a b c \alpha \beta \gamma,
\\ %%%%%%%%%%%%%%%%%%%%%%%%%%
\gaxrot{\theta} h_8  \hat e =
\polaroddity b a c \beta \alpha \gamma,
\\
\gaxrot{\theta} h_{11}   \hat e =
\polaroddity b c a \beta \gamma \alpha,
    \end{array}
    \label{eq:just6polarblack}
\end{align}
needing to equal
$$
\gyzpi\cdot\mathscr{P} \cdot \hat e
=
\left[\begin{array}{cc|ccc}
     m_p & m_q & m_a & m_b & m_c \\
     0 & 0 & |\mathbf{a}|e^{i(\pi-\alpha)} & |\mathbf{b}|e^{i(\pi-\beta)} & |\mathbf{c}|e^{i(\pi-\gamma)} \\
     p & -p & a_z & b_z & c_z
\end{array}\right].
$$

\subsubsection{Case $h_1$}
$h_1= 1 $.

$
((\mathbf{a}=0) \lor (\theta+\alpha = \pi-\alpha + 2 n_a \pi))
\land
((\mathbf{b}=0) \lor (\theta+\beta = \pi-\beta + 2 n_b \pi))
\land
((\mathbf{c}=0) \lor (\theta+\gamma = \pi-\gamma + 2 n_c \pi))
$

$
\implies ((\mathbf{a}=0) \lor (\alpha = \frac 1 2 (\pi-\theta) +  n_a \pi))
\land
((\mathbf{b}=0) \lor (\beta = \frac 1 2 (\pi-\theta) +  n_b \pi))
\land
((\mathbf{c}=0) \lor (\gamma = \frac 1 2 (\pi-\theta) +  n_c \pi))
$

$\implies$ the three-momenta $\vec a$, $\vec b$ and $\vec c$ must all live in a common plane containing the beam axis, but are otherwise unconstrained. This can be written in a frame-independent way using the \textbf{Lorentz contraction shorthand notation} of \eqref{eq:lorentzcontractionsepsilonhorthand}  as:
\begin{align}
\epsLor_{abpq}=
\epsLor_{acpq}=
\epsLor_{bcpq}=0.
\end{align}
\followUpInFuture[Rigorous proof is Missing]
A rigorous proof of the above statement is left as an exercise for the reader.

\subsubsection{Case $h_2$}
$h_2=(ab)  $.

$
(m_a=m_b)
\land
(a_z=b_z)
\land
(|\mathbf{a}|=|\mathbf{b}|)
\land
((\mathbf{a}=\mathbf{b}=0) \lor (\theta+\beta = \pi-\alpha + 2 n_a \pi))
\land
((\mathbf{a}=\mathbf{b}=0) \lor (\theta+\alpha = \pi-\beta + 2 n_b \pi))
\land
((\mathbf{c}=0) \lor (\theta+\gamma = \pi-\gamma + 2 n_c \pi))
$

$
\implies 
(m_a=m_b)
\land
(a_z=b_z)
\land
(|\mathbf{a}|=|\mathbf{b}|)
\land
(
  (\mathbf{a}=\mathbf{b}=0) 
  \lor 
  (
    (\alpha+\beta=\pi-\theta+2n_a  \pi)
    \land
    (\alpha+\beta=\pi-\theta+2n_b  \pi)
  )
)
\land
((\mathbf{c}=0) \lor (\gamma = \frac 1 2 (\pi-\theta) +  n_c \pi))
$

$
\implies 
(m_a=m_b)
\land
(a_z=b_z)
\land
(|\mathbf{a}|=|\mathbf{b}|)
\land
(
  (\mathbf{a}=\mathbf{b}=0) 
  \lor 
  (\alpha+\beta=\pi-\theta+2n_{ab}  \pi)
)
\land
((\mathbf{c}=0) \lor (\gamma = \frac 1 2 (\pi-\theta) +  n_c \pi))
$

Expanding over the four cases ($\mathbf{a}=0 \land \mathbf{c}=0$), ($\mathbf{a}\ne0 \land \mathbf{c}=0$)
 ($\mathbf{a}=0 \land \mathbf{c}\ne0$) and
  ($\mathbf{a}\ne0 \land \mathbf{c}\ne0$):\footnote{Here and in similar places some terms have been highlighted in blue. This is to provide either emphasis to changed terms, or to assist in identifying pairs of opening and closing brackets, \textit{etc}.}

$
\implies 
(m_a=m_b)
\land
(a_z=b_z)
\land
(|\mathbf{a}|=|\mathbf{b}|)
\land
(
{\color{blue}[}
  (\mathbf{a}=\mathbf{b}=\mathbf{c}=0)
{\color{blue}]}
\lor
{\color{blue}[}
  (\mathbf{a}=\mathbf{b}=0)
  \land
  (|\mathbf{c}|>0)
  \land
  (\gamma = \frac 1 2 (\pi-\theta) +  n_c \pi)
{\color{blue}]}
\lor
{\color{blue}[}
  (|\mathbf{a}|>0)
  \land
  (\mathbf{c}=0)
  \land
  (\alpha+\beta=\pi-\theta+2n_{ab}  \pi)
{\color{blue}]}
\lor
{\color{blue}[}
  (|\mathbf{a}|>0)
  \land
  (|\mathbf{c}|>0)
  \land
  (\alpha+\beta=\pi-\theta+2n_{ab}  \pi)
  \land
  (\gamma = \frac 1 2 (\pi-\theta) +  n_c \pi)
{\color{blue}]}
)
$

Removing the constraints which only constrain something outside of $\hat e$ and scaling the second last constraint:

$
\implies 
(m_a=m_b)
\land
(a_z=b_z)
\land
(|\mathbf{a}|=|\mathbf{b}|)
\land
(
{\color{blue}[}
  (\mathbf{a}=\mathbf{b}=\mathbf{c}=0)
{\color{blue}]}
\lor
{\color{blue}[}
  (\mathbf{a}=\mathbf{b}=0)
  \land
  (|\mathbf{c}|>0)
{\color{blue}]}
\lor
{\color{blue}[}
  (|\mathbf{a}|>0)
  \land
  (\mathbf{c}=0)
{\color{blue}]}
\lor
{\color{blue}[}
  (|\mathbf{a}|>0)
  \land
  (|\mathbf{c}|>0)
  \land
  (\frac 1 2 (\alpha+\beta)=\frac 1 2 (\pi-\theta)+n_{ab}  \pi)
  \land
  (\gamma = \frac 1 2 (\pi-\theta) +  n_c \pi)
{\color{blue}]}
)
$

Replacing second-last constraint with an equivalent constraint:

$
\implies 
(m_a=m_b)
\land
(a_z=b_z)
\land
(|\mathbf{a}|=|\mathbf{b}|)
\land
(
{\color{blue}[}
  (\mathbf{a}=\mathbf{b}=\mathbf{c}=0)
{\color{blue}]}
\lor
{\color{blue}[}
  (\mathbf{a}=\mathbf{b}=0)
  \land
  (|\mathbf{c}|>0)
{\color{blue}]}
\lor
{\color{blue}[}
  (|\mathbf{a}|>0)
  \land
  (\mathbf{c}=0)
{\color{blue}]}
\lor
{\color{blue}[}
  (|\mathbf{a}|>0)
  \land
  (|\mathbf{c}|>0)
  \land
{\color{blue}
  (\frac 1 2 (\alpha+\beta)-\gamma=(n_{ab}-n_c)  \pi)
}
  \land
  (\gamma = \frac 1 2 (\pi-\theta) +  n_c \pi)
{\color{blue}]}
)
$

Removing, again, constraints that only fix things outside of $\hat e$:

$
\implies 
(m_a=m_b)
\land
(a_z=b_z)
\land
(|\mathbf{a}|=|\mathbf{b}|)
\land
(
{\color{blue}[}
  (\mathbf{a}=\mathbf{b}=\mathbf{c}=0)
{\color{blue}]}
\lor
{\color{blue}[}
  (\mathbf{a}=\mathbf{b}=0)
  \land
  (|\mathbf{c}|>0)
{\color{blue}]}
\lor
{\color{blue}[}
  (|\mathbf{a}|>0)
  \land
  (\mathbf{c}=0)
{\color{blue}]}
\lor
{\color{blue}[}
  (|\mathbf{a}|>0)
  \land
  (|\mathbf{c}|>0)
  \land
  (\frac 1 2 (\alpha+\beta)-\gamma=(n_{ab}-n_c)  \pi)
{\color{blue}]}
)
$

Removing constraints that are just more stringent versions of things we already have covered in Case $\gaxrot{\theta} 1 $ , since we only need \textbf{new} events:
\followUpInFuture[Add more rigour to make sure that such copies are not thrown away by mistake, and that the link is done at the maths not the interpretive level?]

$
\implies 
(m_a=m_b)
\land
(a_z=b_z)
\land
(|\mathbf{a}|=|\mathbf{b}|)
\land
(
{\color{blue}[}
  (|\mathbf{a}|>0)
  \land
  (\mathbf{c}=0)
{\color{blue}]}
\lor
{\color{blue}[}
  (|\mathbf{a}|>0)
  \land
  (|\mathbf{c}|>0)
  \land
  (\frac 1 2 (\alpha+\beta)-\gamma=(n_{ab}-n_c)  \pi)
{\color{blue}]}
)
$

Factorizing:

$
\implies 
(m_a=m_b)
\land
(a_z=b_z)
\land
(|\mathbf{a}|=|\mathbf{b}|
{\color{blue}>0})
\land
(
{\color{blue}[}
  (\mathbf{c}=0)
{\color{blue}]}
\lor
{\color{blue}[}
  (|\mathbf{c}|>0)
  \land
  (\frac 1 2 (\alpha+\beta)-\gamma=(n_{ab}-n_c)  \pi)
{\color{blue}]}
)
$

\begin{align}
\implies 
(m_a=m_b)
\land
(a_z=b_z)
\land
(|\vec a|=|\vec b |)
\land
(|\mathbf{a}|{\color{blue}>0})
\qquad\qquad\qquad\qquad\nonumber\\ \qquad
\land
(
{\color{blue}[}
  (\mathbf{c}=0)
{\color{blue}]}
\lor 
{\color{blue}[}
  (|\mathbf{c}|>0)
  \land
  (\frac 1 2 (\alpha+\beta)-\gamma=(n_{ab}-n_c)  \pi)
{\color{blue}]}
)
\label{eq:annoyingoneA}
\end{align}
or alternatively (using Lemma~\ref{lem:modamodbisdotprodrel})
\begin{align}
\implies 
(m_a=m_b)
\land
(a_z=b_z)
\land
((\vec a-\vec b).(\vec a+\vec b)=0)
\land
(|\mathbf{a}|{\color{blue}>0})
\qquad\qquad\qquad\qquad\nonumber\\ \qquad
\land
(
{\color{blue}[}
  (\mathbf{c}=0)
{\color{blue}]}
\lor
{\color{blue}[}
  (|\mathbf{c}|>0)
  \land
  (\frac 1 2 (\alpha+\beta)-\gamma=(n_{ab}-n_c)  \pi)
{\color{blue}]}
).
\label{eq:annoyingoneB}
\end{align}

$\implies$ the most general event of this type may be constructed by: (i) defining a plane $P$ containing the beam axis; (ii) putting the four-momentum $c$ anywhere in $P$; (iii) placing the four-momentum $a$ anywhere; (iv) giving $b$ the same mass as $a$; and (v) and positioning $\vec b$  such that it is the reflection of $\vec a$ in $P$. 
\followUpInFuture[WHY DO WE HAVE TO WORK IN A FRAME RATHER THAN IN NONE?]

Writing all of (\ref{eq:annoyingoneA}) or (\ref{eq:annoyingoneB}) in frame-invariant form is annoying, although the first four terms of either are easy. 
Using Lemmas~\ref{lem:adotbinsomeframe}, \ref{lem:threemominframe} and \ref{lem:transmomforcolevent}
they become:
\begin{align}\label{eq:whatwealreadyhaveA}
\easyPartOfAxialCatapultConstraintABCVersionA a b c \end{align}
or using Lemma~\ref{lem:momdifferencesidentity} as 
\begin{align}\label{eq:whatwealreadyhaveB}
\easyPartOfAxialCatapultConstraintABCVersionB a b c \end{align}
and these restrict $\vec a$ and $\vec b$ to run from the origin to points on the circumference of a circle which lies in the plane $z=z_0$ and which is centred on the point $(x,y,z)=(0,0,z_0)$, for some $z_0$.
The problematic part is the last part which should place $\vec c$ on the plane in which $\vec a$ mirrors $\vec b$:
\begin{align}
(
{\color{blue}[}
  (\mathbf{c}=0)
{\color{blue}]}
\lor
{\color{blue}[}
  (|\mathbf{c}|>0)
  \land
  (\frac 1 2 (\alpha+\beta)-\gamma=(n_{ab}-n_c)  \pi)
{\color{blue}]}
).\label{eq:trickycontstraint34}
\end{align}
It is tempting to try to use $(\epsLor_{(a+b)cpq}=0)$ since it says that in the $(p+q)$ rest frame $(\vec a+\vec b)$, $\vec c$ and $\vec p$ all lie in a common plane!
Alas, while this correctly constrains $\vec c$ under most circumstances, it fails to constrain $\vec c$ when $\vec a + \vec b=0$.   This is a problem because in such circumstances $\vec c$ needs still to be on the plane in which $\vec a$ mirrors $\vec b$.
Another tempting idea is to seek to assert that $(\vec a-\vec b).\vec c$ should be zero in the $(p+q)$ rest frame.  Again, this successfully constrains $c$ under most circumstances, but (alas) fails to do so when $\vec a=\vec b$.

Since each of the two previous ideas works  when the other fails, one could sum their squares, and require that the result is zero, but this would both (i) be ugly, and (ii) would square away the $(pq)$-oddness of the former and the $(ab)$-oddness of the latter, which isn't usually ideal as (as will be seen later) we will need $(ab)$-oddness here.

However, the failure of $(\vec a-\vec b).\vec c=0$ to constrain $\vec c$ when $\vec a=\vec b$ could be countered by deliberately excluding that case from consideration.  Specifically, in the case where $\vec a=\vec b$ then $\vec c$ is supposed to sit in the same plane as $\vec a$ (or $\vec b$) and $\vec p$ --- but events of that type are already included in case $h_1=1$, and we therefore do not need to include them here.  The best way to excluding events for which $\vec a=\vec b$, given that we already have a constraint here that $m_a=m_b$, is just to exclude events in which $a=b$ (that's a comparison between four vectors).\footnote{Other options which just target $\vec a=\vec b$ could include demanding that $|\vec a-\vec b|^2\ne0$ or that $\epsLor_{abpq}\ne 0$.  The former could be written using Lemma~\ref{lem:threemominframe} as $
\SYMGRAMTWO{a-b}{p+q}\ne 0
$ which is $(ab)$-even and $(pq)$-even.  The latter is $(ab)$-odd and $(pq)$-odd.  Alas, excluding events with  $\epsLor_{abpq}= 0$ may exclude MORE than we want to exclude. For example, it also excludes events in which $((\vec a_T=-\vec b_T\ne0) \land (a_z=b_z))$ for which it would be acceptable to have $\vec c$ at right angles to both $\vec a_T$ and $\vec b_T$, and which would consequently place such an event outside of Case $h_1$.} 
Therefore, in order to implement (\ref{eq:trickycontstraint34}) as an add-on to (\ref{eq:whatwealreadyhaveA}) or (\ref{eq:whatwealreadyhaveB}) we append the following::
\begin{align}
\label{eq:addon}
\left(\PerpABF{a-b}{c}{p+q}\right)\land(\theaddon a b).
\end{align}
The form of (\ref{eq:addon}) fits better with (\ref{eq:whatwealreadyhaveB}) than with (\ref{eq:whatwealreadyhaveA}) since the linearity of Gram Determinants ensures that
\begin{align*}
\left(
\left(\PerpABF {a-b} {a+b} {p+q}\right)
\land
\left(\PerpABF {a-b} {c} {p+q}\right)
\right)
\iff \qquad\qquad\qquad
\\
\qquad\qquad\qquad \iff
\left(
\left(\PerpABF {a-b} {a+b} {p+q}\right)
\land
\left(\PerpABF {a-b} {a+b+c} {p+q}\right)
\right).
\end{align*}
in which last term is (perhaps) nicer than the term it replaces insofar as the former contains an $a+b+c$, which is invariant under the discrete part of our symmetry group, while the latter contains a lone $c$, which is not.  We therefore can say that the all the new \textbf{non-chiral} events brought by case $h_2$ are contained within the condition:
\begin{align*}
\axialCatapultConstraintABCPreFinal a b c.
\end{align*}
Is it possible to remove the ``$(\SYMGRAMTHR a p q \color{blue}{}>0)$'' term?  If we were to remove this part, we would allow $\vec a$ and $\vec b$ to be axial.  If $\vec a$ and $\vec b$ were axial, then given that they have the same $z$-value under our constraint, they would actually be co-incident. This is already forbidden by the $(\theaddon a b)$ term above. Therefore the ``$(\SYMGRAMTHR a p q \color{blue}{}>0)$'' term is redundant and can be removed, leaving us with:
\begin{align*}
\axialCatapultConstraintABC a b c.
\end{align*}

\subsubsection{Case $h_3$}
$h_3= (ac)  $. 

Everything here works the same as with $h_2 = (ab)  $ except with $b\leftrightarrow c$.

\subsubsection{Case $h_4$}
$h_4= (bc)  $.

Everything here works the same as with $h_2 = (ab)$ except with $a\leftrightarrow c$.

\subsubsection{Case $h_5$} 
$h_5= (abc).$

$
(m_a=m_b=m_c)
\land
(|\mathbf{a}|=|\mathbf{b}|=|\mathbf{c}|)
\land
(a_z=b_z=c_z)
\land
(
    (\mathbf{a}=\mathbf{b}=\mathbf{c}=0)
    \lor
    (
      (\theta+\beta=\pi-\alpha+2 n_1 \pi)
      \land
       (\theta+\gamma=\pi-\beta+2 n_2 \pi)
      \land
       (\theta+\alpha=\pi-\gamma+2 n_3 \pi)
    )
)
$

$
\implies
(m_a=m_b=m_c)
\land
(|\mathbf{a}|=|\mathbf{b}|=|\mathbf{c}|)
\land
(a_z=b_z=c_z)
\land
(
    (\mathbf{a}=\mathbf{b}=\mathbf{c}=0)
    \lor
    (
      (\alpha+\beta=\pi-\theta+2 n_1 \pi)
      \land
       (\beta+\gamma=\pi-\theta+2 n_2 \pi)
      \land
       (\alpha+\gamma=\pi-\theta+2 n_3 \pi)
    )
)
$

But
$$(\alpha+\beta=\pi-\theta+2 n_1 \pi) \land
       (\beta+\gamma=\pi-\theta+2 n_2 \pi)$$
       $\implies$
$$\alpha-\gamma=2 (n_1-n_2) \pi=0\mod 2\pi$$
and similarly
$$\alpha-\beta=0\mod 2\pi$$
$$\gamma-\beta=0\mod 2\pi$$
and so
$$\alpha=\beta=\gamma\mod 2\pi$$
meaning that the four-momenta $a$, $b$ and $c$ must all share a common mass, and must have $\vec a=\vec b=\vec c$. 
The above constraint is a more restrictive sub-case of Case $h_1= 1 $ and therefore \textbf{this case may hereafter be discarded}.

\subsubsection{Case $h_6$} $h_6= (cba).$

\textbf{This case may hereafter be discarded} for the same reasons as Case $h_6= (abc).$

\subsubsection{Case $h_7$} $h_7= \gyzpi(pq).$

$
(m_p=m_q)
\land
(a_z=b_z=c_z=0)
\land
((\mathbf{a}=0) \lor (\theta-\alpha = \pi-\alpha + 2 n_a\pi))
\land
((\mathbf{b}=0) \lor (\theta-\beta = \pi-\beta + 2 n_b\pi))
\land
((\mathbf{c}=0) \lor (\theta-\gamma = \pi-\gamma + 2 n_c\pi))
$

$
\implies
(m_p=m_q)
\land
(a_z=b_z=c_z=0)
\land
((\mathbf{a}=0) \lor (\theta = \pi + 2 n_a\pi))
\land
((\mathbf{b}=0) \lor (\theta = \pi + 2 n_b\pi))
\land
((\mathbf{c}=0) \lor (\theta = \pi + 2 n_c\pi))
$

So far as $\hat e$ is concerned (i.e.~ignoring constraints that only affect $\theta$, $n_1$, $n_2$ and $n_3$) we have:

$
\implies
(m_p=m_q)
\land
(a_z=b_z=c_z=0)
\land
((\mathbf{a}=0) \lor (-)
\land
((\mathbf{b}=0) \lor (-)
\land
((\mathbf{c}=0) \lor (-)
$

$
\implies
(m_p=m_q)
\land
(a_z=b_z=c_z=0)
$

$\implies$ $m_p=m_q$, while the four-vectors $a$, $b$ and $c$ may lie anywhere on collision event's \textbf{transverse plane} (see Definition~\ref{def:transverseplane}).  This condition may be re-phrased in a frame-independent way as follows:
$$
(p^2=q^2)
\land
\left(\PerpABF {a} {p-q} {p+q}\right)
\land
\left(\PerpABF {b} {p-q} {p+q}\right)
\land
\left(\PerpABF {c} {p-q} {p+q}\right).
$$

\subsubsection{Case $h_8$} $h_8= \gyzpi (pq)(ab).$

$
(m_p=m_q)
\land
(m_a=m_b)
\land
(a_z+b_z=c_z=0)
\land
(|\mathbf{a}|=|\mathbf{b}|)
\land
(
   (\mathbf{a}=\mathbf{b}=0)
   \lor
   (
    (\theta-\beta=\pi-\alpha+2n_1\pi)
    \land
    (\theta-\alpha=\pi-\beta+2n_2\pi)
   )
)
\land
(
   (|\mathbf{c}|=0)
   \lor
   (\theta-\gamma = \pi-\gamma + 2n_3\pi)
)
$

Alphas and betas to LHS:

$
\implies
(m_p=m_q)
\land
(m_a=m_b)
\land
(a_z+b_z=c_z=0)
\land
(|\mathbf{a}|=|\mathbf{b}|)
\land
(
   (\mathbf{a}=\mathbf{b}=0)
   \lor
   (
    (\alpha-\beta=\pi-\theta+2n_1\pi)
    \land
    (\beta-\alpha=\pi-\theta+2n_2\pi)
   )
)
\land
(
   (\mathbf{c}=0)
   \lor
   (\theta = \pi + 2n_3\pi)
)
$

De-duplicating alphas and betas:

$
\implies
(m_p=m_q)
\land
(m_a=m_b)
\land
(a_z+b_z=c_z=0)
\land
(|\mathbf{a}|=|\mathbf{b}|)
\land
(
   (\mathbf{a}=\mathbf{b}=0)
   \lor
   (
    (\pi-\theta+(n_1+n_2)\pi=0)
    \land
    (\alpha-\beta=(n_1-n_2)\pi)
   )
)
\land
(
   (\mathbf{c}=0)
   \lor
   (\theta = \pi + 2n_3\pi)
)
$

Expanding over $\mathbf{c}=0$ vs $\mathbf{c}\ne0$:

$
\implies
{\color{blue}[}
(\mathbf{c}=0)
\land
(m_p=m_q)
\land
(m_a=m_b)
\land
(a_z+b_z=c_z=0)
\land
(|\mathbf{a}|=|\mathbf{b}|)
\land
(
   (\mathbf{a}=\mathbf{b}=0)
   \lor
   (
    (\pi-\theta+(n_1+n_2)\pi=0)
    \land
    (\alpha-\beta=(n_1-n_2)\pi)
   )
)
\land
(
   (\mathbf{c}=0)
   \lor
   (\theta = \pi + 2n_3\pi)
)
{\color{blue}]}
\lor
{\color{blue}[}
(\mathbf{c}\ne 0)
\land
(m_p=m_q)
\land
(m_a=m_b)
\land
(a_z+b_z=c_z=0)
\land
(|\mathbf{a}|=|\mathbf{b}|)
\land
(
   (\mathbf{a}=\mathbf{b}=0)
   \lor
   (
    (\pi-\theta+(n_1+n_2)\pi=0)
    \land
    (\alpha-\beta=(n_1-n_2)\pi)
   )
)
\land
(
   (\mathbf{c}=0)
   \lor
   (\theta = \pi + 2n_3\pi)
)
{\color{blue}]}
$

Capitalising on the above expansion:

$
\implies
{\color{blue}[}
(\vec c=0)
\land
(m_p=m_q)
\land
(m_a=m_b)
\land
(a_z+b_z=0)
\land
(|\mathbf{a}|=|\mathbf{b}|)
\land
(
   (\mathbf{a}=\mathbf{b}=0)
   \lor
   (
    (\pi-\theta+(n_1+n_2)\pi=0)
    \land
    (\alpha-\beta=(n_1-n_2)\pi)
   )
)
{\color{blue}]}
\lor
{\color{blue}[}
(\mathbf{c}\ne 0)
\land
(m_p=m_q)
\land
(m_a=m_b)
\land
(a_z+b_z=c_z=0)
\land
(|\mathbf{a}|=|\mathbf{b}|)
\land
(
   (\mathbf{a}=\mathbf{b}=0)
   \lor
   (
    (\pi-\theta+(n_1+n_2)\pi=0)
    \land
    (\alpha-\beta=(n_1-n_2)\pi)
    \land
    (\theta -\pi = 2n_3\pi)
   )
)
\land
   (\theta = \pi + 2n_3\pi)
{\color{blue}]}
$

De-duplicating theta-phi terms in $\mathbf{c}\ne 0$ part:

$
\implies
{\color{blue}[}
(\vec c=0)
\land
(m_p=m_q)
\land
(m_a=m_b)
\land
(a_z+b_z=0)
\land
(|\mathbf{a}|=|\mathbf{b}|)
\land
(
   (\mathbf{a}=\mathbf{b}=0)
   \lor
   (
    (\pi-\theta+(n_1+n_2)\pi=0)
    \land
    (\alpha-\beta=(n_1-n_2)\pi)
   )
)
{\color{blue}]}
\lor
{\color{blue}[}
(\mathbf{c}\ne 0)
\land
(m_p=m_q)
\land
(m_a=m_b)
\land
(a_z+b_z=c_z=0)
\land
(|\mathbf{a}|=|\mathbf{b}|)
\land
(
   (\mathbf{a}=\mathbf{b}=0)
   \lor
   (
    (\alpha-\beta=(n_1-n_2)\pi)
    \land
    (n_1+n_2=2n_3)
    \land
    (2(\pi-\theta)+(n_1+n_2)\pi=-2n_3\pi)
   )
)
\land
   (\theta = \pi + 2n_3\pi)
{\color{blue}]}
$

But $n_1+n_2=2n_3\implies n_1-n_2=2(n_3-n_2)$, while knowledge that `$n_1+n_2=2n_3$' allows you to re-write `$2(\pi-\theta)+(n_1+n_2)\pi=-2n_3\pi$' as `$\pi-\theta+2 n_3\pi=0$' so:

$
\implies
{\color{blue}[}
(\vec c=0)
\land
(m_p=m_q)
\land
(m_a=m_b)
\land
(a_z+b_z=0)
\land
(|\mathbf{a}|=|\mathbf{b}|)
\land
(
   (\mathbf{a}=\mathbf{b}=0)
   \lor
   (
    (\pi-\theta+(n_1+n_2)\pi=0)
    \land
    (\alpha-\beta=(n_1-n_2)\pi)
   )
)
{\color{blue}]}
\lor
{\color{blue}[}
(\mathbf{c}\ne 0)
\land
(m_p=m_q)
\land
(m_a=m_b)
\land
(a_z+b_z=c_z=0)
\land
(|\mathbf{a}|=|\mathbf{b}|)
\land
(
   (\mathbf{a}=\mathbf{b}=0)
   \lor
   (
    (\alpha-\beta=2(n_3-n_2)\pi)
    \land
    (n_1+n_2=2n_3)
    \land
    (\pi-\theta+2 n_3\pi=0)
   )
)
\land
   (\theta = \pi + 2n_3\pi)
{\color{blue}]}
$

Removing redundancy:

$
\implies
{\color{blue}[}
(\vec c=0)
\land
(m_p=m_q)
\land
(m_a=m_b)
\land
(a_z+b_z=0)
\land
(|\mathbf{a}|=|\mathbf{b}|)
\land
(
   (a=b=0)
   \lor
   (
    (\pi-\theta+(n_1+n_2)\pi=0)
    \land
    (\alpha-\beta=(n_1-n_2)\pi)
   )
)
{\color{blue}]}
\lor
{\color{blue}[}
(\mathbf{c}\ne 0)
\land
(m_p=m_q)
\land
(m_a=m_b)
\land
(a_z+b_z=c_z=0)
\land
(|\mathbf{a}|=|\mathbf{b}|)
\land
(
   (\mathbf{a}=\mathbf{b}=0)
   \lor
   (
    (\alpha-\beta=2(n_3-n_2)\pi)
    \land
    (n_1+n_2=2n_3)
   )
)
\land
   (\theta = \pi + 2n_3\pi)
{\color{blue}]}
$

In second blue block, $\theta$ is now mentioned only once, and we are not interested in $\theta$ so it may be removed as it is now only constraining itself:

$
\implies
{\color{blue}[}
(\vec c=0)
\land
(m_p=m_q)
\land
(m_a=m_b)
\land
(a_z+b_z=0)
\land
(|\mathbf{a}|=|\mathbf{b}|)
\land
(
   (\mathbf{a}=\mathbf{b}=0)
   \lor
   (
    (\pi-\theta+(n_1+n_2)\pi=0)
    \land
    (\alpha-\beta=(n_1-n_2)\pi)
   )
)
{\color{blue}]}
\lor
{\color{blue}[}
(\mathbf{c}\ne 0)
\land
(m_p=m_q)
\land
(m_a=m_b)
\land
(a_z+b_z=c_z=0)
\land
(|\mathbf{a}|=|\mathbf{b}|)
\land
(
   (\mathbf{a}=\mathbf{b}=0)
   \lor
   (
    (\alpha-\beta=2(n_3-n_2)\pi)
    \land
    (n_1+n_2=2n_3)
   )
)
{\color{blue}]}
$

Now, in the second blue block $n_1$ is only constraining itself, and $\theta$ is only constraining itself in the first red block, so more things can be removed:

$
\implies
{\color{blue}[}
(\vec c=0)
\land
(m_p=m_q)
\land
(m_a=m_b)
\land
(a_z+b_z=0)
\land
(|\mathbf{a}|=|\mathbf{b}|)
\land
(
   (\mathbf{a}=\mathbf{b}=0)
   \lor
   (\alpha-\beta=(n_1-n_2)\pi)
)
{\color{blue}]}
\lor
{\color{blue}[}
(\textbf{c}\ne 0)
\land
(m_p=m_q)
\land
(m_a=m_b)
\land
(a_z+b_z=c_z=0)
\land
(|\mathbf{a}|=|\mathbf{b}|)
\land
(
   (\mathbf{a}=\mathbf{b}=0)
   \lor
   (\alpha-\beta=2(n_3-n_2)\pi)
)
{\color{blue}]}
$

Factorizing:

$
\implies
(m_p=m_q)
\land
(m_a=m_b)
\land
(a_z+b_z=0)
\land
{\color{magenta}[}
{\color{blue}[}
(\vec c=0)
\land
(|\mathbf{a}|=|\mathbf{b}|)
\land
(
   (\mathbf{a}=\mathbf{b}=0)
   \lor
   (\alpha-\beta=(n_1-n_2)\pi)
)
{\color{blue}]}
\lor
{\color{blue}[}
(\mathbf{c}\ne 0)
\land
(c_z=0)
\land
(|\mathbf{a}|=|\mathbf{b}|)
\land
(
   (\mathbf{a}=\mathbf{b}=0)
   \lor
   (\alpha-\beta=2(n_3-n_2)\pi)
)
{\color{blue}]}
{\color{magenta}]}
$

This implies that \textbf{non-chiral collider events} in this case must have $(m_p=m_q)$ and $(m_a=m_b)$ and:
\begin{enumerate}
    \item particle $c$ should be at rest in the $(p+q)$ rest frame, while $\vec a$ and $\vec b$ are equal and opposite but are otherwise unconstrained; or
    \item particle $c$ should be at rest in the $(p+q)$ rest frame, and $\vec a$ is free to point wherever it likes,  so long as $\vec b$ is the reflection of $\vec a$ in the \textbf{transverse plane} of the event; or
    \item particle $c$ may lie anywhere on the \textbf{transverse plane} other than at the origin, and $\vec a$ is free to point wherever it likes,  so long as $\vec b$ is the reflection of $\vec a$ in the \textbf{transverse plane} of the event.
 \end{enumerate}
Enumerated point 1 above is a more restrictive sub-case of pre-existing Case $h_1=1 $, and so \textbf{it may hereafter be disregarded}.

Furthermore, enumerated points 2 and 3 above above may be combined into a single simpler case, giving us instead the final result that:
\begin{quote}
The \textbf{non-chiral collider events} provided by this case that are not  already covered by earlier cases must have $m_p=m_q$ and $m_a=m_b$, but allow  $\vec c$ to lie anywhere on the \textbf{transverse plane}, and permit $\vec a$ to point wherever it likes, so long as $\vec b$ is then the reflection of $\vec a$ in the \textbf{transverse plane} of the event.
\end{quote}
Using Lemma~\ref{lem:TransPlaneReflectionLemma}, the above constraint may be re-phrased in a frame-independent way as follows:
$$
\cFingerConstraintABCPreFinal  a b    c,
$$
which (using the linear properties of Gram Determinants) can be re-written more symmetrically as
$$
\cFingerConstraintABC  a b    c.
$$

%\comCGL{CAN ONE DO THE WHOLE CALCULATION IN A FRAME INDEPENDENT WAY?}

%\comCGL{IS THERE A WAY OF REDUCING THE LAST TWO CONSTRAINTS TO A SINGLE NICER ONE?  I HAVE NO PROOF THIS CANNOT BE DONE.  Making them into a single constraint is easy (sum the squares of each LHS). The tricky part is making something nice -- which typically means retaining $p-q$ oddness. }

\subsubsection{Case $h_9$} $h_9= \gyzpi (pq) (ac)$.

Results here are the same as those of Case  $h_8=\gyzpi (pq) (ac)$ except with $b\leftrightarrow c$.

\subsubsection{Case $h_{10}$} $h_{10}=\gyzpi (pq) (bc)$.

Results here are the same as those of Case $h_8= \gyzpi (pq) (ac)$ except with $a\leftrightarrow c$.

\subsubsection{Case $h_{11}$} $h_{11}= \gyzpi (pq)(abc) $.

$
\implies
(m_p=m_q)
\land
(m_a=m_b=m_c)
\land
(|\mathbf{a}|=|\mathbf{b}|=|\mathbf{c}|)
\land
(a_z=-b_z=c_z=-a_z)
\land
(
   (\mathbf{a}=\mathbf{b}=\mathbf{c}=0)
   \lor 
   (
     (\theta-\beta=\pi-\alpha+2n_1\pi)
     \land 
     (\theta-\gamma=\pi-\beta+2n_2\pi)
     \land 
     (\theta-\alpha=\pi-\gamma+2n_3\pi)  
   )
)
$

\noindent which is evidently just a more restrictive version of Case $h_7= \gyzpi (pq) $, and so \textbf{may subsequently be ignored}.

\subsubsection{Case $h_{12}$} $h_{12}=\gyzpi(pq)(cba)$.

This case \textbf{may subsequently be ignored} for the same reason as Case $h_{11}=\gyzpi(pq)(abc)$.
\subsubsection{Summary of all cases}
The above results, when brought together, and when using the $[a,b,p,q]$ square-bracket notation explained in equation~(\ref{eq:epscontractionnotation} and the $G(\cdots)$ and $\Delta_3(\cdots)$ Gram Determinant notation defined in Appendix~\ref{app:gramnotation}), establish that a
\textbf{collision event} $e\in\pqabcColEvent$ is \textbf{non-chiral} if and only if:
\begin{gather*}
\bignonchiralcond.
\end{gather*}
The above result, once negated, concludes the proof of Theorem~\ref{thm:nonchiralpqabcEvents}.\label{ss:endofGIANTPROOF001}
\subsubsection{End of proof of Theorem~\ref{thm:nonchiralpqabcEvents}}
\label{sec:longproofend}

%\end{proof}%%MATCHES GIANTPROOF001

\section{Continuous Lorentz-invariant permutation-invariant parity-odd event variables for chiral events in various classes}
\label{sec:wherewearenow}

Each of the remaining parts of this section is named according to the class of \textbf{events} for which the results within it are relevant.  

For each class of \textbf{event} we seek to identify one or more sets of \textbf{continuous}, parity-odd, Lorentz-invariant, appropriately permutation-invariant event variables about which strong statements can be made concerning their \textbf{necessity} and \textbf{sufficiency}.\footnote{The definitions of \textbf{continuity}, \textbf{necessity} and \textbf{sufficiency} which are being used in this context may be found in Section~\ref{sec:ourobjectives} of the Introduction.}  

The \textbf{continuity}
requirement is met by our assembling higher-order event variables out of simpler ones (such as invariant masses, scalar products, and elemental pseudoscalars) by only ever using  the operations of multiplication, addition and subtraction. Since the simpler variables are continuous, and since addition, subtraction and multiplication (unlike division) maintain continuity (which is why all polynomials are continuous) the resulting high-level event-variables are also continuous.

We begin in Section~\ref{sec:exploratorysectionforpqabeventvars} by considering a very simple case: \textbf{events} in $\pqabAllEvent$.  For those the event variables may be written down with minimal effort. Though these variables are novel, the result is presented here less on account of its potential usefulness to physics than for its ability to serve as an example which illustrates, in miniature, the structure of the arguments being used to construct sufficient sets of parity-odd variables for the classes of \textbf{event} which then follow.  

In Section~\ref{sec:exploratorysectionforpqabceventvars}, we  address the more complex task of creating a set of \textbf{necessary} and  \textbf{sufficient} parity-odd variables which can ascribe parities first to \textbf{chiral} \textbf{collision events} of the form $e\in\pqabcColEvent$, and which are invariant under both the $S_2$ permutation of $(pq)$ and the $S_3$ permutation of $(abc)$.  Sections~\ref{sec:bringbothnoncolltypestogeterinpqabc} and \ref{sec:pqabcbothcollandnoncollvars} extend these variables to a set which is able to ascribe a non-zero parity to any \textbf{ events} in $e\in\pqabcAllEvent$, regardless of whether the $(pq)$-component is or is not in the initial state.

Finally, Section~\ref{sec:supportmaterials} describes resources which may help readers wishing to reproduce calcualtions from places within Section~\ref{sec:wherewearenow}.

\subsection{\textbf{Events} in $\pqabAllEvent$}
\label{sec:exploratorysectionforpqabeventvars}

%\subsubsection{\textbf{Non-collision events} in $\pqabEvent$}

%\subsubsection{\textbf{Collision events} in $\pqabEvent$}

\begin{lemma}[\textbf{Sufficiency} of $X_1$, $X_2$ and $X_3$]
\label{lem:threevarsforpqabevents}
For any $u_1,u_2\in \mathbb{R}$ with $u_1\ne0$ and $u_2\ne 0$ at least one of the three quantities:
\begin{align}
X_1 &=  \mathscr{P}^{ab}_{pq}(0)  G^{ab}_{pq}(2),\label{eq:unhattedx1} \\
X_2 &=  \mathscr{P}^{ab}_{pq}(0) \Re\left[ \left(   f^{ab} + i u_1 g^{a-b}_{a+b} \right)  
\left(  G^{\phantom{ab}}_{pq}(0) + i u_2  G^{\phantom{ab}}_{pq}(1)   \right) \right],\label{eq:unhattedx2} \qquad\text{and}\\
X_3 &=  \mathscr{P}^{ab}_{pq}(0) \Im\left[ \left(   f^{ab} + i u_1 g^{a-b}_{a+b} \right)  
\left(  G^{\phantom{ab}}_{pq}(0) + i u_2  G^{\phantom{ab}}_{pq}(1)   \right) \right]\label{eq:unhattedx3}
\end{align}
will be non-zero for any \textbf{event} $e\in\pqabAllEvent$ which is \textbf{chiral}.
\begin{proof}
\textbf{Events} in $\pqabAllEvent$ are either \textbf{collision events} or \textbf{non-collision events}. Lemma~\ref{lem:everynoncollevinpqabxallisnonchiral} has already shown  that the only \textbf{events} in $\pqabAllEvent$ capable of being \textbf{chiral} are \textbf{collision events}. Attempts to construct parity odd variables may therefore, without any loss of generality, focus entirely on \textbf{collision events}.  With that in mind, the desired result follows from inspection of  \eqref{eq:secondtimeBisdefined}, together with the fact that the product of two complex numbers is non-zero if and only if each of the complex numbers themselves is non-zero.  Note that the quantities $u_1$ and $u_2$  permit quantities like $X_2$ and $X_3$ to be dimensionally self-consistent. The quantities $u_1$ and $u_2$ therefore  represent a sort of `gauge freedom' within the definitions of $X_1$, $X_2$ and $X_3$.  The choice of $u_1$ and $u_2$ values can affect \textit{which} of  $X_1$, $X_2$ and $X_3$ is/are non-zero for any given event, but their choice cannot affect the statement that \textit{at least one of $X_1$, $X_2$ or $X_3$ is non-zero for each \textbf{chiral} \textbf{event}  $e\in\pqabAllEvent$.}
\end{proof}
\end{lemma}
\begin{lemma}[\textbf{\Independence} of $X_1$, $X_2$ and $X_3$]
None of the quantities $X_1$, $X_2$ $X_3$ can be omitted from the statement of Lemma~\ref{lem:threevarsforpqabevents} without invalidating it.
\begin{proof}
To demonstrate that $X_j$ cannot be omitted from the statement of Lemma~\ref{lem:threevarsforpqabevents} without invalidating it, it is sufficient to find an \textbf{event}  $e\in\pqabAllEvent$ for which $X_j\ne 0$ and for which $X_k=0$ for all $k\ne j$. To this end we observe (making use of the different types of Lorentz-vector notation described in Section~\ref{sec:lorentzvectornotation}) that:
\begin{itemize}
\item
$X_1$ may not be omitted since:
\begin{align}
\left(
\begin{array}{l}
a^\mu=\fourvecMassForm 0 {a_x} {+y} {+z},\\
b^\mu=\fourvecMassForm 0 {b_x} {-y} {-z},\\
p^\mu=\fourvecT e 0 0 {+p},\\
q^\mu=\fourvecT e 0 0 {-p}
\end{array}
\right)
\implies
\left(
\begin{array}{l}
X_1 = -32 (a_x + b_x) e^3 p^2 y z,\\
X_2 = 0, \\
X_3 = 0
\end{array}
\right)\label{eq:exampleeventforX1version1}
\end{align}
or
\begin{align}
\left(
\begin{array}{l}
a^\mu=\fourvecMassForm 0 {\lambda \cos\alpha} {+y} {\lambda \sin\alpha},\\
b^\mu=\fourvecMassForm 0 {\lambda \cos\beta} {-y} {\lambda \sin\beta},\\
p^\mu=\fourvecT e 0 0 {+p},\\
q^\mu=\fourvecT e 0 0 {-p}
\end{array}
\right)
\implies
\left(
\begin{array}{l}
X_1 = -16 e^3 p^2 y \lambda^2 (\cos\alpha + 
   \cos\beta) (\sin\alpha - \sin\beta),\\
X_2 = 0, \\
X_3 = 0
\end{array}
\right)\label{eq:exampleeventforX1version2}
\end{align}
\item 
while
$X_2$ may not be omitted since:
\begin{align}
\left(
\begin{array}{l}
a^\mu=\fourvecMassForm 0 {a_x} {+y} {z},\\
b^\mu=\fourvecMassForm 0 {b_x} {-y} {z},\\
p^\mu=\fourvecT e 0 0 {+p},\\
q^\mu=\fourvecT e 0 0 {-p}
\end{array}
\right)
\implies
\left(
\begin{array}{l}
X_1 = 0,\\
X_2 = -128 (a_x - b_x) (a_x + b_x)^2 e^5 p^2 u_1 u_2 y z, \\
X_3 = 0
\end{array}
\right)
\end{align}
or
\begin{align}
\left(
\begin{array}{l}
a^\mu=\fourvecMassForm {m_a} {a_x} {+y} {z_0},\\
b^\mu=\fourvecMassForm {m_b} {b_x} {-y} {z_0},\\
p^\mu=\fourvecT {e_p} 0 0 {+p},\\
q^\mu=\fourvecT {e_q} 0 0 {-p}
\end{array}
\right)
\implies
\left(
\begin{array}{l}
X_1 = 0,\\
X_2 = (\text{something that can be non-zero}), \\
X_3 = 0
\end{array}
\right)
\end{align}
when $z_0=%(-a_x^2 e_p^2 u_1 + b_x^2 e_p^2 u_1 + a_x^2 e_q^2 u_1 -  b_x^2 e_q^2 u_1)/(4 (m_a^2 - m_b^2) p u_2)
-\frac{(a_x^2 - b_x^2)  (e_p^2 - e_q^2)  u_1}{
 4 (m_a^2 - m_b^2)p u_2}$,
 \item
 and finally $X_3$ may not be omitted since:
\begin{align}
\left(
\begin{array}{l}
a^\mu=\fourvecMassForm 0 {a_x} {+y} {0},\\
b^\mu=\fourvecMassForm 0 {b_x} {-y} {0},\\
p^\mu=\fourvecT {e_p} 0 0 {+p},\\
q^\mu=\fourvecT {e_q} 0 0 {-p}
\end{array}
\right)
\implies
\left(
\begin{array}{l}
X_1 = 0,\\
X_2 = 0, \\
X_3 = -(a_x - b_x) (a_x + b_x)^2 (e_p - e_q) (e_p + e_q)^4 p u_1 y
\end{array}
\right)
\end{align}
or
\begin{align}
\left(
\begin{array}{l}
a^\mu=\fourvecMassForm {m_a} {x} {+y} {z},\\
b^\mu=\fourvecMassForm {m_b} {x} {-y} {z},\\
p^\mu=\fourvecT {e} 0 0 {+p},\\
q^\mu=\fourvecT {e} 0 0 {-p}
\end{array}
\right)
\implies
\left(
\begin{array}{l}
X_1 = 0,\\
X_2 = 0, \\
X_3 =-64 e^3 (m_a - m_b) (m_a + m_b) p^2 u_2 x y z
\end{array}
\right).
\end{align}
\end{itemize}

\end{proof}
\end{lemma}

\begin{corollary}
The parity-even part of any event selection is not unimportant.  This statement applies regardless of whether the class of events is $\pqabAllEvent$ or some other class of events.
\label{cor:theneedforparityeveneventselectionsaswellaspoddvars}
\begin{proof}
It should not need saying that  one cannot discover new physics if one throws away the events that contain the signs of new physics.  But a more subtle point is that a poor event selection can render as unobservable signs of new physics for which evidence is present in the events selected.  The events parameterised in \eqref{eq:exampleeventforX1version1} and \eqref{eq:exampleeventforX1version2} provide a nice concrete illustration of this point. Specifically:  suppose that there were a hypothetical parity-violating  new-physics process ``$C$'' which had a cross section ``$\sigma$'' for producing events of the form shown in \eqref{eq:exampleeventforX1version1}  with the parameters taking the values $a_x=y=z=1$ and $b_x=1$.  An event of that type (we will call it $c$) would look like this:
\begin{align}
c=\left(
\begin{array}{l}
a^\mu=\fourvecMassForm 0 1 {+1} {+1},\\
b^\mu=\fourvecMassForm 0 1 {-1} {-1},\\
p^\mu=\fourvecT e 0 0 {+p},\\
q^\mu=\fourvecT e 0 0 {-p}
\end{array}
\right)
\end{align}
for which
\begin{align}
\left(
\begin{array}{l}
X_1(c) = -64e^3 p^2,\\
X_2(c) = 0, \\
X_3(c) =0
\end{array}
\right).
\end{align}
Furthermore, suppose that there is a different parity-violating new-physics process ``$D_1$'' which has the same cross section $\sigma$ but which only produces events of form shown in \eqref{eq:exampleeventforX1version2}  with the parameters taking the values $\alpha=0$, $\beta=\frac \pi 2$, $y=1$ and $\lambda=2$.  An event of that type (we will call it $d_1$) would look like this:
\begin{align}
d_1=\left(
\begin{array}{l}
a^\mu=\fourvecMassForm 0 2 {1} 0,\\
b^\mu=\fourvecMassForm 0 0 {-1} {2},\\
p^\mu=\fourvecT e 0 0 {+p},\\
q^\mu=\fourvecT e 0 0 {-p}
\end{array}
\right)
\end{align}
for which
\begin{align}
\left(
\begin{array}{l}
X_1(d_1) = +64 e^3 p^2,\\
X_2(d_1) = 0, \\
X_3(d_1) =0
\end{array}
\right)
.
\end{align}
One readily sees that because the values of $X_1$ are equal and opposite for events of type $c$ and $d_1$, and because the cross sections for both new-physics processes are identical, a histogram of $X_1$-values built from a dataset containing events of both types would be symmetric and would therefore provide no evidence of parity violation.  The failure to see the parity violation is not, however, the fault of $X_1$.  The inability to discover parity violation stems from failure to separate two very distinguishable sources of parity violation.  Concretely: consider the parity-even Lorentz scalar $K$ defined\footnote{Note that although $K$ is parity-even it otherwise has the same `\textbf{invariance under specific symmetries}' which we required of the parity-odd event variables $X_n$.  Specifically $K$ is: (i) Lorentz-invariant, (ii) invariant with respect to  $a^\mu\leftrightarrow b^\mu$ and (iii)  invariant with respect to  $p^\mu\leftrightarrow q^\mu$. Furthermore, $K$ is \textbf{continuous}.}
by\begin{align}
K=(a^\mu(p_\mu-q_\mu))^2 + (b^\mu(p_\mu-q_\mu))^2.\end{align}  
Because
\begin{align}
K(c)&=+1,\qquad\text{and}\\
K(d_1)&=-1
\end{align}
an event selection which kept only events for which $K>0$ would therefore generate an asymmetric distribution for $X_1$ (specifically a delta function at $X_1=-64 e^3 p^2$) and so would find evidence of parity-violation  caused by process $C$. Similarly evidence for parity violation could also have been found if we had instead selected only events with $K<0$.  This time the observed parity violation would be due to process $D_1$.

We found above that a non-trivial event selection (in this case on a variable $K$) was needed to find the parity violation which was present in $C$ and $D_1$.  The $X_n$ variables have done what was claimed of them in the preamble to this paper: an asymmetry in at least one of their distributions provided evidence of observable parity violation given an appropriate selection.

But at this point, one might ask:
\begin{quote}`Were we lucky in the example just given?  Could a situation have arisen in which there was a third parity-violating new-physics processes $D_2$ which (like $D_1$) could hide its parity violation when present with $C$, but whose events \textbf{could not}, alas, be separated from those of $C$ by the use of an appropriate event selection?  In this case would the variables $X_n$ fail to do what is claimed of them?'\end{quote}

The answer here is a simple `no'.  At least one of the $X_n$ is always capable of seeing any  source of parity violation which is present in the dataset of input momenta (given enough luminosity and an appropriate event selection). The only way one source of parity violation, $C$, can completely hide another, $D_2$, is if for every type of event $e$ which $C$ produces, process $D$ produces the parity-inversion of  $e$ (up to a rotation or an element of the symmetry group  imposed on the events) at the same rate. In such a situation the totality of $C$ together with $D_2$ would amount to a parity conserving theory overall. Process $D_2$ would (in such a case) be the parity-complement of process $C$. Put differently: what it means to say that there is `observable parity violation within a dataset' is that there is an event selection for which a parity-odd variable can show an asymmetry.  This is simply true by definition. If no sufficiently good selection exits, then there was no observable parity violation in the first place!

We conclude this discussion by noting that our sister papers \cite{Gripaios:2020ori} and \cite{Gripaios:2020hya} exist to explain how (for any class of events) one may create a finite set $S=\{K_1,\ldots,K_m\}$ containing event variables $K_1,\ldots,K_m$ which collectively are guaranteed to be able to perform the role which $K$ performed in the example above. In summary:
\begin{itemize}
\item
the present paper shows how the number of parity-odd event-variables whose asymmetries  need to be measured (if one wishes to maintain full sensitivity to all forms of observable parity-violation) can be reduced from infinity to a finite number, while still ensuring that those parity-odd event-variables have the properties listed in  Section~\ref{sec:ourobjectives}; while
\item
our sister papers \cite{Gripaios:2020ori} and \cite{Gripaios:2020hya} demonstrate how the number of parity-even event-variables which are needed to form event-selections (if one is to make use of the parity-odd discovery variables mentioned in the last bullet point) can be reduced from infinity to a finite number, while still ensuring that those parity-even event-variables have the properties listed in  Section~\ref{sec:ourobjectives}; and so
\item
the goals of the present paper rely on the results of the sister papers.
\end{itemize}
\end{proof}
\end{corollary}

\begin{corollary}
The quantities $X_1$,  $X_2$ and $X_3$ defined in Lemma~\ref{lem:threevarsforpqabevents} are defined in a co-ordinate free way, and may be computed both for \textbf{collision events} and \textbf{non-collision events}. However, if one is willing to work just with \textbf{collision events}, one could choose to work in the frame in which $\vec p$ is equal and opposite to $\vec q$, with both $\vec p$ and $\vec q$ parallel to the $z$-axis, and with the axes oriented such that $p_z>0$, (see  Lemmas~\ref{lem:collisioncond} and \ref{lem:collisionshaveaxis}).  By glancing at the form of $X_1$, $X_2$ and $X_3$ in that frame one may write down three alternative variables $\hat X_1$,  $\hat X_2$ and $\hat X_3$ which may be easier to understand but share the key properties of  $X_1$, $X_2$ and $X_3$.
\begin{align}
\hat X_1 &=  ((\vec a \times \vec b )\cdot \hat {\vec p})%\mathscr{P}^{ab}_{pq}(0) 
( (\vec a - \vec b)\cdot \hat {\vec p} ) \nonumber \\
&=  (a_x b_y - a_y b_x) (a_z - b_z) \label{eq:hattedx1}
%G^{ab}_{pq}(2)
,\\
\hat X_2 &=  ((\vec a \times \vec b )\cdot \hat {\vec p})%\mathscr{P}^{ab}_{pq}(0) 
\ \Re\left[ \left(   (m_a-m_b) + i u_1 (|\vec a|-|\vec b|) \right)  
\left(  (m_p-m_q) + i u_2  (\vec a+\vec b)\cdot \hat{\vec p}   \right) \right] \nonumber
\\
%&=
%(a_x b_y - a_y b_x) \ \Re\left[ \left(   (m_a-m_b) + i u_1 (|\vec a|-|\vec b|) \right)  
%\left(  (m_p-m_q) + i u_2  (a_z+b_z)   \right) \right] %\nonumber
%\\
&=
(a_x b_y - a_y b_x) %\mathscr{P}^{ab}_{pq}(0) 
 \left(  (m_a-m_b)(m_p-m_q) -   u_1 u_2   (|\vec a|-|\vec b|) (a_z+b_z)   \right),\label{eq:hattedx2}
\\
\hat X_3 &=  ((\vec a \times \vec b )\cdot \hat {\vec p})%\mathscr{P}^{ab}_{pq}(0) 
\ \Im\left[ \left(   (m_a-m_b) + i u_1 (|\vec a|-|\vec b|) \right)  
\left(  (m_p-m_q) + i u_2  (\vec a+\vec b)\cdot \hat{\vec p}   \right) \right] \nonumber
\\
%&=
%(a_x b_y - a_y b_x) \ \Im\left[ \left(   (m_a-m_b) + i u_1 (|\vec a|-|\vec b|) \right)  
%\left(  (m_p-m_q) + i u_2  (a_z+b_z)   \right) \right]
%\nonumber
%\\
&=
(a_x b_y - a_y b_x) 
 \left(     u_1 (|\vec a|-|\vec b|) (m_p-m_q) +  u_2 (m_a-m_b) (a_z+b_z)   \right). \label{eq:hattedx3}
\end{align}
Note, these hatted variables are not identical to the un-hatted ones $X_1$,  $X_2$ and $X_3$.  Rather they are a different set of variables which nevertheless satisfy their own equivalent of   Lemma~\ref{lem:threevarsforpqabevents}
(provided that $e$ is a \textbf{collision event} in $\pqabColEvent$).
\end{corollary}

\subsection{%\textbf{chiral}
\textbf{Collision events} in $\pqabcColEvent$}

\label{sec:exploratorysectionforpqabceventvars}

We will first identify a set $S_{19}$ of 19 variables which, between them, will assign at least one non-zero parity to every \textbf{chiral} \textbf{collision event} $e\in\pqabcColEvent$, that is to say to any \textbf{collision event} satisfying the constraint $C$ of \eqref{eq:secondtimeCisdefined}.  $S_{19}$ will therefore have the \textbf{sufficiency} property mentioned in the introduction.  The equivalent task for events in $\pqabColEvent$ was performed (in Lemma~\ref{lem:threevarsforpqabevents}) largely by inspection.  With the addition of $c$ the task in hand will take longer to accomplish.

Once the set $S_{19}$ has been constructed, we will prove that it has the \textbf{irreducible} property mentioned in the introduction.

\subsubsection{Four sub-cases of $C$:  $C_1$, $C_2$, $C_3$ and $C_4$.}

\begin{figure}
\begin{centering}
\includegraphics[width=0.95\textwidth]{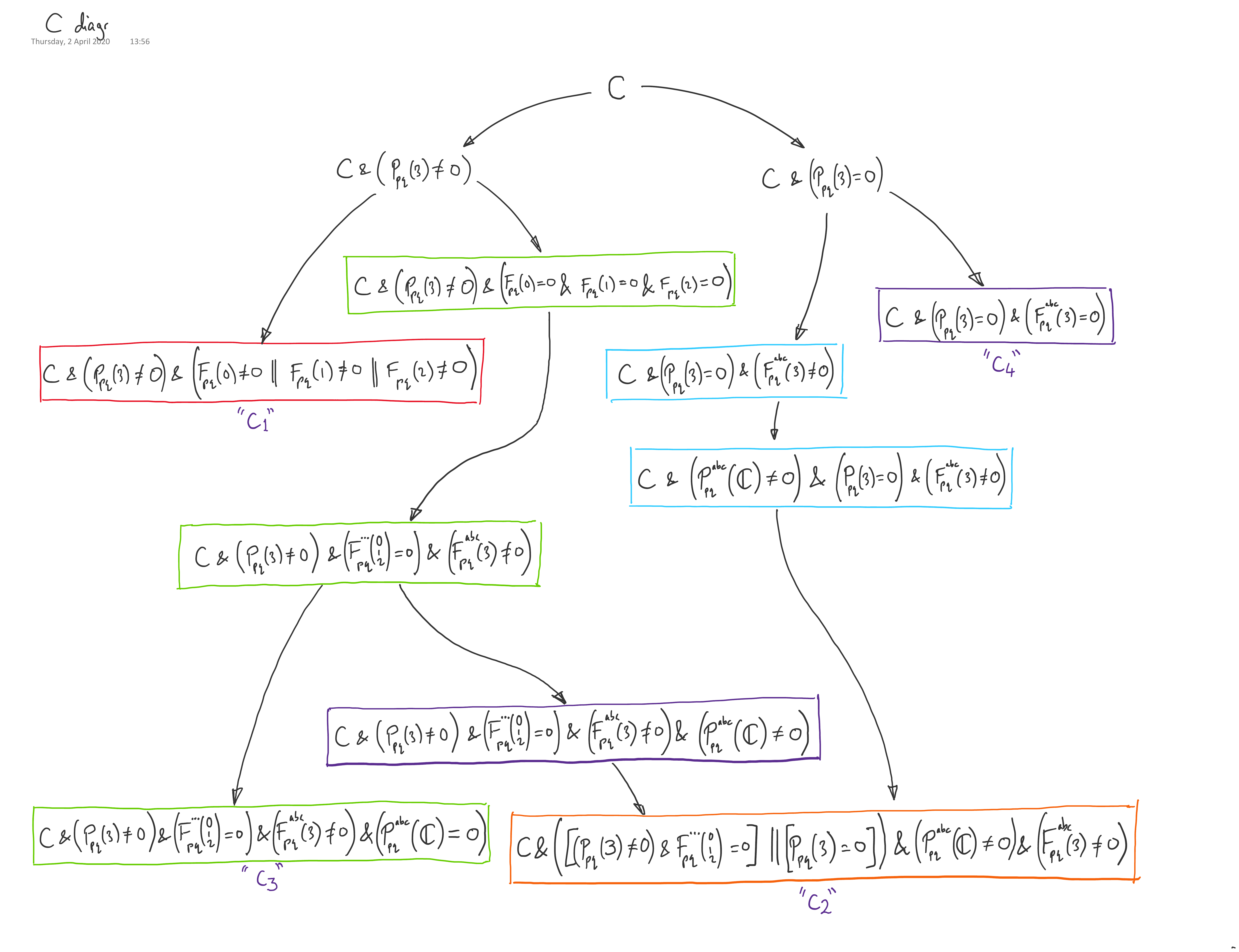}
\caption{The method by which $C$ is split into disjoint sub-cases, $C_1$, $C_2$, $C_3$ and $C_4$.  Note that the development of the second cyan box out of the first follows from the requirement within $C$ that $(\POne\ne0)\lor(\PTwo\ne0)\lor(\PThr\ne0)$. \label{fig:C_split_big}}
\end{centering}
\end{figure}

We begin by partitioning $C$ into four disjoint sub-cases, $C_1$, $C_2$, $C_3$ and $C_4$ as shown  in Figure~\ref{fig:C_split_big}. For each sub-case $C_i$ we will find a set of variables $S_{19}^{(i)}$ that will have the \textbf{sufficiency} property for any  \textbf{collisions event} in $\pqabcColEvent$ satisfying $C_i$.  The set $S_{19}$ will then be defined to be the union of these sets,
$S_{19}=S_{19}^{(1)} \cup S_{19}^{(2)} \cup S_{19}^{(3)} \cup S_{19}^{(4)}$, and will therefore, by construction, have the desired \textbf{sufficiency} property for any event satisfying $C$.

\begin{definition}
\label{def:CPARTBEFOREC1C2}
The following five \hyperlink{link:allpqabcvars}{variables} $V_1,\ldots,V_5$ are all \textbf{parity-odd}, are manifestly Lorentz-invariant, are manifestly invariant under permutations of $a$, $b$ and $c$, and are manifestly invariant under exchange of $p$ and $q$: 
\begin{align}
V_1 
    &=  
    \mathscr{P}^{abc}_{pq}(1)
    F^{abc}_{pq}(3),\label{eq:V1DEF}
\\
V_2
    &=  
    \mathscr{P}^{abc}_{pq}(2)
    F^{abc}_{pq}(3),
\\
V_3
    &=  
    \mathscr{P}^{\phantom{abc}}_{pq}(3)
    F^{\phantom{abc}}_{pq}(0),
\\
V_4
    &=  
    \mathscr{P}^{\phantom{abc}}_{pq}(3)
    F^{\phantom{abc}}_{pq}(1),
\\
V_5
    &=  
    \mathscr{P}^{\phantom{abc}}_{pq}(3)
    F^{\phantom{abc}}_{pq}(2),\label{eq:V5DEF}
\end{align}
%\begin{proof}
%It is plain to see from the $\mathscr P$-row and $F$-row of $C$ that these variables are non-zero, while their symmetries (indicated by their indices) show them to be as required.
%The indices on these variables show  $V_1$ to $V_5$ to have the symmetries stated.
%\end{proof}
\end{definition}

\subsubsection{A set of variables $S_{19}^{(1)}$ to cover $C_1$}

\begin{lemma}
\label{lem:s19ONE}
The set
\begin{align}
S_{19}^{(1)} = \{ V_3, V_4, V_5 \} \label{eq:s19ONE}
\end{align}
%(see Definition~\ref{def:CPARTBEFOREC1C2}) 
assigns at least one non-zero parity to every event in $C_1$ as required.
\begin{proof}
The result follows trivially from the definition of $C_1$.  [Recall that $C_1 = C \land (\mathscr{P}^{\phantom{abc}}_{pq}(3)\ne 0) \land
\left(
(F^{\phantom{abc}}_{pq}(0)\ne0) \lor
(F^{\phantom{abc}}_{pq}(1)\ne0) \lor
(F^{\phantom{abc}}_{pq}(2)\ne0)
\right).
$]
\end{proof}
\end{lemma}

\subsubsection{A set of variables $S_{19}^{(2)}$ to cover $C_2$}
\label{lem:s19TWO}
\begin{lemma}
The set
\begin{align}
S_{19}^{(2)} = \{ V_1, V_2 \}\label{eq:s19TWO}
\end{align}
%(see Definition~\ref{def:CPARTBEFOREC1C2}) 
assigns at least one non-zero parity to every event in $C_2$ as required.
\begin{proof}
Recall that 
\begin{align*}
C_2 
&= 
\left(
C \land (\text{another constraint})\land(\mathscr{P}^{{abc}}_{pq}(\mathbb{C})\ne 0) \land
(F^{{abc}}_{pq}(3)\ne0)\right)
\\
&=
\left(
C \land (\text{another constraint})\land\left(
(
\mathscr{P}^{{abc}}_{pq}(1)\ne0
)
\lor
(
\mathscr{P}^{{abc}}_{pq}(2)\ne0
)
\right) \land
(F^{{abc}}_{pq}(3)\ne0)\right)
\end{align*}
which is sufficient to ensure that at least one of $V_1$ and $V_2$ is non-zero on $C_2$.
\end{proof}
\end{lemma}

\subsubsection{A set of variables $S_{19}^{(3)}$ to cover $C_3$}
\label{sec:wemushtgivethissecaname}

\begin{definition}
Suppose $f(x,y;z)$ is a function which takes as arguments three four-momenta ($x$, $y$ and $z$) and which returns are real or complex Lorentz scalar (and not a pseudoscalar!).  Suppose also that $f$ is antisymmetric in the first two arguments: $f(x,y;z)=-f(y,x;z)$.  Given such an $f$, define the event variable $\mathscr{P}(f)$ as follows:\footnote{Perhaps it would be better called an event \textit{functional}.}
\begin{align}
\mathscr{P}(f)&\equivdef \sum_{
\tiny\begin{array}{c}
  \text{Perms of} 
      \\
      \{a,b,c\}
\end{array}} 
\left\{
\mathscr{P}^{\phantom{abc}}_{pq}(3)
\cdot
F^{abc}_{pq}(3)
\cdot
g^a_{p-q} \cdot g^c_{p-q} \cdot f(a,c;b)
\vphantom{\int}
\right\}.\label{eq:bigpf}
\end{align}
By construction such an event variable is parity-odd, Lorentz invariant, and invariant under permutations of both $(abc)$ and $(pq)$.
Since there $3!$ perms of $\{a,b,c\}$ such  a function is a sum of six terms.
\end{definition}

\begin{definition}
Define two auxiliary functions of four-momenta, $f_6(x,y;z)$ and $f_7(x,y;z)$ as follows:
\begin{align}
f_6(x,y;z) &= \epsShort^{xz}_{pq} - \epsShort^{yz}_{pq}, \\
f_7(x,y;z) &= x^2 - y^2.
\end{align}
In terms of those auxiliary functions defined two new event variables $V_6$ and $V_7$ as follows:
\begin{align}
V_6 \label{eq:v6def} &\equiv \mathscr{P}(f_6) \equiv \sum_{
\tiny\begin{array}{c}
  \text{Perms of} 
      \\
      \{a,b,c\}
\end{array}} 
\left\{
\mathscr{P}^{\phantom{abc}}_{pq}(3)
\cdot
F^{abc}_{pq}(3)
\cdot
g^a_{p-q} \cdot g^c_{p-q} \cdot \left(
(\epsShort^{ab}_{pq})^2
-
(\epsShort^{cb}_{pq})^2
\right)
\vphantom{\int}
\right\}, \qquad\text{and}\\
V_7 \label{eq:v7def} &\equiv \mathscr{P}(f_7) \equiv \sum_{
\tiny\begin{array}{c}
  \text{Perms of} 
      \\
      \{a,b,c\}
\end{array}} 
\left\{
\mathscr{P}^{\phantom{abc}}_{pq}(3)
\cdot
F^{abc}_{pq}(3)
\cdot
g^a_{p-q} \cdot g^c_{p-q} \cdot (a^2-c^2)
\vphantom{\int}
\right\}. 
\end{align}
\end{definition}

\begin{remark}
\noindent  The remainder of  Section~\ref{sec:wemushtgivethissecaname} will be directed toward proving Lemma~\ref{lem:C1maintheorem}. This lemma will show that the set of event variables $S_{19}^{(3)}=\{V_6, V_7\}$  is sufficient to assign a non-zero parity to any \textbf{chiral} \textbf{collision event} in $C_3$.  A first step toward doing so is to consider the conditions $C_3$ in more detail.
\end{remark}

\noindent Writing out $C_3$ explicitly we see that:
\begin{gather*}
C_3 \\
\equiv \\
\left(
    \PThr\ne 0)
\land
    (\FZer=0) \land (\FOne=0) \land (\FTwo=0)
\land
    (\FThr\ne0)
\land
    (\PCom=0)
\right) 
        \land
        C\\
        =\\
\toperator \\
    (\PCom=0)
\land
    (\PThr\ne 0)
\land
    (\FZer=0)
\land
    (\FOne=0)
\land
    (\FTwo=0)
\land
    (\FThr\ne0)
\\
\separatorlor \land
\\
 %%%%%%%%%%%%%%%%%% h1 %%%%%%%%%%%%%%%%%%%
 (
\isFalse{
  \POne
  \ne0}
)
\lor
(
\isFalse{
\PTwo
  \ne0
  }
)
\lor
(
\isTrue{
\PThr
  \ne0
  }
)
\\ 
\separatorlor \land
\\ %%%%%%%%%%%%%%%%%% h2,3,4: %%%%%%%%%%%%%%%%%%%
\left\{
\begin{array}{c}
(f^{ab}\ne0) \lor \isTrue{(g^{a-b}_{p-q}\ne0)} \lor (g^{a-b}_{a+b}\ne0) \lor (g^{a-b}_{a+b+c}\ne0) \lor (a=b) \\ 
\smallseparatorlor \land \\
(f^{bc}\ne0) \lor \isTrue{(g^{b-c}_{p-q}\ne0)} \lor (g^{b-c}_{b+c}\ne0) \lor (g^{b-c}_{a+b+c}\ne0) \lor (b=c) \\ 
\smallseparatorlor \land \\
(f^{ca}\ne0) \lor \isTrue{(g^{c-a}_{p-q}\ne0)} \lor (g^{c-a}_{c+a}\ne0) \lor (g^{c-a}_{a+b+c}\ne0) \lor (c=a) 
\end{array}
\right\}
\\
\separatorlor \land
\\ %%%%%%%%%%%%%%%%%% h7: %%%%%%%%%%%%%%%%%%%
(\isFalse{
 \FZer
  \ne0}
)
\lor
(\isFalse{
  \FOne
  \ne0}
)
\lor
(\isFalse{
\FTwo
  \ne0}
)
\lor
(
\isTrue{
\FThr
\equivdef
  g^{a-b}_{p-q} 
  g^{b-c}_{p-q} 
  g^{c-a}_{p-q}
  \ne0
  }
)
\\ 
\separatorlor \land
\\ %%%%%%%%%%%%%%%%%% h8,9,10: %%%%%%%%%%%%%%%%%%%
(\isFalse{
 \FZer
  \ne0}
)
\lor
(\isFalse{
  \FOne
  \ne0}
)
\lor
\left\{
\begin{array}{c}
    (f^{ab}\ne0) \lor (g^{a+b}_{p-q}\ne0) \lor (\Delta_3(a-b,p,q)\ne0) \\
    \tinyseparatorlor \land \\
    (f^{bc}\ne0) \lor (g^{b+c}_{p-q}\ne0) \lor (\Delta_3(b-c,p,q)\ne0) \\
    \tinyseparatorlor \land \\
    (f^{ca}\ne0) \lor (g^{c+a}_{p-q}\ne0) \lor (\Delta_3(c-a,p,q)\ne0) 
\end{array}
\right\}
.
\\
\toperator
\end{gather*}
in which the second, third and fourth `TRUE' remarks follow from the true `$(
F^{abc}_{pq}(3)
\equivdef
  g^{a-b}_{p-q} 
  g^{b-c}_{p-q} 
  g^{c-a}_{p-q}
  \ne0
)$' statement.  We may therefore simplify the above expression for $C_3$, re-writing it as:
\begin{gather*}
%%%%%%%%%%%%%%%%%%%%%%%%%%%%%%%%%%%%%%%%%%%%%%%%%%%%%%%%%%%%%%%%%%%%%%%%%%%%%%%%%%%%%%%%%%%%%%%%%%%%%%%%%%%%%%%%%%%%%%%%%%%%%%%%%%%%%%%%%%%%%%%%%%%%%%%%%%%%%%%%%%%%%%%%%%%%%%%%%%%%%%%
C_3 = \numberthis \label{eq:newC3kkk}\\
\toparotor {0.5\textwidth}
\\
        (\FZerOneTwo=0)
        \land
        (\FThr\ne0)
\\
\separotor \land {0.22\textwidth}
\\
        (\POne=0)
        \land
        (\PTwo=0)
        \land
        (\PThr\ne0)
\\
\separotor \land {0.22\textwidth}
\\ %%%%%%%%%%%%%%%%%% h8,9,10: %%%%%%%%%%%%%%%%%%%
\left\{
\begin{array}{c}
    (f^{ab}\ne0) \lor (g^{a+b}_{p-q}\ne0) \lor (\Delta_3(a-b,p,q)\ne0) \\
    \tinyseparatorlor \land \\
    (f^{bc}\ne0) \lor (g^{b+c}_{p-q}\ne0) \lor (\Delta_3(b-c,p,q)\ne0) \\
    \tinyseparatorlor \land \\
    (f^{ca}\ne0) \lor (g^{c+a}_{p-q}\ne0) \lor (\Delta_3(c-a,p,q)\ne0) 
\end{array}
\right\}\numberthis \label{con:cOneExtras}
\\
\toparotor {0.5\textwidth}.
\end{gather*}

\begin{lemma}
\label{lem:weakerpocketbarnowl}
When $C_3$ is satisfied, it is necessarily the case that one of the elements of $\{g^{a}_{p-q},g^{b}_{p-q},g^{c}_{p-q}\}$ is zero while the other two are non-zero and are unequal to each other.
\begin{proof}
Inspection of  \eqref{eq:newC3kkk} shows that $
C_3
\implies  \left(
        (\FTwo =0)
        \land 
        (\FThr\ne0) \right).
$
The result therefore follows from an application of Lemma~\ref{lem:weakercentralcolumnlem} using the assignments 
 $x=g^{a}_{p-q}$, $y=g^{b}_{p-q}$ and $z=g^{c}_{p-q}$, having noted the definitions of  $\FTwo$ and $\FThr$ in  \eqref{eq:F2def} and \eqref{eq:F3abcpqdef}.
\end{proof}
\end{lemma}

\begin{lemma}
\label{lem:pocketbarnowl}
When $C_3$ is satisfied, it is necessarily the case that one of the elements of $\{g^{a}_{p-q},g^{b}_{p-q},g^{c}_{p-q}\}$ is zero while the other two sum to zero but are themselves non-zero.
\begin{proof}
Inspection of  \eqref{eq:newC3kkk} shows that $
C_3
\implies  \left(
        (\FOne=0)
        \land 
        (\FTwo=0)
        \land 
        (\FThr\ne0) \right).
$
The result therefore follows from an application of Lemma~\ref{lem:centralcolumnlem} using the assignments 
 $x=g^{a}_{p-q}$, $y=g^{b}_{p-q}$ and $z=g^{c}_{p-q}$, having noted the definitions of $\FOne$,  $\FTwo$ and $\FThr$ in  \eqref{eq:F1def} to \eqref{eq:F3abcpqdef}.
\end{proof}
\end{lemma}
\begin{lemma}
\label{lem:pocketeagle}
When $C_3$ is satisfied, it is necessarily the case that one of the elements of $\{\epsShort^{ab}_{pq}, \epsShort^{bc}_{pq}, \epsShort^{ca}_{pq} \}$ is zero while the other two sum to zero but are themselves non-zero.
\begin{proof}
Inspection of  \eqref{eq:newC3kkk} shows that $
C_3
\implies  \left(
        (\POne=0)
        \land 
        (\PTwo=0)
        \land 
        (\PThr\ne0) \right).
$
The result therefore follows from an application of Lemma~\ref{lem:centralcolumnlem} using the assignments 
 $x=\epsShort^{ab}_{pq}$, $y=\epsShort^{bc}_{pq}$ and $z=\epsShort^{ca}_{pq}$, having noted the definitions of $\POne$,  $\PTwo$ and $\PThr$ in  \eqref{eq:p1def} to \eqref{eq:p3def}.
\end{proof}
\end{lemma}

\begin{lemma}
\label{eq:whenpfsarenonzero}
For an event satisfying $C_3$:
\begin{gather*}
\left[
\begin{array}{c}
\left( (f(b,c;a)\ne 0)\land(g^a_{p-q}=0)  \right) %\lor 
\\
%\lor \\
 \tinyseparatorlor \lor \\
%\qquad 
\left((f(c,a;b)\ne 0)\land(g^b_{p-q}=0)  \right) %\lor 
\\
%\lor \\
 \tinyseparatorlor \lor \\
%\qquad\qquad
\left((f(a,b;c)\ne 0)\land(g^c_{p-q}=0)  \right)
\end{array}
\right]
\implies 
 \left[ \mathscr{P}(f) \ne 0 \right].
\end{gather*}
\begin{proof}
For an event satisfying $C_3$ Lemma~\ref{lem:weakerpocketbarnowl} has already shown us that one of the elements of $\{g^{a}_{p-q},g^{b}_{p-q},g^{c}_{p-q}\}$ is zero while the other two are non-zero.  The $g$ which vanishes will `switch off'  four of the six terms in the sum which defines $\mathscr{P}(f)$ in \eqref{eq:bigpf}.  Suppose, without loss of generality, that $g^b_{p-q}=0$ and
\begin{align}
g^a_{p-q}\ne &0 \ne g^c_{p-q} \label{eq:whatwejustsupposed}
\end{align}  In such a case:
\begin{align*}
%f_4\equivdef
%%f_5\equivdef
\left.\vphantom{\int}\mathscr{P}(f)\right|_{g^b_{p-q}=0}&=
\sum_{
\tiny\begin{array}{c}
  \text{Perms of} 
      \\
      \{a,c\}
\end{array}} 
\left\{
\PThr
\cdot
F^{abc}_{pq}(3)
\cdot
g^a_{p-q} \cdot g^c_{p-q} 
f(a,c;b)
\vphantom{\int}
\right\}
\\
&= 
\PThr
\cdot
g^a_{p-q} \cdot g^c_{p-q} \cdot
\sum_{
\tiny\begin{array}{c}
  \text{Perms of} 
      \\
      \{a,c\}
\end{array}} 
\left\{
F^{abc}_{pq}(3)
\cdot
f(a,c;b)
\vphantom{\int}
\right\}
\\
&= 
\PThr
\cdot
g^a_{p-q} \cdot g^c_{p-q} \cdot
\left\{
F^{abc}_{pq}(3)
\cdot
f(a,c;b)
+
F^{cba}_{pq}(3)
\cdot
f(c,a;b)
\vphantom{\int}
\right\}
\\
&= 
\PThr
\cdot
g^a_{p-q} \cdot g^c_{p-q} \cdot
\left\{
F^{abc}_{pq}(3)
\cdot
f(a,c;b)
+
(
-F^{abc}_{pq}(3)
)
\cdot
(-f(a,c;b))
\vphantom{\int}
\right\}
\\
&= 
\PThr
\cdot
g^a_{p-q} \cdot g^c_{p-q} \cdot
2\cdot \FThr
\cdot \numberthis \label{eq:ourlastprod}
f(a,c;b)%\qquad\text{(using $f(x,y)$)}
.
\end{align*}
We recognise, therefore, that in such a case $\mathscr{P}(f)$  is non-zero if $f(a,c;b)$ is non-zero since the remaining terms in
the product \eqref{eq:ourlastprod} are non-zero either due to \eqref{eq:whatwejustsupposed} or due to constraints in $C_3$.
\end{proof}
\end{lemma}

\hide{
\begin{lemma}
\label{lem:veryshortlem}
\begin{gather*}
C_1 \implies \left(\text{there exists a non-zero $X\in\mathb{R}$ such that}
\left[
\begin{array}{c}
      (\epsShort^{bc}_{pq}=0) \land ( \epsShort^{ca}_{pq} = - \epsShort^{ab}_{pq} = X) \\
      \tinyseparatorlor \lor
      \\
    (\epsShort^{ca}_{pq}=0) \land ( \epsShort^{ab}_{pq} = - \epsShort^{bc}_{pq} = X)
    \\
    \tinyseparatorlor \lor 
    \\
( \epsShort^{ab}_{pq}=0) \land (\epsShort^{bc}_{pq} = - \epsShort^{ca}_{pq} = X)
\end{array}
\right]
\right).
\end{gather*}
\begin{proof}
The $\mathscr{P}$ constraints of \eqref{eq:newC3kkk}, together with Lemma~\ref{lem:centralcolumnlem}, tell us that of the elements of $\{\epsShort^{ab}_{pq},\epsShort^{bc}_{pq},\epsShort^{ca}_{pq}\}$ only one is zero, while the other two \textit{sum} to zero but are not zero themselves.
\end{proof}
\end{lemma}
}

\begin{lemma}
\label{lem:shouldreallybesandwhichlemma}
For an any \textbf{collision event} $e$ in $\pqabcColEvent$ the following result holds:
\begin{align*}
\hide{
\nibbledchocolatewafer {} \numberthis \label{eq:nibbledchocwafer} \\
\implies \\
}
C_3 &\implies \nibbledchocolatewafer {} 
\\
&\implies\left[
\begin{array}{l}
    \left[ (\epsShort^{ab}_{pq}=0)\land (\epsShort^{bc}_{pq}=-\epsShort^{ca}_{pq}\ne0)   \land   \left\{\text{$(\vec a-\vec b)\propto \vec p$ in the $(p+q)$ rest frame}\right\} \right] \lor \qquad\qquad
    \\
      \qquad\left[ (\epsShort^{bc}_{pq}=0) \land (\epsShort^{ca}_{pq}=-\epsShort^{ab}_{pq}\ne0)  
      \land
      \left\{\text{$(\vec b-\vec c)\propto \vec p$ in the $(p+q)$ rest frame}\right\} \right] \lor \qquad
    \\
      \qquad\qquad\left[ (\epsShort^{ca}_{pq}=0)\land (\epsShort^{ab}_{pq}=-\epsShort^{bc}_{pq}\ne0)   \land   \left\{\text{$(\vec c
      \vphantom{\vec b} %To get all the lines to have the same height
      -\vec a) \propto \vec p$ in the $(p+q)$ rest frame}\right\}
     \right] 
     \end{array}
    \right].%\numberthis\label{eq:dummyeqlabel3234234}
\end{align*}
\begin{proof}
%\label{lem:chocwafer}
%	For any \textbf{collision event} $e$ in $\pqabcColEvent$ the following result\footnote{Note that \eqref{eq:chocwaferstatment} contains a $\implies$ not an $\iff$ symbol. This is intentional as its RHS is compatible with $\mathscr{P}^{abc}_{pq}(2)\ne0$.} holds:
\noindent 
The existence of the first implication has already been noted in the proofs of Lemmas~\ref{lem:weakerpocketbarnowl} and \ref{lem:pocketeagle}.
Regarding the second implication:
Lemma~\ref{lem:weakerpocketbarnowl}
tells us that of the elements of $\{g^{a}_{p-q},g^{b}_{p-q},g^{c}_{p-q}\}$, exactly one is zero, while the other two are neither zero nor equal to each other.  Since the statement of the lemma which we are trying to prove is invariant under any permutation of $a$, $b$ and $c$ we may assume, \textbf{without loss of generality}, that \begin{align}
g^b_{p-q}=0\label{eq:wherewebrokethesymm}\qquad\qquad
\text{and}
\qquad\qquad 0\ne g^a_{p-q} \ne g^c_{p-q} \ne 0.
\end{align}
%for some $K\in\mathbb R$ with $K\ne 0$.  
Lemma~\ref{lem:pocketeagle}
tells us that of the elements of $\{\epsShort^{ab}_{pq},\epsShort^{bc}_{pq},\epsShort^{ca}_{pq}\}$ only one is zero, while the other two \textit{sum} to zero but are not zero themselves. 
There are then three separate possibilities for the epsilons, each of which can be parameterized by some $Y\in\mathbb R$ with $Y\ne0$:
\begin{align}
\left[
\begin{array}{c!{:}c!{\text{and}}c}
     \text{Case $a$} &  \epsShort^{bc}_{pq}=0 & \epsShort^{ca}_{pq} = - \epsShort^{ab}_{pq} = Y, \\
     \text{Case $b$} &  \epsShort^{ca}_{pq}=0 & \epsShort^{ab}_{pq} = - \epsShort^{bc}_{pq} = Y, \\
     \text{Case $c$} &  \epsShort^{ab}_{pq}=0 & \epsShort^{bc}_{pq} = - \epsShort^{ca}_{pq} = Y .
\end{array}
\right]\label{eq:theepsiloncases}
\end{align}
The result of the Lemma~\ref{lem:shouldreallybesandwhichlemma} will be proved if it is proved in each of the three cases listed above.  The proof of the Lemma for `Case $b$' will proceed differently to those for `Case $a$' and `Case $c$' 
on account of our having made $b$ special in \eqref{eq:wherewebrokethesymm}.  Proofs in each of the three sub-cases follow.  Each proof proceeds in the $(p+q)$ rest-frame, which Lemma~\ref{lem:collisionshaveaxis} tells us always exists.%, and all use the principle that the desired result will follow from showing that $\vec a\times\vec b\cdot\vec c\ne0$ in that frame.
\begin{subproof}[Case $a$]
Making use of Lemma~\ref{lem:abcvectripprodinLIform} this case asserts that \begin{align}
\vec b\times \vec c \cdot \vec p&=0,\label{eq:bcp0}
\\
\vec c\times \vec a \cdot \vec p&= X,\label{eq:capX}
\\
\vec b\times \vec a \cdot \vec p&=X\label{eq:bapX}
\end{align}
for some $X\in\mathbb R$ with $X\ne0$. %The fact that the RHS of the last two equations is non-zero immediately tells us that none of $\vec a$, $\vec b$, $\vec c$ or $\vec p$ is the zero vector. 
Subtracting the last two equations gives:
\begin{align}
(\vec c-\vec b)\times \vec a \cdot \vec p = 0.
\end{align}
This tells us that $(\vec c-\vec b)$, $\vec a$ and $\vec p$ are linearly dependent, \textit{i.e.}\ there exists some constants $\sigma, \lambda,\mu$ (not all zero) such that 
\begin{align}
\sigma (\vec c-\vec b) + \lambda \vec a + \mu \vec p = 0.\label{eq:firstlindepstatement}
\end{align}
If $\lambda\ne0$ then we would have
$$
\vec a = \frac \sigma \lambda (\vec b-\vec c)-\frac \mu \lambda \vec p
$$
from which we could deduce that
$$\vec a \times \vec c \cdot \vec p 
=
\frac \sigma \lambda \vec b \times \vec c \cdot \vec p = 0\text{\qquad(using \eqref{eq:bcp0})}
$$
which would contradict \eqref{eq:capX}.  It must therefore be the case that $\lambda=0$.   As $\lambda =0$, we know that: (i) $\sigma$ and $\mu$ are not both zero; and (ii) equation~\eqref{eq:firstlindepstatement} can be re-written as:
\begin{align}
    \sigma (\vec c-\vec b) + \mu \vec p = 0.\label{eq:notsurewhattocallthis}
\end{align}
Neither $\vec c-\vec b$ nor $\vec p$ is the zero vector (since $\vec c=\vec b$ would contradict \eqref{eq:wherewebrokethesymm} while $\vec p=0$ would contradict both \eqref{eq:capX} and the assumption that $e$ is a \textbf{collision event}) so both $\sigma$ \textbf{and} $\mu$ must be non-zero.  Therefore
\begin{align}
\vec c-\vec b = -\frac \mu \sigma \vec p\qquad\qquad\text{($\sigma\ne0$,$\mu\ne0$).} \label{eq:anotherparallalelsdfsdf}
\end{align}
\end{subproof}

\begin{subproof}[Case $b$]

This case asserts that \begin{align}
\vec c\times \vec a \cdot \vec p&=0,\label{eq:cap0}
\\
\vec a\times \vec b \cdot \vec p&= X,\label{eq:abpX}
\\
\vec c\times \vec b \cdot \vec p&=X\label{eq:cbpX}
\end{align}
for some $X\in\mathbb R$ with $X\ne0$. %The fact that the RHS of the last two equations is non-zero immediately tells us that none of $\vec a$, $\vec b$, $\vec c$ or $\vec p$ is the zero vector. 
Subtracting the last two equations gives:
\begin{align}
(\vec a-\vec c)\times \vec b \cdot \vec p = 0.
\end{align}
This tells us that $(\vec a-\vec c)$, $\vec b$ and $\vec p$ are linearly dependent, i.e.~there exists some constants $\alpha, \beta,\gamma$ (not all zero) such that 
\begin{align}
\alpha (\vec a-\vec c) + \beta \vec b + \gamma \vec p = 0.\label{eq:secondlindepstatement}
\end{align}
We will now prove that $\gamma\ne0$ by  assuming that $\gamma=0$ and getting a contradiction.  If it were the case that $\gamma=0$ we could re-write \eqref{eq:secondlindepstatement} as:
\begin{align}
\alpha (\vec c-\vec a) = \beta \vec b \label{eq:thirdlindepstatement}
\end{align}
	together with the new requirement that at least one of $\alpha$ and $\beta$ is non-zero.  Note, however, that neither $\vec c-\vec a$ nor $\vec b$ is the zero vector (since $\vec a=\vec c$ would contradict \eqref{eq:wherewebrokethesymm} while $\vec p=0$ would contradict both \eqref{eq:abpX} and the assumption that $e$ is a \textbf{collision event}) so having one of $\alpha$ and $\beta$ non-zero and the other zero is impossible.  We can therefore say that $\gamma=0$ implies $\alpha\ne0$ and $\beta\ne0$.   Evaluating $\GRAMTWO \bullet {p+q} {p-q} {p+q}$ on each side of  \eqref{eq:thirdlindepstatement}, and using the linear properties of Gram Determinants, gives:
\begin{align}
    \alpha g^c_{p-q} - \alpha g^a_{p-q} = \beta g^b_{p-q}
\end{align}
which (using \eqref{eq:wherewebrokethesymm}) may be re-written as:
\begin{align}
	\alpha \left( g^c_{p-q} - g^a_{p-q} \right) = 0
\end{align}
	which contradicts either \eqref{eq:wherewebrokethesymm} or the already established requirements that $\alpha\ne0$.  Accordingly, from this contradiction we conclude that $\gamma\ne0$ as originally desired.

Next we will prove that $\alpha\ne0$, also by contradiction.  If it were the case that $\alpha=0$, then since we now know that $\gamma\ne0$ we could re-write \eqref{eq:secondlindepstatement} as $\vec p = -\frac\beta\gamma \vec b$ which would contradict \eqref{eq:abpX}.  Accordingly $\alpha\ne0$ as required.

 Knowing that $\alpha\ne0$ and $\gamma\ne0$ we may now rewrite \eqref{eq:secondlindepstatement} as \begin{align}
\vec a - \vec c= -\frac\beta\alpha \vec b -\frac \gamma\alpha \vec p\qquad\text{($\alpha\ne0, \gamma\ne 0$)} \label{eq:magicIneedlater}
\end{align} 
 so that
\begin{align}
    \vec a\times\vec b\cdot\vec c
    &= (\vec a-\vec c)\times\vec b\cdot\vec c \nonumber
    \\
    &=
    \left( -\frac\beta\alpha \vec b -\frac \gamma\alpha\vec p\right)\times\vec b\cdot\vec c \nonumber
    \\
    &=
     -\frac \gamma\alpha\vec p \times\vec b\cdot\vec c \nonumber
     \\
    &=
     \frac \gamma\alpha X \qquad\text{(by \eqref{eq:cbpX})}\label{eq:thirdmagic}
\end{align}
which is non-zero since all of $\gamma$, $\alpha$ and $X$ are themselves non-zero.
%\begin{align}
%  \vec a\times\vec b\cdot\vec c
%    &\ne 0\label{eq:thirdmagiczero}
%\end{align}

\noindent
%Having demonstrated above, as desired, that $\vec a\times\vec b\cdot\vec c\ne0$ in `Case $b$', it would be possible to move straight to the `Case $c$' proof which follows.  However, one further result relation to `Case $b$' will first be found here, as it will be useful later in the document. 

We will now show that $\beta=0$.  To do this, note that \eqref{eq:magicIneedlater} can be arranged to read:
\begin{align}
    \vec p = \frac{\alpha}{\gamma}( \vec c-\vec a) - \frac \beta \gamma \vec b 
   \qquad \qquad\text{($\alpha\ne0, \gamma\ne0$)}\label{eq:secondmagic}
\end{align}
and thus
\begin{alignat}{2}
0 &= 
    \vec c\times \vec a \cdot \vec p &&{\text{(from \eqref{eq:cap0})}} \nonumber
    \\
    &=
    \vec c\times \vec a \cdot \left( \frac{\alpha}{\gamma}( \vec c-\vec a) - \frac \beta \gamma \vec b \right)\qquad&&{\text{(using \eqref{eq:secondmagic})}}\nonumber
    \\
      &=
   - \frac \beta \gamma \vec c\times \vec a \cdot   \vec b \nonumber
   \\
   &= -\frac{\beta X}{\gamma \alpha}  && \text{(using \eqref{eq:thirdmagic})}\nonumber
   \\
   &\implies (\beta=0) &&\text{(since $X\ne0$, $\gamma\ne0$ and $\alpha\ne0$).}
\end{alignat}

Putting the above result together with \eqref{eq:magicIneedlater} it is seen that in `Case $b$' it is always the case that:
\begin{align}
    \vec a - \vec c= -\frac \gamma\alpha \vec p\qquad\text{($\alpha\ne0, \gamma\ne 0$).} \label{eq:bigMagicIneedlater}
\end{align}
\end{subproof}
\begin{subproof}[Case $c$]
The proof of the lemma for `Case $c$' is exactly the same as the proof of the lemma for `Case $a$' but with $a\leftrightarrow c$.
\end{subproof}
\noindent This concludes the proof of Lemma~\ref{lem:shouldreallybesandwhichlemma}.
\end{proof}
\end{lemma}

\hide{
\begin{corollary}
\label{cor:safeujgh}
$C_1  \implies [a,b,c,p+q]\ne 0$.
\begin{proof}
This result follows from Lemma~\ref{lem:chocwafer} together with the $\mathscr{P}$- and $F$-rows of $C_1$ in (\ref{con:cOneExtras}).
\end{proof}
\end{corollary}
}

\begin{corollary}
\label{cor:notcurrentlyusedbutcouldbeuseful}
For an any \textbf{collision event} $e$ in $\pqabcColEvent$ the following result holds:
\begin{gather*}
C_3\implies
\nibbledchocolatewafer {} 
\implies ([a,b,c,(p+q)]\ne0).
\end{gather*}
\begin{proof}
The proof of Lemma~\ref{lem:shouldreallybesandwhichlemma} was divided into three cases: (a), (b) and (c).
In  \eqref{eq:anotherparallalelsdfsdf} of case (a) it was seen that \begin{align}
\vec c-\vec b = -\frac \mu \sigma \vec p\qquad\qquad\text{($\sigma\ne0$,$\mu\ne0$)}
\end{align}
which shows that
$$
\vec a\times \vec b\cdot \vec c 
=
\vec a\times \vec b\cdot (\vec c -\vec b)
=
\vec a\times \vec b\cdot (-\frac \mu \sigma \vec p)
=
\frac \mu \sigma \vec b\times \vec a\cdot \vec p
=
\frac \mu \sigma X \qquad\text{(by \eqref{eq:bapX})}
$$ which is non-zero since none of $\mu$, $\sigma$ or $X$ is zero.
In \eqref {eq:thirdmagic} of case (b) it was already observed that $\vec a\times \vec b\cdot \vec c \ne 0$.
Case (c) would proceed exactly as case (a) but with $a\leftrightarrow c$.
In all cases, therefore, we have seen that $\vec a\times \vec b\cdot \vec c \ne 0$.  This, together with Lemma~\ref{lem:abcvectripprodinLIform} concludes the proof.
\end{proof}
\end{corollary}

\begin{definition}[\textbf{$C_3$ brace shorthand}]
\label{def:braceshorthand}Symbolically, we shall denote the curly braced part of (\ref{con:cOneExtras}),  reproduced here:
\begin{align}
\left\{
\begin{array}{c}
    (f^{ab}\ne0) \lor (g^{a+b}_{p-q}\ne0) \lor (\Delta_3(a-b,p,q)\ne0) \\
    \tinyseparatorlor \land \\
    (f^{bc}\ne0) \lor (g^{b+c}_{p-q}\ne0) \lor (\Delta_3(b-c,p,q)\ne0) \\
    \tinyseparatorlor \land \\
    (f^{ca}\ne0) \lor (g^{c+a}_{p-q}\ne0) \lor (\Delta_3(c-a,p,q)\ne0) ,
\end{array}
\right\}
\end{align}
with this symbol:
$$
\tritribox {\phantom{\square}}{\phantom{\square}}{\phantom{\square}}{\phantom{\square}}{\phantom{\square}}{\phantom{\square}}{\phantom{\square}}{\phantom{\square}}{\phantom{\square}}
%\square  \square   \square 
 %          \square  \square   \square 
  %         \square  \square   \square
  .
$$
We will fill this symbol with ticks ($\tmark$), crosses ($\xmark$) or circles ($\circ$) to indicate which of the corresponding conditions are true, false or unconstrained respectively.  For example,
$$
\tritribox \tmark \circ \circ
           \circ \circ \xmark
           \circ \circ \circ
$$
would mean that $(f^{ab}\ne0)$ and  $(\Delta_3(b-c,p,q)=0)$, with the state of all other constraints unspecified.

\end{definition}

\begin{lemma}
\label{lem:thingstodowithmiddlecolumns}
If the constraints of $C_3$ in  (\ref{con:cOneExtras}) are satisfied, then the \textbf{brace shorthand} of Definition~\ref{def:braceshorthand} must  take one of the following six forms:
$$
\tritribox \circ \tmark \circ
           \circ \tmark \circ
           \tmark \xmark \circ,
\tritribox \circ \tmark \circ
           \circ \tmark \circ
           \circ \xmark \tmark,
\tritribox \circ \tmark \circ 
           \tmark \xmark \circ
           \circ \tmark \circ,
\tritribox \circ \tmark \circ
           \circ \xmark \tmark
           \circ \tmark \circ,
\tritribox \tmark \xmark \circ
           \circ \tmark \circ 
           \circ \tmark \circ,
\tritribox \circ \xmark \tmark
           \circ \tmark \circ 
           \circ \tmark \circ,
$$
which may be summarised more succinctly as
$$
\left[
\tritribox \circ \tmark \circ
           \circ \tmark \circ
           \tmark \xmark \circ,
\tritribox \circ \tmark \circ
           \circ \tmark \circ
           \circ \xmark \tmark,
           \text{and all row permutations thereof}\right].
$$
\begin{proof}
That the centre column need always contain two ticks and one cross is easy to show. Lemma~\ref{lem:pocketbarnowl} tells us that of the elements of  $\{g^{a}_{p-q},g^{b}_{p-q},g^{c}_{p-q},\}$ only one is zero, while the other two \textit{sum} to zero but are not zero themselves. Suppose, without loss of generality, that $g^b_{p-q}=0$ while the other two are non-zero but sum to zero $g^{a}_{p-q}+g^{c}_{p-q}=0$.  In such a case it is trivial to see that the central column of constraints is:
\begin{align}
\left(
\begin{array}{c}
 g^{a+b}_{p-q}\ne 0  \\
 g^{b+c}_{p-q}\ne 0  \\
 g^{c+a}_{p-q}\ne 0  
\end{array}
\right)
=
\left(
\begin{array}{c!{\ne}l}
 g^{a}_{p-q}& 0  \\
 g^{c}_{p-q}& 0  \\
 0& 0  
\end{array}
\right)
=
\left(
\begin{array}{c}
 \tmark \\
 \tmark  \\
 \xmark   
\end{array}
\right).\label{eq:cenrecolumnfun}
\end{align}
The proof is completed by noting that in order for the $C_1$ constraints to be satisfied, there must be at least one tick in every row of the \textbf{brace shorthand}.
\end{proof}
\end{lemma}

\begin{lemma}
\label{lem:importantp4lemma}
$V_6$ of \eqref{eq:v6def} is non-zero when the $C_3$ constraint of \eqref{con:cOneExtras} is satisfied by an event having one of the following three types:
$$
\tritribox \circ \tmark \circ
           \circ \tmark \circ
           \circ \xmark \tmark _D,
\tritribox \circ \tmark \circ
           \circ \xmark \tmark
           \circ \tmark \circ _E,
\tritribox \circ \xmark \tmark
           \circ \tmark \circ
           \circ \tmark \circ _F.
$$
\begin{proof}
If we can show the result for $\{\phantom{M} \}_D$ it will follow for $\{\phantom{M} \}_E$ and $\{\phantom{M} \}_F$ by symmetry.
The proof of Lemma~\ref{lem:thingstodowithmiddlecolumns} has already shown that the central column of $\{\phantom{M} \}_D$ implies that
%\begin{align}
$g^b_{p-q}=0
%\qquad\text{and}\qquad g^a_{p-q}=-g^c_{p-q}\ne0\label{eq:clincher}
$
%\end{align}
and so Lemma~\ref{eq:whenpfsarenonzero} tells us that $V_6$ will be non-zero if we can show that  $\left(\epsShort^{ab}_{pq}\right)^2-\left(\epsShort^{bc}_{pq}\right)^2\ne0$ when $\{\phantom{M} \}_D$ pertains.
To prove that $\left(\epsShort^{ab}_{pq}\right)^2-\left(\epsShort^{bc}_{pq}\right)^2$ is not zero, we use Lemma~\ref{lem:shouldreallybesandwhichlemma}.  This lemma tells us that we may choose to divide the problem into three sub-cases, one for each of the line in this expression:
\begin{align}
\label{expr:thisexprhasthreelines}
\left[
\begin{array}{l}
    \left[ (\epsShort^{ab}_{pq}=0)\land (\epsShort^{bc}_{pq}=-\epsShort^{ca}_{pq}\ne0)   \land   \left\{\text{$(\vec a-\vec b)\propto \vec p$ in the $(p+q)$ rest frame}\right\} \right] \lor \qquad\qquad
    \\
      \qquad\left[ (\epsShort^{bc}_{pq}=0) \land (\epsShort^{ca}_{pq}=-\epsShort^{ab}_{pq}\ne0)  
      \land
      \left\{\text{$(\vec b-\vec c)\propto \vec p$ in the $(p+q)$ rest frame}\right\} \right] \lor \qquad
    \\
      \qquad\qquad\left[ (\epsShort^{ca}_{pq}=0)\land (\epsShort^{ab}_{pq}=-\epsShort^{bc}_{pq}\ne0)   \land   \left\{\text{$(\vec c
      \vphantom{\vec b} %To get all the lines to have the same height
      -\vec a) \propto \vec p$ in the $(p+q)$ rest frame}\right\}
     \right] 
     \end{array}
    \right].
\end{align}

\noindent In the case shown in the \textit{first} line of \eqref{expr:thisexprhasthreelines} we see that
$\left(\epsShort^{ab}_{pq}\right)^2-\left(\epsShort^{bc}_{pq}\right)^2
=
-\left(\epsShort^{bc}_{pq}\right)^2 \ne 0$ as desired.

\noindent In the case shown in the \textit{second} line of \eqref{expr:thisexprhasthreelines} we see that
$\left(\epsShort^{ab}_{pq}\right)^2-\left(\epsShort^{bc}_{pq}\right)^2
=
\left(\epsShort^{ab}_{pq}\right)^2 \ne 0$ as desired.

\noindent The case shown in the \textit{third} line of \eqref{expr:thisexprhasthreelines} is more problematic, however.  Na\"ive substitution gives 
$\left(\epsShort^{ab}_{pq}\right)^2-\left(\epsShort^{bc}_{pq}\right)^2
=0$ which is not the result we desire.  We must show, therefore, that the case shown in the \textit{third} line above is incompatible with $\{\phantom{M} \}_D$ and is therefore impossible.  
Specifically, the incompatibility comes from the tick in the bottom right hand corner of $\{\phantom{M} \}_D$.  This tick asserts that $\Delta_3(a-c,p,q)\ne 0$ which,
given Lemma~\ref{lem:transmomforcolevent}, tells us that $\vec a -\vec c$ cannot be parallel to $\vec p$ in the $(p+q)$ rest frame.  This statement, however,  is incompatible with the $((\vec c-\vec a) \propto \vec p)$ requirement in the third line of \eqref{expr:thisexprhasthreelines}.
In all the cases which are compatible with $\{\phantom{M} \}_D$ we have shown that $\left(\epsShort^{ab}_{pq}\right)^2-\left(\epsShort^{bc}_{pq}\right)^2
\ne0$. This concludes the proof.
\end{proof}
\end{lemma}

%\noindent \comCGL{Low priority, but next two lemmas would flow  better if I re-named $V_6 \leftrightarrow V_7$.}

\begin{lemma}
\label{lem:importantp5lemma}
$V_7$ of \eqref{eq:v7def} is non-zero when the $C_3$ constraint of \eqref{con:cOneExtras} is satisfied by an event having one of the following three types:
$$
\tritribox \circ \tmark \circ
           \circ \tmark \circ
           \tmark \xmark \circ_A,
\tritribox \circ \tmark \circ
           \tmark \xmark \circ
           \circ \tmark \circ_B,
\tritribox \tmark \xmark \circ
           \circ \tmark \circ
           \circ \tmark \circ_C.
$$
\begin{proof}
If we can show the result for $\{\phantom{M} \}_A$ it will follow for $\{\phantom{M} \}_B$ and $\{\phantom{M} \}_C$ by symmetry.
The proof of Lemma~\ref{lem:thingstodowithmiddlecolumns} has already shown that the central column of $\{\phantom{M} \}_A$ implies that
%\begin{align}
$g^b_{p-q}=0
%\qquad\text{and}\qquad g^a_{p-q}=-g^c_{p-q}\ne0\label{eq:clincher}
$
%\end{align}
and so Lemma~\ref{eq:whenpfsarenonzero} tells us that $V_7$ will be non-zero if we can show that  $a^2-c^2\ne0$ when $\{\phantom{M} \}_A$ pertains.  However this is precisely what the tick in the bottom left corner of $\{\phantom{M} \}_A$ guarantees.
\end{proof}
\end{lemma}

\begin{lemma}
\label{lem:C1maintheorem}
\label{lem:s19THR}
Given the definitions of $V_6$ and $V_7$ in \eqref{eq:v6def} and \eqref{eq:v7def} respectively, the variables within the set
\begin{align}
S_{19}^{(3)} = \{ V_6, V_7 \} \label{eq:s19THR}
\end{align}
assign at least one non-zero parity to every event in $C_3$ as required.
\begin{proof}
This theorem is just a synthesis of Lemmas~\ref{lem:importantp5lemma} and
\ref{lem:importantp4lemma}.
\end{proof}
\end{lemma}

\begin{remark}
We do not currently use Corollary~\ref{cor:notcurrentlyusedbutcouldbeuseful}.  In principle it would permit $[a,b,c,(p+q)]$ to be used in place of the product
$$
\PThr
\cdot
\FThr$$
	in \eqref{eq:bigpf}, \eqref{eq:v6def} and \eqref{eq:v7def}.  This substitution would change (and simplify!) the definition of $\mathscr{P}(f)$, which in turn would change (and simplify!) the definitions of $V_6$ and $V_7$.  This substitution would not affect the \textbf{sufficiency} proofs associated with $V_6$ and $V_7$, however, it complicates the proofs associated with their joint \textbf{necessity} and so we therefore avoid making it.  It remains possible that the difficulty demonstrating the \textbf{\independence} of the final set of nineteen variables (when this substitution is made and we have a new $V_6$ and $V_7$) is a sign that this alternative set of 19 variables is not \textbf{\independent}, and that therefore our own $S_{19}$ is not a \textbf{minimal} set.  This line of enquiry should be investigated further. \followUpInFuture[Take note Gripaios and Hadaddin!  We won't be wrong as we don't claim $S_{19}$ is globally minimal.  But it would be a shame to miss a set smaller than $S_{19}$ if this is evidence pointing to the existence of one.]
\end{remark}

\subsubsection{A set of variables $S_{19}^{(4)}$ to cover $C_4$}
Recall that
\begin{gather*}
C_4 \\
= \\
\left( (\PThr=0)
        \land
        (\FThr=0)\right)\land C
\\
=
\\
\toperator \\
        (\PThr=0)
        \land
        (\FThr=0)
\\
\separatorlor \land
\\
 %%%%%%%%%%%%%%%%%% h1 %%%%%%%%%%%%%%%%%%%
(
  \POne
  \ne0
)
\lor
(
\PTwo
  \ne0
)
\lor
(
\isFalse{
\PThr
  \ne0}
)
\\ 
\separatorlor \land
\\ %%%%%%%%%%%%%%%%%% h2,3,4: %%%%%%%%%%%%%%%%%%%
\left\{
\begin{array}{c}
(f^{ab}\ne0) \lor (g^{a-b}_{p-q}\ne0) \lor (g^{a-b}_{a+b}\ne0) \lor (g^{a-b}_{a+b+c}\ne0) \lor (a=b) \\ 
\smallseparatorlor \land \\
(f^{bc}\ne0) \lor (g^{b-c}_{p-q}\ne0) \lor (g^{b-c}_{b+c}\ne0) \lor (g^{b-c}_{a+b+c}\ne0) \lor (b=c) \\ 
\smallseparatorlor \land \\
(f^{ca}\ne0) \lor (g^{c-a}_{p-q}\ne0) \lor (g^{c-a}_{c+a}\ne0) \lor (g^{c-a}_{a+b+c}\ne0) \lor (c=a) 
\end{array}
\right\}
\\
\separatorlor \land
\\ %%%%%%%%%%%%%%%%%% h7: %%%%%%%%%%%%%%%%%%%
(
 \FZer
  \ne0
)
\lor
(
  \FOne
  \ne0
)
\lor
(
\FTwo
  \ne0
)
\lor
(
\isFalse{
\FThr
  \ne0}
)
\\ 
\separatorlor \land
\\ %%%%%%%%%%%%%%%%%% h8,9,10: %%%%%%%%%%%%%%%%%%%
(
 \FZer
  \ne0
)
\lor
(
 \FOne
  \ne0
)
\lor
\left\{
\begin{array}{c}
    (f^{ab}\ne0) \lor (g^{a+b}_{p-q}\ne0) \lor (\Delta_3(a-b,p,q)\ne0) \\
    \tinyseparatorlor \land \\
    (f^{bc}\ne0) \lor (g^{b+c}_{p-q}\ne0) \lor (\Delta_3(b-c,p,q)\ne0) \\
    \tinyseparatorlor \land \\
    (f^{ca}\ne0) \lor (g^{c+a}_{p-q}\ne0) \lor (\Delta_3(c-a,p,q)\ne0) 
\end{array}
\right\}
\\
\toperator
\end{gather*}
and thus
\begin{gather*}
C_4 = \numberthis \label{con:c4first} \\
\toperator \\
        (\FThr
        \equivdef
  g^{a-b}_{p-q} 
  g^{b-c}_{p-q} 
  g^{c-a}_{p-q}=0)
\\
\separatorlor \land
\\
 %%%%%%%%%%%%%%%%%% h1 %%%%%%%%%%%%%%%%%%%
 \left[
(
  \POne
  \ne0
)
\lor
(
\PTwo
  \ne0
)
\right]
\land
(
\PThr
  =0
)\numberthis \label{con:c2firstppart}
\\ 
\separatorlor \land
\\ %%%%%%%%%%%%%%%%%% h2,3,4: %%%%%%%%%%%%%%%%%%%
\left\{
\begin{array}{c}
(f^{ab}\ne0) \lor (g^{a-b}_{p-q}\ne0) \lor (g^{a-b}_{a+b}\ne0) \lor (g^{a-b}_{a+b+c}\ne0) \lor (a=b) \\ 
\smallseparatorlor \land \\
(f^{bc}\ne0) \lor (g^{b-c}_{p-q}\ne0) \lor (g^{b-c}_{b+c}\ne0) \lor (g^{b-c}_{a+b+c}\ne0) \lor (b=c) \\ 
\smallseparatorlor \land \\
(f^{ca}\ne0) \lor (g^{c-a}_{p-q}\ne0) \lor (g^{c-a}_{c+a}\ne0) \lor (g^{c-a}_{a+b+c}\ne0) \lor (c=a) 
\end{array}
\right\} \numberthis \label{con:blockOfFours}
\\
\separatorlor \land
\\ %%%%%%%%%%%%%%%%%% h7: %%%%%%%%%%%%%%%%%%%
(
 \FZer
  \ne0
)
\lor
(
 \FOne
  \ne0
)
\lor
(
\FTwo
  \ne0
)
\\ 
\separatorlor \land
\\ %%%%%%%%%%%%%%%%%% h8,9,10: %%%%%%%%%%%%%%%%%%%
(
 \FZer
  \ne0
)
\lor
(
  \FOne
  \ne0
)
\lor
\left\{
\begin{array}{c}
    (f^{ab}\ne0) \lor (g^{a+b}_{p-q}\ne0) \lor (\Delta(a-b,p,q)\ne0) \\
    \tinyseparatorlor \land \\
    (f^{bc}\ne0) \lor (g^{b+c}_{p-q}\ne0) \lor (\Delta_3(b-c,p,q)\ne0) \\
    \tinyseparatorlor \land \\
    (f^{ca}\ne0) \lor (g^{c+a}_{p-q}\ne0) \lor (\Delta_3(c-a,p,q)\ne0) 
\end{array}
\right\}
\\
\toperator
.
\end{gather*}

\hide{
\begin{lemma}
Note that \eqref{eq:machineone}, \eqref{eq:CothersDEF} and \eqref{eq:machinetwo} imply that
\begin{align}
    C=C_0\lor C_1 \lor C_2.\label{eq:allthecsthingaddsup}
\end{align}
\end{lemma}
}

\begin{lemma}
$C_4$ of \eqref{con:c4first} is incompatible with $a=b$ and is incompatible with $b=c$ and is incompatible with $c=a$.
\begin{proof}
Given the $(a,b,c)$-symmetry of $C_4$ it is sufficient to prove that $C_4$ is incompatible with $a=b$.  Suppose therefore, without loss of generality, that $a=b$.  In this case, $\epsShort^{ab}_{pq}=0$ and so $\mathscr{P}^{abc}_{pq}(2)=\epsShort^{ab}_{pq} \epsShort^{bc}_{pq} \epsShort^{ca}_{pq}=0$.
Independently, $a=b$ implies that $\epsShort^{bc}_{pq} = -\epsShort^{ca}_{pq}$. This in turn means that $\mathscr{P}^{abc}_{pq}(1) = \epsShort^{ab}_{pq}+ \epsShort^{bc}_{pq} +\epsShort^{ca}_{pq}=0-\epsShort^{ca}_{pq} +\epsShort^{ca}_{pq}  =0$.  These two results ($\mathscr{P}^{abc}_{pq}(1)=0$ and $\mathscr{P}^{abc}_{pq}(2)=0$) are in conflict with \eqref{con:c2firstppart} which itself requires that at least one of them be non-zero.
\end{proof}
\end{lemma}

\begin{corollary}
We may therefore re-write $C_4$ in an arguably simpler and more explicit form in which the `$a=b$' options are removed and replaced by explicit (albeit redundant) statements to the contrary:
\begin{gather*}
C_4 = \numberthis \label{con:c2smaller} \\
\toparotor  {0.60\textwidth}
\\
  (a\ne b)
  \land (b\ne c)
  \land (c\ne a)
  \\
  \separotor \land {0.28\textwidth}
  \\
        (\FThr
        \equivdef
  g^{a-b}_{p-q} 
  g^{b-c}_{p-q} 
  g^{c-a}_{p-q}=0)
\\
\separotor \land {0.28\textwidth}
\\
 %%%%%%%%%%%%%%%%%% h1 %%%%%%%%%%%%%%%%%%%
 \left[
(
  \POne
  \ne0
)
\lor
(
\PTwo
  \ne0
)
\right]
\land
(
\PThr
  =0
)
\\ 
\separotor \land {0.28\textwidth}
\\ %%%%%%%%%%%%%%%%%% h2,3,4: %%%%%%%%%%%%%%%%%%%
\timesavingsandwichfilling
\numberthis \label{con:blockOfFourssmaller}
\\
\separotor \land {0.28\textwidth}
\\ %%%%%%%%%%%%%%%%%% h7: %%%%%%%%%%%%%%%%%%%
(
 \FZer
  \ne0
)
\lor
(
  \FOne
  \ne0
)
\lor
(
\FTwo
  \ne0
)
\\ 
\separotor \land {0.28\textwidth}
\\ %%%%%%%%%%%%%%%%%% h8,9,10: %%%%%%%%%%%%%%%%%%%
(
 \FZer
  \ne0
)
\lor
(
  \FOne
  \ne0
)
\lor
\left\{
\begin{array}{c}
    (f^{ab}\ne0) \lor (g^{a+b}_{p-q}\ne0) \lor (\Delta_3(a-b,p,q)\ne0) \\
    \tinyseparatorlor \land \\
    (f^{bc}\ne0) \lor (g^{b+c}_{p-q}\ne0) \lor (\Delta_3(b-c,p,q)\ne0) \\
    \tinyseparatorlor \land \\
    (f^{ca}\ne0) \lor (g^{c+a}_{p-q}\ne0) \lor (\Delta_3(c-a,p,q)\ne0) 
\end{array}
\right\}
\\
\toparotor  {0.60\textwidth}
.
\end{gather*}
\end{corollary}

\begin{remark}
The present goal is to find a set of parity-odd variables, at least one of which is provably non-zero for any event satisfying $C_4$.
One way to do so is to find a set of parity-odd variables which has that same desired property, but on a wider (i.e.~less restrictive) class of events.  We therefore now define a less restrictive class of events, $C_{4A}$, with a view to finding Parity-odd variables which work for any event satisfying its requirements:
\end{remark}
\begin{definition}[ $C_4=C_{4A}\land C_{4B}$]
Noting that $C_4$ is a chain of `anded' requirements, we split the constraint into two parts $C_{4A}$ and $C_{4B}$ such that 
$C_4=C_{4A}\land C_{4B}$:
%%%%%%%%%%%%%%%%%%%%%%%%%%%%%%%%%%%%%%%%%%%%%
\begin{gather*}
C_{4A} \equivdef \numberthis \label{con:C2Asmaller} \left[\begin{array}{c}
\toparotor  {0.60\textwidth}
\\
 %%%%%%%%%%%%%%%%%% h1 %%%%%%%%%%%%%%%%%%%
(
  \POne
  \ne0
)
\lor
(
\PTwo
  \ne0
)
\\ 
\separotor \land {0.28\textwidth}
\\ %%%%%%%%%%%%%%%%%% h2,3,4: %%%%%%%%%%%%%%%%%%%
\timesavingsandwichfilling
\\
\separotor \land {0.28\textwidth}
\\ %%%%%%%%%%%%%%%%%% h7: %%%%%%%%%%%%%%%%%%%
(
 \FZer
  \ne0
)
\lor
(
  \FOne
  \ne0
)
\lor
(
\FTwo
  \ne0
)
\\ 
\toparotor  {0.60\textwidth}
\end{array}
\right]
\end{gather*}
and
\begin{gather*}
C_{4B} \equivdef 
\\
\toparotor  {0.60\textwidth}
\\
  (a\ne b)
  \land (b\ne c)
  \land (c\ne a)
  \\
  \separotor \land {0.28\textwidth}
  \\
        (\FThr
        \equivdef
  g^{a-b}_{p-q} 
  g^{b-c}_{p-q} 
  g^{c-a}_{p-q}=0) \numberthis \label{eq:theawkwardfconstraintghtzdfsd}
\land
(
\PThr
  =0
)
  \\
  \separotor \land {0.28\textwidth}
  \\
 (
 \FZer
  \ne0
)
\lor
(
 \FOne
  \ne0
)
\lor
\left\{
\begin{array}{c}
    (f^{ab}\ne0) \lor (g^{a+b}_{p-q}\ne0) \lor (\Delta_3(a-b,p,q)\ne0) \\
    \tinyseparatorlor \land \\
    (f^{bc}\ne0) \lor (g^{b+c}_{p-q}\ne0) \lor (\Delta_3(b-c,p,q)\ne0) \\
    \tinyseparatorlor \land \\
    (f^{ca}\ne0) \lor (g^{c+a}_{p-q}\ne0) \lor (\Delta_3(c-a,p,q)\ne0) 
\end{array}
\right\}
\\ 
\toparotor  {0.60\textwidth}
.
\end{gather*}
\end{definition}
\begin{lemma}
\label{lem:allticksinonecolumnresult}
If $f(x,y)$ is a real or complex valued function which changes sign under exchange of its arguments (i.e.~$f(x,y)=-f(y,x)$) then the function
\begin{align}
    \pi(a,b,c) \equivdef 
      f(a,b) f(b,c) f(c,a)
\end{align}
 is totally antisymmetric in $a$, $b$ and $c$.    If, furthermore, $f(a,b)\ne0$, $f(b,c)\ne0$, $f(c,a)\ne0$, then $\pi(a,b,c)$ is non-zero.
\begin{proof}
Since the function $\pi(a,b,c)$ is evidently invariant under cyclic permutations of three labels ($a$, $b$ and $c$) it will be totally antisymmetric if it changes sign under the transposition of two of those labels.  Without loss of generality, therefore, consider  only 
$\pi(b,a,c) = 
      f(b,a) f(a,c) f(c,b)
    =
      -f(a,b) f(c,a) f(b,c)=-\pi(a,b,c)$ proving anti-symmetry.
Under the additional conditions stated, the function is evidently non-zero as it is a product of non-zero quantities required. [\textit{NB: Lemma~\ref{lem:allticksinonecolumnresult} is just a special case of Lemma~\ref{lem:someticksinonecolumnresult}. One may therefore prove Lemma~\ref{lem:allticksinonecolumnresult} as a corollary to Lemma~\ref{lem:someticksinonecolumnresult} by simply setting $g(x,y)$ equal to $f(x,y)$ therein.}]
\end{proof}
\end{lemma}
\begin{lemma}
\label{lem:someticksinonecolumnresult}
If $f(x,y)$ and $g(x,y)$ are real or complex valued functions which change sign under exchange of their arguments (i.e.~$f(x,y)=-f(y,x)$ and     $g(x,y)=-g(y,x)$) then the function
\begin{align}
    \sigma(a,b,c) \equivdef 
      f(a,b) f(b,c) g(c,a)
    + f(b,c) f(c,a) g(a,b)
    + f(c,a) f(a,b) g(b,c) 
\end{align}
 is totally antisymmetric in $a$, $b$ and $c$.    If, furthermore, $f(a,b)\ne0$, $f(b,c)\ne0$, $f(c,a)=0$ and $g(c,a)\ne0$, then $\sigma(a,b,c)$ is non-zero.
\begin{proof}
Since the function $\sigma(a,b,c)$ is a sum of terms related by cyclic permutations of three labels ($a$, $b$ and $c$) it is totally antisymmetric if it changes sign under the transposition of two of those labels.  Without loss of generality, therefore, consider  only 
\begin{alignat}{3}
\sigma(b,a,c) &= \hphantom{+}
      f(b,a) f(a,c) g(c,b)
    + f(a,c) f(c,b) g(b,a)
    + f(c,b) f(b,a) g(a,c) %done
    \nonumber \\
    &=
     - f(a,b) f(c,a) g(b,c)
    - f(c,a) f(b,c) g(a,b)
    - f(b,c) f(a,b) g(c,a)  && \ \text{(antisymmetry: $f,g$)} %done
    \nonumber \\
  &=
      -(f(a,b) f(c,a) g(b,c)
    + f(c,a) f(b,c) g(a,b)
    + f(b,c) f(a,b) g(c,a) ) && \ \text{(factorizing)} %done
    \nonumber \\
  &=
      -(
       f(b,c) f(a,b) g(c,a)
     + f(c,a) f(b,c) g(a,b)
     + f(a,b) f(c,a) g(b,c)
    ) && \ \text{(reordering sum)} %done
    \nonumber \\
  &=
     -(
        f(a,b) f(b,c)g(c,a)
     + f(b,c) f(c,a) g(a,b)
     + f(c,a) f(a,b) g(b,c)
    )   && \ \text{(commuting products)} \nonumber
     \\
    &=-\sigma(a,b,c) 
\end{alignat}
proving anti-symmetry.
Under the additional conditions stated, $\sigma(a,b,c)= f(a,b) f(b,c) g(c,a)
    + 0
    + 0=f(a,b) f(b,c) g(c,a)\ne0$ as required. 
    \begin{remark}
    The proof would fail if $f$ and $g$ were quaternion-valued functions.  This is because quaternions do not always commute, and so the last step in the proof of anti-symmetry would not be valid.
    \end{remark}
\end{proof}
\end{lemma}
\begin{definition}\label{def:fourfunctionswhichhelpus}
%For every $k\in\mathbb R$ with $k\ne0$  
Define six functions $f_i(x)$ as follows:
\begin{align}
   f_{1}(x)&\equivdef g^x_x = \Delta_2(x,p+q),\\
    f_{2}(x)&\equivdef g^x_{a+b+c}\label{eq:f4shorthandvardef}\\
    f_{3}(x)&\equivdef x^2,\\
    f_{4}(x)&\equivdef  g^x_{p-q},\label{eq:f2shorthandvardef},\\ %UPDATED
    f_{5}(x)&\equivdef\SYMGRAMTWO x \Sigma\equiv\SYMGRAMTWO x {a+b+c},
    \label{eq:f5shorthanddef}\\
    f_{6}(x)&\equivdef\GRAMTWO x \Sigma {p+q} \Sigma \equiv \GRAMTWO x {a+b+c} {p+q} {a+b+c} \label{eq:f6shorthanddef}
\end{align}
and let $\delta f_{i}^{xy}\equivdef f_{i}(x)-f_{i}(y)$ denote their differences.

Furthermore, %for every $k\in\mathbb R$ with $k\ne 0$ 
define two complex functions $f_{1\mathbb C}(x)$ and $f_{2\mathbb C}(x)$:
\begin{align}
    f_{1\mathbb C}(x) &\equivdef  f_{1}(x) + i f_{2}(x), \\
    f_{2\mathbb C}(x) &\equivdef  f_{3}(x) + i f_{4}(x),
\end{align}
and again let $\delta f_{i\mathbb C}(x,y)\equivdef f_{i\mathbb C}(x)-f_{i\mathbb C}(y)$ denote their differences.
\end{definition}
\begin{lemma}
Using the functions of Definition~\ref{def:fourfunctionswhichhelpus} we may write the curly-brace-enclosed constraints of \eqref{con:C2Asmaller} more symmetrically as either three sets of four real constraints (\eqref{eq:conasfourreals}) or as three sets of two complex constraints (\eqref{eq:conastwocomplex}):
\begin{align}
    C_{4A\{\}}&\equivdef\timesavingsandwichfilling \nonumber
    \\
    \nonumber
    \\
    &\equiv
    \left\{
    \begin{array}{c}
        (\delta f_1^{ab}\ne 0) \lor
        (\delta f_2^{ab}\ne 0) \lor
        (\delta f_3^{ab}\ne 0) \lor
        (\delta f_4^{ab}\ne 0)\\
        \separotor \land {3.5cm}\\
        (\delta f_1^{bc}\ne 0) \lor
        (\delta f_2^{bc}\ne 0) \lor
        (\delta f_3^{bc}\ne 0) \lor
        (\delta f_4^{bc}\ne 0)\\
        \separotor \land {3.5cm}\\
        (\delta f_1^{ca}\ne 0) \lor
        (\delta f_2^{ca}\ne 0) \lor
        (\delta f_3^{ca}\ne 0) \lor
        (\delta f_4^{ca}\ne 0)
    \end{array}
    \right\} \label{eq:conasfourreals}
    \\ 
    \nonumber \\
&\equiv
    \left\{
    \begin{array}{c}
        (\delta f_{1\mathbb C}(a,b)\ne 0) \lor (\delta f_{2\mathbb C}(a,b)\ne 0)\\
        \separotor \land {2cm}\\
        (\delta f_{1\mathbb C}(b,c)\ne 0) \lor (\delta f_{2\mathbb C}(b,c)\ne 0)\\
        \separotor \land {2cm}\\
        (\delta f_{1\mathbb C}(c,a)\ne 0) \lor (\delta f_{2\mathbb C}(c,a)\ne 0)
    \end{array}
    \right\}.\label{eq:conastwocomplex}
\end{align}
\begin{proof}
The proof of equivalence for column three of \eqref{eq:conasfourreals} is trivial; the proof for columns two and four follows from the linearity of Gram Determinants; while the proof for column one requires Lemma~\ref{lem:momdifferencesidentity}.  \end{proof}
\end{lemma}
\begin{definition}
In a similar manner to Definition~\ref{def:braceshorthand} we use shapes like:
\begin{gather}
\label{eq:examplefourthree}
\fourthreebox {\tmark & \circ & \circ & \circ}
              {\circ & \circ & \xmark & \circ}
              {\circ & \circ & \circ & \circ}
              \qquad\text{and}\qquad
\twothreebox {\tmark & \circ }
              {\circ & \xmark}
              {\circ & \circ }
\end{gather}
to represent, respectively, the status of the twelve real constraints in \eqref{eq:conasfourreals} and the six complex constraints in \eqref{eq:conastwocomplex}.
\end{definition}
\begin{lemma}
Self-consistency dictates that no column of constraints in \eqref{eq:examplefourthree} can contain one tick and two crosses. 
\begin{proof}
Consider, for example, a configuration like:
\begin{align*}
\fourthreebox {\circ & \circ & \tmark & \circ}
              {\circ & \circ & \xmark & \circ}
              {\circ & \circ & \xmark & \circ}
              .\end{align*}
It asserts that $a^2\ne b^2$ and $b^2=c^2$ and $c^2=a^2$.  This is not self consistent. The same argument works for any column, because in each column every constraint takes the form $\delta f(a,b)\equiv f(a)-f(b)\ne 0$ for some $f$. 
\end{proof}
\end{lemma}

\begin{corollary}
Any event satisfying $C_{4A}$ matches at least one of the following ten cases (or row permutations thereof):
\begin{align}
\begin{array}{cccc}
    \fourthreebox
    {\tmark & \circ & \circ & \circ}
    {\tmark & \circ & \circ & \circ}
    {\tmark & \circ & \circ & \circ}_1 &    
            \fourthreebox
            {\tmark & \circ & \circ & \circ}
            {\tmark & \circ & \circ & \circ}
            {\xmark & \tmark & \circ & \circ}_2 &    
                      \fourthreebox
                      {\tmark & \circ & \circ & \circ}
                      {\tmark & \circ & \circ & \circ}
                      {\xmark & \xmark & \tmark & \circ}_3 &    
                            \fourthreebox
                            {\tmark & \circ & \circ & \circ}
                            {\tmark & \circ & \circ & \circ}
                            {\xmark & \xmark & \xmark & \tmark}_4
    \\ \\
    \fourthreebox
    {\xmark & \tmark & \circ & \circ}
    {\xmark & \tmark & \circ & \circ}
    {\xmark & \tmark & \circ & \circ}_5 &    
            \fourthreebox
            {\xmark & \tmark & \circ & \circ}
            {\xmark & \tmark & \circ & \circ}
            {\xmark & \xmark & \tmark & \circ}_6 &    
                      \fourthreebox
                      {\xmark & \tmark & \circ & \circ}
                      {\xmark & \tmark & \circ & \circ}
                      {\xmark & \xmark & \xmark & \tmark}_7  
    & %%%%
    \fourthreebox
    {\xmark & \xmark & \tmark & \circ}
    {\xmark & \xmark & \tmark & \circ}
    {\xmark & \xmark & \tmark & \circ}_8 \\
    \\    
            \fourthreebox
            {\xmark & \xmark & \tmark & \circ}
            {\xmark & \xmark & \tmark & \circ}
            {\xmark & \xmark & \xmark & \tmark}_9   
  &%  \\ \\
    \fourthreebox
    {\xmark & \xmark & \xmark & \tmark}
    {\xmark & \xmark & \xmark & \tmark}
    {\xmark & \xmark & \xmark & \tmark}_{10},  
\end{array}\label{eq:tencasesofticks}
\end{align}
or equivalently such an event matches at least one of the following three complexly-constrained cases (or row permutations thereof):
\begin{align}
%\begin{array}{cccc}
    \twothreebox
    {\tmark & \circ}
    {\tmark & \circ}
    {\tmark & \circ}_{1\mathbb C}, \quad%&    
                        \twothreebox
                        {\tmark & \circ}
                        {\tmark & \circ}
                        {\xmark & \tmark}_{2\mathbb C},  \quad%\\   \\
    \twothreebox
    {\xmark & \tmark}
    {\xmark & \tmark}
    {\xmark & \tmark}_{3\mathbb C}\label{eq:threecasesofcomplexticks}.
%\end{array}
\end{align}
Note that case $\{\phantom{i}\}_{10}$, though needed for $C_{4A}$, is not  needed for $C_{4}$ since it is in conflict with the $F$ constraint shown to the left of \eqref{eq:theawkwardfconstraintghtzdfsd}. \label{cor:somesmalldiscussionofthings}
\end{corollary}
\begin{lemma}
\label{lem:antisymfuncsindivisionalgebras}
For any of the ten cases, $\{\phantom{i}\}_i$, shown in \eqref{eq:tencasesofticks} (or for any of the three cases, $\{\phantom{i}\}_{i\mathbb C}$,  shown in  \eqref{eq:threecasesofcomplexticks}) it is possible to construct a real or complex valued function $f_{i}(a,b,c)$ (respectively $f_{i\mathbb C}(a,b,c)$) which is totally antisymmetric in $a$, $b$ and $c$, and which is never zero when the associated constraint ($\{\phantom{i}\}_i$ or   $\{\phantom{i}\}_{i\mathbb C}$) is true.
\begin{proof}
Every case $i$ %$\{\phantom{i}\}_i$ 
has one of two properties.  Either: (i) it has three ticks in a single column - call it column $j$; or (ii) it has two ticks and one cross in one column (call it $j$) and a third tick in a different column (call it $s$) in the same row as the cross in column $j$.   Lemmas~\ref{lem:allticksinonecolumnresult} is sufficient to guarantee that for real-constrained cases with `property (i)' then
\begin{align}
f_{i}(a,b,c)
\equivdef
3 \delta f_{j}^{ab}  \delta f_{j}^{bc}  \delta f_{j}^{ca}
\end{align}
is an antisymmetric function with the desired properties, while Lemma~\ref{lem:someticksinonecolumnresult} is  sufficient to guarantee that for real-constrained cases having `property (ii)' that 
\begin{align}
f_{i}(a,b,c)
\equivdef
      \delta f_{j}^{ab} \delta f_{j}^{bc} \delta f_{s}^{ca} +
      \delta f_{j}^{ab} \delta f_{s}^{bc} \delta f_{j}^{ca} +
      \delta f_{s}^{ab} \delta f_{j}^{bc} \delta f_{j}^{ca}
\end{align}
will work.  Similarly, for the complex cases one has correspondingly success with
\begin{align}
f_{i\mathbb C}(a,b,c)
\equivdef
      3  \delta f_{j\mathbb C}^{ab} \delta f_{j\mathbb C}^{bc} \delta f_{j\mathbb C}^{ca}
\end{align}
and
\begin{align}
f_{i\mathbb C}(a,b,c)
\equivdef
      \delta f_{j\mathbb C}^{ab} \delta f_{j\mathbb C}^{bc} \delta f_{s\mathbb C}^{ca} +
      \delta f_{j\mathbb C}^{ab} \delta f_{s\mathbb C}^{bc} \delta f_{j\mathbb C}^{ca} +
      \delta f_{s\mathbb C}^{ab} \delta f_{j\mathbb C}^{bc} \delta f_{j\mathbb C}^{ca}.
\end{align}
\end{proof}
\end{lemma}
\begin{definition}
\label{def:Uabc}
To aid in emphasising commonalities in the above, we can define a universal functional $U^{abc}_{pq}[f(x,y),g_{pq}(x,y)]$, which takes two antisymmetric\footnote{By antisymmetric we mean that $f(x,y)=-f(y,x)$ and $g_{pq}(x,y)=-g_{pq}(y,x)$.} functions $f(x,y)$ and $g_{pq}(x,y)$ as inputs, and which returns a totally antisymmetric function of $a$, $b$ and $c$.  Specifically:
\hide{
\begin{align}
U^{abc}_{pq}[f(x,y),g_{pq}(x,y)] : (a,b,c) &\mapsto U^{abc}_{pq}[f(x,y),g_{pq}(x,y)](a,b,c) \nonumber
\end{align}
with
}
\begin{align*}
   & U^{abc}_{pq}[f(x,y),g_{pq}(x,y)]%(a,b,c) 
   = 
 %  \nonumber \\
%&\qquad\qquad   
f(a,b)f(b,c)g_{pq}(c,a) +
    f(a,b)g_{pq}(b,c)f(c,a) +
    g(a,b)_{pq}f(b,c)f(c,a).
\end{align*}
Note that because $U^{abc}_{pq}$ is linear in its second argument, $g_{pq}(x,y)$, it inherits whatever $(pq)$-evenness or $(pq)$-oddness is carried by that second argument.  Because we have presented the definition of $U^{abc}_{pq}$ using an implicitly $(pq)$-odd function $g_{pq}(x,y)$ we have given $U^{abc}_{pq}$ a subscript which indicates that inherited symmetry. However $U$ could just as easily be used with a $(pq)$-even function $g(x,y)$, in which case it would be denoted $U^{abc}$ rather than $U^{abc}_{pq}$.
\end{definition}

\begin{definition}
\begin{align}
\label{eq:usingAThirdPowerTrick}
%\mathscr{P}^{abc}_{pq}(\mathbb{C}(u))
\PComUS u
&\equivdef\PComU{u},\qquad\text{and}\\
\label{eq:notUsingAThirdPowerTrick}
%F_{pq}(\mathbb{C}(u))
\FComUS u
&\equivdef\FComU{u}.
\end{align}
\end{definition}
\begin{definition}
\label{def:oursixcomplexfs}
The following six complex-valued functions:
\begin{alignat*}{9}
v[1\mathbb{C}] 
&=
\PComUS {u_{2\phantom{1}}}
\cdot {}
&&U^{abc}_{pq}[
\delta f_1^{xy} + i u_1 \delta f_2^{xy}
,\ \ (
&&\delta f_1^{xy} + i u_1 \delta f_2^{xy})\cdot \FComUS {u_{3\phantom{1}}}
&&]
,
\\
v[2\mathbb{C}] 
&=
\PComUS {u_{5\phantom{1}}} 
\cdot {}
&&U^{abc}_{pq}[
\delta f_1^{xy} + i u_4 \delta f_2^{xy}
,\ \ (
&&\delta f_1^{xy} + i u_4 \delta f_2^{xy})\cdot \FZer
&&]
,
\\
v[3\mathbb{C}] 
&=
\PComUS {u_{8\phantom{1}}}
\cdot  {}
&&U^{abc}_{pq}[
\delta f_1^{xy} + i u_6 \delta f_2^{xy}
,
&&\delta f_3^{xy}\cdot \FComUS {u_{7\phantom{1}}}
&&]
,
\\
v[4\mathbb{C}] 
&=
\PComUS {u_{11}}
\cdot  {}
&&U^{abc}_{pq}[
\delta f_1^{xy} + i u_{9} \delta f_2^{xy}
,
&& \delta f_3^{xy} \cdot \FZer + i 
 u_{10} \cdot ( \delta f_4^{xy})_{pq}
&&],
\\
v[5\mathbb{C}] 
&=
\PComUS {u_{12}}
\cdot {}
&&U^{abc}_{pq}[
\delta f_3^{xy},
&&\delta f_3^{xy}\cdot  \FComUS {u_{13}}
&&]\qquad\text{and}
\\
v[6\mathbb{C}] 
&=
\PComUS {u_{15}}
\cdot  {}
&&U^{abc}_{pq}[
\delta f_3^{xy} 
,
&&  \delta f_3^{xy}\cdot \FZer + i 
 u_{14} \cdot(\delta f_4^{xy})_{pq}
&&]
\end{alignat*}
are defined in terms of fifteen real constants,  $\{u_1,\ldots,u_{15}\}$, whose only constraint is that each be non-zero. 
In the definition above, the two occurrences of the expression $\delta f_4^{xy}$ have each been wrapped in subscripted brackets: $\delta f_4^{xy}\rightarrow(\delta f_4^{xy})_{pq}$. This wrapping is purely to emphasise the $(pq)$-oddness which $\delta f_4^{xy}$ has and which is not shared by $\delta f_1^{xy}$, $\delta f_2^{xy}$ or $\delta f_4^{xy}$; no operation is performed by this bracketing.

\begin{remark}[Purpose of the $u_i$]
Later on we will be using the twelve real and imaginary parts of $v[1\mathbb{C}]$ to $v[6\mathbb{C}]$ as real parity-odd event variables which cover $C_4$.
The reason for using intermediate complex functions here, even though the ultimate goal is real variables is to permit the use of an efficient `trick' that is exploited in the second bullet point of the proof of Lemma~\ref{lem:howthecomplexesbreakdown}. As the remarks at the end of Lemma~\ref{lem:lemaboutproductsofthreecomplexnumbers} explain in more detail, the `magic' of that trick is that the the algebraic structure of the complex numbers allows a chain of anded pairwise ors to be expressed as no more than two real constraints, and so their use tends to result in fewer real variables at the end.  Nonetheless, it is little more than a `trick' and so it should not be a surprise if later a better approach is found that results in a smaller set of event variables.  As soon as complex numbers are introduced, the $u_i$ appear as a `gauge freedom' of sorts -- either out of necessity if unit consistency is required -- or to permit secondary optimisation.  The next remark discusses both of these.
\end{remark}

\begin{remark}[Choosing values for $u_i$]Where this document requires each $u_i$ to given a concrete value, we shall set each of them to 1. However, in practical applications it would be sensible to set each $u_i$ so that it balances the scales and units of the real and imaginary parts of the complex number it is used to build. 

For example, $u_3$ is used to create the complex number
$\FComUS {u_3}=\FComU {u_3}.$
Since $\FTwo$ has the same units as the \textbf{cube} of $\FOne$, it is likely that setting $u_3=1$ will lead to $\FComU {u_3}$ being dominated by either its real or its imaginary part.  This is not desirable. Ideally $u_3$ would be better chosen to be something similar to the ratio of the root mean square values of $\FOne$ and $\FTwo$, averaged on the data for which usage is expected. Alternatively the $u_i$ could simply be treated as hyperparameters for arbitrary optimisation. Although there is nothing wrong with this latter approach, it may  still be advisable to initialise any hyperparameter optimisation process  with values of the $u_i$ that set sensible scales (as described above) in order to reduce the chance of that process getting stuck in a local optimum rather than the global optimum.

Yet another approach which can reduce (but not remove) some of the difficulty in finding good values for the $u_i$ quantities is illustrated by the use of the power three seen in the definition of $\PComUS {u}$ in  \eqref{eq:usingAThirdPowerTrick}.  This power at least ensures that $u$ is dimensionless, thereby removing a reason for its optimal scale to vary in response to unit-changes in inputs.   A similar trick could be applied in \eqref{eq:notUsingAThirdPowerTrick}, however inserting a third power there makes it much harder to find the example events needed in  Lemma~\ref{lem:independenceLemmaForS19} when proving that the variables in $S_{19}$ are all \textbf{necessary}. For that reason we omit an extra power of three in \eqref{eq:notUsingAThirdPowerTrick}, however this does not mean that such a definition must be avoided by any end user.  
\end{remark}
\begin{remark}[Using fewer $u_i$]
If a user decides to optimise the $u_i$ for a particular purpose but is concerned about the potential for over training, then it seems likely that the eighteen $u_i$ values above could perhaps be reduced to four degrees of freedom with little loss of performance by setting the following internal constraints:
\begin{gather}
u_1=u_4=u_6=u_9,\\
u_2=u_5=u_8=u_{11}=u_{12}=u_{15},\\
u_3=u_7=u_{13},\\
u_{10}=u_{14}.
\end{gather}
This reduction may work since each of $u_1$, $u_4$, $u_6$ and $u_9$ exist to scale $\delta f^{xy}_1$ with respect to $\delta f^{xy}_2$ and so it seems likely that a common value might work for all of them.  Similar arguments may be advanced for the other groups.
\end{remark}
\end{definition}

\begin{corollary}
\label{cor:thesefunctionsarenice}
Each of the six functions  in Definition~\ref{def:oursixcomplexfs} is Lorentz-invariant, parity-odd, and invariant with respect to permutations of $(abc)$ or $(pq)$.
\begin{proof}
The Lorentz-invariant is trivially inherited from each function's ingredients.  The parity-oddness results from each variable having one and only one parity odd factor --- which in all cases has the form $(\PComU u)$ for some non-zero $u\in\mathbb R$. 
The invariance with respect to permutations of $(abc)$ or $(pq)$ may be seen by checking that every product contains both an even number of $(abc)$-odd terms and an even number of $(pq)$-odd terms as indicated by their superscripts or subscripts.
\end{proof}
\end{corollary}

\begin{lemma}
\label{lem:howthecomplexesbreakdown}
With reference to functions written in Definition~\ref{def:oursixcomplexfs} and the brace notation of \eqref{eq:tencasesofticks}, the following claims regarding \textbf{sufficiency} coverage are made:
\begin{itemize}
\item
$v[1\mathbb{C}]$ and $v[2\mathbb{C}]$ collectively cover $\{\phantom{i}\}_{1}$, $\{\phantom{i}\}_{2}$ and $\{\phantom{i}\}_{5}$,
\item
$v[3\mathbb{C}]$ and $v[4\mathbb{C}]$ collectively cover $\{\phantom{i}\}_{3}$, $\{\phantom{i}\}_{6}$, and $\{\phantom{i}\}_{7}$,
\item
$v[4\mathbb{C}]$ alone covers $\{\phantom{i}\}_{4}$, and
\item
$v[5\mathbb{C}]$ and
$v[6\mathbb{C}]$ 
collectively cover $\{\phantom{i}\}_{8}$ and $\{\phantom{i}\}_{9}$.
\end{itemize}
\begin{proof}
The proof appeals to:
\begin{itemize}
\item
Lemma~\ref{lem:antisymfuncsindivisionalgebras};
\item
the fact that any $n$ real or complex numbers are \textit{all} non-zero if and only if their product is non-zero;
\item
the fact that $C_4$ guarantees that $\PComUS{u}\ne 0$ for any non-zero $u\in\mathbb{R}$;
\item
the fact that $C_4$ guarantees that $\left((\FZer\ne0)\lor(\FComUS{u}\ne 0)\right)$ for any non-zero $u\in\mathbb{R}$; and 
\item
the locations within each brace $\{\phantom{i}\}_{i}$ from which the ingredients of each function is drawn.
\end{itemize}
Regarding the last of those bullet points, we present below a visual guide which uses round brackets to show which pairs of real constraints  have been combined into a single complex constraint, and which (by their absence) which have not:
\begin{itemize}
\item
$v[1\mathbb{C}]$ and $v[2\mathbb{C}]$ both draw from $
\fourthreebox
    {(\tmark & \tmark) & \circ & \circ}
    {(\tmark & \tmark) & \circ & \circ}
    {(\tmark & \tmark) & \circ & \circ}
    $,
\item
$v[3\mathbb{C}]$ draws from $
\fourthreebox
    {(\tmark & \tmark) & \circ & \circ}
    {(\tmark & \tmark) & \circ & \circ}
    {(\xmark & \xmark) & \tmark & \circ}
    $,
\item
$v[4\mathbb{C}]$ draws from $
\fourthreebox
    {(\tmark & \tmark) & \circ & \circ}
    {(\tmark & \tmark) & \circ & \circ}
    {(\xmark & \xmark) & (\tmark & \tmark)}
    $,
\item
$v[5\mathbb{C}]$ draws from $
\fourthreebox
    {\xmark & \xmark & \tmark & \circ}
    {\xmark & \xmark & \tmark & \circ}
    {\xmark & \xmark & \tmark & \circ}
    $,
\item
$v[6\mathbb{C}]$ draws separately from $
\fourthreebox
    {\xmark & \xmark & \tmark & \circ}
    {\xmark & \xmark & \tmark & \circ}
    {\xmark & \xmark & \tmark & \circ}
    $ and from $
\fourthreebox
    {\xmark & \xmark & \tmark & \circ}
    {\xmark & \xmark & \tmark & \circ}
    {\xmark & \xmark & \xmark & \tmark}
    $.
\end{itemize}
\end{proof}
\end{lemma}

\begin{corollary}
\label{cor:s19FOU}
If
\begin{alignat}{6}
V_{\phantom{1}8} &= \Re[v[1\mathbb C]], & \phantom{M} &
V_{\phantom{1}9} &&= \Im[v[1\mathbb C]],
\\
V_{10} &= \Re[v[2\mathbb C]], & \phantom{M} &
V_{11} &&= \Im[v[2\mathbb C]],
\\
V_{12} &= \Re[v[3\mathbb C]], & \phantom{M} &
V_{13} &&= \Im[v[3\mathbb C]],
\\
V_{14} &= \Re[v[4\mathbb C]], & \phantom{M} &
V_{15} &&= \Im[v[4\mathbb C]],
\\
V_{16} &= \Re[v[5\mathbb C]], & \phantom{M} &
V_{17} &&= \Im[v[5\mathbb C]],
\\
V_{18} &= \Re[v[6\mathbb C]], & \phantom{M} &
V_{19} &&= \Im[v[6\mathbb C]]
\end{alignat}
then the set 
\begin{align}
S_{19}^{(4)} = \{
V_{8}
,
V_{9},
V_{10},
V_{11},
V_{12},
V_{13},
V_{14},
V_{15},
V_{16},
V_{17},
V_{18},
V_{19}
\} \label{eq:s19FOU}
\end{align}
will ascribe at least one non-zero real parity to every \textbf{collision event} $e$ satisfying $C_4$, as required.
\begin{proof}
The proof is follows directly from Corollary~\ref{cor:thesefunctionsarenice} and Lemma~\ref{lem:howthecomplexesbreakdown}.
\end{proof}
\end{corollary}

\subsubsection{The sufficiency and necessity of the variables in $S_{19}$}

\label{sec:forfinallywedefines19properl}

\begin{corollary}[\textbf{Sufficiency}]

\label{cor:finallywedefines19properl}
   If the set of parity-odd invariant event variables $S_{19}$ is defined as follows:
\begin{align}
\label{eq:s19ALL}
S_{19} &= 
S_{19}^{(1)} \cup
S_{19}^{(2)} \cup
S_{19}^{(3)} \cup
S_{19}^{(4)}
\end{align}
in terms of subsets defined in  \eqref{eq:s19ONE},
\eqref{eq:s19TWO},
\eqref{eq:s19THR} and
\eqref{eq:s19FOU}, then Lemmas~\ref{lem:s19ONE}, \ref{lem:s19TWO} and \ref{lem:s19THR} together with Corollary \ref{cor:s19FOU} prove that the nineteen variables contained in $S_{19}$  will assign at least one non-zero parity to every \textbf{chiral} \textbf{collision event} $e\in\pqabcColEvent$.
\end{corollary}

\begin{lemma}[\textbf{Necessity}]Each \label{lem:idependencelemfors19} of the 19 event variables in $S_{19}$ can be shown to be \textbf{necessary} (at least in the case in which each of the constants $u_1,\ldots,u_{18}$ in
Definition \ref{def:oursixcomplexfs} are set equal to 1 in whatever unit system is used) by the explicit construction of a \textbf{collision event} $e\in\pqabcColEvent$ which results in a non-zero parity for that variable, and a zero-parity for every other variable in $S_{19}$.\label{lem:independenceLemmaForS19}
\begin{proof}
	Below are listed nineteen special \textbf{collision events} within $\pqabcColEvent$ which are named $\text{Event}[1]$ to $\text{Event}[19]$. The numbering is such that $\text{Event}[n]$ has the property that each of the variables in $\{V_1,\ldots,V_{19}\}\setminus\{V_n\}$ evaluates to zero while $V_n$ alone is non-zero.  $\text{Event}[n]$ thus serves to  demonstrate the necessity of retaining $V_n$ within $S_{19}$.  Note that although the aforementioned properties are independent of the choice of the non-zero real (but otherwise free) constants $u_1,\ldots,u_{18}$ in
Definition \ref{def:oursixcomplexfs},  some valid values of those constants cause the example events to have complex momenta.  Nonetheless, it may be observed that the example events have been chosen such that no event has complex momenta when all those $u_i$ constants take the value of 1. This is demonstrated (for any events which retain dependence on any $u_i$) by showing what form the momenta in those events take for unit $u_i$ values.  It may thus be seen that all example events may, simultaneously, be built of real and timelike momenta in $\mathbb{V}$, as required. Lorentz Vectors in the following list use the display conventions of Appendix~\ref{sec:lorentzvectornotation}.

\newcommand\pullInExample[1]{
\input{examples/independence/pqabc/collisionEvents/var#1.tex}
\text{Event}[#1]&=\left[
\begin{array}{l}
p^\mu=
\left(\ugrfguijoooopE,\left(\ugrfguijoooopX,\ugrfguijoooopY,\ugrfguijoooopZ\right)\right),\\
q^\mu=
\left(\ugrfguijooooqE,\left(\ugrfguijooooqX,\ugrfguijooooqY,\ugrfguijooooqZ\right)\right),\\
a^\mu=
\left(\ugrfguijooooaE,\left(\ugrfguijooooaX,\ugrfguijooooaY,\ugrfguijooooaZ\right)\right),\\
b^\mu=
\left(\ugrfguijoooobE,\left(\ugrfguijoooobX,\ugrfguijoooobY,\ugrfguijoooobZ\right)\right),\\
c^\mu=
\left(\ugrfguijoooocE,\left(\ugrfguijoooocX,\ugrfguijoooocY,\ugrfguijoooocZ\right)\right)
\end{array}
\right]
}

\newcommand\pullInExampleWDU[1]{

\input{examples/independence/pqabc/collisionEvents/var#1.tex}
\text{Event}[#1]&=\left[
\begin{array}{l}
p^\mu=
\left(\ugrfguijoooopE,\left(\ugrfguijoooopX,\ugrfguijoooopY,\ugrfguijoooopZ\right)\right),\\
q^\mu=
\left(\ugrfguijooooqE,\left(\ugrfguijooooqX,\ugrfguijooooqY,\ugrfguijooooqZ\right)\right),\\
a^\mu=
\left(\ugrfguijooooaE,\left(\ugrfguijooooaX,\ugrfguijooooaY,\ugrfguijooooaZ\right)\right),\\
b^\mu=
\left(\ugrfguijoooobE,\left(\ugrfguijoooobX,\ugrfguijoooobY,\ugrfguijoooobZ\right)\right),\\
c^\mu=
\left(\ugrfguijoooocE,\left(\ugrfguijoooocX,\ugrfguijoooocY,\ugrfguijoooocZ\right)\right)
\end{array}
\right]
 \nonumber \\
&=\left[
\begin{array}{l}
p^\mu=
\left(\ugrfguijoooopEdu,\left(\ugrfguijoooopXdu,\ugrfguijoooopYdu,\ugrfguijoooopZdu\right)\right),\\
q^\mu=
\left(\ugrfguijooooqEdu,\left(\ugrfguijooooqXdu,\ugrfguijooooqYdu,\ugrfguijooooqZdu\right)\right),\\
a^\mu=
\left(\ugrfguijooooaEdu,\left(\ugrfguijooooaXdu,\ugrfguijooooaYdu,\ugrfguijooooaZdu\right)\right),\\
b^\mu=
\left(\ugrfguijoooobEdu,\left(\ugrfguijoooobXdu,\ugrfguijoooobYdu,\ugrfguijoooobZdu\right)\right),\\
c^\mu=
\left(\ugrfguijoooocEdu,\left(\ugrfguijoooocXdu,\ugrfguijoooocYdu,\ugrfguijoooocZdu\right)\right)
\end{array}
\right]\qquad\text{(when $u_i\rightarrow 1$)}
}

\allowdisplaybreaks
\begin{align}

% Do not edit
% This LaTeX was generated by
% /Users/lester/LESTERHOME/proj/papers/2018LORENTZ/20200407a-full-S2S30-03.nb
% on Tue 30 Jun 2020 09:44:20
% It contains some momentum values that should zero V1 to V19 with the exception of V1
% A copy of the above workbook may also have been saved as notebook.nb in the current directory.
\gdef\ugrfguijoooopE{1}
\gdef\ugrfguijoooopX{0}
\gdef\ugrfguijoooopY{0}
\gdef\ugrfguijoooopZ{1}
\gdef\ugrfguijoooopEdu{1}
\gdef\ugrfguijoooopXdu{0}
\gdef\ugrfguijoooopYdu{0}
\gdef\ugrfguijoooopZdu{1}
\gdef\ugrfguijooooqE{1}
\gdef\ugrfguijooooqX{0}
\gdef\ugrfguijooooqY{0}
\gdef\ugrfguijooooqZ{-1}
\gdef\ugrfguijooooqEdu{1}
\gdef\ugrfguijooooqXdu{0}
\gdef\ugrfguijooooqYdu{0}
\gdef\ugrfguijooooqZdu{-1}
\gdef\ugrfguijooooaE{\frac{3 \sqrt{17}}{4}}
\gdef\ugrfguijooooaX{-\frac{1}{2}}
\gdef\ugrfguijooooaY{0}
\gdef\ugrfguijooooaZ{-\frac{\sqrt{5}}{4}}
\gdef\ugrfguijooooaEdu{\frac{3 \sqrt{17}}{4}}
\gdef\ugrfguijooooaXdu{-\frac{1}{2}}
\gdef\ugrfguijooooaYdu{0}
\gdef\ugrfguijooooaZdu{-\frac{\sqrt{5}}{4}}
\gdef\ugrfguijoooobE{\frac{1}{4} \sqrt{113+40 \sqrt{3}}}
\gdef\ugrfguijoooobX{\frac{3 \sqrt{3}}{8}}
\gdef\ugrfguijoooobY{\frac{3}{8}}
\gdef\ugrfguijoooobZ{0}
\gdef\ugrfguijoooobEdu{\frac{1}{4} \sqrt{113+40 \sqrt{3}}}
\gdef\ugrfguijoooobXdu{\frac{3 \sqrt{3}}{8}}
\gdef\ugrfguijoooobYdu{\frac{3}{8}}
\gdef\ugrfguijoooobZdu{0}
\gdef\ugrfguijoooocE{\frac{\sqrt{73}}{4}}
\gdef\ugrfguijoooocX{\frac{1}{2}}
\gdef\ugrfguijoooocY{0}
\gdef\ugrfguijoooocZ{\frac{\sqrt{5}}{4}}
\gdef\ugrfguijoooocEdu{\frac{\sqrt{73}}{4}}
\gdef\ugrfguijoooocXdu{\frac{1}{2}}
\gdef\ugrfguijoooocYdu{0}
\gdef\ugrfguijoooocZdu{\frac{\sqrt{5}}{4}}

\text{Event}[1]&=\left[
\begin{array}{l}
p^\mu=
\left(\ugrfguijoooopE,\left(\ugrfguijoooopX,\ugrfguijoooopY,\ugrfguijoooopZ\right)\right),\\
q^\mu=
\left(\ugrfguijooooqE,\left(\ugrfguijooooqX,\ugrfguijooooqY,\ugrfguijooooqZ\right)\right),\\
a^\mu=
\left(\ugrfguijooooaE,\left(\ugrfguijooooaX,\ugrfguijooooaY,\ugrfguijooooaZ\right)\right),\\
b^\mu=
\left(\ugrfguijoooobE,\left(\ugrfguijoooobX,\ugrfguijoooobY,\ugrfguijoooobZ\right)\right),\\
c^\mu=
\left(\ugrfguijoooocE,\left(\ugrfguijoooocX,\ugrfguijoooocY,\ugrfguijoooocZ\right)\right)
\end{array}
\right]
,\\

% Do not edit
% This LaTeX was generated by
% /Users/lester/LESTERHOME/proj/papers/2018LORENTZ/20200407a-full-S2S30-03.nb
% on Tue 30 Jun 2020 09:44:20
% It contains some momentum values that should zero V1 to V19 with the exception of V2
% A copy of the above workbook may also have been saved as notebook.nb in the current directory.
\gdef\ugrfguijoooopE{2}
\gdef\ugrfguijoooopX{0}
\gdef\ugrfguijoooopY{0}
\gdef\ugrfguijoooopZ{1}
\gdef\ugrfguijoooopEdu{2}
\gdef\ugrfguijoooopXdu{0}
\gdef\ugrfguijoooopYdu{0}
\gdef\ugrfguijoooopZdu{1}
\gdef\ugrfguijooooqE{2}
\gdef\ugrfguijooooqX{0}
\gdef\ugrfguijooooqY{0}
\gdef\ugrfguijooooqZ{-1}
\gdef\ugrfguijooooqEdu{2}
\gdef\ugrfguijooooqXdu{0}
\gdef\ugrfguijooooqYdu{0}
\gdef\ugrfguijooooqZdu{-1}
\gdef\ugrfguijooooaE{\frac{1}{4} \sqrt{\frac{3}{2} \left(9+\sqrt{3}\right)}}
\gdef\ugrfguijooooaX{-\frac{\sqrt{3}}{2}}
\gdef\ugrfguijooooaY{0}
\gdef\ugrfguijooooaZ{\frac{1}{4} \sqrt{\frac{3}{2} \left(1+\sqrt{3}\right)}}
\gdef\ugrfguijooooaEdu{\frac{1}{4} \sqrt{\frac{3}{2} \left(9+\sqrt{3}\right)}}
\gdef\ugrfguijooooaXdu{-\frac{\sqrt{3}}{2}}
\gdef\ugrfguijooooaYdu{0}
\gdef\ugrfguijooooaZdu{\frac{1}{4} \sqrt{\frac{3}{2} \left(1+\sqrt{3}\right)}}
\gdef\ugrfguijoooobE{\frac{1}{4} \sqrt{\frac{3}{2} \left(7+\sqrt{3}\right)}}
\gdef\ugrfguijoooobX{-\frac{3 \sqrt{3}}{8}}
\gdef\ugrfguijoooobY{-\frac{3}{8}}
\gdef\ugrfguijoooobZ{-\frac{1}{4} \sqrt{\frac{3}{2} \left(1+\sqrt{3}\right)}}
\gdef\ugrfguijoooobEdu{\frac{1}{4} \sqrt{\frac{3}{2} \left(7+\sqrt{3}\right)}}
\gdef\ugrfguijoooobXdu{-\frac{3 \sqrt{3}}{8}}
\gdef\ugrfguijoooobYdu{-\frac{3}{8}}
\gdef\ugrfguijoooobZdu{-\frac{1}{4} \sqrt{\frac{3}{2} \left(1+\sqrt{3}\right)}}
\gdef\ugrfguijoooocE{\frac{\sqrt{3}}{2}}
\gdef\ugrfguijoooocX{-\frac{\sqrt{3}}{4}}
\gdef\ugrfguijoooocY{-\frac{3}{4}}
\gdef\ugrfguijoooocZ{0}
\gdef\ugrfguijoooocEdu{\frac{\sqrt{3}}{2}}
\gdef\ugrfguijoooocXdu{-\frac{\sqrt{3}}{4}}
\gdef\ugrfguijoooocYdu{-\frac{3}{4}}
\gdef\ugrfguijoooocZdu{0}

\text{Event}[2]&=\left[
\begin{array}{l}
p^\mu=
\left(\ugrfguijoooopE,\left(\ugrfguijoooopX,\ugrfguijoooopY,\ugrfguijoooopZ\right)\right),\\
q^\mu=
\left(\ugrfguijooooqE,\left(\ugrfguijooooqX,\ugrfguijooooqY,\ugrfguijooooqZ\right)\right),\\
a^\mu=
\left(\ugrfguijooooaE,\left(\ugrfguijooooaX,\ugrfguijooooaY,\ugrfguijooooaZ\right)\right),\\
b^\mu=
\left(\ugrfguijoooobE,\left(\ugrfguijoooobX,\ugrfguijoooobY,\ugrfguijoooobZ\right)\right),\\
c^\mu=
\left(\ugrfguijoooocE,\left(\ugrfguijoooocX,\ugrfguijoooocY,\ugrfguijoooocZ\right)\right)
\end{array}
\right]
,\\

% Do not edit
% This LaTeX was generated by
% /Users/lester/LESTERHOME/proj/papers/2018LORENTZ/20200407a-full-S2S30-03.nb
% on Tue 30 Jun 2020 09:44:21
% It contains some momentum values that should zero V1 to V19 with the exception of V3
% A copy of the above workbook may also have been saved as notebook.nb in the current directory.
\gdef\ugrfguijoooopE{2}
\gdef\ugrfguijoooopX{0}
\gdef\ugrfguijoooopY{0}
\gdef\ugrfguijoooopZ{1}
\gdef\ugrfguijoooopEdu{2}
\gdef\ugrfguijoooopXdu{0}
\gdef\ugrfguijoooopYdu{0}
\gdef\ugrfguijoooopZdu{1}
\gdef\ugrfguijooooqE{3}
\gdef\ugrfguijooooqX{0}
\gdef\ugrfguijooooqY{0}
\gdef\ugrfguijooooqZ{-1}
\gdef\ugrfguijooooqEdu{3}
\gdef\ugrfguijooooqXdu{0}
\gdef\ugrfguijooooqYdu{0}
\gdef\ugrfguijooooqZdu{-1}
\gdef\ugrfguijooooaE{\sqrt{3}}
\gdef\ugrfguijooooaX{1}
\gdef\ugrfguijooooaY{0}
\gdef\ugrfguijooooaZ{1}
\gdef\ugrfguijooooaEdu{\sqrt{3}}
\gdef\ugrfguijooooaXdu{1}
\gdef\ugrfguijooooaYdu{0}
\gdef\ugrfguijooooaZdu{1}
\gdef\ugrfguijoooobE{\sqrt{3}}
\gdef\ugrfguijoooobX{1}
\gdef\ugrfguijoooobY{0}
\gdef\ugrfguijoooobZ{-1}
\gdef\ugrfguijoooobEdu{\sqrt{3}}
\gdef\ugrfguijoooobXdu{1}
\gdef\ugrfguijoooobYdu{0}
\gdef\ugrfguijoooobZdu{-1}
\gdef\ugrfguijoooocE{\sqrt{11}}
\gdef\ugrfguijoooocX{1}
\gdef\ugrfguijoooocY{1}
\gdef\ugrfguijoooocZ{0}
\gdef\ugrfguijoooocEdu{\sqrt{11}}
\gdef\ugrfguijoooocXdu{1}
\gdef\ugrfguijoooocYdu{1}
\gdef\ugrfguijoooocZdu{0}

\text{Event}[3]&=\left[
\begin{array}{l}
p^\mu=
\left(\ugrfguijoooopE,\left(\ugrfguijoooopX,\ugrfguijoooopY,\ugrfguijoooopZ\right)\right),\\
q^\mu=
\left(\ugrfguijooooqE,\left(\ugrfguijooooqX,\ugrfguijooooqY,\ugrfguijooooqZ\right)\right),\\
a^\mu=
\left(\ugrfguijooooaE,\left(\ugrfguijooooaX,\ugrfguijooooaY,\ugrfguijooooaZ\right)\right),\\
b^\mu=
\left(\ugrfguijoooobE,\left(\ugrfguijoooobX,\ugrfguijoooobY,\ugrfguijoooobZ\right)\right),\\
c^\mu=
\left(\ugrfguijoooocE,\left(\ugrfguijoooocX,\ugrfguijoooocY,\ugrfguijoooocZ\right)\right)
\end{array}
\right]
,\\

% Do not edit
% This LaTeX was generated by
% /Users/lester/LESTERHOME/proj/papers/2018LORENTZ/20200407a-full-S2S30-03.nb
% on Tue 30 Jun 2020 09:44:21
% It contains some momentum values that should zero V1 to V19 with the exception of V4
% A copy of the above workbook may also have been saved as notebook.nb in the current directory.
\gdef\ugrfguijoooopE{\frac{5}{4}}
\gdef\ugrfguijoooopX{0}
\gdef\ugrfguijoooopY{0}
\gdef\ugrfguijoooopZ{1}
\gdef\ugrfguijoooopEdu{\frac{5}{4}}
\gdef\ugrfguijoooopXdu{0}
\gdef\ugrfguijoooopYdu{0}
\gdef\ugrfguijoooopZdu{1}
\gdef\ugrfguijooooqE{\frac{5}{4}}
\gdef\ugrfguijooooqX{0}
\gdef\ugrfguijooooqY{0}
\gdef\ugrfguijooooqZ{-1}
\gdef\ugrfguijooooqEdu{\frac{5}{4}}
\gdef\ugrfguijooooqXdu{0}
\gdef\ugrfguijooooqYdu{0}
\gdef\ugrfguijooooqZdu{-1}
\gdef\ugrfguijooooaE{\sqrt{2}}
\gdef\ugrfguijooooaX{0}
\gdef\ugrfguijooooaY{1}
\gdef\ugrfguijooooaZ{0}
\gdef\ugrfguijooooaEdu{\sqrt{2}}
\gdef\ugrfguijooooaXdu{0}
\gdef\ugrfguijooooaYdu{1}
\gdef\ugrfguijooooaZdu{0}
\gdef\ugrfguijoooobE{\sqrt{2}}
\gdef\ugrfguijoooobX{1}
\gdef\ugrfguijoooobY{0}
\gdef\ugrfguijoooobZ{0}
\gdef\ugrfguijoooobEdu{\sqrt{2}}
\gdef\ugrfguijoooobXdu{1}
\gdef\ugrfguijoooobYdu{0}
\gdef\ugrfguijoooobZdu{0}
\gdef\ugrfguijoooocE{\sqrt{3}}
\gdef\ugrfguijoooocX{1}
\gdef\ugrfguijoooocY{0}
\gdef\ugrfguijoooocZ{1}
\gdef\ugrfguijoooocEdu{\sqrt{3}}
\gdef\ugrfguijoooocXdu{1}
\gdef\ugrfguijoooocYdu{0}
\gdef\ugrfguijoooocZdu{1}

\text{Event}[4]&=\left[
\begin{array}{l}
p^\mu=
\left(\ugrfguijoooopE,\left(\ugrfguijoooopX,\ugrfguijoooopY,\ugrfguijoooopZ\right)\right),\\
q^\mu=
\left(\ugrfguijooooqE,\left(\ugrfguijooooqX,\ugrfguijooooqY,\ugrfguijooooqZ\right)\right),\\
a^\mu=
\left(\ugrfguijooooaE,\left(\ugrfguijooooaX,\ugrfguijooooaY,\ugrfguijooooaZ\right)\right),\\
b^\mu=
\left(\ugrfguijoooobE,\left(\ugrfguijoooobX,\ugrfguijoooobY,\ugrfguijoooobZ\right)\right),\\
c^\mu=
\left(\ugrfguijoooocE,\left(\ugrfguijoooocX,\ugrfguijoooocY,\ugrfguijoooocZ\right)\right)
\end{array}
\right]
,\\

% Do not edit
% This LaTeX was generated by
% /Users/lester/LESTERHOME/proj/papers/2018LORENTZ/20200407a-full-S2S30-03.nb
% on Tue 30 Jun 2020 09:44:21
% It contains some momentum values that should zero V1 to V19 with the exception of V5
% A copy of the above workbook may also have been saved as notebook.nb in the current directory.
\gdef\ugrfguijoooopE{\frac{5}{4}}
\gdef\ugrfguijoooopX{0}
\gdef\ugrfguijoooopY{0}
\gdef\ugrfguijoooopZ{1}
\gdef\ugrfguijoooopEdu{\frac{5}{4}}
\gdef\ugrfguijoooopXdu{0}
\gdef\ugrfguijoooopYdu{0}
\gdef\ugrfguijoooopZdu{1}
\gdef\ugrfguijooooqE{\frac{5}{4}}
\gdef\ugrfguijooooqX{0}
\gdef\ugrfguijooooqY{0}
\gdef\ugrfguijooooqZ{-1}
\gdef\ugrfguijooooqEdu{\frac{5}{4}}
\gdef\ugrfguijooooqXdu{0}
\gdef\ugrfguijooooqYdu{0}
\gdef\ugrfguijooooqZdu{-1}
\gdef\ugrfguijooooaE{\frac{3}{2}}
\gdef\ugrfguijooooaX{1}
\gdef\ugrfguijooooaY{0}
\gdef\ugrfguijooooaZ{\frac{1}{2}}
\gdef\ugrfguijooooaEdu{\frac{3}{2}}
\gdef\ugrfguijooooaXdu{1}
\gdef\ugrfguijooooaYdu{0}
\gdef\ugrfguijooooaZdu{\frac{1}{2}}
\gdef\ugrfguijoooobE{\frac{3}{2}}
\gdef\ugrfguijoooobX{1}
\gdef\ugrfguijoooobY{0}
\gdef\ugrfguijoooobZ{\frac{1}{2}}
\gdef\ugrfguijoooobEdu{\frac{3}{2}}
\gdef\ugrfguijoooobXdu{1}
\gdef\ugrfguijoooobYdu{0}
\gdef\ugrfguijoooobZdu{\frac{1}{2}}
\gdef\ugrfguijoooocE{2}
\gdef\ugrfguijoooocX{1}
\gdef\ugrfguijoooocY{1}
\gdef\ugrfguijoooocZ{-1}
\gdef\ugrfguijoooocEdu{2}
\gdef\ugrfguijoooocXdu{1}
\gdef\ugrfguijoooocYdu{1}
\gdef\ugrfguijoooocZdu{-1}

\text{Event}[5]&=\left[
\begin{array}{l}
p^\mu=
\left(\ugrfguijoooopE,\left(\ugrfguijoooopX,\ugrfguijoooopY,\ugrfguijoooopZ\right)\right),\\
q^\mu=
\left(\ugrfguijooooqE,\left(\ugrfguijooooqX,\ugrfguijooooqY,\ugrfguijooooqZ\right)\right),\\
a^\mu=
\left(\ugrfguijooooaE,\left(\ugrfguijooooaX,\ugrfguijooooaY,\ugrfguijooooaZ\right)\right),\\
b^\mu=
\left(\ugrfguijoooobE,\left(\ugrfguijoooobX,\ugrfguijoooobY,\ugrfguijoooobZ\right)\right),\\
c^\mu=
\left(\ugrfguijoooocE,\left(\ugrfguijoooocX,\ugrfguijoooocY,\ugrfguijoooocZ\right)\right)
\end{array}
\right]
,\\

% Do not edit
% This LaTeX was generated by
% /Users/lester/LESTERHOME/proj/papers/2018LORENTZ/20200407a-full-S2S30-03.nb
% on Tue 30 Jun 2020 09:44:21
% It contains some momentum values that should zero V1 to V19 with the exception of V6
% A copy of the above workbook may also have been saved as notebook.nb in the current directory.
\gdef\ugrfguijoooopE{\frac{5}{4}}
\gdef\ugrfguijoooopX{0}
\gdef\ugrfguijoooopY{0}
\gdef\ugrfguijoooopZ{1}
\gdef\ugrfguijoooopEdu{\frac{5}{4}}
\gdef\ugrfguijoooopXdu{0}
\gdef\ugrfguijoooopYdu{0}
\gdef\ugrfguijoooopZdu{1}
\gdef\ugrfguijooooqE{\frac{5}{4}}
\gdef\ugrfguijooooqX{0}
\gdef\ugrfguijooooqY{0}
\gdef\ugrfguijooooqZ{-1}
\gdef\ugrfguijooooqEdu{\frac{5}{4}}
\gdef\ugrfguijooooqXdu{0}
\gdef\ugrfguijooooqYdu{0}
\gdef\ugrfguijooooqZdu{-1}
\gdef\ugrfguijooooaE{\sqrt{3}}
\gdef\ugrfguijooooaX{1}
\gdef\ugrfguijooooaY{0}
\gdef\ugrfguijooooaZ{1}
\gdef\ugrfguijooooaEdu{\sqrt{3}}
\gdef\ugrfguijooooaXdu{1}
\gdef\ugrfguijooooaYdu{0}
\gdef\ugrfguijooooaZdu{1}
\gdef\ugrfguijoooobE{\sqrt{5}}
\gdef\ugrfguijoooobX{1}
\gdef\ugrfguijoooobY{0}
\gdef\ugrfguijoooobZ{0}
\gdef\ugrfguijoooobEdu{\sqrt{5}}
\gdef\ugrfguijoooobXdu{1}
\gdef\ugrfguijoooobYdu{0}
\gdef\ugrfguijoooobZdu{0}
\gdef\ugrfguijoooocE{2}
\gdef\ugrfguijoooocX{1}
\gdef\ugrfguijoooocY{1}
\gdef\ugrfguijoooocZ{-1}
\gdef\ugrfguijoooocEdu{2}
\gdef\ugrfguijoooocXdu{1}
\gdef\ugrfguijoooocYdu{1}
\gdef\ugrfguijoooocZdu{-1}

\text{Event}[6]&=\left[
\begin{array}{l}
p^\mu=
\left(\ugrfguijoooopE,\left(\ugrfguijoooopX,\ugrfguijoooopY,\ugrfguijoooopZ\right)\right),\\
q^\mu=
\left(\ugrfguijooooqE,\left(\ugrfguijooooqX,\ugrfguijooooqY,\ugrfguijooooqZ\right)\right),\\
a^\mu=
\left(\ugrfguijooooaE,\left(\ugrfguijooooaX,\ugrfguijooooaY,\ugrfguijooooaZ\right)\right),\\
b^\mu=
\left(\ugrfguijoooobE,\left(\ugrfguijoooobX,\ugrfguijoooobY,\ugrfguijoooobZ\right)\right),\\
c^\mu=
\left(\ugrfguijoooocE,\left(\ugrfguijoooocX,\ugrfguijoooocY,\ugrfguijoooocZ\right)\right)
\end{array}
\right]
,\\

% Do not edit
% This LaTeX was generated by
% /Users/lester/LESTERHOME/proj/papers/2018LORENTZ/20200407a-full-S2S30-03.nb
% on Tue 30 Jun 2020 09:44:21
% It contains some momentum values that should zero V1 to V19 with the exception of V7
% A copy of the above workbook may also have been saved as notebook.nb in the current directory.
\gdef\ugrfguijoooopE{\frac{5}{4}}
\gdef\ugrfguijoooopX{0}
\gdef\ugrfguijoooopY{0}
\gdef\ugrfguijoooopZ{1}
\gdef\ugrfguijoooopEdu{\frac{5}{4}}
\gdef\ugrfguijoooopXdu{0}
\gdef\ugrfguijoooopYdu{0}
\gdef\ugrfguijoooopZdu{1}
\gdef\ugrfguijooooqE{\frac{5}{4}}
\gdef\ugrfguijooooqX{0}
\gdef\ugrfguijooooqY{0}
\gdef\ugrfguijooooqZ{-1}
\gdef\ugrfguijooooqEdu{\frac{5}{4}}
\gdef\ugrfguijooooqXdu{0}
\gdef\ugrfguijooooqYdu{0}
\gdef\ugrfguijooooqZdu{-1}
\gdef\ugrfguijooooaE{\sqrt{3}}
\gdef\ugrfguijooooaX{1}
\gdef\ugrfguijooooaY{0}
\gdef\ugrfguijooooaZ{1}
\gdef\ugrfguijooooaEdu{\sqrt{3}}
\gdef\ugrfguijooooaXdu{1}
\gdef\ugrfguijooooaYdu{0}
\gdef\ugrfguijooooaZdu{1}
\gdef\ugrfguijoooobE{\frac{\sqrt{21}}{2}}
\gdef\ugrfguijoooobX{\frac{1}{2}}
\gdef\ugrfguijoooobY{1}
\gdef\ugrfguijoooobZ{0}
\gdef\ugrfguijoooobEdu{\frac{\sqrt{21}}{2}}
\gdef\ugrfguijoooobXdu{\frac{1}{2}}
\gdef\ugrfguijoooobYdu{1}
\gdef\ugrfguijoooobZdu{0}
\gdef\ugrfguijoooocE{\sqrt{11}}
\gdef\ugrfguijoooocX{1}
\gdef\ugrfguijoooocY{0}
\gdef\ugrfguijoooocZ{-1}
\gdef\ugrfguijoooocEdu{\sqrt{11}}
\gdef\ugrfguijoooocXdu{1}
\gdef\ugrfguijoooocYdu{0}
\gdef\ugrfguijoooocZdu{-1}

\text{Event}[7]&=\left[
\begin{array}{l}
p^\mu=
\left(\ugrfguijoooopE,\left(\ugrfguijoooopX,\ugrfguijoooopY,\ugrfguijoooopZ\right)\right),\\
q^\mu=
\left(\ugrfguijooooqE,\left(\ugrfguijooooqX,\ugrfguijooooqY,\ugrfguijooooqZ\right)\right),\\
a^\mu=
\left(\ugrfguijooooaE,\left(\ugrfguijooooaX,\ugrfguijooooaY,\ugrfguijooooaZ\right)\right),\\
b^\mu=
\left(\ugrfguijoooobE,\left(\ugrfguijoooobX,\ugrfguijoooobY,\ugrfguijoooobZ\right)\right),\\
c^\mu=
\left(\ugrfguijoooocE,\left(\ugrfguijoooocX,\ugrfguijoooocY,\ugrfguijoooocZ\right)\right)
\end{array}
\right]
,\\
\pullInExampleWDU{8},\\
\pullInExampleWDU{9},\\
\pullInExampleWDU{10},\\

% Do not edit
% This LaTeX was generated by
% /Users/lester/LESTERHOME/proj/papers/2018LORENTZ/20200407a-full-S2S30-03.nb
% on Tue 30 Jun 2020 09:44:22
% It contains some momentum values that should zero V1 to V19 with the exception of V11
% A copy of the above workbook may also have been saved as notebook.nb in the current directory.
\gdef\ugrfguijoooopE{\frac{5}{4}}
\gdef\ugrfguijoooopX{0}
\gdef\ugrfguijoooopY{0}
\gdef\ugrfguijoooopZ{1}
\gdef\ugrfguijoooopEdu{\frac{5}{4}}
\gdef\ugrfguijoooopXdu{0}
\gdef\ugrfguijoooopYdu{0}
\gdef\ugrfguijoooopZdu{1}
\gdef\ugrfguijooooqE{2}
\gdef\ugrfguijooooqX{0}
\gdef\ugrfguijooooqY{0}
\gdef\ugrfguijooooqZ{-1}
\gdef\ugrfguijooooqEdu{2}
\gdef\ugrfguijooooqXdu{0}
\gdef\ugrfguijooooqYdu{0}
\gdef\ugrfguijooooqZdu{-1}
\gdef\ugrfguijooooaE{\sqrt{3}}
\gdef\ugrfguijooooaX{1}
\gdef\ugrfguijooooaY{1}
\gdef\ugrfguijooooaZ{0}
\gdef\ugrfguijooooaEdu{\sqrt{3}}
\gdef\ugrfguijooooaXdu{1}
\gdef\ugrfguijooooaYdu{1}
\gdef\ugrfguijooooaZdu{0}
\gdef\ugrfguijoooobE{\sqrt{3}}
\gdef\ugrfguijoooobX{\sqrt{2}}
\gdef\ugrfguijoooobY{0}
\gdef\ugrfguijoooobZ{0}
\gdef\ugrfguijoooobEdu{\sqrt{3}}
\gdef\ugrfguijoooobXdu{\sqrt{2}}
\gdef\ugrfguijoooobYdu{0}
\gdef\ugrfguijoooobZdu{0}
\gdef\ugrfguijoooocE{\sqrt{3}}
\gdef\ugrfguijoooocX{-1}
\gdef\ugrfguijoooocY{-1}
\gdef\ugrfguijoooocZ{0}
\gdef\ugrfguijoooocEdu{\sqrt{3}}
\gdef\ugrfguijoooocXdu{-1}
\gdef\ugrfguijoooocYdu{-1}
\gdef\ugrfguijoooocZdu{0}

\text{Event}[11]&=\left[
\begin{array}{l}
p^\mu=
\left(\ugrfguijoooopE,\left(\ugrfguijoooopX,\ugrfguijoooopY,\ugrfguijoooopZ\right)\right),\\
q^\mu=
\left(\ugrfguijooooqE,\left(\ugrfguijooooqX,\ugrfguijooooqY,\ugrfguijooooqZ\right)\right),\\
a^\mu=
\left(\ugrfguijooooaE,\left(\ugrfguijooooaX,\ugrfguijooooaY,\ugrfguijooooaZ\right)\right),\\
b^\mu=
\left(\ugrfguijoooobE,\left(\ugrfguijoooobX,\ugrfguijoooobY,\ugrfguijoooobZ\right)\right),\\
c^\mu=
\left(\ugrfguijoooocE,\left(\ugrfguijoooocX,\ugrfguijoooocY,\ugrfguijoooocZ\right)\right)
\end{array}
\right]
,\\
\pullInExampleWDU{12},\\
\pullInExampleWDU{13},\\

% Do not edit
% This LaTeX was generated by
% /Users/lester/LESTERHOME/proj/papers/2018LORENTZ/20200407a-full-S2S30-03.nb
% on Tue 30 Jun 2020 09:44:22
% It contains some momentum values that should zero V1 to V19 with the exception of V14
% A copy of the above workbook may also have been saved as notebook.nb in the current directory.
\gdef\ugrfguijoooopE{2}
\gdef\ugrfguijoooopX{0}
\gdef\ugrfguijoooopY{0}
\gdef\ugrfguijoooopZ{1}
\gdef\ugrfguijoooopEdu{2}
\gdef\ugrfguijoooopXdu{0}
\gdef\ugrfguijoooopYdu{0}
\gdef\ugrfguijoooopZdu{1}
\gdef\ugrfguijooooqE{1}
\gdef\ugrfguijooooqX{0}
\gdef\ugrfguijooooqY{0}
\gdef\ugrfguijooooqZ{-1}
\gdef\ugrfguijooooqEdu{1}
\gdef\ugrfguijooooqXdu{0}
\gdef\ugrfguijooooqYdu{0}
\gdef\ugrfguijooooqZdu{-1}
\gdef\ugrfguijooooaE{\sqrt{2}}
\gdef\ugrfguijooooaX{1}
\gdef\ugrfguijooooaY{0}
\gdef\ugrfguijooooaZ{0}
\gdef\ugrfguijooooaEdu{\sqrt{2}}
\gdef\ugrfguijooooaXdu{1}
\gdef\ugrfguijooooaYdu{0}
\gdef\ugrfguijooooaZdu{0}
\gdef\ugrfguijoooobE{\sqrt{2}}
\gdef\ugrfguijoooobX{0}
\gdef\ugrfguijoooobY{1}
\gdef\ugrfguijoooobZ{0}
\gdef\ugrfguijoooobEdu{\sqrt{2}}
\gdef\ugrfguijoooobXdu{0}
\gdef\ugrfguijoooobYdu{1}
\gdef\ugrfguijoooobZdu{0}
\gdef\ugrfguijoooocE{\frac{\sqrt{13}}{2}}
\gdef\ugrfguijoooocX{-1}
\gdef\ugrfguijoooocY{0}
\gdef\ugrfguijoooocZ{0}
\gdef\ugrfguijoooocEdu{\frac{\sqrt{13}}{2}}
\gdef\ugrfguijoooocXdu{-1}
\gdef\ugrfguijoooocYdu{0}
\gdef\ugrfguijoooocZdu{0}

\text{Event}[14]&=\left[
\begin{array}{l}
p^\mu=
\left(\ugrfguijoooopE,\left(\ugrfguijoooopX,\ugrfguijoooopY,\ugrfguijoooopZ\right)\right),\\
q^\mu=
\left(\ugrfguijooooqE,\left(\ugrfguijooooqX,\ugrfguijooooqY,\ugrfguijooooqZ\right)\right),\\
a^\mu=
\left(\ugrfguijooooaE,\left(\ugrfguijooooaX,\ugrfguijooooaY,\ugrfguijooooaZ\right)\right),\\
b^\mu=
\left(\ugrfguijoooobE,\left(\ugrfguijoooobX,\ugrfguijoooobY,\ugrfguijoooobZ\right)\right),\\
c^\mu=
\left(\ugrfguijoooocE,\left(\ugrfguijoooocX,\ugrfguijoooocY,\ugrfguijoooocZ\right)\right)
\end{array}
\right]
,\\

% Do not edit
% This LaTeX was generated by
% /Users/lester/LESTERHOME/proj/papers/2018LORENTZ/20200407a-full-S2S30-03.nb
% on Tue 30 Jun 2020 09:44:22
% It contains some momentum values that should zero V1 to V19 with the exception of V15
% A copy of the above workbook may also have been saved as notebook.nb in the current directory.
\gdef\ugrfguijoooopE{\frac{5}{4}}
\gdef\ugrfguijoooopX{0}
\gdef\ugrfguijoooopY{0}
\gdef\ugrfguijoooopZ{1}
\gdef\ugrfguijoooopEdu{\frac{5}{4}}
\gdef\ugrfguijoooopXdu{0}
\gdef\ugrfguijoooopYdu{0}
\gdef\ugrfguijoooopZdu{1}
\gdef\ugrfguijooooqE{2}
\gdef\ugrfguijooooqX{0}
\gdef\ugrfguijooooqY{0}
\gdef\ugrfguijooooqZ{-1}
\gdef\ugrfguijooooqEdu{2}
\gdef\ugrfguijooooqXdu{0}
\gdef\ugrfguijooooqYdu{0}
\gdef\ugrfguijooooqZdu{-1}
\gdef\ugrfguijooooaE{\sqrt{3}}
\gdef\ugrfguijooooaX{1}
\gdef\ugrfguijooooaY{1}
\gdef\ugrfguijooooaZ{0}
\gdef\ugrfguijooooaEdu{\sqrt{3}}
\gdef\ugrfguijooooaXdu{1}
\gdef\ugrfguijooooaYdu{1}
\gdef\ugrfguijooooaZdu{0}
\gdef\ugrfguijoooobE{\sqrt{2}}
\gdef\ugrfguijoooobX{1}
\gdef\ugrfguijoooobY{0}
\gdef\ugrfguijoooobZ{0}
\gdef\ugrfguijoooobEdu{\sqrt{2}}
\gdef\ugrfguijoooobXdu{1}
\gdef\ugrfguijoooobYdu{0}
\gdef\ugrfguijoooobZdu{0}
\gdef\ugrfguijoooocE{\sqrt{11}}
\gdef\ugrfguijoooocX{1}
\gdef\ugrfguijoooocY{-1}
\gdef\ugrfguijoooocZ{0}
\gdef\ugrfguijoooocEdu{\sqrt{11}}
\gdef\ugrfguijoooocXdu{1}
\gdef\ugrfguijoooocYdu{-1}
\gdef\ugrfguijoooocZdu{0}

\text{Event}[15]&=\left[
\begin{array}{l}
p^\mu=
\left(\ugrfguijoooopE,\left(\ugrfguijoooopX,\ugrfguijoooopY,\ugrfguijoooopZ\right)\right),\\
q^\mu=
\left(\ugrfguijooooqE,\left(\ugrfguijooooqX,\ugrfguijooooqY,\ugrfguijooooqZ\right)\right),\\
a^\mu=
\left(\ugrfguijooooaE,\left(\ugrfguijooooaX,\ugrfguijooooaY,\ugrfguijooooaZ\right)\right),\\
b^\mu=
\left(\ugrfguijoooobE,\left(\ugrfguijoooobX,\ugrfguijoooobY,\ugrfguijoooobZ\right)\right),\\
c^\mu=
\left(\ugrfguijoooocE,\left(\ugrfguijoooocX,\ugrfguijoooocY,\ugrfguijoooocZ\right)\right)
\end{array}
\right]
,\\

% Do not edit
% This LaTeX was generated by
% /Users/lester/LESTERHOME/proj/papers/2018LORENTZ/20200407a-full-S2S30-03.nb
% on Tue 30 Jun 2020 09:44:22
% It contains some momentum values that should zero V1 to V19 with the exception of V16
% A copy of the above workbook may also have been saved as notebook.nb in the current directory.
\gdef\ugrfguijoooopE{\frac{3}{2}}
\gdef\ugrfguijoooopX{0}
\gdef\ugrfguijoooopY{0}
\gdef\ugrfguijoooopZ{1}
\gdef\ugrfguijoooopEdu{\frac{3}{2}}
\gdef\ugrfguijoooopXdu{0}
\gdef\ugrfguijoooopYdu{0}
\gdef\ugrfguijoooopZdu{1}
\gdef\ugrfguijooooqE{\frac{3}{2}}
\gdef\ugrfguijooooqX{0}
\gdef\ugrfguijooooqY{0}
\gdef\ugrfguijooooqZ{-1}
\gdef\ugrfguijooooqEdu{\frac{3}{2}}
\gdef\ugrfguijooooqXdu{0}
\gdef\ugrfguijooooqYdu{0}
\gdef\ugrfguijooooqZdu{-1}
\gdef\ugrfguijooooaE{\sqrt{\frac{29}{6}}}
\gdef\ugrfguijooooaX{1}
\gdef\ugrfguijooooaY{0}
\gdef\ugrfguijooooaZ{\frac{2}{\sqrt{3}}}
\gdef\ugrfguijooooaEdu{\sqrt{\frac{29}{6}}}
\gdef\ugrfguijooooaXdu{1}
\gdef\ugrfguijooooaYdu{0}
\gdef\ugrfguijooooaZdu{\frac{2}{\sqrt{3}}}
\gdef\ugrfguijoooobE{\sqrt{\frac{19}{3}}}
\gdef\ugrfguijoooobX{1}
\gdef\ugrfguijoooobY{1}
\gdef\ugrfguijoooobZ{-\frac{1}{\sqrt{3}}}
\gdef\ugrfguijoooobEdu{\sqrt{\frac{19}{3}}}
\gdef\ugrfguijoooobXdu{1}
\gdef\ugrfguijoooobYdu{1}
\gdef\ugrfguijoooobZdu{-\frac{1}{\sqrt{3}}}
\gdef\ugrfguijoooocE{\sqrt{\frac{10}{3}}}
\gdef\ugrfguijoooocX{1}
\gdef\ugrfguijoooocY{-1}
\gdef\ugrfguijoooocZ{-\frac{1}{\sqrt{3}}}
\gdef\ugrfguijoooocEdu{\sqrt{\frac{10}{3}}}
\gdef\ugrfguijoooocXdu{1}
\gdef\ugrfguijoooocYdu{-1}
\gdef\ugrfguijoooocZdu{-\frac{1}{\sqrt{3}}}

\text{Event}[16]&=\left[
\begin{array}{l}
p^\mu=
\left(\ugrfguijoooopE,\left(\ugrfguijoooopX,\ugrfguijoooopY,\ugrfguijoooopZ\right)\right),\\
q^\mu=
\left(\ugrfguijooooqE,\left(\ugrfguijooooqX,\ugrfguijooooqY,\ugrfguijooooqZ\right)\right),\\
a^\mu=
\left(\ugrfguijooooaE,\left(\ugrfguijooooaX,\ugrfguijooooaY,\ugrfguijooooaZ\right)\right),\\
b^\mu=
\left(\ugrfguijoooobE,\left(\ugrfguijoooobX,\ugrfguijoooobY,\ugrfguijoooobZ\right)\right),\\
c^\mu=
\left(\ugrfguijoooocE,\left(\ugrfguijoooocX,\ugrfguijoooocY,\ugrfguijoooocZ\right)\right)
\end{array}
\right]
,\\
\pullInExampleWDU{17},\\

% Do not edit
% This LaTeX was generated by
% /Users/lester/LESTERHOME/proj/papers/2018LORENTZ/20200407a-full-S2S30-03.nb
% on Tue 30 Jun 2020 09:44:23
% It contains some momentum values that should zero V1 to V19 with the exception of V18
% A copy of the above workbook may also have been saved as notebook.nb in the current directory.
\gdef\ugrfguijoooopE{\frac{3}{2}}
\gdef\ugrfguijoooopX{0}
\gdef\ugrfguijoooopY{0}
\gdef\ugrfguijoooopZ{1}
\gdef\ugrfguijoooopEdu{\frac{3}{2}}
\gdef\ugrfguijoooopXdu{0}
\gdef\ugrfguijoooopYdu{0}
\gdef\ugrfguijoooopZdu{1}
\gdef\ugrfguijooooqE{\frac{3}{2}}
\gdef\ugrfguijooooqX{0}
\gdef\ugrfguijooooqY{0}
\gdef\ugrfguijooooqZ{-1}
\gdef\ugrfguijooooqEdu{\frac{3}{2}}
\gdef\ugrfguijooooqXdu{0}
\gdef\ugrfguijooooqYdu{0}
\gdef\ugrfguijooooqZdu{-1}
\gdef\ugrfguijooooaE{\sqrt{5}}
\gdef\ugrfguijooooaX{1}
\gdef\ugrfguijooooaY{\sqrt{3}}
\gdef\ugrfguijooooaZ{1}
\gdef\ugrfguijooooaEdu{\sqrt{5}}
\gdef\ugrfguijooooaXdu{1}
\gdef\ugrfguijooooaYdu{\sqrt{3}}
\gdef\ugrfguijooooaZdu{1}
\gdef\ugrfguijoooobE{\sqrt{6}}
\gdef\ugrfguijoooobX{1}
\gdef\ugrfguijoooobY{-\sqrt{3}}
\gdef\ugrfguijoooobZ{1}
\gdef\ugrfguijoooobEdu{\sqrt{6}}
\gdef\ugrfguijoooobXdu{1}
\gdef\ugrfguijoooobYdu{-\sqrt{3}}
\gdef\ugrfguijoooobZdu{1}
\gdef\ugrfguijoooocE{\sqrt{6}}
\gdef\ugrfguijoooocX{1}
\gdef\ugrfguijoooocY{0}
\gdef\ugrfguijoooocZ{-2}
\gdef\ugrfguijoooocEdu{\sqrt{6}}
\gdef\ugrfguijoooocXdu{1}
\gdef\ugrfguijoooocYdu{0}
\gdef\ugrfguijoooocZdu{-2}

\text{Event}[18]&=\left[
\begin{array}{l}
p^\mu=
\left(\ugrfguijoooopE,\left(\ugrfguijoooopX,\ugrfguijoooopY,\ugrfguijoooopZ\right)\right),\\
q^\mu=
\left(\ugrfguijooooqE,\left(\ugrfguijooooqX,\ugrfguijooooqY,\ugrfguijooooqZ\right)\right),\\
a^\mu=
\left(\ugrfguijooooaE,\left(\ugrfguijooooaX,\ugrfguijooooaY,\ugrfguijooooaZ\right)\right),\\
b^\mu=
\left(\ugrfguijoooobE,\left(\ugrfguijoooobX,\ugrfguijoooobY,\ugrfguijoooobZ\right)\right),\\
c^\mu=
\left(\ugrfguijoooocE,\left(\ugrfguijoooocX,\ugrfguijoooocY,\ugrfguijoooocZ\right)\right)
\end{array}
\right]
,\\

% Do not edit
% This LaTeX was generated by
% /Users/lester/LESTERHOME/proj/papers/2018LORENTZ/20200407a-full-S2S30-03.nb
% on Tue 30 Jun 2020 09:44:23
% It contains some momentum values that should zero V1 to V19 with the exception of V19
% A copy of the above workbook may also have been saved as notebook.nb in the current directory.
\gdef\ugrfguijoooopE{\frac{4}{3}}
\gdef\ugrfguijoooopX{0}
\gdef\ugrfguijoooopY{0}
\gdef\ugrfguijoooopZ{1}
\gdef\ugrfguijoooopEdu{\frac{4}{3}}
\gdef\ugrfguijoooopXdu{0}
\gdef\ugrfguijoooopYdu{0}
\gdef\ugrfguijoooopZdu{1}
\gdef\ugrfguijooooqE{\frac{4}{3}}
\gdef\ugrfguijooooqX{0}
\gdef\ugrfguijooooqY{0}
\gdef\ugrfguijooooqZ{-1}
\gdef\ugrfguijooooqEdu{\frac{4}{3}}
\gdef\ugrfguijooooqXdu{0}
\gdef\ugrfguijooooqYdu{0}
\gdef\ugrfguijooooqZdu{-1}
\gdef\ugrfguijooooaE{\sqrt{3}}
\gdef\ugrfguijooooaX{\sqrt{2}}
\gdef\ugrfguijooooaY{0}
\gdef\ugrfguijooooaZ{0}
\gdef\ugrfguijooooaEdu{\sqrt{3}}
\gdef\ugrfguijooooaXdu{\sqrt{2}}
\gdef\ugrfguijooooaYdu{0}
\gdef\ugrfguijooooaZdu{0}
\gdef\ugrfguijoooobE{\sqrt{3}}
\gdef\ugrfguijoooobX{0}
\gdef\ugrfguijoooobY{1}
\gdef\ugrfguijoooobZ{1}
\gdef\ugrfguijoooobEdu{\sqrt{3}}
\gdef\ugrfguijoooobXdu{0}
\gdef\ugrfguijoooobYdu{1}
\gdef\ugrfguijoooobZdu{1}
\gdef\ugrfguijoooocE{\frac{\sqrt{17}}{2}}
\gdef\ugrfguijoooocX{0}
\gdef\ugrfguijoooocY{-1}
\gdef\ugrfguijoooocZ{1}
\gdef\ugrfguijoooocEdu{\frac{\sqrt{17}}{2}}
\gdef\ugrfguijoooocXdu{0}
\gdef\ugrfguijoooocYdu{-1}
\gdef\ugrfguijoooocZdu{1}

\text{Event}[19]&=\left[
\begin{array}{l}
p^\mu=
\left(\ugrfguijoooopE,\left(\ugrfguijoooopX,\ugrfguijoooopY,\ugrfguijoooopZ\right)\right),\\
q^\mu=
\left(\ugrfguijooooqE,\left(\ugrfguijooooqX,\ugrfguijooooqY,\ugrfguijooooqZ\right)\right),\\
a^\mu=
\left(\ugrfguijooooaE,\left(\ugrfguijooooaX,\ugrfguijooooaY,\ugrfguijooooaZ\right)\right),\\
b^\mu=
\left(\ugrfguijoooobE,\left(\ugrfguijoooobX,\ugrfguijoooobY,\ugrfguijoooobZ\right)\right),\\
c^\mu=
\left(\ugrfguijoooocE,\left(\ugrfguijoooocX,\ugrfguijoooocY,\ugrfguijoooocZ\right)\right)
\end{array}
\right]
.
\end{align}

\end{proof}
\end{lemma}

\begin{remark}
	Note that since we have shown \textbf{necessity} only in the case that the  $u_i$ values are chosen in appropriate ranges (those which keep the event momenta of these particular example events real) then it remains possible that for certain cleverly chosen values of the $u_i$ some of the variables $V_k$ might no longer perform tasks not covered by other variables, and so the set may then be \textbf{reducible}.   Whether this possibility can be used to shrink $S_{19}$ to a smaller and more \textbf{minimal} set of event variables has not been fully investigated. Nonetheless, preliminary tests suggest that there is nothing to be gained by such a procedure.  For example: the only $u_i$ for which the $u_i=1$ step of the proof was needed were these:
\begin{centering}
\begin{tabular}{|l|l|}
\hline
$u_1$ and $u_3$ & for  $V_8$ and $V_9$ \\
\hline
$u_4$  & for  $V_{10}$ \\
\hline
$u_7$ and $u_8$ & for  $V_{12}$ and $V_{13}$ \\
\hline
$u_{12}$ and $u_{13}$ & for  $V_{17}$\\
\hline
\end{tabular}.
\end{centering}
	A quick investigation focusing only on the \textbf{necessity} of $V_9$ has shown that for any $(u_1,u_3)\ne0$ it is possible to construct an event for which $V_9\ne0$ and every other $V_k=0$.\footnote{No proof of the statement just made is presented in this document.  However a demonstration of how one can construct events which make $V_9$ \textbf{necessary} for any $(u_1,u_3)$ may be found in the python reference implementation \texttt{reference-implementation.py} supplied with this paper, and a mathematica notebook \texttt{parity-variables.nb} which accompanies it provides some background to explain how these reference events were constructed. For more information see Section~\ref{sec:supportmaterials}.}   The process for constructing such events is tiresome and results in no tangible benefit (beyond confirming that a given $V_k$ is universally \textbf{necessary}) and so the same task has not been re-attempted for the remaining variables $V_8$, $V_{10}$. $V_{12}$, $V_{13}$ and $V_{17}$. However, the authors do not think it would be impossible to construct such events if the need were strong enough. We therefore make the claim (without proof) that for any $\{u_k\}$ all the set $S_{19}=\{V_1,\ldots,V_{19}\}$ is \textbf{irreducible} on \textbf{collision events} $e\in\pqabcColEvent$.
\followUpInFuture[One should get a student to find events which prove that $V_1$ to $V_{19}$ really are necessary for any $u_i$.]
\end{remark}

\subsubsection{Simplifications to the set $S_{19}$ valid for LHC three photon events}

We note, without proof, that of the nineteen variables in $S_{19}$, all but nine of them are non-zero when $m_p=m_q$ (as it is at the LHC) and $m_a=m_b=m_c=0$ (as it is if none of the reconstructed particles carry mass information).

An LHC collaboration looking at massless three-jet or three-photon events could therefore use just the nine variables in:
\begin{align}
S_9=\{ 
V_1,
V_2,
V_4,
V_5,
V_6,
V_8,
V_9,
V_{14},
V_{15}
\}.\label{eq:shortersnineformasslessLHC3j}
\end{align}

\subsubsection{Discussion concerning the \textbf{Minimality} of the set $S_{19}$}
\label{sec:discOfMinimialityOfS19}
\begin{remark}[\textbf{Minimality}]
In Corollary~\ref{cor:finallywedefines19properl} and Lemma~\ref{lem:idependencelemfors19}
we saw that the variables of $S_{19}$ are \textbf{necessary} and \textbf{sufficient} for the purposes of ascribing parities to \textbf{chiral} \textbf{collision events} in $\pqabcColEvent$. But are they \textbf{minimal}?  Is there a smaller set of variables that can accomplish the same goals?

The simple answer to the above question is that we do not know. However, despite this, there are a number of pieces of circumstantial evidence which suggest that they may not comprise a \textbf{minimal} set:
\begin{itemize}
\item
	In the first remark following the proof of Lemma~\ref{lem:idependencelemfors19} it has already been noted that the possibility has not been excluded that there exist values of $u_i$ not all equal to 1 (in appropriate units) such that some of the variables in $S_{19}$ might be un-\textbf{necessary} and the set itself \textbf{reducible}.
\item
The fact that a \textbf{sufficient} set of parity-odd event variables might ever need to contain more than one variable could be said to stem from the fact that any given parity-odd variable can have roots (i.e.~can evaluate to zero) on events which are \textbf{chiral} as well as \textbf{non-chiral}.\footnote{Parity-odd variables necessarily evaluate to zero on \textbf{non-chiral} events, but this does not preclude them to having zeros on other events.}  The greater the order of a polynomial the larger is the number of roots it contains, and correspondingly the greater is the chance that some of these additional unrequested roots will occur on \textbf{chiral} \textbf{events} in $\pqabcColEvent$ rather than in unphysical places. All of the variables contained in $S_{19}$ are polynomial functions of momentum components, and many of them have high orders, and so have enlarged chances of having undesirable `chiral roots' -- which in turn necessitate the presence of yet more more parity-odd variables to  assign parities to those roots.  At no point in our argument have we demonstrated that the constructions of our variables represent polynomials with the least possible order -- or that they have related beneficial properties.  There is therefore no reason to believe that the solution found need be \textbf{minimal}.
\item
A  specific example of the argument just made was already alluded to in the remark following the proof of Lemma~\ref{lem:s19THR} where it was noted that it had not been possible to demonstrate \textbf{\independence} for a variant $S_{19}$ in which the (lower order) quantity $[a,b,c,(p+q)]$ had been used in place of the (much higher order) ingredient $
\PThr
\cdot
\FThr$
 in the parity-odd variables  covering events in $C_3$. This suggests that variables $V_6$ and $V_7$ are more complicated than they need to be, or (equivalently) that there is a better partition of $C$ into than the $C_1\lor C_2 \lor C_3 \lor C_4$ split used here.  The reasons that $C_3$ does not fit well with $[a,b,c,(p+q)]$ is that Corollary~\ref{cor:notcurrentlyusedbutcouldbeuseful} shows only that $C_3 \implies ([a,b,c,(p+q)]\ne0)$. It does not show that $([a,b,c,(p+q)]\ne0)\implies C_3$.  In fact, $[a,b,c,(p+q)]$ can even be non-zero on \textbf{non-collision} events that are not even in $C$, which is a fact used by all of the variables in \eqref{eq:snoncolextranoncollvarset}!
\item 
    The variables in $S_{19}$ cover a broader class of events  than was originally required, since no use was made of some of the constraints in $C_{4B}$ when covering $C_4$.  It is possible that exploitation of the unused constraints in $C_{4B}$. could allow a smaller set of variables to be created.
\item
   The complex `tricks' used in Definition~\ref{def:oursixcomplexfs} and discussed in the remarks following it and in Lemma~\ref{lem:lemaboutproductsofthreecomplexnumbers}, while helpful, seem rather arbitrary.    Were there a four-dimensional commutative division algebra (e.g.~were the quaternions, $\mathbb H$,  commutative in multiplication) then we could have used an even more powerful trick -- and would have ended up with a smaller set than $S_{19}$. Alas the quaternions are not commutative in multiplication, and we know of no four dimensional alternatives which are.\footnote{Frobenius's Theorem relating to real division algebras presumably prevents any progress in this direction.}
Nonetheless, the reduction provided by this `trick' feels wholly accidental and unrelated to the nature of the problem being solved, so it seems unlikely it is the best `trick' that could have been used.
\item
	Finally (and this is a very subjective point!) the `gauge freedom' presented by the $u_i$ seems ugly, undesirable and unnecessarily complicated.  The \textbf{\independence} and \textbf{sufficiency} properties of the \textit{set} of all variables is unaffected by the $u_i$ values, and yet the individual parities are strongly influenced by the $u_i$. This does not seem to be a welcome state of affairs.
\end{itemize}
In summary: it would not be surprising to discover that a set of event variables  could be constructed which was smaller than $S_{19}$ and yet shared all its other desired properties.
\end{remark}

\subsection{\textbf{Non-collision events}
$e\in\pqabcAllEvent$}
\label{sec:bringbothnoncolltypestogeterinpqabc}

\begin{definition}
\label{def:noncollmassivefunctions}
The following three functions
\begin{alignat}{9}
v[1,\mathbb{C},\text{non-coll}, m_p+m_q>0] 
&=
[a,b,c,p+q]
\cdot {}
&&U^{abc}_{\phantom{pq}}[
\delta f_1^{xy} + i u^{ncm}_1 \delta f_2^{xy}
,\ 
&&\delta f_1^{xy} + i u^{ncm}_1 \delta f_2^{xy}
&&]
,
\\
v[2,\mathbb{C},\text{non-coll},m_p+m_q>0] 
&=
[a,b,c,p+q]
\cdot  {}
&&U^{abc}_{\phantom{pq}}[
\delta f_1^{xy} + i u^{ncm}_2 \delta f_2^{xy}
,
&&\delta f_3^{xy}
&&],
\\
v[1,\mathbb{R},\text{non-coll},m_p+m_q>0] 
&=
[a,b,c,p+q]
\cdot {}
&&U^{abc}_{\phantom{pq}}[
\delta f_3^{xy}
,
&&\delta f_3^{xy} && ]
\end{alignat}
are defined in terms of two real constants:  $u^{ncm}_1\ne0$ and $u^{ncm}_2\ne0$.
\end{definition}

\begin{definition}
\label{def:noncollmasslessfunctions}
The following three functions
\begin{alignat}{9}
v[1,\mathbb{C},\text{non-coll}, m_p=m_q=0] 
&=
[a,b,c,p+q]
\cdot {}
&&U^{abc}_{\phantom{pq}}[
\delta f_5^{xy} + i u^{ncn}_1 \delta f_6^{xy}
,\ 
&&\delta f_5^{xy} + i u^{ncn}_1 \delta f_6^{xy}
&&]
,
\\
v[2,\mathbb{C},\text{non-coll},m_p=m_q=0] 
&=
[a,b,c,p+q]
\cdot  {}
&&U^{abc}_{\phantom{pq}}[
\delta f_5^{xy} + i u^{ncn}_2 \delta f_6^{xy}
,
&&\delta f_3^{xy}
&&],
\\
v[1,\mathbb{R},\text{non-coll},m_p=m_q=0] 
&=
[a,b,c,p+q]
\cdot {}
&&U^{abc}_{\phantom{pq}}[
\delta f_3^{xy}
,
&&\delta f_3^{xy} && ]
\end{alignat}
are defined in terms of two real constants:  $u^{ncn}_1\ne0$ and $u^{ncn}_2\ne0$ and the functions $\delta f_5^{xy}$ and $\delta f_6^{xy}$ defined just after  \eqref{eq:f5shorthanddef} and \eqref{eq:f6shorthanddef}.
\end{definition}

\begin{definition}
Since $v[1,\mathbb{R},\text{non-coll},m_p+m_q>0]$ and $v[1,\mathbb{R},\text{non-coll},m_p=m_q=0]$ have identical function definitions, we may remove the parts of their name listing constraints on $m_p$ and $m_q$ and instead refer both of them as 
\begin{align}
v[1,\mathbb{R},\text{non-coll}]
&\equivdef
v[1,\mathbb{R},\text{non-coll},m_p+m_q\ge 0]
\\
&\equiv
v[1,\mathbb{R},\text{non-coll},m_p=m_q=0]\nonumber
.
\end{align}
\end{definition}

\begin{lemma}
\label{lem:somevarsworkdfornoncollmassive}
By putting Corollary~\ref{cor:whennoncollpqabnoncollevisnonchiral} together arguments similar to those used in Corollary~\ref{cor:thesefunctionsarenice} and Lemma~\ref{lem:howthecomplexesbreakdown}, the three functions of Definition~\ref{def:noncollmassivefunctions} may be shown to be parity-odd event variables, and that at least one of them will be non-zero when evaluated on a any \textbf{chiral non-collision event} for which at least one of $p$ or $q$ is massive.
\end{lemma}

\begin{lemma}
\label{lem:somevarsworkdfornoncollmassless}
By putting Corollary~\ref{cor:whennoncollpqmasslessabnoncollevisnonchiral} together arguments similar to those used in Corollary~\ref{cor:thesefunctionsarenice} and Lemma~\ref{lem:howthecomplexesbreakdown}, the three functions of Definition~\ref{def:noncollmasslessfunctions} may be shown to be parity-odd event variables, and that at least one of them will be non-zero when evaluated on a any \textbf{chiral non-collision event} for which both $p$ and $q$ are massless.
\end{lemma}

\begin{definition}
The variables
\begin{alignat}{6}
V_{20} &= \Re[v[1,\mathbb{C},\text{non-coll}, m_p+m_q>0]], & \phantom{M} &
V_{21} &&= \Im[v[1,\mathbb{C},\text{non-coll}, m_p+m_q>0]],
\\
V_{22} &= \Re[v[2,\mathbb{C},\text{non-coll}, m_p+m_q>0]], & \phantom{M} &
V_{23} &&= \Im[v[2,\mathbb{C},\text{non-coll}, m_p+m_q>0]],
\\
V_{25} &= \Re[v[1,\mathbb{C},\text{non-coll}, m_p=m_q=0]], & \phantom{M} &
V_{26} &&= \Im[v[1,\mathbb{C},\text{non-coll}, m_p=m_q=0]],
\\
V_{27} &= \Re[v[2,\mathbb{C},\text{non-coll}, m_p=m_q=0]], & \phantom{M} &
V_{28} &&= \Im[v[2,\mathbb{C},\text{non-coll}, m_p=m_q=0]],
\end{alignat}
together with
\begin{align}
V_{24} = v[1,\mathbb{R},\text{non-coll}]
\end{align}
permit the definition of the set
\begin{align}
S^{(\text{non-coll})} = \{
V_{20},
V_{21},
V_{22},
V_{23},
V_{24},
V_{25},
V_{26},
V_{27},
V_{28}
\}.\label{eq:snoncolextranoncollvarset}
\end{align}
\end{definition}

\begin{corollary}
\label{cor:noncolsufficiency}
The nine variables in the set $S^{(\text{non-coll})}$ of \eqref{eq:snoncolextranoncollvarset} 
 assign at least one non-zero real parity to every \textbf{non-collision event} $e\in\pqabcAllEvent$.
 \begin{proof}
 Lemma~\ref{lem:somevarsworkdfornoncollmassive} shows that variables $V_{20}$, 
$V_{21}$, 
$V_{22}$, 
$V_{23}$ and 
$V_{24}$ achieve the stated aim for \textbf{chiral non-collision events} for which at least one of $p$ and $q$ is massive, while Lemma~\ref{lem:somevarsworkdfornoncollmassless} shows that variables $V_{24}$, 
$V_{25}$, 
$V_{26}$, 
$V_{27}$ and
$V_{28}$ achieve the stated aim for \textbf{chiral non-collision events} for which both $p$ and $q$ are massless.  There are no other types of  \textbf{chiral non-collision event} to be considered.
 \end{proof}
\end{corollary}

\subsection{\textbf{Events} $e\in\pqabcAllEvent$ %(both \textbf{collision events} and \textbf{non-collision events})
}
\label{sec:pqabcbothcollandnoncollvars}

\subsubsection{Sufficiency of the variables in $S_{28}$}

\begin{lemma}
If the set of parity-odd invariant event variables $S_{28}$ is defined as follows:
\begin{align}
\label{eq:s28ALL}
S_{28} &= 
S_{19}^{(1)} \cup
S_{19}^{(2)} \cup
S_{19}^{(3)} \cup
S_{19}^{(4)} \cup
S^{(\text{non-coll})}
\end{align}
in terms of subsets defined in  \eqref{eq:s19ONE},
\eqref{eq:s19TWO},
\eqref{eq:s19THR},
\eqref{eq:s19FOU} and
\eqref{eq:snoncolextranoncollvarset}, then Lemmas~\ref{lem:s19ONE}, \ref{lem:s19TWO} and \ref{lem:s19THR} together with Corollaries \ref{cor:s19FOU} and \ref{cor:noncolsufficiency} prove that the twenty-eight variables in $S_{28}$ will assign at least one non-zero parity to every \textbf{chiral} \textbf{event} $e\in\pqabcAllEvent$.
\end{lemma}

\subsubsection{Necessity of the variables in $S_{28}$ }

\textbf{Necessity} has not been demonstrated for each of the variables in  $S_{28}$ over \textbf{events}  $e\in\pqabcAllEvent$.  Equivalently, the set $S_{28}$ has not been shown to be \textbf{\independent}.

\subsection{Support materials and a reference implementation for the discussed event variables}
\label{sec:supportmaterials}

A python script \texttt{reference-implementation.py} has been placed on the arXiv in both (i) the ancillary support files location belonging to this paper \cite{Lester:2020jrg}, and (ii) independently with other support materials at \url{https://www.hep.phy.cam.ac.uk/~lester/parity/index.html}.  The implementation provided  can calculate the variables $X_1$ to $X_3$ and $V_1$ to $V_{28}$ defined in this paper, and can do so with arbitrary precision when the input four momenta have integer values for $E$, $p_x$, $p_y$ and $p_z$.  The implementation is primarily intended to allow users to validate other implementations which they may build for their own purposes, rather than to provide a library for direct use in analyses.  It is therefore optimised for readability rather than speed or ability to be vectorised or compiled.  If there is sufficient interest, the ancillary files may later be modified to contain other implementations (or links to implementations elsewhere) which are better optimised for other purposes.

The ancillary files of \cite{Lester:2020jrg} also contain a \textit{Mathematica} \cite{Mathematica} notebook which implements definitions of the event variables $V_1$ to $V_{28}$ as well as many test cases. It may be possible for some readers to add code to that notebook which could export implementations of the event variables in other formats of their choice.

%The variables $\hat X_1, \hat X_2$ and $\hat X_3$ are not implemented in any of the ancillary files at the first release, but could appear in a later releases of the ancillary files at the request of readers.

\section{Discussion}
\label{sec:discussion}

\hide{\color{purple} [\textbf{Note added November 2021:}  The original text of this `Discussion' section was  written shortly before the paper's first upload to the arXiv in August 2020. Other papers made public in Nov 2021 
\cite{Lester:2021kur,Lester:2021aks,Tombs:2021wae} have since provided concrete answers to some of the matters which were discussed herein.  To aid future readers, forward references to those works have been added here in purple notes like this one.]
}

\subsection{Are multiple geometric parities unavoidable?}

Arguably, the most surprising of our results is that for none of the non-trival \textbf{event classes} which we have considered have we managed to construct a continuous parity-odd variable without `unwanted' zeros  on \textbf{chiral} events.   This is equivalent to noting that for none of the non-trivial \textbf{event classes} have we managed to find a \textbf{minimal} set containing only a \textit{single} parity-odd event variable.
This observation is noteworthy in that it \textit{suggests} that for at least some \textbf{classes of event} (even rather simple ones!) one cannot talk about:
\begin{quote}
`the geometric \textit{parity} of an event'
\end{quote}
but must instead talk about:
\begin{quote}
`the geometric \textit{parities} of an event'.
\end{quote}
%It appears that some events may unavoidably have more than one non-ignorable notion of parity stemming from the geometry of the its constituent momenta.
In short: the geometry of events cannot always be boiled down to a single continuous parity.\footnote{Footnote added April 2022: Approximately one year after this paper was first put on the arXiv, but before this paper was accepted for publication, another paper of ours, \cite{Lester:2021kur},   provided a rigorous demonstration that there are many circumstances in which multiple real parities are an unavoidable consequence of demanding \textbf{sufficiency} and \textbf{continuity} in measures of chirality.  Indeed, the circumstances in which such multiple parities are seen to arise in \cite{Lester:2021kur} are so wide ranging that it is not overblown to interpret the result therein as indicating that \textit{almost any object of sufficiently complexity} (an event in a collider is just a simple example) will require a plurality of real and continuous parities to fully classify its chirality.} 

\vspace{2mm}
\noindent
A skeptic might counter the last statement with the following objection:
\begin{quote}
In no case have the authors proved that any of their sets of derived parity-odd event variables are \textbf{minimal}. Perhaps, despite their best efforts, the authors have simply failed to spot a single event variable $V$ whose sign always represents \textit{the} geometric parity of any event in the given class.
\end{quote}
This first concern is worth taking seriously, in part for the reasons already given in Section~\ref{sec:discOfMinimialityOfS19}.  Even though we do not ourselves believe that \textbf{events} in all \textbf{classes} \textit{can} be given a single all-encompassing and continuous parity, it may still be the case that we have simply not been creative enough in the \textbf{event classes} we have considered in this document.  That possibility cannot be excluded.
Emboldened, such a skeptic may continue with a second claim\footnote{In the first versions of this paper sent to the arXiv this claim contained a mistake which has now been corrected.} as follows:
\begin{quote}
Furthermore, I contend that the authors bring the proliferation of parities entirely upon themselves via their introduction of an unnecessary and self-imposed requirement to consider only \textbf{continuous} event variables. Without that burdensome requirement, for any event class, no matter how complex, there will always exist a parity `bit' which ascribes a +1 or a -1 to any genuinely \textbf{chiral} event and ascribes a zero to all others.
\end{quote}
The above claim is also true as we will now show.
\subsubsection{Multiple parities are a consequence of continuity!}

The proof of the above statement is a very simple. From any finite set of \textbf{continuous} parity-odd variables  which is \textbf{sufficient} for a given class of events one may generate a \textit{single} (albeit \textbf{discontinuous}) parity-odd variable which is also \textbf{sufficient} for that same class of events. One example of such a construction simply requires the first non-zero element of the set to be returned. For example: the three \textbf{continuous} variables  $X_1$, $X_2$ and $X_3$ which defined in  \eqref{eq:unhattedx1} to \eqref{eq:unhattedx3} and which are \textbf{sufficient} for events in $\pqabAllEvent$ can be replaced by a single (albeit \textbf{discontinuous}) parity-odd variable $X$ defined by:
\begin{align}
X=\begin{cases}
X_1 & \text{if $X_1\ne 0$,} \\
X_2 & \text{if $X_2\ne 0$,} \\
X_3 & \text{if $X_3\ne 0$} \\
0 & \text{otherwise}
\end{cases}
\end{align}
which is \textbf{sufficient} for the same class of events.   The jury is currently out on how important the \textbf{continuity} requirement is. More studies will have to show whether we are right or wrong to believe that the benefits of continuity in event variables will outweigh the complexity of dealing with them!

\subsubsection{Is this new?}

If multiple continuous event parities are indeed unavoidable for some \textbf{event classes}, then almost certainly there will already exist a theorem proving it (or an equivalent result) somewhere within the mathematical literature.   But even if such a result is well known to the mathematical  community, its existence currently remains unknown to both the authors of this note and the community of particle physicist colleagues whom they have canvassed.  As such, disseminating awareness of the issue and promoting further discussions in the area of tests of non-standard parity violation is presumably worthwhile.\footnote{Footnote added April 2022:
Approximately one year after this paper was submitted to the arXiv, but around a year before it was accepted for publication, Ref.~\cite{Lester:2021kur} found and documented the existence  of papers in the Chemical/Chemistry literature which may be regarded containing proofs that there exist molecules whose whose parities (if defined analogously to the parities used for particle collisions herein) would necessarily require more than one real variable to describe them.  Ref.~{Lester:2021kur} contains many more results on continuity in parities than we can report here. }

\subsection{Could $\pqabColEvent$ results have been derived from $\pqabcColEvent$ results?} 

In principle we could have derived parity-odd event variables to cover events in $\pqabColEvent$ (or respectively in $\pqabAllEvent$) by setting $c^\mu \rightarrow 0$ in variables suited to events in $\pqabcColEvent$ (or respectively in $\pqabcAllEvent$) since the zero Lorentz-vector is unique in being the only Lorentz-vector which is constant in all frames. However, a set of variables which is \textbf{minimal} or \textbf{\independent} for events in $\pqabcColEvent$ need not necessarily be \textbf{minimal} or \textbf{\independent} for events in $\pqabColEvent$, after the $c^\mu \rightarrow 0$ replacement. This, and the fact that the resulting event variables (and their derivation) is much simpler for  $\pqabColEvent$ than for events in $\pqabcColEvent$, motivated their separate treatment.

\subsection{Illustrations of potential utility of the proposed framework}

\subsubsection{$pp\rightarrow jj+X$ or $pp\rightarrow \gamma\gamma+X$ at the LHC}

An experimental collaboration such as ATLAS \cite{Aad:2008zzm} or CMS \cite{Chatrchyan:2008aa} seeking to find evidence of non-standard sources of parity violation \cite{Lester:2019bso} in two-jet or two-photon events could, in principle:
\begin{itemize}
\item
decide upon some kind of parity-blind selection\footnote{E.g.~this could be done by using parity-even variables from our sister papers \cite{Gripaios:2020ori} and \cite{Gripaios:2020hya} as explained in Corollary~\ref{cor:theneedforparityeveneventselectionsaswellaspoddvars}.} (more on this later),
\item
and could then pick two real constants $u_1\ne 0$ and $u_2\ne 0$, 
\item
and then in terms of $u_1$ and $u_2$ it could fix upon either (i) the set $S$ of three un-hatted variables $S=\{X_1,X_2,X_3\}$ of \eqref{eq:unhattedx1} to \eqref{eq:unhattedx3}, or (ii) the set $\hat S$ of three hatted variables $\hat S=\{\hat X_1,\hat X_2,\hat X_3\}$ 
 of \eqref{eq:hattedx1} to \eqref{eq:hattedx3}, 
 \item
 and then could look for a data-data asymmetry between the positive and negative halves of each of the $X_i$ (or $\hat X_i$) distributions resulting from application of the chosen selection.
\end{itemize}
It is likely that such a collaboration would quickly realise that their $X_3$ (or $\hat X_3$) values were always be identically zero and so might omit it from the above programme\footnote{This would occur if inputs are supplied for which $m_p=m_q$ (which should be the case in a proton-proton collider) and if $m_a$ and $m_b$ are set equal to zero (or are the same as each other -- which is highly likely for reconstructed jets and/or photons at the LHC).}
which need only use the two remaining variables. In such a situation  $X_2$ (and $\hat X_2$) turns out to have a pre-factor of $u_1 u_2$ (and no other $u_1$ or $u_2$ dependence). Consequently the choice $u_1=u_2=1$ is then (literally!) as good as any other and would probably be adopted.

[Aside: In the remainder of this section the discussion will be written as if the un-hatted set of variables had been chosen, however every statement may be interpreted as applying to both hatted and un-hatted choices.]

The collaboration would know that, despite confining itself to just the two variables in $\{X_1,X_2\}$, it would not be reducing the number of sources of non-standard parity violation to which it might be sensitive. On the contrary, one of the results of this paper is the proof that any source of non-standard parity violation discoverable in two jet or two photon events would \textit{have} to show up in at least one of these two variables given a clever enough event-selection.

Of course, this does not mean the task of discovery is simple.  There is no such thing as a free lunch!   Importantly: nothing in our analysis has  shown what event selection(s) will be required to enable those sources to be seen with the strongest significance above any backgrounds.  Furthermore, making a parity blind selection can itself be a tricky thing to do even if one has access to the parity-even variables defined in our sister papers \cite{Gripaios:2020ori} and \cite{Gripaios:2020hya}.  The first asymmetries seen would likely be evidence of parity-violating mis-calibrations of the detector. There is therefore (just as one would expect) no `magic button' that can make a non-standard parity violation stand out effortlessly.  Physicists will not be made redundant!  The only benefit of our calculation, in this specific example, is that it has reduced the scope of the problem that the physicist needs to tackle from one which is unimaginably large to one of finite scope.  

Prior to our work the number of variables one might have needed to test for non-standard parity violation in $pp\rightarrow \gamma\gamma$ events could have been unbounded in length.  Our paper shows that, on the contrary, the experimentalist need focus on creating innovative selections for (in principle) no more than two test distributions.

\subsubsection{$pp\rightarrow jjj+X$ or $pp\rightarrow \gamma\gamma\gamma+X$ at the LHC} 

Consider the wrangling and unhappy discussions contained within \cite{Lester:2019bso} concerning the choice of event variable to use in a three-jet test of non-standard parity violation in CMS open data.  Instead of settling on the arbitrary and poorly motivated choice made therein, the authors of that study could, using the results of our paper, now choose to concentrate only on the nine  variables contained within $S_{9}$ of \eqref{eq:shortersnineformasslessLHC3j} safe in the knowledge that (in principle) doing so would not limit their discovery potential.  [Recall that $S_9 \subset S_{19}$.]

\subsubsection{$pp\rightarrow jjjj+e^+e^-$ at the LHC} 

An analysis team wishing to look for sources of non-standard parity violation in \textbf{events} with  permutation (or other) symmetries we have not yet considered would have to  derive the parity-odd event variables appropriate to their \textbf{event class}.  They would need to 
follow a similar set of steps to those presented in this paper, hopefully benefiting from components of the framework presented here.  \footnote{Footnote added April 2022: Approximately one year after this paper was first put on the arXiv, but about one year before it was accepted for publication, an alternative method for testing for parity violation (a method which avoids asking for  provably \textbf{sufficient} or \textbf{necessary} sets of variables) was proposed in Ref.~\cite{Lester:2021aks}. The approach described in \cite{Lester:2021aks} projects outgoing-momenta onto a cylinder which can be unrolled into something resembling an image with a non-trivial topology. The image can then be turned into a parity by supplying it to any sufficiently general parity-odd function whose form and parameters may well need to be machine learned to increase its usefulness.  The functions considered are required to respect any desired symmetries (e.g.~rotational invariance) by construction.  Each class of particles  can be represented as a different `colour' layer in such an image. For example, a  $pp\rightarrow jjjj+e^+e^-$ event might be represented in a three-colour image with the jet momenta imprinted in one colour-channel, the $e^+$ momenta in a second channel, and the $e^-$ momenta in a third.  Hit pixels record energy deposits but do not come with labels saying `which jet' hit each pixel. Images are therefore naturally invariant with respect to changes to particle-orderings, and so the resulting parities computed from them have the same property. A second paper, \cite{Tombs:2021wae}, shows how such techniques can be applied to detectors which are imperfect and have efficiencies which vary from one location to another or which have parts of the detector which are inoperable or lack acceptance altogether.} 

\section{Conclusion}
\label{sec:conclusion}

The first two papers in this series \cite{Gripaios:2020ori,Gripaios:2020hya} set out general methods for constructing sets of continuous \textit{parity-even} event variables (with given permutation, Lorentz and other symmetries) having the property that those sets could be proven to contain all physically relevant information contained within the momenta of an event \textit{barring any information associated with the handedness of the event}.

This paper has shown how those those results could be extended via the creation of additional sets of continuous but this time \textit{parity-odd} event variables (with given permutation, Lorentz or other symmetries) which may be proven to have sensitivity to \textit{all potentially discoverable sources of non-standard parity violation lurking in the data.}

It is somewhat unusual (and beneficial!) that the natural way in which such variables might be used would be see physicists searching for asymmetries between the  positive and negative halves of their distributions.  This is in contrast to the more frequent approach of having to compare events recorded in data with those simulated by Monte-Carlo event generators.  Such first-order data-data sensitivity (rather than data-MC sensitivity) is particularly rare in LHC Beyond Standard Model searches.

To illustrate the process of constructing such variables, particular symmetry groups were chosen corresponding to simple processes, and sets of variables satisfying the required properties were generated.

In one example, it was noted that collider events having three important jets in the final state ($pp\rightarrow jjj+X$) could be characterised by an $S_2\times S_3$ symmetry where the $S_2$ is a permutation symmetry over the two colliding particles, and $S_3$ is a permutation symmetry over the final-state three jets.  For such events a set of \textbf{nineteen} event variables was constructed and was proved to be able to parameterize \textit{any} Lorentz-invariant source of non-standard parity violation discoverable from within them.  Furthermore, the variables in this set were shown to be \independent\ in the sense that none of them could be omitted without rendering one or more forms of non-standard parity violation invisible.

Smaller sets of variables were shown to be able to accomplish similar objectives in cases where, say, the masses of the colliding particles and/or final state jets could be neglected.  Indeed, in the case of the symmetric colliders like the LHC (colliding identical protons) it was noted that only \textbf{nine} of the above nineteen variables would be needed if the jets were recorded as massless particles (or if they were photons).

Appropriate sets of variables were also computed for processes in which  there were fewer particles (e.g.~only an $S_2\times S_2$ symmetry).  In the simplest case of all -- entirely massless  $pp\rightarrow jj+X$ --  just \textbf{two} variables were shown to be sufficient for sensitivity to all non-standard parity violating processes.

Larger sets of variables (\textbf{twenty-eight} in the case of $S_2\times S_3$) were shown to be able to accomplish a similar objective and a more general case in which no pairs of the momenta were assumed (necessarily) to represent colliding particles.  Such general cases are relevant to, for example, direct dark-matter searches in which permutation symmetry might apply to non-colliding final-state momenta (e.g.~$X\rightarrow \gamma \gamma + j j j + Y$).
 
The authors are not aware of any similar attempts to produce sets of parity-odd variables having the provable coverage/sensitivity properties of those reported here.
It is hoped, therefore, that the sorts of variables proposed here  will facilitate the growth of a new generation of exhaustive yet model independent searches for sources of non-standard parity violation -- searches which hitherto were only possible to approach in model-dependent ways.

\section*{Acknowledgements}
We thank Scott Melville, Tom Gillam and other members of the Cambridge Pheno Working Group for helpful advice and comments.   This work has been partially supported by STFC consolidated grants ST/P000681/1 and ST/S505316/1. WH is supported by the Cambridge Trust.

\appendix
 \appendixpage
 \addappheadtotoc

\section{Notation %not strongly specific to this paper
}

\label{appendix:notation}
This Appendix describes notation which we have used (mostly for core concepts like Lorentz vectors) and for which an earlier description would have  interrupted the flow of text.    
\subsection{Logical connectives (`\textbf{OR}' and `\textbf{AND}').}
The logical connectives \textbf{OR} and \textbf{AND} are denoted with the symbols $\lor$ and $\land$ respectively. For example the solution to the equation $$x^2=4$$ could be written as $$(x=2)\lor(x=-2)$$ while the solution to $$x^2+y^2=0\text{\qquad with\qquad}x,y\in\mathbb{R}$$ could be written as  $$(x=0) \land(y=0).$$
Formally $\land$ takes precedence over $\lor$, so that $a\lor b \land c$ means $a\lor (b \land c)$ rather than $(a\lor b) \land c$, however this article aims to use explicit parentheses wherever doubt might arise. 
\subsection{Lorentz vector notation}
\label{sec:lorentzvectornotation}
A Lorentz-vector $p^\mu$ having energy $E$ and three-momentum $(p_x,p_y,p_z)$ may be denoted component-wise either as
\begin{align}
p^\mu=\fourvec E {p_x} {p_y} {p_z}\qquad\text{or}\qquad p^\mu=\fourvecT E {p_x} {p_y} {p_z}
\end{align}
depending on the available space.
Sometimes it is more convenient to parameterize the four-vector by specifying its invariant mass $m$ together with its three momentum. In those cases the following alternative notation is used:
\begin{align}
p^\mu=\fourvecMassForm m {p_x} {p_y} {p_z}
\end{align}
and the energy is therefore taken to be give by $E=\sqrt{m^2+p_x^2 +p_y^2 +p_z^2}$.

\subsection{Epsilon notation}
\label{sec:epsilonNotation}

There are two alternating tensors we may come across which we will distinguish by using $\epsEuc$ with roman indices for one, and $\epsLor$ with Greek indices for the other.  Here is a summary:
\subsubsection{The Euclidean epsilon: $\epsEuc$}
\begin{itemize}
\item
\textbf{Symbol}: $\epsEuc$;
\item
\textbf{Indices}: always Roman;
\item
\textbf{Rank}: Can have any rank.  All these are valid: 
$\epsEuc_{ab}$,  
$\epsEuc_{abc}$,
$\epsEuc_{abcd}$,
$\epsEuc_{abcde}$,
  \ldots;
\item
\textbf{Index Range}: Indices will be integers in the set $\{1,2,\ldots,n\}$ where $n$ is the rank;
\item
\textbf{Raising/Lowering}: means nothing: $\epsEuc^{abcde}=\epsEuc_{a\phantom{b}cde}^{\phantom{a}b\phantom{cde}}=\epsEuc_{abcde}$;
\item
\textbf{Sign convention}: $\epsEuc_{12345\ldots }=+1$;
\item
\textbf{Complete antisymmetry}: $\epsEuc_{\ldots i \ldots j \ldots }=-\epsEuc_{\ldots j \ldots i \ldots }$;
\item
\textbf{Relationship to determinants}
\begin{align}\left|\left(\begin{array}{cccc}
a_1 & b_1 & c_1 & \cdots \\
a_2 & b_2 & c_2 & \cdots \\
a_3 & b_3 & c_3 & \cdots \\
\vdots & \vdots & \vdots & \ddots 
\end{array}\right)\right|
&=\epsEuc_{ijk\cdots}\  a_i b_j c_j \cdots
\end{align}
\end{itemize}

\subsubsection{The Lorentz epsilon: $\epsLor$}
\begin{itemize}
\item
\textbf{Symbol}: $\epsLor$;
\item
\textbf{Indices}: always Greek, (however see the `\textbf{Lorentz vector contraction shorthand}' bullet point further down);
\item
\textbf{Rank}: Always four, as in:
$\epsLor_{\mu\nu\sigma\tau}$;
\item
\textbf{Index Range}: Indices are taken to be integers in $\{0,1,2,3\}$ with $0$ representing the time component;
\item
\textbf{Raising/Lowering}: is meaningful and done with the metric, e.g. : $\epsLor^{\mu}_{\phantom{\mu}\nu\sigma\tau} =\epsLor_{\alpha\nu\sigma\tau} g^{\mu\alpha}$;
\item
\textbf{Complete antisymmetry}: $\epsLor_{\ldots \mu \ldots \nu \ldots }=-\epsLor_{\ldots \nu \ldots \mu \ldots }$;
\item
\textbf{Sign convention}: Because the metric has an odd number of negative components, $\epsLor^{1234}=-\epsLor_{1234}$ so $\epsLor^{1234}$ and $\epsLor_{1234}$ cannot both be $+1$ although one of them is defined to be so.  We never reveal which of them is plus one and which is minus one.  Instead we leave an explicit $\epsLor_{0123}$ in any expression (e.g.~\eqref{eq:somethingdependingoneps0123}) which depends on the convention.
\item
\textbf{Lorentz vector contraction shorthand}:
Contraction with Lorentz vectors having Roman labels is sometimes written in a short form which places the Lorentz vectors themselves where the indices would be.  For example i.e.~$\epsLor_{abc(p+q)}$ is sometimes used as a shorthand notation for $\epsLor_{\mu\nu\sigma\tau} a^\mu b^\nu c^\sigma (p+q)^\tau$.  Where such notation might lead to confusion, the same quantity would instead be represented by the expression $[a,b,c,p+q]$ which makes use of square brackets which we define to represent contractions between epsilons and Lorentz vectors as follows:
\begin{align}
[a,b,c,d]\equiv
\epsLor_{\mu\nu\sigma\tau} a^\mu b^\nu c^\sigma d^\tau
\equiv \epsLor_{abcd}.\label{eq:epscontractionnotation}
\end{align}
\end{itemize}

\subsubsection{Summary of epsilon-related notation}
\begin{align}
\epsEuc_{ijklmnopq} &\qquad \text{could be valid notation}
\\
\epsLor_{\alpha\beta\gamma\delta\zeta\eta\theta\iota\kappa} &\qquad \text{cannot be valid notation,}
\\
\epsLor_{\mu\nu\sigma\tau} &= - \epsLor^{\mu\nu\sigma\tau},
\\
[a,b,c,d]
&\equiv 
\epsLor_{\mu\nu\sigma\tau} a^\mu b^\nu c^\sigma d^\tau 
\equiv 
\epsLor_{abcd} 
, \label{eq:lorentzcontractionsepsilonhorthand}
\\
\gamma^5
&=
i
\gamma^0
\gamma^1
\gamma^2
\gamma^3
=
\gamma_5 ,
\\
(g^{00} ,g^{11} ,g^{22} ,g^{33})
&=
(1,-1,-1,-1).
\hide{,
\\
\gamma^\mu
\gamma^\nu
\gamma^\sigma
&=
g^{\mu\nu} \gamma^\sigma
-
g^{\mu\sigma} \gamma^\nu
+
g^{\sigma\nu} \gamma^\mu
+
i \gamma^5 \epsLor^{\mu\nu\sigma\tau} \gamma_\tau \epsLor_{0123}
\\
\Tr\left[
\gamma^\mu \gamma^\nu \gamma^\sigma \gamma^\tau
\right]
&=
4 g^{\mu\nu} g^{\sigma\tau}
-
4 g^{\mu\sigma} g^{\nu\tau}
+
4 g^{\mu\tau} g^{\nu\sigma},
\\
\Tr \left[
\gamma^\mu
\gamma^\nu
\gamma^\sigma
\gamma^\tau
\gamma^5
\right]
&=
4 i \epsLor^{\mu\nu\sigma\tau} \epsLor_{0123}
\\
\epsilon^{\alpha\beta\gamma\delta}
\epsilon_{\mu\nu\sigma\tau}
&= 
\sum_{\text{perms $\pi$ of $(\mu,\nu,\sigma,\tau)$}}
\sgn(\pi) 
\delta^\alpha_\mu
\delta^\beta_\nu
\delta^\gamma_\sigma
\delta^\delta_\tau \text{\comCGL{Fix signs of euc $\delta$ vs lor $\delta$}}
.}
\end{align}
Note that we do not choose which of
$\epsLor^{0123}$
or
$\epsLor_{0123}$
is $+1$ but instead simply leave $\epsLor_{0123}$ in any expression which depends on the convention.

\hide{
\subsection{Other notation}
Use the notation $$\{a,b,c\}^\mu = \epsLor^{\mu\nu\sigma\tau} a_\nu b_\sigma c_\tau.$$
For the following pseudoscalar use the notation \begin{align}[a,b,c,d] = \epsLor^{\mu\nu\sigma\tau} a_\mu b_\nu c_\sigma d_\tau
\end{align} already defined in \eqref{eq:lorentzcontractionsepsilonhorthand}.
Both $\{a,b,c\}^\mu$ and $[a,b,c,d]$ change sign on interchange of any two arguments.  Since the former has three arguments it therefore has cyclic symmetry:
$$\{a,b,c\}^\mu = \{b,c,a\}^\mu$$ while the latter has four arguments and so is anti-cyclic:
$$[a,b,c,d]=-[b,c,d,a].$$
Note therefore that:
\begin{align}
a.\{b,c,d\}\phantom{.d} &= + [a,b,c,d]\qquad\text{but} \\
\{a,b,c\}.d &= - [a,b,c,d].
\end{align}
}

\subsection{Gram Determinants}
\label{app:gramnotation}

\begin{definition}
\label{def:gram}
Following Byckling and Kajantie \cite{nla.cat-vn1953978} we define the {\it Gram determinant} of the vectors $p_1,\ldots,p_n;q_1,\ldots,p_n$ to be the determinant of the matrix $M$ containing the scalar products of the vectors, $(M)_{ij} = p_i.q_j$. Specifically we define the {\it Gram determinant} $G$ of those vectors as follows:
\begin{align}
G\left(\begin{array}{cc}p_1,\ldots,p_n\\q_1,\ldots,q_n\end{array}\right) 
&=
\left| 
\begin{array}{cccc}
p_1.q_1 & p_1.q_2 & \cdots & p_1.q_n \\
\vdots & \vdots  &  & \vdots \\
p_n.q_1 & p_n.q_2 & \cdots & p_n.q_n 
\end{array}
\right|. \label{eq:gram}
\end{align}
\end{definition}

\begin{definition}
\label{def:symgram}
Also following Byckling and Kajantie \cite{nla.cat-vn1953978} we denote the {\it symmetric Gram determinant}, $\Delta_n$, of the vectors $p_1,\ldots,p_n$ as follows:
\begin{align}
\Delta_n(p_1,\ldots,p_n) = G\left(
\begin{array}{cc}p_1,\ldots,p_n\\p_1,\ldots,p_n\end{array}\right).\label{eq:symmgram}
\end{align}
\end{definition}

\section{Mathematical Identities, \textit{etc.}
}

\label{app:mathematicalidentitissdf}

\begin{lemma}
\label{lem:momdifferencesidentity}
The following is an identity between Gram determinants:
\begin{align}
\SYMGRAMTWO a \Lambda - \SYMGRAMTWO b \Lambda  \equiv \GRAMTWO {a-b} \Lambda {a+b} \Lambda\end{align}
which, in the notation of \eqref{eq:defofgramshorthand}, reads as follows: \begin{align}
g^a_a - g^b_b \equiv g^{a-b}_{a+b}.\label{eq:gaagbbdifferencelemma}
\end{align}
\begin{proof}
\begin{align}
\GRAMTWO {a-b} \Lambda {a+b} \Lambda
&=\left|\begin{array}{cc}
     (a-b).(a+b) & (a-b).\Lambda  \\
     \Lambda.(a+b)  & \Lambda.\Lambda
\end{array}
\right|
\\
&= (a-b)\cdot(a+b)\Lambda^2 - (a-b).\Lambda (a+b).\Lambda
\\
&= a^2 \Lambda^2-b^2\Lambda^2 - (a.\Lambda)^2+(b.\Lambda)^2 
\\
&= \SYMGRAMTWO a \Lambda - \SYMGRAMTWO b \Lambda.
\end{align}
\end{proof}
\end{lemma}

\begin{lemma}
\label{lem:tetanotherlemmaaboutgramdetsfughfjdn}
The following is an identity among Gram determinants:
\begin{gather}
\left[
\left(
\SYMGRAMTWO a s =\SYMGRAMTWO b s 
\right)
\land
\left(
\GRAMTWO a s c s = \GRAMTWO b s c s
\right)
\right]\nonumber
\\
 \iff 
\\
\left[
\left(
\SYMGRAMTWO a s =\SYMGRAMTWO b s 
\right)
\land
\left(
\GRAMTWO {a-b} s {a+b+c} s = 0
\right)
\right]\nonumber.
\end{gather}
[Aside: when $s\in\mathbb{M}$ and $a$ and $b$ are both in $\mathbb{V}$ then this lemma has the interpretation that, in the rest frame of $s$, the following is true: $[(|\vec a|=|\vec b|)\land((\vec a-\vec b)\cdot \vec c=0) ]\iff[(|\vec a|=|\vec b|)\land((\vec a-\vec b)\cdot (\vec a+\vec b+\vec c)=0)]$.]
\begin{proof}
When $\SYMGRAMTWO a s =\SYMGRAMTWO b s $ we know that \begin{align}
\GRAMTWO a s a s =\GRAMTWO b s b s.\end{align}
In this case 
\begin{align*}
\left[
\GRAMTWO a s c s =\GRAMTWO b s c s
\right]
& \iff\left[
\GRAMTWO a s c s + \GRAMTWO a s {a+b} s =\GRAMTWO b s c s+ \GRAMTWO a s {a+b} s
\right]\\
&\iff\left[
\GRAMTWO a s {a+b+c} s  =\GRAMTWO b s c s+ \GRAMTWO a s {a}s + \GRAMTWO a s {b} s
\right]\\
&\iff\left[
\GRAMTWO a s {a+b+c} s  =\GRAMTWO b s c s+ \GRAMTWO b s {b}s + \GRAMTWO b s {a} s
\right]\\
& \iff\left[
\GRAMTWO a s {a+b+c} s  =\GRAMTWO b s {a+b+c} s
\right]\\
& \iff\left[
\GRAMTWO {a-b} s {a+b+c} s  =0
\right].
\end{align*}
\end{proof}
\end{lemma}

\begin{lemma}
\label{lem:notusedbutused}
%{\color{red}[Move this result earlier?]}
$((x\ne0) \lor (y\ne0) \lor (z\ne0))\iff((x+y+z\ne0) \lor (xy+yz+zx\ne0) \lor (xyz\ne0)).$
\begin{proof}
Negating both sides, may be seen that it is sufficient to prove that $((x=0) \land (y=0) \land (z=0))\iff((x+y+z=0) \land (xy+yz+zx=0) \land (xyz=0)).$  The `$\implies$' direction is evidently trivial. We consider therefore only the `$\impliedby$' direction.  We begin therefore with $((x+y+z=0) \land (xy+yz+zx=0) \land (xyz=0))$.
The $xyz=0$ part tells us (w.l.o.g.) that $z=0$.  Putting this in the remaining two equations gives $(x+y=0)\land(xy=0)$.  This in turn shows that (w.l.o.g.) $y=0$, and one more iteration shows us that  $x=0$.  Clearly a similar argument works for any complete set of \textbf{elementary symmetric polynomials}, however big the order. 
\end{proof}
\end{lemma}
\begin{remark}
Lemma~\ref{lem:notusedbutused} turns the statement ``$(x\ne0) \lor (y\ne0) \lor (z\ne0)$'' into another one involving the properties of two polynomials of odd order ($x+y+z$ and $xyz$) and one polynomial of even order ($xy+yz+zx$). The next lemma (Lemma~\ref{lem:tripledecomp}) is a variant which turns the same statement into another which uses \textbf{only} polynomials of odd order.  The oddness of the polynomials will make Lemma~\ref{lem:tripledecomp} much more useful to us than Lemma~\ref{lem:notusedbutused}.
\end{remark}
\begin{lemma}
%{\color{red}[Move this result earlier?]}
\label{lem:tripledecomp}
$((x\ne0) \lor (y\ne0) \lor (z\ne0))\iff((x+y+z\ne0) \lor (x-y)(y-z)(z-x)\ne0) \lor (xyz\ne0)).$
\begin{proof}
Negating both sides, may be seen that it is sufficient to prove that $((x=0) \land (y=0) \land (z=0))\iff((x+y+z=0) \land (x-y)(y-z)(z-x)=0) \land (xyz=0)).$ The `$\implies$' direction is evidently trivial. We consider therefore only the `$\impliedby$' direction.  The $xyz=0$ part tells us (w.l.o.g.) that $z=0$.  Putting this in the remaining two equations gives $(x+y=0)\land(xy(x-y)=0)$.  The first of those tells us that $y=-x$. Substituting this into the $xy(x-y)=0$ gives $-2x^3=0$ which shows us that $x=0$ and thence that $x=y=z=0$.
\end{proof}
\end{lemma}

\begin{lemma}
\label{lem:centralcolumnlem}\begin{gather}
 (x+y+z=0)\land(xyz=0)\land((x-y)(y-z)(z-x)\ne0) \label{expr:tophalf}\\
 \iff \nonumber \\
 ((x=-y\ne0) \land (z=0))\lor((y=-z\ne0) \land (x=0))\lor((z=-x\ne0 )\land (y=0)).
 ) \label{expr:bothalf} 
 \end{gather}
\begin{proof}
%Left as exercise for reader.
For the `$(\ref{expr:tophalf})\impliedby(\ref{expr:bothalf})$' direction, and given the cyclic $(x,y,z)$-symmetry, suppose w.l.o.g.~that  $((x=-y\ne0) \land (z=0))$.  This supposition would make (\ref{expr:tophalf}) read $(0=0)\land(0=0)\land((2x)(-x)(-x)\ne0)$ which is evidently  true under the supposition.
For the `$(\ref{expr:tophalf})\implies(\ref{expr:bothalf})$' direction, we can see that the middle requirement of (\ref{expr:tophalf}) forces at least one of $x$, $y$ and $z$ to be zero, which therefore makes (\ref{expr:bothalf}) trivially true.
\end{proof}
\end{lemma}

\begin{lemma}
(This is a weaker version of Lemma~\ref{lem:centralcolumnlem}.)
\label{lem:weakercentralcolumnlem}\begin{gather}
 (xyz=0)\land((x-y)(y-z)(z-x)\ne0)\label{expr:weakertophalf}\\
 \iff \nonumber \\
 ((0\ne x\ne y\ne0) \land (z=0))\lor
 ((0\ne y\ne z\ne0) \land (x=0))\lor
 ((0\ne z\ne x\ne0 )\land (y=0))
 ). \label{expr:weakerbothalf} 
 \end{gather}
\begin{proof}
For the `$(\ref{expr:weakertophalf})\implies(\ref{expr:weakerbothalf})$' direction, we can see that the middle requirement of (\ref{expr:weakertophalf}) forces at least one of $x$, $y$ and $z$ to be zero.  Suppose, without loss of generality, that it forces $x=0$.  This leaves the other condition of (\ref{expr:weakertophalf}) reading $-y(y-z)z \ne 0$, which implies that neither $y$ nor $z$ can be zero, and that $y\ne z$. This makes (\ref{expr:weakerbothalf}) trivially true.
For the `$(\ref{expr:weakertophalf})\impliedby(\ref{expr:weakerbothalf})$' direction,  and given the cyclic $(x,y,z)$-symmetry, suppose w.l.o.g.~that $((0\ne y\ne z\ne0) \land (x=0))$.  This supposition would make (\ref{expr:weakertophalf}) read $(0=0)\land((-y)(y-z)(z)\ne0)$ which is evidently  true under the supposition.
\end{proof}
\end{lemma}

\begin{lemma}
If $R$ is an orthogonal matrix for which $R\vec z=\vec z$,  if $n$ is a positive integer, and if $\vec a=-R^n \vec a$ then $\vec a\cdot \vec z=0$.  \label{lem:whenonplaneperttorot}[Remark: We will use this lemma mostly when $R$ is a rotation matrix with axis $\vec z$.]
\begin{proof}
As $R$ is an orthogonal matrix ($R^{-1} = R^T$) we may deduce from the supplied conditions  that $\vec z=R^T \vec z$ and so $\vec z=(R^T)^n \vec z = (R^n)^T \vec z$.   Therefore $
\vec a\cdot \vec z
=
(-R^n \vec a)\cdot\vec z 
=
-(R^n \vec a)\cdot\vec z 
=
-(R^n \vec a)^T \vec z 
=
-(\vec a^T (R^n)^T)\vec z
=
-\vec a^T ((R^n)^T \vec z)
=
-\vec a^T \vec z
=
-\vec a\cdot \vec z.
$
\end{proof}
\end{lemma}

\begin{lemma}
\label{lem:lemaboutproductsofthreecomplexnumbers}
For $(x_a,y_a),(x_b,y_b),(x_c,y_c)\in\mathbb{R}^2$ the following is true:
 \begin{gather}
 \left[
     (x_a,y_a)\ne0 \land (x_b,y_b)\ne0 \land (x_c,y_c)\ne0
     \right] \nonumber
     \\
     \iff
     \\
\hide{     \left[(
  y_a y_b x_c 
+ y_a x_b y_c 
+ x_a y_b y_c 
- x_a x_b x_c 
\ne 0 )\land(
   y_a x_b x_c
+  x_a y_b x_c 
+  x_a x_b y_c 
-  y_a y_b y_c \ne 0)
     \right]}
      \left[
     \begin{array}{c}
     \toparotor{8cm} \\
  y_a y_b x_c 
+ y_a x_b y_c 
+ x_a y_b y_c
- x_a x_b x_c 
\ne 0
\\
\separotor \lor {3.5cm}
\\
   y_a x_b x_c
+  x_a y_b x_c 
+  x_a x_b y_c
-
 y_a y_b y_c 
 \ne 0
\\
\toparotor{8cm}
\end{array}
     \right]
     .
 \end{gather}\label{eq:endofsecondsandwwich}
 \begin{proof}
Any $n$ complex numbers are all non-zero if and only if their product is non-zero.  Therefore three two-vectors $(x_a,y_a)$, $(x_b,y_b)$ and $(x_c,y_c)$ in $\mathbb{R}^2$ are all non-zero vectors if and only if $\Re(p)\ne0$ or $\Im(p)\ne0$ where \begin{align}
p&=(x_a+i y_a)(x_b + i y_b)(x_c+i y_c)
\\
&=
(
  x_a x_b x_c 
- x_c y_a y_b 
- x_b y_a y_c 
- x_a y_b y_c 
)
+i(
  x_b x_c y_a
+  x_a x_c y_b 
+  x_a x_b y_c 
-  y_a y_b y_c
)
 .\label{eq:reimpartsoftriplecomplexprod}\end{align} 
\end{proof}
\begin{remark}
Note that although this lemma has been stated with \textit{three} anded constraints, it is evident that any number of anded constraints could have been tied together in a similar way, and that independently of the length of that product, only two real constraints would have appeared in the equivalent of \eqref{eq:endofsecondsandwwich}.  Note that one could have avoided complex numbers altogether if one had chosen to prove the following less useful (though still valid) result:
\begin{gather}
 \left[
     (x_a,y_a)\ne0 \land (x_b,y_b)\ne0 \land (x_c,y_c)\ne0
     \right] \nonumber
     \\
     \iff
     \\
     \left[
    ( x_a x_b x_c \ne 0) \lor
    ( x_a x_b y_c \ne 0) \lor
    ( x_a y_b x_c \ne 0) \lor
    ( x_a y_b y_c \ne 0) \lor \qquad
    \right.
    \\
    \left.
  \qquad  ( y_a x_b x_c \ne 0) \lor
    ( y_a x_b y_c \ne 0) \lor
    ( y_a y_b x_c \ne 0) \lor
    ( y_a y_b y_c \ne 0) 
     \right].
 \end{gather}
The reason this latter result is often less useful for the construction of event variables, however, is that the number of real variables obtained (eight in this example)  grows as two to the power of the number of anded constraints!  It is for this reason (a desire to reduce the final number  of event variables) that complex products and complex numbers have been used in the second bullet point in the proofs of Lemmas~\ref{lem:threevarsforpqabevents} and \ref{lem:howthecomplexesbreakdown} or in Definition~\ref{def:oursixcomplexfs} and in similar places.
\end{remark}
\end{lemma}
\begin{lemma}
%Suppose $\lambda$ is a non-zero real number which scales the $x$-values of Lemma~\ref{lem:lemaboutproductsofthreecomplexnumbers}, while $\mu$ is a non-zero real number which scales the $y$-values. Then we have the following:

If $(x_a,y_a),(x_b,y_b),(x_c,y_c)\in\mathbb{R}^2$, and if $\sgn(x)$ evaluates to: $+1$ if $x>0$, to $0$ if $x=0$, and to $-1$ if $x<0$, then the following is true:
 \begin{gather*}
 \left[
     (x_a,y_a)\ne0 \land (x_b,y_b)\ne0 \land (x_c,y_c)\ne0
     \right] \nonumber
     \\
     \impliedby
     \\
     \left[
     \begin{array}{c}
     \toparotor{7cm} \\
     \sgn(
  y_a y_b x_c 
+ y_a x_b y_c 
+ x_a y_b y_c)
\ne \sgn( x_a x_b x_c) 
\\
\separotor \lor {3cm}
\\
  \sgn( y_a x_b x_c
+  x_a y_b x_c 
+  x_a x_b y_c) 
\ne \sgn(  y_a y_b y_c)
\\
\toparotor{7cm}
\end{array}
     \right]\numberthis\label{eq:destinationapp1}
     .
 \end{gather*}
 \begin{proof}
 Let $\lambda,\nu\in\mathbb{R}$ be any two  non-zero real numbers, $\lambda\ne0\ne\mu$, and let $R$ be the square of their ratio, $R=(\lambda/\mu)^2>0.$  Then:
 \begin{gather}
 \left[
     (x_a,y_a)\ne0 \land (x_b,y_b)\ne0 \land (x_c,y_c)\ne0
     \right] \nonumber
     \\
     \iff \nonumber
     \\
     \mirrorcondleft{
     \left[
     (\lambda x_a,\mu y_a)\ne0 \land (\lambda x_b,\mu y_b)\ne0 \land (\lambda x_c,\mu y_c)\ne0
     \right]}{\quad}{\forall \lambda\ne0,\mu\ne 0,} \nonumber
     \\
     \mirrorcondright{\iff}{\qquad\qquad\qquad\qquad}{\text{(by Lemma~\ref{lem:lemaboutproductsofthreecomplexnumbers})}}
    %\phantom{\text{(by Lemma~\ref{lem:lemaboutproductsofthreecomplexnumbers})}} \qquad\iff\qquad\text{(by Lemma~\ref{lem:lemaboutproductsofthreecomplexnumbers})}
    \nonumber \\
    \mirrorcondleft{
      \left[
     \begin{array}{c}
     \toparotor{7cm} \\
  \lambda \mu^2 \left( 
  y_a y_b x_c 
+ y_a x_b y_c 
+ x_a y_b y_c \right)
\ne \lambda^3 x_a x_b x_c 
\\
\separotor \lor {3cm}
\\
  \lambda^2\mu( y_a x_b x_c
+  x_a y_b x_c 
+  x_a x_b y_c)
\ne \mu^3
 y_a y_b y_c 
\\
\toparotor{7cm}
\end{array}
     \right] }{\quad}{\forall \lambda\ne 0,\mu\ne0,}\nonumber
\\
\iff \nonumber
\\
     \mirrorcondleft{ \left[
     \begin{array}{c}
     \toparotor{7cm} \\
   \left( 
  y_a y_b x_c 
+ y_a x_b y_c 
+ x_a y_b y_c \right)
\ne R x_a x_b x_c 
\\
\separotor \lor {3cm}
\\
  R( y_a x_b x_c
+  x_a y_b x_c 
+  x_a x_b y_c)
\ne 
 y_a y_b y_c 
\\
\toparotor{7cm}
\end{array}
     \right]}{\quad}{\forall R>0, }.\label{eq:sourceapp1}
\end{gather}
It remains, therefore, to show that \eqref{eq:sourceapp1} is implied by \eqref{eq:destinationapp1}.  We will do so by proving the contrapositive, namely that `not \eqref{eq:sourceapp1}' implies `not \eqref{eq:destinationapp1}'
. Using the abbreviations:
\begin{align*}
y_a y_b x_c 
+ y_a x_b y_c 
+ x_a y_b y_c &\rightarrow A,  \\
 x_a x_b x_c  &\rightarrow B,  \\
  y_a x_b x_c
+  x_a y_b x_c 
+  x_a x_b y_c &\rightarrow C,\quad\text{and}\\
 y_a y_b y_c & \rightarrow D,
 \end{align*} we may write: 
\begin{alignat}{1}
 \text{(not \eqref{eq:sourceapp1})}
 \iff& \left\{
     \exists R>0 \text{\ s.t.\ } \left[
     (A= R B )\land( R C = D)
         \right]\right\}\nonumber %\label{eq:notsourceapp1} 
         \\
 \implies &
 \left[(\sgn A=\sgn B) \land (\sgn C = \sgn D)\right]\nonumber
 \\
\iff&\text{(not \eqref{eq:destinationapp1})}.\nonumber
\end{alignat}
\end{proof}
\end{lemma}

%\section{Proof of Theorem~\ref{thm:nonchiralpqabcEvents}
%On the chirality of  $(pq)\rightarrow (abc)+X$ events at colliders
%}
%\input{long_proof.tex}

%\input{unused_results.tex}

%Here \cite{Barter:2018xbc} is a citation.
% FOR BIBTEX
%\bibliographystyle{unsrt}
%\bibliography{main}

%% FOR BIBLATEX
\printbibliography

\end{document}